\documentclass[11pt]{scrartcl}
\usepackage[english]{babel}
\usepackage[T1]{fontenc}
\usepackage[utf8]{inputenc}
\usepackage{color,amssymb,amsfonts,mathrsfs,graphicx,setspace,amsmath}
\usepackage{amsthm}
\usepackage{caption}
\usepackage[left=3cm,right=3cm,top=2.5cm,bottom=3cm]{geometry}
\usepackage{mathtools}
\usepackage{thmtools}
\usepackage[mathscr]{eucal}
\usepackage{fullpage}
\usepackage{xcolor}
\usepackage{tikz}
\usetikzlibrary{shapes.geometric,hobby,calc} \usetikzlibrary{decorations.pathmorphing}
\usetikzlibrary{decorations.text}
\usetikzlibrary{shapes.misc}
\usetikzlibrary{calc}
\usetikzlibrary{arrows.meta}
\usetikzlibrary{decorations,shapes}
\usepackage{hyperref}
\usepackage{nameref,cleveref}
\usepackage[square, numbers, comma, sort&compress]{natbib}  \usepackage{macros}
\usepackage{tikzmacros}
\usepackage{enumitem}
\usepackage{thm-restate}
\setlist{nosep}

\usepackage{dsfont}
\usepackage{soul}
\usepackage{xspace}
\usepackage[prependcaption,colorinlistoftodos]{todonotes}
\usepackage{titling}
\usepackage{marginnote}
\usepackage{makeidx}
\usepackage{subcaption}

\newcommand{\Symbol}[1]{}
\newcommand{\Index}[1]{\index{#1}}

\hypersetup{
	colorlinks=true,
	linkcolor=myBlue,
	citecolor=myBlue,
	urlcolor=myLightBlue,
	bookmarksopen=true,
	bookmarksnumbered,
	bookmarksopenlevel=2,
	bookmarksdepth=3,
	hypertexnames=false }

\usetikzlibrary[positioning]
\usetikzlibrary{patterns}
\usetikzlibrary{shapes}
\usetikzlibrary{decorations.markings}
\usetikzlibrary{arrows}

\graphicspath{{./}{../Figures/}}

\title{Cycles of Well-Linked Sets I: an Elementary Bound for Directed Cycle Packing\thanks{The results in this manuscript were also presented in Milani's PhD thesis \cite{Milani2024}. A preliminary version of this paper was published at \emph{FOCS 2024}\cite{hkmm24cows}.}}

\predate{}
\date{}
\postdate{}

\preauthor{}

\DeclareRobustCommand{\authorthing}{
	\begin{center}
	    Meike Hatzel\thanks{Partially supported by the Institute for Basic Science (IBS-R029-C1).}
            \smallskip\\[0em] Technical University Darmstadt, Germany\\[0em]
	    \href{mailto:research@meikehatzel.com}{research@meikehatzel.com}\\[0em]
	    \bigskip
	    Stephan Kreutzer\smallskip\\[0em]
        Technische Universität Berlin, Germany\\[0em]
        \href{mailto:stephan.kreutzer@tu-berlin.de}{stephan.kreutzer@tu-berlin.de}
	    \\[0em]
	    \bigskip
			Marcelo Garlet Milani\thanks{Supported by Japan Science and Technology Agency (JST) as part of Adopting Sustainable Partnerships for Innovative Research Ecosystem (ASPIRE), Grant Number JPMJAP2302.}\smallskip
			\\[0em]
 			National Institute of Informatics, Tokyo, Japan
			\\[0em]
			\href{mailto:research@mgarletmilani.com}{research@mgarletmilani.com}
	    \\[0em]
	    \bigskip
	    Irene Muzi\smallskip\\[0em]
            Universität Hamburg, Germany\\[0em]
            \href{mailto:irene.muzi@gmail.com}{irene.muzi@gmail.com}
          \end{center}}
\author{\authorthing}
\postauthor{}

\newenvironment{cenv}{\begin{list}{}{            \setlength{\labelwidth}{1.5em}      \setlength{\leftmargin}{\labelwidth}      \addtolength{\leftmargin}{\labelsep}      \setlength{\listparindent}{0em}      \setlength{\topsep}{10pt}      \setlength{\itemsep}{5pt}      \setlength{\parsep}{0pt}    }
  }{
  \end{list}
}

\newcounter{claimcounter}

\newcommand{\DrawPicture}[1]{}
\newcommand{\DeferTask}[1]{}
\renewcommand{\DeferTask}[1]{}
\renewcommand{\DrawPicture}[1]{}
\renewcommand*{\Task}[2][mediumpriority]{}
\renewcommand*{\TaskPerson}[3][]{}

\begin{document}
\setlength{\abovedisplayskip}{2pt}
\setlength{\belowdisplayskip}{2pt}
\setlength{\abovedisplayshortskip}{0pt}
\setlength{\belowdisplayshortskip}{0pt}
\maketitle

\begin{abstract}

In 1996, Reed, Robertson, Seymour and Thomas [Combinatorica 1996] proved Younger's Conjecture, which states that, for all directed graphs $D$, there exists a function $f$ such that, if $D$ does not contain $k$ disjoint cycles, then $D$ contains a feedback vertex set, i.e.~a subset of vertices whose deletion renders the graph acyclic, of size bounded by $f(k)$.
However, the function obtained by Reed, Robertson, Seymour and Thomas in their paper is enormous and, in fact, not even elementary.
We prove the first elementary upper bound for the function \(f\) above, showing it is upper-bounded by a power tower of height 8.

Our proof is inspired by the breakthrough result of Chekuri and Chuzhoy [J.~ACM 2016], who proved a polynomial bound for the Excluded Grid Theorem for undirected graphs.
We translate a key concept of their proof to directed graphs by introducing \emph{paths of well-linked sets (PWS)}, and show that any digraph of large directed \treewidth contains a large PWS, which in turn contains a large fence.

We believe that the theoretical tools developed in this work may find applications beyond the results above, in a similar way as the path-of-sets-system framework due to Chekuri and Chuzhoy [J.~ACM 2016] did for undirected graphs (see, for example, Hatzel, Komosa, Pilipczuk and Sorge [Discret.~Math.~Theor.~Comput.~Sci.~2022], Chekuri and Chuzhoy [SODA 2015] and Chuzhoy and Nimavat [arXiv 2019]).
Indeed, in a follow-up paper, we apply this framework to improve the bounds of the Directed Grid Theorem.

		\end{abstract}

\clearpage

\newpage

\tableofcontents

\newpage

\section{Introduction}

	The Erd\H{o}s-P\'{o}sa Theorem~\cite{posa1965independent} is one of the first results in structural graph theory considering the relationship between \say{packing} multiple copies of some structure in a graph and the size of a set \say{covering} all such copies.
	More precisely, they proved that every undirected graph contains $k$ pairwise vertex-disjoint cycles or it contains a feedback vertex set, i.e.~a subset of vertices whose deletion renders the graph acyclic, of size bounded by $\Oh(k\log{k})$.
	They also prove that the function they obtained is tight up to constant factors.

	The \emph{packing-covering duality} above inspired a long line of research on similar results for several combinatorial objects.
	We refer the reader to \cite{RaymondT17} for a survey on the topic and to~\cite{AhnGHK25,MoussetNSW17,BatenburgHJR19} for more recent results.
	For a (to the best of our knowledge) complete and up-to-date list, see \cite{raymondEPlist}.

	The validity of an analogous property for directed cycles was conjectured by Younger~\cite{younger1973} in 1973 and proven by Reed, Robertson, Seymour and Thomas~\cite{reed1996packing} in 1996.
    However, the function obtained by~\cite{reed1996packing} is very far from the bound found by Erd\H{o}s and P\'{o}sa for undirected cycles and is, in fact, not even elementary. 
	
	The main result of this paper is a new proof of the result above due to~\cite{reed1996packing}, but providing instead an elementary bound for the dependency between the maximum number of disjoint directed cycles and the size of a minimum feedback vertex set.
	More precisely, we prove the following, where \(\bound{statement:elementary-younger}{f}{k}\) is a \emph{power tower} of height 8 which we define later.

	\begin{restatable}{theorem}{younger}
		\label{statement:elementary-younger}
										Let \(D\) be a digraph.
		Then \(D\) has \(k\) pairwise vertex-disjoint cycles or there is some \(X \subseteq \V{D}\) of size at most \(\bound{statement:elementary-younger}{f}{k}\) such that \(D - X\) is acyclic.
	\end{restatable}

	The proof in \cite{posa1965independent} relies on a theorem due to \cite[Theorem 4]{ep1962} which states that \textsl{there is a constant \(c_3\) such that every graph on \(n\) vertices and \(n + \ell\) edges contains \(c_3 \cdot \ell / \log \ell\) pairwise edge-disjoint cycles.}
	The proof of this statement, in turn, relies on another result from the same paper stating that \textsl{every graph of minimum degree 3 contains a cycle of at most \(\Oh(\log n)\) edges.}
	These statements are, however, not true for digraphs, as one can take an acyclic orientation of a clique as a counter-example to the former and the \emph{cylindrical grid} of order \(k\) as a counter-example for the latter.
	Hence, already showing the existence of some (not necessarily elementary) function \(\bound{statement:elementary-younger}{f}{}\) requires different techniques from those used in the undirected setting.

	To overcome the issues above, \cite{reed1996packing} used a directed grid structure called a \emph{fence}\footnote{For subtle technical reasons, a fence \((\mathcal{P}, \mathcal{Q})\) in our notation corresponds to a fence \((\mathcal{Q}, \mathcal{P})\) in the notation of \cite{reed1996packing,kawarabayashi2022directed}.} to reroute paths in a digraph in order to close cycles.
	More precisely, they prove the following, where
	\(\bound{statement:fence-plus-back-linkage-implies-cycles}{q}{k}\) and
	\(\bound{statement:fence-plus-back-linkage-implies-cycles}{r}{k}\) are polynomials defined later.

	\begin{restatable}[{\cite[Statement (3.1)]{reed1996packing}}]{lemma}{fenceToCycles}
		\label{statement:fence-plus-back-linkage-implies-cycles}
		Let \(k, p, q, r \geq 1\) be integers.
		Let \((\mathcal{P}, \mathcal{Q})\) be a \((p, q)\)-fence in a digraph \(D\) where \(q \geq \bound{statement:fence-plus-back-linkage-implies-cycles}{q}{k}\) and \(p \geq r\), and let \(\mathcal{R}\) be an \(\End{\mathcal{P}}\)-\(\Start{\mathcal{P}}\) linkage of order \(r \geq \bound{statement:fence-plus-back-linkage-implies-cycles}{r}{k}\).
		Then, \(D\) contains \(k\) pairwise vertex-disjoint cycles.
	\end{restatable}

	Unfortunately, their method of obtaining a fence results in a non-elementary bound.
	Additionally, they used a recursive argument.
	As any super-polynomial step in such a recursion yields a non-elementary bound in the end, and as not all of our bounds are polynomial, we cannot use recursive arguments in the same way as they did.
    Since we cannot follow the structure of their proof, we instead draw inspiration from~\cite{chekuri2016polynomial,amiri2016erdos, amiri2017,kawarabayashi2015directed}.

	The starting point of our proof is the following result due to \cite{amiri2016erdos, amiri2017}.
	
	\begin{lemma}[{\cite[Lemma 4.2]{amiri2016erdos}, \cite[Lemma 4.3.1]{amiri2017}}]
		\label{lem:akkw2}
		Let $D$ be a digraph with $\dtw{D} \leq w$. For each strongly connected digraph $H$, the digraph $D$ either has $k$ disjoint copies of $H$ as a topological minor, or contains a set $T$ of at most $k \cdot (w+1)$ vertices such that $H$ is not a topological minor of $D - T$. 
	\end{lemma}

	By setting \(H = \Ck{2}\) in the statement above, we obtain that digraphs of small directed \treewidth either contain many pairwise vertex-disjoint cycles or admit a feedback vertex set of small size.
	Hence, for our results, it suffices to show that digraphs of sufficiently large directed \treewidth contain many pairwise vertex-disjoint cycles.

	To this end, we follow the structure of the proof by~\cite{kawarabayashi2015directed}, which was split into roughly three parts:
	First, they obtain a grid-like structure called a \emph{web}.
	Then, they construct a fence from this web where the \say{end} of the fence is well-linked back to its \say{beginning}.
	Finally, they use such a \emph{back-linkage} to close the cycles of the cylindrical grid.
	We note that~\cite{reed1996packing} also obtained their fence through a web, and, while the bounds of~\cite{kawarabayashi2015directed} are also non-elementary for all three parts, the overall structure of the latter proof is not recursive, making it more suitable for improvements towards elementary bounds.
	Hence, our first technical contribution is obtaining a web in a given digraph whose directed \treewidth is large enough but still upper-bounded by an elementary function of the size of the web (\cref{section:constructing-web}).

	To obtain a fence from the web, we draw inspiration from the \emph{path-of-sets system} framework due to Chekuri and Chuzhoy~\cite{chekuri2016polynomial}, which played an important role in their proof of a polynomial bound for the (undirected) Grid Theorem.
	In rough terms, a path-of-sets system describes a sequence of highly connected sets organised in a path-like structure, a concept which can be adapted to the directed setting.
	In order to handle all the cases that appear in the directed setting, we need to consider two types of highly connected sets, namely \emph{well-linked} and \emph{order-linked} sets.
    This, in turn, leads us to our definitions of \emph{paths of well-linked sets}, \emph{paths of order-linked sets} and \emph{cycles of well-linked sets}.
	These three concepts naturally capture the connectivity properties provided by \emph{fences}, \emph{acyclic grids} and \emph{cylindrical grids}, respectively.

	As undirected connectivity is considerably different from its directed counterpart, it is unsurprising that most of the proof techniques of~\cite{chekuri2016polynomial} cannot be easily adapted to the directed setting.
	Thus, in order to obtain the connectivity properties required above, we develop a framework based on the concept of \emph{temporal digraphs} (see, e.g.~\cite{CasteigtsHMZ20,Molter20} for an overview of results on temporal graphs).
	Informally, a temporal digraph is a sequence of digraphs (on potentially different vertex or arc sets) where each digraph represents the state of the temporal digraph at a given point in time.
	Walks in temporal digraphs need to respect the time at which each arc is available.
	In our setting, temporal digraphs arise naturally when we consider disjoint paths intersecting a sequence of disjoint subgraphs in the same order.

	One of our main contributions in this context is introducing a novel concept called \(H\)\emph{-routings} for digraphs and temporal digraphs, which is a weaker relation than strong immersions or butterfly minors for digraphs.
	That is, if a digraph \(D\) contains \(H\) as a strong immersion or as a butterfly minor, then it also contains an \(H\)-routing.
	A central part of our proof consists in establishing a relation between connectivity in digraphs and the existence of certain \(H\)-routings in temporal digraphs.
	Using the framework above, we obtain one of our main results (the meaning of \(A(S_0)\) and \(B(S_\ell)\) is given later in~\cref{def:path-of-well-linked-sets}).
	
	\begin{customthm}{9.11}{thm:high_dtw_to_POSS_plus_back-linkage}
				Every digraph $D$ with $\dtw{D} \geq \bound{short:thm:high_dtw_to_POSS_plus_back-linkage}{t}{w, \ell} \in \PowerTower{7}{\Polynomial{25}{w, \ell}}$ contains a path of well-linked sets $(\mathcal{S} = (S_0,$ $S_1,$ $\dots,$ $S_{\ell}),\mathscr{P})$ of width $w$ and length $\ell$ such that $B(S_\ell)$ is well-linked to $A(S_0)$ in $D$.
	\end{customthm}
	
	Since we also show that a large path of well-linked sets contains a large fence, by using \cref{statement:fence-plus-back-linkage-implies-cycles,thm:high_dtw_to_POSS_plus_back-linkage}, we can solve the large directed \treewidth case of~\cref{statement:elementary-younger}, completing the proof of our main result.

	\Cref{short:thm:high_dtw_to_POSS_plus_back-linkage} also plays an important role in obtaining an elementary bound for the Directed Grid Theorem.
	Intuitively, a digraph consisting of a fence \(F\) with an internally disjoint linkage \(\mathcal{L}\) from the \say{end} of \(F\) back to its beginning is very close to a cylindrical grid.
	In \cite{COSSII}, we prove that one can indeed find such a disjoint linkage \(\mathcal{L}\) starting from a path of well-linked sets as obtained in \cref{short:thm:high_dtw_to_POSS_plus_back-linkage}.

	Our modular approach facilitates the transfer of the intermediate results in our proof to other settings.
	Well-linked sets play an important role in several results in the theory of digraphs (for example, in~\cite{reed1996packing,johnson2001directed,edwards2017half,kawarabayashi2015directed}), and our framework provides additional tools for obtaining such sets.

	From an algorithmic perspective, our intermediate concepts allow us to divide the problem of finding an acyclic grid in a digraph into subproblems such as finding long walks in temporal digraphs, constructing \(H\)-routings and obtaining well-linked sets.
	These subproblems can be studied independently of each other, simplifying the process of identifying bottlenecks and computational obstacles for the original problem.
 		
	The paper is organised as follows.
	\Cref{sec:preliminaries,sec:grids} contain preliminary definitions used throughout the paper.
	In~\cref{sec:proof overview}, we provide an overview of the proof.
	We construct a web in a digraph of large directed \treewidth in \cref{section:constructing-web}, improving the corresponding step of the proof of~\cite{kawarabayashi2015directed} from a non-elementary to an elementary bound.
	Our framework on temporal digraphs is introduced in \cref{sec:temporal}, where we also obtain the \(H\)-routings from which we construct our order-linked and well-linked sets.
	In \cref{sec:order-linked,sec:well-linked}, we introduce the concepts of paths of order-linked sets and paths of well-linked sets, respectively, and show how to obtain the corresponding grid type from each of them.
	In~\cref{sec:constructing-pows}, we apply the framework developed in the sections above in order to construct a path of well-linked sets, which we then use to obtain the disjoint cycles, completing the proof of our main result, \cref{statement:elementary-younger}.
	Finally, we present our final remarks and open questions in \cref{sec:conclusion}.

\section{Preliminaries}
\label{sec:preliminaries}

In this section, we establish our notation and recall standard concepts and results from the literature used throughout the paper.

\paragraph{Sequences, sets and functions.}
Given sequences $S_1 \coloneqq (x_1, x_2, \dots, x_{j})$ and $S_2 \coloneqq (y_1, y_2, \dots$, $y_{k})$, we write $S_1 \cdot S_2$ for the sequence $S_3 \coloneqq (x_1, x_2, \dots, x_{j}, y_1, y_2, \dots, y_{k})$. 
We say that $S_1 \cdot S_2$ is a \emph{decomposition} of $S_3$.
The following is a well-known theorem about sequences of numbers due to Erd\H{o}s and Szekeres.
\begin{theorem}[\cite{erdosszekeres1935}]
	\label{thm:erdos_szekeres}
	Let $r,s \in\N$.
	Every sequence of distinct numbers of length at least $\Brace{r-1}\Brace{s-1}+1$ contains a monotonically increasing subsequence of length $r$ or a monotonically decreasing subsequence of length $s$.
\end{theorem}

An \emph{ordered set} is a sequence $A = \Brace{a_1, \ldots, a_{k}}$ such that all elements of $A$ are distinct.
The order $\leq_A : A \times A$ induced by $A$ is defined by $a_i \leq_A a_j$ for all $1 \leq i \leq j \leq k$.
An \emph{ordered subset} $A' \subseteq A$ then is just a subsequence of $A$, that is, the order of the elements is preserved.
If we obtain an ordered set $A'$ from a set $A$ by fixing an order, we call $A'$ an \emph{ordering} of $A$.

\paragraph{Power towers and polynomials}
Let \(d\) be an integer and \(V = \{x_1, \ldots, x_k\}\) a set of variables.
A \emph{polynomial} of degree \(d\) over \(V\) is a function \(p(x_{1}, x_{2}, \ldots, x_{k})\) of the form \(p(x_{1}, x_{2}, \ldots, x_{k}) = \sum_{i=1}^n (c_i \prod_{j=1}^k x_j^{e_{j,i}})\), where for each \(1 \leq i \leq n\) and each \(1 \leq j \leq k\) we have that \(c_i \in \Reals\), \(e_{j,i} \in \Naturals\) and \(\sum_{j=1}^k e_{j,i} \leq d\).
We write \(\Polynomial{d}{x_{1}, x_{2}, \ldots, x_{k}}\) for the set of all functions \(f\) for which there is a polynomial \(p\) of degree \(d\) over the variable set \(x_{1}, x_{2}, \ldots, x_{k}\) such that \(f(x_{1}, x_{2}, \ldots, x_{k}) \in \Oh(p(x_{1}, x_{2}, \ldots, x_{k}))\).

We define \emph{power towers} as follows.
Given an integer \(h\) and a set of functions \(F\) over a set of variables \(V\), we define a set of functions \(\PowerTower{h}{F}\) recursively as follows.
We set \(\PowerTower{0}{F} = F\) and define \(\PowerTower{h}{F}\) as  \(\Set{f : \Reals^{\Abs{V}} \to \Reals \mid f \in \Oh(2^{g(V)}), g \in \PowerTower{h-1}{F}}\) for \(h > 1\).
If \(F = \Polynomial{d}{V}\), we say that a function \(f \in \PowerTower{h}{F}\) is a \emph{power tower} of height \(h\).

\paragraph{Graphs and digraphs.} 
We denote by $E(G)$ the edge/arc set of a graph $G$, directed or not, and by $V(G)$ its vertex set. We often use $G$ for undirected and $D$ for directed graphs (often referred to as \emph{digraphs}). 

Let $D$ be a digraph. 
Given a set of vertices $X \subseteq \V{D}$, we write $D - X$ for the digraph $(Y \coloneqq \V{D} \setminus X$, $\A{D} \cap (Y \times Y))$. 
Similarly, given a set of arcs $F \subseteq \A{D}$, we write $D - F$ for the digraph $(\V{D}, \A{D} \setminus F)$.

If $D$ is a digraph and $v \in V(D)$, then  $\InN{D}{v} \coloneqq \{ u \in V \mid (u,v) \in \A{D}\}$ is the set of \emph{in-neighbours} and $\OutN{D}{v} \coloneqq \{ u \in V \mid (v,u) \in \A{D}\}$ the set of \emph{out-neighbours} of $v$.
By $\Indeg{D}{v} \coloneqq \Abs{\InN{}{v}}$ we denote the \emph{in-degree} of $v$ and by $\Outdeg{D}{v} \coloneqq \Abs{\OutN{}{v}}$  its out-degree.
When working with a set or another structure $X$ containing digraphs, we write $\ToDigraph{X}$ to mean the digraph obtained by taking the union of all digraphs in $X$.

\paragraph{Paths and walks.}
A \emph{walk} of length $\ell$ in a digraph $D$ is a sequence of
vertices \(W \coloneqq \Brace{v_0, v_1, \dots, v_{\ell}}\) such that
$\Brace{v_i, v_{i+1}} \in \A{D}$, for all $0 \leq i < \ell$.
We write $\Start{W}$ for $v_0$ and $\End{W}$ for $v_\ell$ and say that
$W$ is a $v_0$-$v_\ell$-walk.

A walk $W \coloneqq \Brace{v_0, v_1, \dots, v_{\ell}}$ is called a \emph{path} if no vertex appears twice in it and it is called a \emph{cycle} if $v_0 = v_\ell$ and $v_i \neq v_j$ for all $0 \leq i < j < \ell$. 

We often identify a walk $W$ in $D$ with the corresponding subgraph and write $V(W)$ and $E(W)$ for the set of vertices and arcs appearing on it.

Given two walks $W_1 \coloneqq (x_1, x_2, \dots, x_{j})$ and $W_2 \coloneqq (y_1, y_2, \dots, y_{k})$ with $\End{W_1} = \Start{W_2}$, we make use of the concatenation notation for sequences and write $W_1 \cdot W_2$ for the walk $W_3 \coloneqq (x_1, x_2, \dots, x_{j}, y_2, y_3, \dots, y_{k})$.
We say that $W_1 \cdot W_2$ is a decomposition of $W_3$.
If $W_1$ or $W_2$ is an empty sequence, then the result of $W_1 \cdot W_2$ is the other walk (or the empty sequence if both walks are empty).

Let $P$ be a path in a digraph $D$ and let $X$ be a set of vertices with $\V{P} \cap X \neq \emptyset$. 
We consider the vertices $p_1,\dots,p_m$ of $P$ ordered by their occurrence on $P$.
Let $i$ be the highest index such that $p_i \in X$ and let $j$ be the smallest index such that $p_j \in X$.
We call $p_i$ the \emph{last vertex of $P$ in $X$} or, depending on the perspective, the last element of $X$ on $P$, and $p_j$ the \emph{first vertex of $P$ in $X$} or the first vertex of $X$ on $P$.

\paragraph{Specific digraphs.}
We denote the digraph of a path on $k$ vertices by $\Pk{k}$.
For the \emph{bidirected path on $k$} vertices, we write $\biPk{k} \coloneqq (\Set{u_1, u_2, \dots, u_{k}}, \{(u_i, u_j) \mid 1 \leq i, j \leq k \text{ and } \Abs{i - j} = 1\})$.
The \emph{cycle on $k$ vertices} is given by $\Ck{k} \coloneqq (\{u_0, u_1, \dots, u_{k - 1}\}, \{(u_{i}, u_{i+1 \mod k}) \mid 0 \leq i < k\})$.
Finally, we write $\biK{k} \coloneqq (\{u_1, u_2, \dots, u_{k}\}, \{(u_i, u_j) \mid 1 \leq i, j \leq k \text{ and } i \neq j\})$ for the \emph{complete digraph on $k$ vertices}.

\paragraph{Connectivity.} A digraph $D$ is said to be \emph{strongly connected} if for every $u,v \in V$ there is a $u$-$v$-path \textbf{and} a $v$-$u$-path in $D$.
We say $D$ is \emph{unilateral} if for every $u,v \in V$ there is a $u$-$v$-path \textbf{or} a $v$-$u$-path in $D$.
Finally, $D$ is \emph{weakly-connected} if the underlying undirected graph of $D$ is connected.

A \emph{feedback vertex set} of $D$ is a set $X \subseteq V(D)$ such that $D - X$ is acyclic.

\paragraph{Linkages and separators.}
Let $A, B \subseteq V(D)$.
An $A$-$B$-walk is a walk $W$ that starts in $A$ and ends in $B$.
A set $X \subseteq \V{D}$ is an \emph{$A$-$B$-separator} if there are no $A$-$B$-paths in $D - X$. 

A \emph{linkage} in $D$ is a set $\LLL$ of pairwise vertex-disjoint paths.
The \emph{order} $|\LLL|$ of $\LLL$ is the number of paths it contains.

An \emph{$A$-$B$-linkage}  of order $k$ is a linkage $\LLL \coloneqq \{L_1, L_2, \dots, L_{k}\}$ such that $\Start{L_i} \in A$ and $\End{L_i} \in B$ for all $1 \leq i \leq k$.
We write $\Start{\mathcal{L}}$ for the set $\{\Start{L_i} \mid L_i \in \mathcal{L}\}$ and $\End{\mathcal{L}}$ for the set $\{\End{L_i} \mid L_i \in \mathcal{L}\}$.
We also extend the notation for path concatenation to linkages.
Given linkages $\mathcal{P} = \Set{P_1, P_2, \dots, P_{k}}$ and $\mathcal{Q} = \Set{Q_1, Q_2, \dots, Q_{k}}$ such that $\End{\mathcal{P}} = \Start{\mathcal{Q}}$, we write $\mathcal{P} \cdot \mathcal{Q}$ for the linkage $\Set{P_a \cdot Q_b \mid P_a \in \mathcal{P}, Q_b \in \mathcal{Q} \text{ and } \End{P_a} = \Start{Q_b}}$. 

It is often convenient to use a linkage $\mathcal{L}$ as a function $\mathcal{L}: \Start{\mathcal{L}} \to \End{\mathcal{L}}$. 
The expression $\Fkt{\mathcal{L}}{a} = b$ then means that $\mathcal{L}$ contains a path starting in $a$ and ending in $b$. 

We frequently use the following classical result by Menger~\cite{menger}.
\begin{theorem}[Menger's Theorem~\cite{menger}]
	\label{thm:menger}
	Let $D$ be a digraph, $A,B \subseteq \V{D}$ with $\Abs{A} = \Abs{B}$.
	There is an $A$-$B$-linkage of size $k$ in $D$ if and only if every $A$-$B$-separator has size at~least~$k$.
\end{theorem}

Throughout the paper, we frequently work with a special kind of linkage that we define next.

\begin{definition}[minimal linkage]
	\label{def:H-minimal}
	Let $D$ be a digraph, $H\subseteq D$ be a subgraph and $\LLL$ be a linkage of order $k$.
    $\LLL$ is \emph{minimal with respect to $H$}, or \emph{$H$-minimal}, if for all arcs $e \in \bigcup_{L\in \LLL} \E{L} \setminus \E{H}$ there is no $\Start{\LLL}$-$\End{\LLL}$-linkage of order $k$ in the graph $(\LLL \cup H) - e$.
\end{definition}

Given a linkage $\mathcal{L}$ in a digraph $D$ and a subgraph $H \subseteq D$, we can always obtain a linkage $\mathcal{L}'$ with same order and same endpoints as $\mathcal{L}$ which is $H$-minimal by iteratively removing arcs $e \in \A{\mathcal{L}} \setminus \A{H}$ for which a $\Start{\mathcal{L}}$-$\End{\mathcal{L}}$-linkage of order $\Abs{\mathcal{L}}$ avoiding $e$ exists.

Minimal linkages were used extensively in~\cite{kawarabayashi2015directed}.
The idea is that, when constructing paths of an $H$-minimal linkage $\LLL$, we always prefer to use arcs of $H$ over arcs not in $E(H)$.
This implies the following property, which we exploit frequently in our proofs.

\begin{definition}[weak minimality]
	\label{def:weak_minimality}
	A linkage $\mathcal{L}$ in a digraph $D$ is \emph{weakly $k$-minimal} with respect to a subgraph $H$ of $D$ if for every $L_1 \cdot e \cdot L_2 \in \mathcal{L}$ with $e \in \E{\mathcal{L}} \setminus \E{H}$ there is a $\V{L_1}$-$\V{L_2}$-separator of size at most $k-1$ in $\Brace{\mathcal{L} \cup H} - e$. 
\end{definition}

\begin{observation}
	\label{obs:H-minimal-implies-weakly-minimal}
	Let $H$ be a subgraph of a digraph $D$ and let $\mathcal{L}$ be a linkage which is $H$-minimal.
	Then $\mathcal{L}$ is weakly $\Abs{\mathcal{L}}$-minimal with respect to $H$.
\end{observation}
\begin{proof}
	Assume towards a contradiction that there is some $L \in \mathcal{L}$ and some $e \in \E{L} \setminus \E{H}$ such that $L$ can be decomposed into $L_1 \cdot e \cdot L_2$ and there is no $\V{L_1}$-$\V{L_2}$-separator of size less than $\Abs{\mathcal{L}}$ in $\ToDigraph{\mathcal{L} \cup H} - e$.
	By~\cref{thm:menger}, there is a $\V{L_1}$-$\V{L_2}$-linkage $\mathcal{Q}$ of order $\Abs{\mathcal{L}}$ in $\ToDigraph{\mathcal{L} \cup H} - e$.
	
	Let $S$ be a minimum $\Start{\mathcal{L}}$-$\End{\mathcal{L}}$-separator in $\ToDigraph{\mathcal{L} \cup H} - e$.
	Because $\mathcal{L}$ is $H$-minimal, we have that $\Abs{S} < \Abs{\mathcal{L}}$.
	Hence, $S$ must hit every path in $\mathcal{L} \setminus \Set{L}$ and must be disjoint from $L$.
	
	Since $\Abs{\mathcal{Q}} = \Abs{\mathcal{L}}$, there is some $Q \in \mathcal{Q}$ which is not hit by $S$.
	Hence, there is a $\Start{L}$-$\End{L}$-path in $\ToDigraph{\mathcal{L} \cup H} - e - S$, a contradiction to the assumption that $S$ is a separator. 
	Thus, $\mathcal{L}$ is weakly $\Abs{\mathcal{L}}$-minimal with respect to $H$.
\end{proof}

We close this part by recalling the definition of well-linkedness, an important property of a central concept in our proof, the path of well-linked sets.

\begin{definition}
    \label{def:well-linked}
	Let $A, B$ be sets of vertices in a digraph $D$.
  We say that \emph{$A$ is well-linked to $B$ in $D$} if for every $A' \subseteq A$ and every $B' \subseteq B$ with $\Abs{A'} = \Abs{B'}$ there is an $A'$-$B'$-linkage of order $\Abs{A'}$ in $D$.
	If \(A\) is well-linked to \(A\), then we say that \(A\) is a \emph{well-linked set}.
\end{definition}

\paragraph{Minors.}
Given a digraph $D$ and an arc $(u,v) \in \A{D}$, we say that $(u,v)$ is \emph{butterfly contractible} if $\Outdeg{}{u} = 1$ or $\Indeg{}{v} = 1$. 
The \emph{butterfly contraction} of $(u,v)$ is the operation which consists in removing $u$ and $v$ from $D$, then adding a new vertex $uv$, together with the arcs $\Set{(w, uv) \mid w \in \Indeg{D}{u}}$ and $\Set{(uv, w) \mid w \in \Outdeg{D}{v}}$. 
Note that, by definition of digraphs, we \emph{remove} duplicated arcs and loops, that is, arcs of the form $(w,w)$.
If there is a subgraph $D'$ of $D$ such that we can construct another digraph $H$ from $D'$ using butterfly contractions, then we say that $H$ is a \emph{butterfly minor of $D$}, or that \emph{$D$ contains $H$ as a butterfly minor}.

\section{Directed \Treewidth and Grids}
\label{sec:grids}

As mentioned before,~\cref{lem:akkw2} allows us to focus on the case that the given digraph has large directed \treewidth.
For this reason, in this section, we recall directed \treewidth and the dual concepts of brambles and cylindrical grids.
We also define various other forms of \say{grids} in directed graphs used in the sequel.

Directed \treewidth was originally introduced by Reed~\cite{reed1999introducing} and by Johnson, Robertson, Seymour and Thomas~\cite{johnson2001directed} (see also~\cite{JohnsonRST2001}).
Adler~\cite{adler2007directed} showed that the original definition in~\cite{johnson2001directed} of directed \treewidth is not closed under butterfly minors.
We, therefore, use the variant of directed \treewidth defined in \cite{kawarabayashi2022directed}, which is closed under taking butterfly minors.

An \emph{arborescence} $T$ is an acyclic directed graph obtained from an undirected rooted tree by orienting all edges away from the root.
That is, $T$ has a vertex $r_0$, called the root of $T$, with the property that for every  $r \in V(T)$ there is a unique directed path from $r_0$ to $r$ in $T$.
For each $r \in V(T)$, we denote the subarborescence of $T$ induced by the set of vertices in $T$ reachable from $r$ by $T_r$.
In particular, $r$ is the root of $T_r$.

\begin{definition}[{\cite[Definition 3.1]{kawarabayashi2022directed}}]
	\label{def:directed-tree-width}
	A \emph{directed \treedecomposition} of a digraph $D$ is a triple $(T, \beta, \gamma)$, where $\beta: V(T) \to 2^{V(D)}$ and $\gamma: E(T) \to 2^{V(D)}$ are functions and $T$ is an arborescence such that
	\begin{enamerate}{T}{item:directed-tree-width:last}
		\item \label{item:directed-tree-width:beta}
		      $\Set{\beta(t) : t \in V(T)}$ is a partition of $V(D)$ into (possibly empty) sets and
		\item \label{item:directed-tree-width:guard}
		      for every $e = (s, t) \in E(T)$, there is no closed directed walk in $D - \gamma(e)$ containing a vertex in $A$ and a vertex in $B$, where $A = \bigcup\Set{\beta(t ): t \in V(T_t)}$ and $B = V(D) \setminus A$.
		\label{item:directed-tree-width:last}
	\end{enamerate}
	For $t \in V(T)$ we define $\Gamma(t) \coloneqq \beta(t) \cup \bigcup\Set{\gamma(e) : e \sim t}$, where $e \sim t$ if $e$ is incident to $t$, and we define $\beta(T_t) \coloneqq \bigcup\Set{\beta(t) : t \in V(T_t)}$.
	The \emph{width} of $(T, \beta, \gamma)$ is the smallest integer $w$ such that $\Abs{\Gamma(t)} \leq w + 1$ for all $t \in V(T)$.
	The \emph{directed \treewidth}\Index{directed \treewidth} of $D$ is the smallest integer $w$ such that $D$ has a directed \treedecomposition of width $w$.
	The sets $\beta(t)$ are called the bags and the sets $\gamma(e)$ are called the guards of the directed \treedecomposition.
\end{definition}

A natural dual to directed \treedecompositions are objects called \emph{brambles}.
The concept of brambles was also introduced by~\cite{johnson2001directed}.
For the same reason as before, we use the variant of brambles defined in~\cite{kawarabayashi2015directed}.

\begin{definition}
	A \emph{bramble} in a digraph $D$ is a set $\mathcal{B}$ of strongly connected subgraphs $B \subseteq D$ such that $B \cap B' \neq \emptyset$ for all $B, B' \in \mathcal{B}$.
	
	A \emph{cover} of $\mathcal{B}$ is a set $X \subseteq \V{D}$ of vertices such that $V(B) \cap X \neq \emptyset$ for all $B \in \mathcal{B}$.
    Finally, the \emph{order} of a bramble $\mathcal{B}$ is the minimum size of a cover for $\mathcal{B}$.
    The bramble number $bn(D)$ of $D$ is the maximum order of a bramble in $D$.
\end{definition}

We also need the following relation between brambles and directed \treewidth.
It can be obtained from results due to~\cite{johnson2001directed} by converting brambles to havens and back, and the statement was proven formally by~\cite{kreutzer2014width}.

\begin{lemma}[{\cite[Corollary 6.4.24]{kreutzer2014width}}]
	\label{state:dtw-and-bramble-number}
		There are constants \(c, c'\) such that for all digraphs \(D\), \(\operatorname{bn}(D) \leq c \cdot \dtw{D} \leq c' \cdot \operatorname{bn}(D)\).
\end{lemma}

By combining the statement (1.1) of \cite{johnson2001directed} and Lemma 6.4.20 of \cite{kreutzer2014width}, we obtain the following.

\begin{corollary}[{\cite{johnson2001directed,kreutzer2014width}}]
	\label{state:dtw-to-bramble}
	Let \(D\) be a digraph.
	If \(\dtw{D} \geq 2k\), then \(D\) contains a bramble of order \(k\).
\end{corollary}

We now define another obstruction to directed \treewidth called \emph{cylindrical grids},
which are illustrated in \cref{fig:cylindrical-grid}.
\begin{definition} \label{def:cylindrical-grid}
	A \emph{cylindrical grid} of order $k$ is a digraph $G_k$ consisting of $k$ pairwise disjoint directed cycles $C_1, C_2, \dots, C_{k}$ of length $2k$, together with a set of $2k$ pairwise vertex-disjoint paths $P_1, P_2, \dots, P_{2k}$ of length $k - 1$ such that 
	\begin{itemize}
		\item each path $P_i$ has exactly one vertex in common with each cycle $C_j$ and both endpoints of $P_i$ are in $\V{C_1} \cup \V{C_k}$, 
		\item the paths $P_1, P_2, \dots, P_{2k}$ appear on each $C_i$ in this order, and
		\item for each $1 \leq i \leq 2k$, if $i$ is odd, then the cycles $C_1, C_2, \dots, C_{k}$ occur on $P_i$ in this order and, if $i$ is even, then the cycles occur in the reverse order $C_k, C_{k-1}, \dots, C_{1}$. 
	\end{itemize}
\end{definition}

Besides cylindrical grids, several different ways of defining \say{directed grids} have been considered in the literature (for example~\cite{reed1996packing}, \cite{johnson2001directed}, \cite{kawarabayashi2015directed}).
Two of these, called \emph{acyclic grids} and \emph{fences} (see~\cref{fig:acyclic-grid-and-fence}), are used at various points in our proof. 
Since we are interested in grids in the context of minors, we define grids by linkages instead of giving explicit vertex and arc sets. 

To motivate the following definitions, let us dissect a cylindrical grid \(((C_1, \dots, C_k), (P_1, \dots,\) \(P_{2k}))\) as follows.
An important difference between cylindrical grids and grids in undirected graphs is that cylindrical grids are \emph{locally acyclic} in the following sense. 
Suppose we delete in each cycle $C_i$ the arc $e_i$ whose head is on the path $P_1$.
The dotted red lines mark these arcs in~\cref{fig:cylindrical-grid}.
The resulting digraph is acyclic and consists of two linkages: the linkage  $\{P_1, \dots, P_{2k}\}$ and the linkage $\{C_1 - e_1, \dots, C_k - e_k\}$ which contains for each cycle $C_i$ the path that remains once the arc $e_i$ is deleted.
Digraphs of this form are called \emph{fences}.
See~\cref{fig:cylindrical-grid} for a drawing of cylindrical grids illustrating how they are constructed from a fence with additional arcs closing the cycles.

\begin{definition}
	\label{def:fence}
	A \emph{$(p,q)$-fence} is a tuple $(\mathcal{P}, \mathcal{Q})$ such that
	\begin{itemize}
		\item \label{fence:P}
		      $\mathcal{P} = \Brace{P_1, P_2, \dots, P_{2p}}$ and $\mathcal{Q} = \Brace{Q_1, Q_2, \dots, Q_{q}}$ are linkages,
		\item \label{fence:P-intersection-Q}
		      for each $1 \leq i \leq 2p$ and each $1 \leq j \leq q$, the digraph $P_i \cap Q_j$ is a path (and therefore non-empty),
		\item \label{fence:P-order-Q}
		      for each $1 \leq j \leq q$, the paths $P_1 \cap Q_j, P_2 \cap Q_j, \dots, P_{2p} \cap Q_j$ appear in this order along $Q_j$, and
		\item \label{fence:Q-order-P}
		      for each $1 \leq i \leq 2p$, if $i$ is odd then the paths $P_i \cap Q_1, P_i \cap Q_2, \dots, P_i \cap Q_{q}$ appear along $P_i$ in this order, and if $i$ is even instead, then the paths $P_i \cap Q_q, P_i \cap Q_{q-1}, \dots, P_i \cap Q_{1}$ appear in this order along $P_i$.
	\end{itemize}
\end{definition}
\begin{figure}[!ht]
	\centering
		\begin{tikzpicture}
        \path[use as bounding box] (-7.9,-3.9) rectangle (7.9,3.9);
        \node (left) at (-4.1,0) {
            \resizebox{0.9\textwidth/2}{!}{\includegraphics{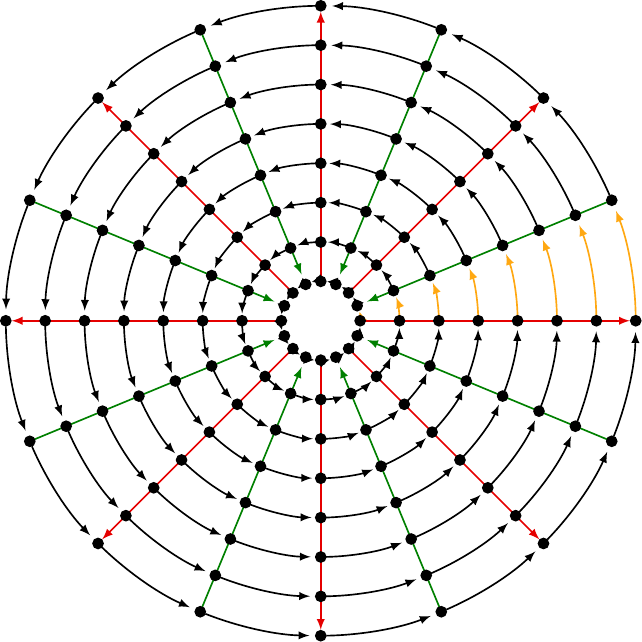}} };
        \node (right) at (4.1,0) {
            \resizebox{0.9\textwidth/2}{!}{\includegraphics{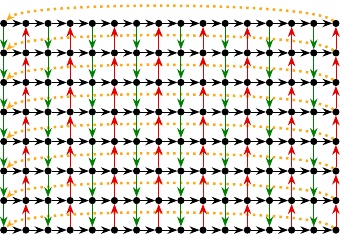}}    };
    \end{tikzpicture}
	\caption{Cylindrical grid $G_8$ of order $8$ drawn in two ways.
        The	drawing on the right illustrates how a cylindrical grid is obtained from a fence.
        The \textcolor{myOrange}{dotted orange} paths symbolise the arcs $e_i$ that close the cycles drawn solid on the left.
	}
	\label{fig:cylindrical-grid}
\end{figure}

See~\cref{fig:fence} for an illustration.
The \say{horizontal} paths, or \emph{rows}, constitute the linkage $\QQQ$ and the \emph{columns} form the linkage $\PPP$. 

A useful property of a $(p, q)$-fence $(\PPP, \QQQ)$ is that if $A \subseteq \Start{\QQQ}$ and $B \subseteq \End{\QQQ}$ are sets with $|A| = |B| \leq p$, then there is an $A$-$B$-linkage $\LLL$ of order $|A|$ in the digraph $\bigcup \PPP \cup \bigcup \QQQ$.

Now, let us further decompose the fence constructed from the cylindrical grid to obtain an even simpler form of directed grid.
In a fence, we can only route from \say{left to right} and we can route \say{upwards} as well as \say{downwards}. 
An even simpler form of a directed grid is obtained if we remove the \say{upwards} paths from a fence, i.e.~every second column.
The resulting digraph is called an \emph{acyclic grid}, illustrated in~\cref{fig:acyclic-grid}.

\begin{definition}
	\label{def:acyclic-grid}
	An \emph{acyclic $(p,q)$-grid} is a pair $(\mathcal{P}, \mathcal{Q})$ such that
	\begin{itemize}
		\item \label{acyclic-grid:P}
		      $\mathcal{P} = \Brace{P_1, P_2, \dots, P_{p}}$ and $\mathcal{Q} = \Brace{Q_1, Q_2, \dots, Q_{q}}$ are linkages,
		\item \label{acyclic-grid:P-intersection-Q}
		      for each $1 \leq i \leq p$ and each $1 \leq j \leq q$, the digraph $P_i \cap Q_j$ is a path (and therefore non-empty),
		\item \label{acyclic-grid:P-order-Q}
		      for each $1 \leq j \leq q$, the paths $P_1 \cap Q_j, P_2 \cap Q_j, \dots, P_{p} \cap Q_j$ appear in this order along $Q_j$, and
		\item \label{acyclic-grid:Q-order-P}
		      for each $1 \leq i \leq p$, the paths $P_i \cap Q_1, P_i \cap Q_2, \dots, P_i \cap Q_{q}$ appear in this order along $P_i$.
	\end{itemize}
\end{definition}

\begin{figure}[!ht]
	\centering
	\begin{subfigure}[b]{0.49\textwidth}
        \centering
		\includegraphics{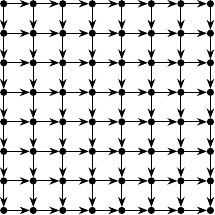}
        \caption{An acyclic $(8,8)$-grid.}
        \label{fig:acyclic-grid}
    \end{subfigure}
    \hfill
    \begin{subfigure}[b]{0.49\textwidth}
        \centering
        \includegraphics{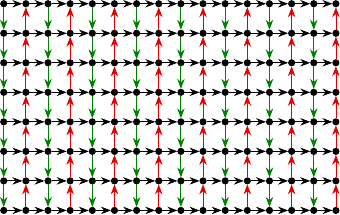}
        \caption{An $(8,8)$-fence.}
        \label{fig:fence}
    \end{subfigure}
	\caption{An acyclic grid and a fence.}
	\label{fig:acyclic-grid-and-fence}
\end{figure}

Acyclic grids only allow routing from top to bottom and from left to right. 

Another grid-like structure we define is the \emph{web}, initially introduced by Reed et al.~in~\cite{reed1996packing}.
Webs are essential in the proof of the Directed Grid Theorem in~\cite{kawarabayashi2015directed} and are equally crucial for our results.

\begin{definition}\label{def:web}
	Let $D$ be a digraph.
	Two linkages $\mathcal{H}$ and $\mathcal{V}$ in $D$ constitute an $\Brace{\Abs{\mathcal{H}},\Abs{\mathcal{V}}}$-web $\Brace{\mathcal{H},\mathcal{V}}$ if every path in $\mathcal{V}$ intersects every path in $\mathcal{H}$.
	
	The set $\Start{\mathcal{V}}$ is called the \emph{top} of the web, while the set $\End{\mathcal{V}}$ is called the \emph{bottom} of the web.
	Finally, $\Brace{\mathcal{H},\mathcal{V}}$ is \emph{well-linked} if $\End{\mathcal{V}}$ is well-linked linked to $\Start{\mathcal{V}}$ in $D$.
\end{definition}

Our proof uses an intermediate object called a \emph{semi-web}.
Unlike a web, we do not require from a semi-web \((\mathcal{H}, \mathcal{V})\) that the paths in \(\mathcal{H}\) and \(\mathcal{V}\) pairwise intersect each other.
Instead, a semi-web has a \emph{degree} which controls how much the two linkages must intersect.

\begin{definition}
	\label{def:semi-web}
    Let $D$ be a digraph.
	Two linkages $\mathcal{H}$ and $\mathcal{V}$ in $D$ form a $\Brace{\Abs{\mathcal{H}},\Abs{\mathcal{V}}}$-\emph{semi-web} $\Brace{\mathcal{H},\mathcal{V}}$ of degree \(d\) if every path in $\mathcal{V}$ intersects at least \(d\) paths in $\mathcal{H}$.

	Finally, $\Brace{\mathcal{H},\mathcal{V}}$ is \emph{well-linked} if $\End{\mathcal{V}}$ is well-linked to $\Start{\mathcal{V}}$ in $D$.
\end{definition}

An $(h,v)$-web $(\mathcal{H}, \mathcal{V})$ of avoidance $d$ from~\cite{kawarabayashi2015directed} corresponds to a $(h,v)$-semi-web $(\mathcal{H}, \mathcal{V})$ of degree \(\Abs{\mathcal{H}} \cdot \frac{d-1}{d}\) where $\mathcal{H}$ is minimal with respect to $\mathcal{V}$ in our notation.
The main reason for this modification is to avoid fractional calculations when determining the bounds of our functions.

\def\appsymb{}

\section{Proof overview}
\label{sec:proof overview}

\newcommand*{\powls}{path of well-linked sets\xspace}
\newcommand*{\powlss}{paths of well-linked sets\xspace}
\newcommand*{\pools}[1]{path of $#1$-order-linked sets\xspace}
\newcommand*{\poolss}[1]{path of $#1$-order-linked sets\xspace}
\newcommand*{\cowls}{cycle of well-linked sets\xspace}
\newcommand*{\cowlss}{cycles of well-linked sets\xspace}

Because~\cref{lem:akkw2} handles the small directed \treewidth case, we focus on the case where the given digraph has large directed \treewidth.
In particular, we are interested in finding specific types of \emph{paths of sets}, as defined below.

\subsection{Paths of sets}
\label{sec:cycles-and-paths-of-sets}
Our first main result consists in obtaining an acyclic grid whose end is well-linked back to its beginning, inside which we can easily find a fence.
To achieve this, we first construct objects with similar connectivity properties to these two types of grids without necessarily being planar.

It is a simple exercise to verify that, in a sufficiently large fence,
the ``left side'' is well-linked to the ``right side'' (see \cref{fig:fence} for the meaning of ``left'' and ``right'' used here).
Hence, the notion of well-linked sets is the connectivity property that we want in our abstraction of fences.
Acyclic grids, however, provide a different type of connectivity.
To capture this, we introduce the notion of $r$-\emph{order-linkedness},
which we formally define later in \cref{def:order-linkedness}.

We can now define the concepts of \emph{paths of well-linked sets}, which behave like fences, and of \emph{paths of $r$-order-linked sets}, which behave like acyclic grids.

\newcommand{\OLtext}[1]{\textcolor{myGreen}{#1}}
\newcommand{\WLtext}[1]{\textcolor{myLightBlue}{#1}}
\begin{customdoubledef}{}{def:path-of-order-linked-sets}{def:path-of-well-linked-sets}[path of \OLtext{$r$-order-linked}/\WLtext{well-linked} sets]
	\label{short:def:path-of-order-linked-sets}
    \label{short:def:path-of-well-linked-sets}
    A \emph{path of \OLtext{$r$-order-linked}/\WLtext{well-linked} sets} of width $w$ and length $\ell$ is a tuple $\Brace{\mathcal{S},\mathscr{P}}$ such that
    \begin{itemize}[noitemsep,topsep=0pt]
        \item $\mathcal{S}$ is a sequence of $\ell + 1$ pairwise disjoint subgraphs $\Brace{S_0,\dots,S_{\ell}},$ which are called \emph{clusters},
        
	    \item for every $0 \leq i \leq \ell$ there are disjoint \OLtext{ordered} sets $A(S_i),B(S_i) \subseteq \V{S_i}$ of size $w$ such that $A(S_i)$ is \OLtext{$r$-order-linked}/\WLtext{well-linked} to $B(S_i)$ in $S_i,$
     
        \item $\mathscr{P}$ is a sequence of $\ell$ pairwise disjoint linkages $\Brace{\mathcal{P}_0, \mathcal{P}_1, \dots, \mathcal{P}_{\ell - 1}}$ such that, for every $0 \leq i < \ell,$ $\mathcal{P}_i$ is a $B(S_i)$-$A(S_{i+1})$-linkage of order $w$ which is internally disjoint from $S_i$ and $S_{i+1}$ as well as disjoint from every $S \in \mathcal{S} \setminus \Set{S_i, S_{i+1}}.$
    \end{itemize}

	Further, a \OLtext{\pools{r}} $\Brace{\mathcal{S}, \mathscr{P}}$ is called \emph{uniform} if for all $0\leq i < \ell$ and for all $b_1,b_2 \in B(S_i)$ we have that $b_1 \leq_{B(S_i)} b_2$ implies $\mathcal{P}_i(b_1) \leq_{A(S_{i+1})} \mathcal{P}_i(b_2).$
	A \WLtext{\powls} is called \emph{strict} if every vertex in $S_i$ lies on an $A(S_i)$-$B(S_i)$-path.
\end{customdoubledef}

We obtain the following connection between paths of order-linked sets and paths of well-linked sets, allowing us to focus on obtaining paths of order-linked sets when proving our main result.
The construction is quite similar to the one used to obtain a fence from an acyclic grid, and the bounds we obtain are essentially the same.

\begin{customlem}{8.3}{proposition:order-linked to path of well-linked sets}
		\label{short:proposition:order-linked to path of well-linked sets}
		Let $\bound{short:proposition:order-linked to path of well-linked sets}{w}{w, \ell} \coloneqq w(\ell + 1)$.
		Every path of $w$-order-linked sets $(\mathcal{S}$, $\mathscr{P})$ of width at least $\bound{short:proposition:order-linked to path of well-linked sets}{w}{w, \ell}$ and length at least $\ell$ contains a path of well-linked sets $\Brace{\mathcal{S}', \mathscr{P}'}$ of width $w$ and length $\ell$.
\end{customlem}

The following two statements allow us to construct acyclic grids and fences from the above objects.
Hence, for our first main result, it suffices to show that digraphs of large directed \treewidth contain a large path of well-linked sets where the last cluster is well-linked to the first.
	
	\begin{customthm}{}{thm:order_linked_to_acyclic_grid}
				\label[thm-abbr]{thm-abbr:order_linked_to_acyclic_grid}
		There is a function \(\bound{thm-abbr:order_linked_to_acyclic_grid}{\ell}{k} \in \PowerTower{1}{\Polynomial{7}{k}}\) such that every path of $1$-order-linked sets of width at least $w = k^2 - 1$ and length at least $\bound{thm-abbr:order_linked_to_acyclic_grid}{\ell}{k}$ contains an acyclic $(k,k)$-grid.
	\end{customthm}

	\begin{customthm}{}{thm:poss-to-fence}
				\label[thm-abbr]{thm-abbr:poss-to-fence}
		There are functions \(\bound{thm:poss-to-fence}{w}{p,q} \in \Oh(p^{5} q^{5})\) and \(\bound{thm:poss-to-fence}{\ell}{p,q} \in \PowerTower{1}{\Polynomial{15}{p, q}}\) such that every path of well-linked sets $\Brace{\mathcal{S}, \mathscr{P}}$ of width at least $\bound{thm-abbr:poss-to-fence}{w}{p,q}$ and length $\ell \geq \bound{thm-abbr:poss-to-fence}{\ell}{p,q}$ contains a $(p,q)$-fence. 
		Moreover, \(\Start{\mathcal{P}} \subseteq A(S_0)\) and \(\End{\mathcal{P}} \subseteq B(S_\ell)\).
	\end{customthm}

This allows us to divide our proof into roughly two main \say{parts}.

The first part (\cref{sec:splits-and-segmentations}) consists in finding a well-linked web \((\mathcal{P}, \mathcal{Q})\) from a bramble of high order where \(\mathcal{P}\) is minimal with respect to \(\mathcal{Q}\).
We emphasise that the requirement of minimality is what makes obtaining such a web very challenging.

After obtaining the well-linked web, we can get a similar structure in which one linkage is \say{ordered} according to the other, constructing objects which are called \emph{splits} and \emph{segmentations} in~\cite{kawarabayashi2015directed}.

In the second part (\cref{sec:obtaining-a-path-of-well-linked-sets}), we construct a path of well-linked sets from the splits and segmentations obtained previously with the additional property that the $B$-set of the last cluster is well-linked to the $A$-set of the first cluster.

\begin{customthm}{9.11}{thm:high_dtw_to_POSS_plus_back-linkage}
    \label{short:thm:high_dtw_to_POSS_plus_back-linkage}
        Every digraph $D$ with $\dtw{D} \geq \bound{short:thm:high_dtw_to_POSS_plus_back-linkage}{t}{w, \ell} \in \PowerTower{7}{\Polynomial{25}{w, \ell}}$ contains a path of well-linked sets $(\mathcal{S} = (S_0,$ $S_1,$ $\dots,$ $S_{\ell}),\mathscr{P})$ of width $w$ and length $\ell$ such that $B(S_\ell)$ is well-linked to $A(S_0)$ in $D$.
\end{customthm}

Finally, we apply results from~\cite{reed1996packing,amiri2016erdos} in order to obtain our final main result.

\begin{customthm}{}{statement:elementary-younger}
	There is a function \(\bound{statement:elementary-younger}{f}{k} \in \PowerTower{8}{\Polynomial{43}{k}}\) such that the following holds.
	Let \(D\) be a digraph.
	Then \(D\) has \(k\) pairwise vertex-disjoint cycles or there is some \(X \subseteq \V{D}\) of size at most \(\bound{statement:elementary-younger}{f}{k}\) such that \(D - X\) is acyclic.
\end{customthm}

\subsection{Splits and segmentations}
\label{sec:splits-and-segmentations}

The construction of our paths of well-linked sets and paths of order-linked sets starts with \emph{splits} and \emph{segmentations},
which in turn are obtained from webs by ensuring that one linkage of the web intersects the other in an ordered fashion.

The proof of~\cite{kawarabayashi2022directed} for obtaining a web works by starting with a \emph{bramble} of high order and then showing that the existence of such a bramble implies the existence of an object called a \emph{\pathsystem}.\footnote{Despite the similarity in names, the \pathsystems we use here and the path-of-sets systems of~\cite{chekuri2016polynomial} are unrelated mathematical objects.}
The intuition behind the \pathsystem is the following.
Say we want to construct a bidirected clique as a butterfly minor.
Towards this end, we take several linkages connecting different subsets.
Of course, these linkages might not be pairwise disjoint.
What we want from our \pathsystem is the property that,
if the linkages we obtained above happen to be mostly disjoint from each other,
and if they intersect the \pathsystem in a ``clean'' way,
then we should be able to easily construct our bidirected clique
(and, in particular, find many disjoint cycles).
The interesting and more challenging part of the proof lies in showing that,
if we fail at finding such a \emph{clean \pathsystem} 
or if we fail at getting the linkages to be disjoint,
then we can obtain one of two useful structures called \emph{split} and \emph{segmentation}.

We define \pathsystems (\cref{def:path-system}), splits and segmentations (\cref{def:split-edge}) later.
\linebreak\Cref{state:path-system-to-clean-path-system} below essentially corresponds to Lemma 4.7 from~\cite{kawarabayashi2022directed}, and our proof is based on theirs.
Our proof relies on a result known as \emph{Lovász Local Lemma} \cite{EL1974,SPENCER197769},
which allowed us to improve the bounds from non-elementary to polynomial functions.

\begin{customlem}{}{state:path-system-to-clean-path-system}
		There are two polynomials \(\bound{state:path-system-to-clean-path-system}{\ell}{p_2,\ell_2,d_1} \in \Oh(\ell_{2} + d_{1} p_{2})\) and \(\bound{state:path-system-to-clean-path-system}{p}{q_1,p_2} \in \Oh(q_1 (p_2)^3)\) such that the following holds.
	Let \(d_1, q_1, p_2, \ell_2\) be integers and \(\mathcal{S} = (\mathcal{P}, \mathcal{L}, \mathcal{A})\) be an \(\ell\)-linked \pathsystem of order \(p\).
	If \(\ell \geq \bound{state:path-system-to-clean-path-system}{\ell}{p_2, \ell_2, d_1}\) and \(p \geq \bound{state:path-system-to-clean-path-system}{p}{q_1, p_2}\), then \(\ToDigraph{\mathcal{S}}\) contains one of the following
	\begin{enamerate}{C}{xitem:path-system-to-clean-path-system:last}
	\item
		a well-linked \(\Brace{\ell, q_1}\)-semi-web \(\Brace{\mathcal{P}_1, \mathcal{Q}_1}\) of degree \(d_1\) where \(\mathcal{P}_1 \in \mathcal{L}\), \(\mathcal{Q}_1 \subseteq \mathcal{P}\) and \(\mathcal{P}_1\) is minimal with respect to \(\mathcal{Q}_1\), or
			\item a clean \(\ell_2\)-linked \pathsystem \((\mathcal{P}_2, \mathcal{L}_2, \mathcal{A}_2)\) of order \(p_2\).
				\label{xitem:path-system-to-clean-path-system:last}
	\end{enamerate}
\end{customlem}

Using further results from~\cite{kawarabayashi2022directed}, we can prove the main result of this section.

\begin{customthm}{}{state:high-dtw-to-segmentation-or-split}
		There is a function \(\bound{state:high-dtw-to-segmentation-or-split}{t}{x,y,q,k} \in \PowerTower{5}{\Polynomial{5}{x,y,q,k}}\) such that the following holds.
	Let \(D\) be a digraph.
	Let \(k, x, y, q\) be integers.
    If \(\dtw{D} \geq \bound{state:high-dtw-to-segmentation-or-split}{t}{x,y,q,k}\), then \(D\) contains one of the following
	\begin{enamerate}{D}{xitem:high-dtw-to-web:last}
	\item
				a cylindrical grid of order \(k\) as a butterfly minor,
	\item 
				a $(y, q)$-split \(\Brace{\mathcal{P}', \mathcal{Q}'}\) of some pair $(\mathcal{P}_1, \mathcal{Q}_1)$ in \(D\), where \(\End{\mathcal{Q}'}\) is well-linked to \(\Start{\mathcal{Q}'}\), or
	\item
				an $(x, q)$-segmentation $(\mathcal{P}', \mathcal{Q}')$ of some pair $(\mathcal{P}_1, \mathcal{Q}_1)$ in \(D\), where \(\End{\mathcal{P}'}\) is well-linked to \(\Start{\mathcal{P}'}\).
		\label{xitem:high-dtw-to-web:last}
	\end{enamerate}
\end{customthm}

As a cylindrical grid of order \(k\) contains \(k\) pairwise vertex-disjoint cycles, the outcome~\ref{item:high-dtw-to-web:grid} above is sufficient for our results.
We still require this stronger property rather than just providing cycles, since we need it to improve the bounds of the Directed Grid Theorem in~\cite{COSSII}.

\subsection{Obtaining a path of well-linked sets}
\label{sec:obtaining-a-path-of-well-linked-sets}
 
	Several steps of our proofs revolve around linkages which intersect certain subgraphs in an ordered fashion.
	To simplify our arguments and reasoning and to avoid repetition, we abstractly model this configuration with the help of the concept of \emph{temporal digraphs}.

	A \emph{temporal digraph} is a pair $T = (V, \AAA)$ consisting of a vertex set $V$ and sequence of arc sets $\mathcal{A} = \Brace{A_1, A_2, \dots A_{\ell}}$ such that $\Layer{T}{t} \coloneqq \Brace{V, A_t}$ is a digraph for all $1 \leq t \leq \ell.$
	We also refer to $\Layer{T}{t}$ as \emph{layer $t$} of $T$ and call $t$ a \emph{\timestep}.
	The \emph{lifetime} of $T$ is given by $\Lifetime{T} \coloneqq \ell.$
	A \emph{temporal walk} of length $n$ from $v_0$ to $v_n$ in a temporal digraph $T$ is a sequence $W \coloneqq (v_0, t_0), (v_1, t_1), \ldots, (v_n, t_n)$ such that $(v_i, v_{i+1}) \in A_{t_{i}}$ and $t_i < t_{i+1} \leq \Lifetime{T}$ for all $0 \leq i \leq n - 1.$
	If such a walk exists, we say that $v_0$ \emph{temporally reaches} $v_n.$
	A temporal walk is said to be a \emph{temporal path} if no vertex appears twice in the sequence.
	Finally, we say that $W$ \emph{departs} at $t_0$ and \emph{arrives} at $t_n,$ and that $t_n - t_0$ is the \emph{duration} of $W.$

	Usage of temporal digraphs arises naturally in a directed setting.
	Consider the example given in~\cref{short:fig:routing_temporal_digraph} of a temporal digraph obtained from a linkage $\PPP \coloneqq \{ P_a, P_b, P_c\}$ passing through three digraphs $\{Q_1, Q_2, Q_3\}$.
	If we want to construct a new linkage starting and ending in a subset of the starting and endpoints of \(\mathcal{P} \coloneqq \Set{P_a, P_b, P_c}\), then as soon as a path in our new linkage visits a vertex in \(Q_2\), it can no longer use vertices from \(Q_1\), as we only have connectivity from \say{left} to \say{right}.
	Further, as we want some way of ensuring that our paths are disjoint and form a linkage, we are interested in the connectivity each \(Q_i\) provides \say{between} the paths in \(\mathcal{P}\) without intersecting other paths in \(\mathcal{P}\).
	This intuition leads us to the following definition.

\begin{figure*}[ht!]
  \centering
  \resizebox{!}{8cm}{  \includegraphics{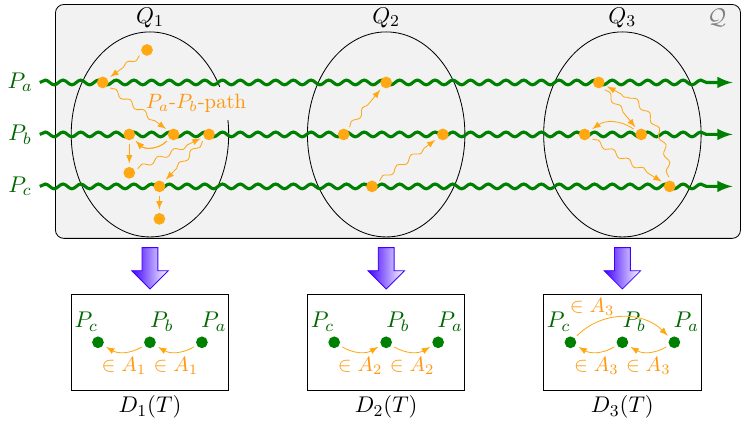}}
  \caption{The layers $D_j(T)$ of the temporal graph $T \coloneqq \Brace{V = \Set{a,b,c}, \mathcal{A} = \Set{A_1,A_2,A_3}}$ constructed from the graphs $Q_j$ displayed above as defined in~\cref{short:def:routing-temporal-digraph}.}
  \label{short:fig:routing_temporal_digraph}
  \label{fig:routing_temporal_digraph}
\end{figure*}

\begin{customdef}{5.3}{def:routing-temporal-digraph}
	\label{short:def:routing-temporal-digraph}
	Let $\mathcal{P}$ be a linkage and let $\mathcal{Q} = \Set{Q_1, Q_2, \dots, Q_{q}}$ be a set of pairwise disjoint digraphs where each path $P_i \in \mathcal{P}$ can be partitioned as $P_i^1 \cdot P_i^2 \cdot \ldots \cdot P_i^{q} = P_i$ such that $V(P_i^j) \cap \V{\mathcal{Q}} \subseteq \V{Q_j}$ for all $1 \leq j \leq q.$
	The \emph{routing temporal digraph $\Brace{V, \mathcal{A}}$ of $\mathcal{P}$ through $\mathcal{Q}$}, which we also refer to as $\RTD{\mathcal{P},\mathcal{Q}}$, is constructed as follows.
	We set $V = \mathcal{P}$ and for each $1 \leq j \leq q$ we define $A_j = \{(P_a, P_b) \mid P_a, P_b \in \mathcal{P}$ and there is a path from \(\V{P_a}\) to \(\V{P_b}\) inside \(Q_j\) which is internally disjoint from \(\mathcal{P}\)\(\}\).
\end{customdef}

In our application of routing temporal digraphs, we want to translate paths (or more general structures we construct) within a routing temporal digraph $T$ of a linkage $\PPP$ through $\QQQ \coloneqq \Set{Q_1, \dots, Q_q}$ into corresponding subgraphs of $\bigcup\PPP \cup \bigcup\QQQ$.
This is made possible by our requirement that a temporal walk in $T$ only uses at most one arc in each layer.
For, if $W$ is a temporal walk $W$ in $T$ and $e = (P_i, P_j) \in \A{W}$ is an arc in layer $A_l$, say, then we can replace $e$ by a path $L(e) \subseteq Q_l$ connecting $P_i$ to $P_j$.
As $W$ contains at most one arc per layer, the paths $L(e)$ and $L(e')$ are pairwise disjoint for distinct arcs $e \neq e'$.
Therefore, the walk $W$ naturally translates into a walk in $\bigcup \PPP\cup \bigcup \QQQ$.

As the example in~\cref{fig:routing_temporal_digraph} demonstrates, this would no longer work if temporal walks were allowed to use more than one arc per layer.
For instance, the digraph $Q_2$ induces the arcs $(P_c, P_b)$ and $(P_b, P_a)$ in the routing temporal digraph $T_2$ of $\PPP$ through $\{Q_2\}$, but the walk $(P_c, P_b, P_a)$ in $T_2$ does not correspond to any $P_c$-$P_a$-walk in $\bigcup\PPP \cup Q_2$.

\textbf{Routings.}

We define a containment relation for digraphs to better describe the kind of connectivity that a routing temporal digraph provides between the paths of the original linkage.
We point out to the reader that connectivity in temporal digraphs is not transitive.
Hence, restricting the following definition to arcs is not sufficient.

\begin{customdef}{5.4}{def:h-routings}
    \label{short:def:h-routings}
	Let $H$ be a digraph, $D$ be a \textcolor{myGreen}{(temporal)} digraph and $S\subseteq \V{D}.$
	An \emph{$H$-routing (over $S$)} is a bijection $\varphi : \V{H} \to S$ such that for each $u$-$v$-path $P$ in $H$ we can find a \textcolor{myGreen}{(temporal)} $\varphi(u)$-$\varphi(v)$-path in $D$ which is disjoint from $S \setminus \Fkt{\varphi}{\V{P}}.$
\end{customdef}

See~\cref{sec:H-routings} for a more detailed discussion of this new containment relation and how it relates to other known ones.

In order to apply the framework defined above, we simplify the notions of splits and segmentations to slightly more general structures we call \emph{ordered} and \emph{folded} webs.

\begin{customdef}{10.1}{def:ordered_web}
	\label{short:def:ordered_web}
	Let $\Brace{\mathcal{H}, \mathcal{V}}$ be an $\Brace{h,v}$-web.
	We say that $\Brace{\mathcal{H}, \mathcal{V}}$ is an \emph{ordered web} if there is an ordering of $\mathcal{V} = \Brace{V_1, V_2, \dots, V_{v}}$ for which each path $H \in \mathcal{H}$ can be decomposed into $H = H_1 \cdot H_2 \cdots H_{v}$ such that $H_i$ intersects $V_j$ if and only if $i = j$.
\end{customdef}

\begin{customdef}{10.5}{def:folded-web}
	\label{short:def:folded-web}
	An $(h,v)$-web $\Brace{\mathcal{H}, \mathcal{V}}$ is a \emph{folded web} if every $V_i \in \mathcal{V}$ can be split as $V_i^a \cdot V_i^b \coloneqq V_i$ such that both $V_i^a$ and $V_i^b$ intersect all paths of $\mathcal{H}$.
\end{customdef}

It is not difficult to show that a $(p,q)$-segmentation \((\mathcal{P}, \mathcal{Q})\) from~\cite{kawarabayashi2022directed} is also an ordered web \((\mathcal{P}, \mathcal{Q})\), and that a $(p,q)$-split \((\mathcal{P}, \mathcal{Q})\) is a folded ordered \((\mathcal{Q}, \mathcal{P})\) web.

Our main results on routing temporal digraphs are about finding \(\Pk{k}\), \(\Ck{k}\) and \(\biPk{k}\)-routings in different contexts.
First, we show that every unilateral temporal digraph with a long enough lifetime contains a walk with many vertices.
The reason why we consider temporal digraphs in which the layers are unilateral is that, given an ordered web \((\mathcal{H}, \mathcal{V})\), each layer of \(\RTD{\mathcal{H}, \mathcal{V}}\) is unilateral.

\begin{customlem}{5.10}{lemma:temporal one-way connected contains walk with many vertices}
    	\label{short:lemma:temporal one-way connected contains walk with many vertices}
    Let $\bound{short:lemma:temporal one-way connected contains walk with many vertices}{\Lifetime{}}{n,k} \coloneqq kn \sum_{i=1}^{kn}n^i$.
	Let $T$ be a temporal digraph with $n$ vertices where each layer is unilateral, and let $S \subseteq \V{T}$ be a set of size $k$.
	If $\Lifetime{T} \geq \bound{short:lemma:temporal one-way connected contains walk with many vertices}{\Lifetime{}}{n,k}$, then $T$ contains a temporal walk $W$ with $S \subseteq \V{W}$.
 \end{customlem}

We can use the walk obtained in~\cref{lemma:temporal one-way connected contains walk with many vertices} to construct a \(\Pk{k}\)-routing.

 \begin{customthm}{5.12}{theorem:one-way connected temporal digraph contains P_k routing}
    	\label{short:theorem:one-way connected temporal digraph contains P_k routing}
	\label{short:thm:one-way connected temporal digraph contains P_k routing}
	Let $\bound{short:theorem:one-way connected temporal digraph contains P_k routing}{\Lifetime{}}{n, k} \coloneqq \bound{short:lemma:temporal one-way connected contains walk with many vertices}{\Lifetime{}}{n, \frac{k^2}{2}}$. 
	Let $T$ be a temporal digraph in which every layer is unilateral.
	If $\Lifetime{T} \geq \bound{short:theorem:one-way connected temporal digraph contains P_k routing}{\Lifetime{}}{n, k}$ and $n \coloneqq \Abs{\V{T}} \geq \frac{k^2}{2} $, then there is some set $S \subseteq \V{T}$ such that $T$ contains a $\Pk{k}$-routing over $S$. 
\end{customthm}

Intuitively, a \(\Pk{k}\)-routing in a routing temporal digraph $\RTD{\mathcal{P},\mathcal{Q}}$ provides connectivity between the paths in \(\mathcal{P}\) which is similar to a column in an acyclic grid, which in turn is related to the concept of order-linkedness defined before.
This intuition is formalised below.

\begin{customlem}{6.6}{lemma:P_k-routing-implies-1-order-linked}
		\label{short:lemma:P_k-routing-implies-1-order-linked}
		Let $h,k$ be integers.
		Let $T$ be the routing temporal digraph of some linkage $\mathcal{L}$ through a sequence $\Brace{H_1, H_2, \dots, H_{h}}$ of disjoint digraphs. 
		Let $\mathcal{L}' \subseteq \mathcal{L}$ be a linkage of order at most $k$.
		If $T$ contains a $\Pk{k}$-routing on the paths $L_1, L_2, \dots, L_{k} \in \mathcal{L}'$, ordered according to their occurrence on the $\Pk{k}$-routing, then $A$ is $1$-order-linked to $B$ in $\ToDigraph{\mathcal{L} \cup \bigcup_{i=1}^{h} H_i}$, where $A = \Set{a_i \mid a_i \text{ is the first vertex of \(L_i\) on \(H_1\)}}$ and $B = \Set{b_i \mid b_i \text{ is the last vertex of \(L_i\) on \(H_h\)}}$. 
	\end{customlem}

It is not difficult to see that, given a folded ordered web \((\mathcal{H}, \mathcal{V})\), each layer of \(\RTD{\mathcal{H}, \mathcal{V}}\) is strongly connected.
Having strongly connected layers allows us to obtain \(\Ck{k}\) or \(\biPk{k}\)-routings instead.
Intuitively, due to these two types of routings providing connectivity in both directions between the paths of our linkage \(\mathcal{P}\), we can use them to obtain well-linked instead of order-linked sets.

\begin{customthm}{5.16}{theorem:strongly connected temporal digraph contains H routing}
	\label{short:theorem:strongly connected temporal digraph contains H routing}
    \label{short:thm-abbr:strongly connected temporal digraph contains H routing}
	Let $T$ be a temporal digraph such that $\Layer{T}{i}$ is strongly connected for all $1 \leq i \leq \Lifetime{T}.$
	If $\Lifetime{T} \geq \bound{short:thm-abbr:strongly connected temporal digraph contains H routing}{\Lifetime{}}{\Abs{\V{T}},k} \in \Oh(k^{11} + k^{5} \Abs{\V{T}}),$ then for every set $S \subseteq \V{T}$ with $\Abs{S} \geq \bound{short:thm-abbr:strongly connected temporal digraph contains H routing}{s}{k} \in \Oh(k^{11})$ there is a subset $S' \subseteq S$ with $\Abs{S'} = k$ such that $D$ contains an $H$-routing over $S'$ for some $H \in \Set{ \Ck{k},  \biPk{k}}.$
\end{customthm}

This allows us to infer the following~\namecref{lemma:routing-implies-well-linked}, stating that if we can find a large enough linkage giving rise to a temporal digraph with long enough lifetime, then we can always find a sublinkage whose start-vertices are well-linked to its end-vertices.

\begin{customprop}{}{lemma:routing-implies-well-linked}
    Let $k$ be an integer, $h \geq \bound{lemma:routing-implies-well-linked}{h}{k} \in \PowerTower{1}{\Polynomial{2}{k}}$, $D$ be a digraph, $\mathcal{L}$ be a linkage of order $\bound{lemma:routing-implies-well-linked}{\ell}{k} \in \Oh(k^{11})$ in $D$ and let $T$ be the routing temporal digraph of $\mathcal{L}$ through $\mathcal{H} \coloneqq \Set{H_1, \dots, H_h}$, where each $H_i$ is a subgraph of $D$.
	If each layer $\Layer{T}{i}$ is strongly connected, then there exists some $\mathcal{L}' \subseteq \mathcal{L}$ of order $k$ such that $\Start{\mathcal{L}'}$ is well-linked to $\End{\mathcal{L}'}$ in $\ToDigraph{\mathcal{L} \cup \mathcal{H}}$.
\end{customprop}

This, in turn, allows us to obtain a path of well-linked sets from a large enough folded ordered web. 
There are functions $\bound{lemma:folded-web-to-pows}{h}{w} \in \Oh(w^{11})$ and $\bound{lemma:folded-web-to-pows}{v}{w, \ell} \in \PowerTower{1}{\Polynomial{2}{w, \ell}}$ such that the following statements hold.

\begin{customlem}{}{lemma:folded-web-to-pows}
    		Let $\Brace{\mathcal{H}, \mathcal{V}}$ be a folded ordered $(h,v)$-web.
		If $h \geq \bound{lemma:folded-web-to-pows}{h}{w}$ and $v \geq \bound{lemma:folded-web-to-pows}{v}{w, \ell}$, then $\Brace{\mathcal{H}, \mathcal{V}}$ contains a path of well-linked sets $\Brace{\mathcal{S} = \Brace{ S_0, S_1, \dots, S_{\ell}}, \mathscr{P}}$ of width $w$ and length $\ell$.
		Additionally, there is a $\Start{\mathcal{H}}$-$\End{\mathcal{H}}$-linkage $\mathcal{L} = \mathcal{L}_1 \cdot \mathcal{L}_2 \cdot \mathcal{L}_3$ using only arcs of $\mathcal{H}$ such that $\mathcal{L}_2$ is an $A(S_0)$-$B(S_\ell)$-linkage of order $w$ inside $\Brace{\mathcal{S}, \mathscr{P}}$ and $\mathcal{L}_1$ and $\mathcal{L}_3$ are internally disjoint from $\Brace{\mathcal{S}, \mathscr{P}}$.
\end{customlem}

\begin{customcor}{}{state:folded-ordered-web-to-well-linked-pows}
	\label{short:state:folded-ordered-web-to-well-linked-pows}
	Let $\Brace{\mathcal{H}, \mathcal{V}}$ be a folded ordered $(h,v)$-web.
	If $h \geq \bound{lemma:folded-web-to-pows}{h}{w}$ and $v \geq \bound{lemma:folded-web-to-pows}{v}{w, \ell}$, then $\Brace{\mathcal{H}, \mathcal{V}}$ contains a path of well-linked sets $\Brace{\mathcal{S} = \Brace{ S_0, S_1, \dots, S_{\ell}}, \mathscr{P}}$ of width $w$ and length $\ell$.
	Additionally, \(A(S_0) \subseteq \Start{\mathcal{H}}\) and \(B(S_\ell) \subseteq \End{\mathcal{H}}\).
\end{customcor}

Similar techniques can be used to obtain a path of order-linked sets from a large ordered web.

\begin{customlem}{}{lemma:ordered-web-to-path-of-order-linked-sets-with-forward-linkage}
				There is a function $\bound{lemma:ordered-web-to-path-of-order-linked-sets-with-forward-linkage}{v}{w,\ell} \in \PowerTower{1}{\Polynomial{13}{\ell, w}}$ such that the following holds.
		Let $\Brace{\mathcal{H}, \mathcal{V}}$ be an ordered $(h,v)$-web where $h = \bound{lemma:ordered-web-to-path-of-order-linked-sets-with-forward-linkage}{h}{w} = w^2 - 1$ and $v \geq \bound{lemma:ordered-web-to-path-of-order-linked-sets-with-forward-linkage}{v}{w,\ell}$.
		Then $(\mathcal{H}, \mathcal{V})$ contains a path of $w$-order-linked sets $\Brace{\mathcal{S} = \Brace{ S_0, S_1, \dots, S_{\ell} }, \mathscr{P}}$ of width $w$ and length $\ell$ with the following additional properties.
		\begin{itemize}
			\item There is a $\Start{\mathcal{H}}$-$\End{\mathcal{H}}$-linkage $\mathcal{L} = \mathcal{L}_1 \cdot \mathcal{L}_2 \cdot \mathcal{L}_3$ of order $w$ contained in $\mathcal{H}$ such that $\mathcal{L}_2$ is an $A(S_0)$-$B(S_\ell)$-linkage and both $\mathcal{L}_1$ and $\mathcal{L}_3$ are internally disjoint from $\Brace{\mathcal{S}, \mathscr{P}}$.
			\item There is a linkage $\mathcal{X} \subseteq \mathcal{V}$ of order $\ell + 1$ and a bijection $\pi : \mathcal{S} \rightarrow \mathcal{X}$ such that $A(S_i) \subseteq \V{\pi(S_i)}$ and $\V{\pi(S_i)} \cap \V{(\mathcal{S}, \mathscr{P})} \subseteq \V{S_i}$ for each $0 \leq i \leq \ell$.
		\end{itemize}
	\end{customlem}

 \begin{customcor}{}{lemma:ordered-web-to-path-of-well-linked-sets-with-side-linkage}
				There are functions \(\bound{lemma:ordered-web-to-path-of-well-linked-sets-with-side-linkage}{h}{w,\ell} \in \Oh(w^{2} \ell^{2})\) and \(\bound{lemma:ordered-web-to-path-of-well-linked-sets-with-side-linkage}{v}{w,\ell} \in \PowerTower{1}{\Polynomial{25}{w, \ell}}\) such that the following holds.
		Let $\Brace{\mathcal{H},\mathcal{V}}$ be an ordered $(h,v)$-web such that $h \geq \bound{lemma:ordered-web-to-path-of-well-linked-sets-with-side-linkage}{h}{w, \ell}$ and $v \geq \bound{lemma:ordered-web-to-path-of-well-linked-sets-with-side-linkage}{v}{w, \ell}$.
		Then, there is a path of well-linked sets $\Brace{\mathcal{S} = \Brace{ S_0, S_1, \dots, S_{\ell} }, \mathscr{P}}$ of width $w$ and length $\ell$ in $\ToDigraph{\mathcal{H} \cup \mathcal{V}}$ such that $B(S_{\ell}) \subseteq \End{\mathcal{H}}$.
		Finally, there is a linkage $\mathcal{X} \subseteq \mathcal{V}$ of order $\ell + 1$ and a bijection $\pi : \mathcal{S} \rightarrow \mathcal{X}$ such that $A(S_i) \subseteq \V{\pi(S_i)}$ and $\V{\pi(S_i)} \cap \V{(\mathcal{S}, \mathscr{P})} \subseteq \V{S_i}$ for each $0 \leq i \leq \ell$.
	\end{customcor}

We can manipulate the path of well-linked sets given by~\cref{lemma:ordered-web-to-path-of-well-linked-sets-with-side-linkage} above in order to ensure that its extremities are contained in the extremities of \(\mathcal{H}\).
This will be useful in the proof of~\cref{thm:high_dtw_to_POSS_plus_back-linkage} when we need the end of the path of well-linked sets to be well-linked to its beginning.

\begin{customlem}{}{state:ordered-web-to-well-linked-pows}
		There are functions \(\bound{state:ordered-web-to-well-linked-pows}{h}{w,\ell} \in \Oh(w^{2} \ell^{2} )\) and \(\bound{state:ordered-web-to-well-linked-pows}{v}{w,\ell} \in \PowerTower{1}{\Polynomial{25}{w, \ell}}\) such that the following holds.
	Let $\Brace{\mathcal{H}, \mathcal{V}}$ be an ordered $(h,v)$-web.
	If $h \geq \bound{state:ordered-web-to-well-linked-pows}{h}{w, \ell}$ and $v \geq \bound{state:ordered-web-to-well-linked-pows}{v}{w, \ell}$, then $\Brace{\mathcal{H}, \mathcal{V}}$ contains a path of well-linked sets $\Brace{\mathcal{S} = \Brace{ S_0, S_1, \dots, S_{\ell}}, \mathscr{P}}$ of width $w$ and length $\ell$.
	Additionally, \(A(S_0) \subseteq \Start{\mathcal{H}}\) and \(B(S_\ell) \subseteq \End{\mathcal{H}}\).
\end{customlem}

We can now use the results above to prove~\cref{short:thm:high_dtw_to_POSS_plus_back-linkage}.

\begin{proof}[Proof sketch of~\cref{short:thm:high_dtw_to_POSS_plus_back-linkage}]
	By~\cref{state:high-dtw-to-segmentation-or-split}, we have one of three cases.
	If $D$ contains a cylindrical grid of order $w \ell$, it contains a \powls of width $w$ and length $\ell$.

	If $D$ contains a \((y,q)\)-split \((\mathcal{V}, \mathcal{H})\), then this split is essentially a folded ordered web \((\mathcal{H}, \mathcal{V})\).
	By~\cref{short:state:folded-ordered-web-to-well-linked-pows}, we can construct a path of well-linked sets from this split.
	Moreover, the beginning and the end of this path of well-linked sets coincide with \(\Start{\mathcal{V}}\) and \(\End{\mathcal{V}}\), respectively.
	As \(\End{\mathcal{V}}\) is well-linked to \(\Start{\mathcal{V}}\), the path of well-linked sets obtained satisfies the restrictions in the statement.

	In the last case, $D$ contains an \((x,q)\)-segmentation \((\mathcal{H}, \mathcal{V})\), which also means that \((\mathcal{H}, \mathcal{V})\) is an ordered web.
	Applying~\cref{state:ordered-web-to-well-linked-pows} to \((\mathcal{H}, \mathcal{V})\) yields a path of well-linked sets.
	Moreover, the beginning and the end of this path of well-linked sets coincide with \(\Start{\mathcal{H}}\) and \(\End{\mathcal{H}}\), respectively.
	As \(\End{\mathcal{H}}\) is well-linked to \(\Start{\mathcal{H}}\), the path of well-linked sets obtained satisfies the restrictions in the statement.
\end{proof}

To obtain our final main result, we first need a statement from \cite{reed1996packing}.

\begin{customlem}{}{statement:fence-plus-back-linkage-implies-cycles}[{\cite[Statement (3.1)]{reed1996packing}}]
		There are polynomials \(\bound{statement:fence-plus-back-linkage-implies-cycles}{p}{k}\) and \(\bound{statement:fence-plus-back-linkage-implies-cycles}{r}{k}\) such that the following holds.
	Let \(k, p, q, r \geq 1\) be integers.
	Let \((\mathcal{P}, \mathcal{Q})\) be a \((p, q)\)-fence in a digraph \(D\) where \(q \geq \bound{statement:fence-plus-back-linkage-implies-cycles}{q}{k}\) and \(p \geq r\), and let \(\mathcal{R}\) be an \(\End{\mathcal{P}}\)-\(\Start{\mathcal{P}}\) linkage of order \(r \geq \bound{statement:fence-plus-back-linkage-implies-cycles}{r}{k}\).
	Then, \(D\) contains \(k\) pairwise vertex-disjoint cycles.
\end{customlem}

Together with our previous results,~\cref{statement:fence-plus-back-linkage-implies-cycles} allows us to handle the case where the digraph has large directed \treewidth.

\begin{customlem}{}{statement:high-dtw-implies-cycles}
		There is a function \(\bound{statement:high-dtw-implies-cycles}{t}{k} \in \PowerTower{8}{\Polynomial{43}{k}}\) such that the following holds.
	If a digraph \(D\) has \(\DTreewidth{D} \geq \bound{statement:high-dtw-implies-cycles}{t}{k}\), then \(D\) contains \(k\) pairwise vertex-disjoint cycles.
\end{customlem}

For the small directed \treewidth case, we need the following statement from~\cite{amiri2016erdos}.

\begin{customlem}{}{lem:akkw2}[{\cite[Lemma 4.2]{amiri2016erdos}}]
        Let $D$ be a digraph with $\dtw{D} \leq w$.
    For every strongly connected digraph $H$, the digraph $D$ either has $k$ disjoint copies of $H$ as a topological minor, or contains a set $T$ of at most $k \cdot (w+1)$ vertices such that $H$ is not a topological minor of $D - T$. 
\end{customlem}

We can now conclude with our final main result.

\begin{customthm}{}{statement:elementary-younger}
	There is a function \(\bound{statement:elementary-younger}{f}{k} \in \PowerTower{8}{\Polynomial{13}{k}}\) such that the following holds.
	Let \(D\) be a digraph.
	Then \(D\) has \(k\) pairwise vertex-disjoint cycles or there is some \(X \subseteq \V{D}\) of size at most \(\bound{statement:elementary-younger}{f}{k}\) such that \(D - X\) is acyclic.
\end{customthm}
\begin{proof}
	Let \(t_1 = \bound{statement:high-dtw-implies-cycles}{t}{k}\).
	Let \(H = \Ck{2}\).
	If \(\dtw{D} \leq t_1 - 1\), then, by~\cref{lem:akkw2}, \(D\) contains \(k\) vertex-disjoint copies of \(H\) as a topological minor or \(D\) contains a set \(X \subseteq \V{D}\) of size at most \(kt_1\) such that \(D - X\) does not contain \(H\) as a topological minor.
	By choice of \(H\), both cases imply the initial statement, and we are done.

	Otherwise, \(\dtw{D} \geq t_1\).
	By~\cref{statement:high-dtw-implies-cycles}, \(D\) contains \(k\) pairwise vertex-disjoint cycles.
\end{proof}

\section{Constructing splits and segmentations}
\label{section:constructing-web}

The construction of our paths of well-linked sets and paths of order-linked sets starts with \emph{splits} and \emph{segmentations} (see~\cref{def:split-edge}), which add more structure to webs by ensuring that one linkage of the web intersects the other in an ordered fashion.
We repeat below the definition of splits and segmentations by Kawarabayashi and Kreutzer~\cite{kawarabayashi2022directed}.

\begin{definition}[{\hspace{1sp}\cite[Definitions 5.6 and 5.7]{kawarabayashi2022directed}}]
	\label{def:split-edge}
    Let $\PPP$ and $\QQQ^\star$ be linkages and let $\QQQ\subseteq \QQQ^\star$ be a sublinkage of order $q$.
    Let $r\geq 0$.
	\begin{enamerate}{X}{xlast-item-split-edge}
    \item \label{def:split} 
        An \emph{$(r, q')$-split of $(\PPP, \QQQ)$ (with respect to $\QQQ^\star$)} is a pair $(\PPP', \QQQ')$ of linkages of order $r = |\PPP'|$ and $q'=|\QQQ'|$ with $\QQQ'\subseteq \QQQ$ such that there is a path $P \in \PPP$ and edges $e_1, \dots, e_{r-1}\in E(P)\setminus E(\QQQ^\star)$ such that $P = P_1e_1P_2\dots e_{r-1} P_r$ and $\PPP' \coloneqq (P_1, \dots, P_r)$ and every $Q\in \QQQ'$ can be divided into subpaths $Q_1, \dots, Q_r$ such that $Q = Q_1e'_1\dots e'_{r-1}Q_r$, for suitable edges $e'_1, \dots, e'_{r-1} \in E(Q)$, and $\emptyset \neq V(Q)\cap V(P_i) \subseteq V(Q_{r+1-i})$, for all $1\leq i \leq r$.
	\item \label{def:other-segmentation}
		A subset $\QQQ' \subseteq \QQQ$ of order $q'$ is called a $q'$-\emph{segmentation} of $P \in \mathcal{P}$ (with respect to $\QQQ^\star$) if there are edges $e_1, \dots, e_{q'-1}\in \A{P} \setminus \A{\QQQ^\star}$ with $P = P_1e_1 \dots P_{q'-1}e_{q'-1}P_{q'}$, for suitable subpaths $P_1, \dots, P_{q'}$, such that $\QQQ'$ can be ordered as $(Q_1, \dots, Q_{q'})$ and $V(Q_i)\cap V(P)\subseteq V(P_i)$.
	\item \label{def:segmentation}
        An \emph{$(r, q')$-segmentation} (with respect to \(\mathcal{Q}^\star\)) is a pair $(\PPP', \QQQ')$ where $\PPP'$ is a linkage of order $r$ and $\QQQ'$ is a linkage of order $q'$ such that $\QQQ'$ is a $q'$-segmentation (with respect to \(\mathcal{Q}^\star\)) of every path $P_i$ into segments $P^i_1e_1P^i_2\dots e_{q'-1}P^i_{q'}$.
    \item \label{def:ordered_segmentation}
        A segmentation $(\PPP', \QQQ')$ is \emph{ordered} if for all $P_i \in \PPP'$ the order $\Brace{Q_1,\dots,Q_{q'}}$ given by the $q'$-segmentation of $P_i$ is the same.
        We say that $(\PPP', \QQQ')$ is an (ordered) $(r, q')$-segmentation of $(\PPP, \QQQ)$ if $\QQQ' \subseteq \QQQ$ and every path in $\PPP'$ is a subpath of a path in $\PPP$.\label{xlast-item-split-edge}
  \end{enamerate}
  An $(r, q)$-split $(\PPP, \QQQ)$ or an $(r, q)$-segmentation $(\PPP, \QQQ)$ is well-linked if $\End{\mathcal{Q}}$ is well-linked to $\Start{\mathcal{Q}}$.
\end{definition}

One can obtain splits and segmentations from weakly-minimal webs using the following result.
We observe that \(\bound{lemma:y-split or x-segmentation}{q}{p,q,x,y,c} \in \PowerTower{2}{\Polynomial{4}{p,q,x,y,c}}\).
\begin{lemma}[{\cite[Lemma~5.13]{kawarabayashi2015directed}}]
	\label{lemma:y-split or x-segmentation}
	Let $p,q,q',r,s,c,x,y \geq 0$ be integers such that $p \geq x$ and $q' \geq \bound{lemma:y-split or x-segmentation}{q}{p,q,x,y,c} \coloneqq (pq(q + c))^{2^{(x - 1)y + 1}}$\boundDef{lemma:y-split or x-segmentation}{q}{p,q,x,y,c}.
	If $D$ contains a $(p, q')$-web $\mathcal{W} \coloneqq (\mathcal{P}, \mathcal{Q})$ where $\mathcal{P}$ is weakly $c$-minimal with respect to $\mathcal{Q}$, then $D$ contains one of the following:
	\begin{enamerate}{S}{item:y-split or x-segmentation:segmentation}
	\item \label{item:y-split or x-segmentation:split}
		a $(y,q)$-split $(\mathcal{P}', \mathcal{Q}')$ of $(\mathcal{P}, \mathcal{Q})$, or
	\item \label{item:y-split or x-segmentation:segmentation}
		an $(x,q)$-segmentation $(\mathcal{P}', \mathcal{Q}')$ of $(\mathcal{P}, \mathcal{Q})$.
	\end{enamerate} 
	Furthermore, if $\mathcal{W}$ is well-linked in $D$, then so is $(\mathcal{P}', \mathcal{Q}')$.
\end{lemma}

With a simple pigeon-hole principle argument, it is possible to construct an ordered segmentation from a segmentation.
The proof provided below is based on some steps of the proof of Lemma 5.19 from~\cite{kawarabayashi2015directed}.
\begin{observation}
    \label{lemma:ordered x-segmentation}
    Let $(\mathcal{P}, \mathcal{Q})$ be a $(p, q)$-segmentation.
    If $p \geq \bound{lemma:ordered x-segmentation}{p}{k, q} \coloneqq (k - 1) q! + 1$\boundDef{lemma:ordered x-segmentation}{p}{k, q}, then there is $\mathcal{P}' \subseteq \mathcal{P}$ such that $(\mathcal{P}', \mathcal{Q})$ is an ordered $(k, q)$-segmentation.
\end{observation}
\begin{proof}
	For each $P_i \in \mathcal{P}$ there is an ordering $\mathcal{Q}_i$ of $\mathcal{Q}$ witnessing that $\mathcal{Q}$ is a $q$-segmentation of $P_i$.
	There are at most $q!$ distinct orderings $\mathcal{Q}_i$.
	Hence, by the pigeon-hole principle, there is some $\mathcal{P}' \subseteq \mathcal{P}$ of size $k$ such that $\mathcal{Q}_i = \mathcal{Q}_j$ for all $P_i, P_j \in \mathcal{P}'$.
\end{proof}

Since the bounds for~\cref{lemma:y-split or x-segmentation} are already elementary, in the remainder of this section we focus on obtaining a web while ensuring all the functions that arise are elementary.

We thereby improve upon the results of~\cite{kawarabayashi2022directed}, who show that digraphs of large directed \treewidth contain a large \emph{well-linked web} or a large cylindrical grid.
The bounds they obtain in this step are, however, non-elementary.
In particular, their proof uses an \emph{iterated Ramsey} argument, and so the bounds obtained are a power tower whose height depends on $h, v$ and $k$ (where $h, v$ and $k$ are defined as in~\cref{thm:dtw_to_web}).
\begin{theorem}[{\cite[Theorem 4.2 + Lemma 3.6 + Lemma 4.10]{kawarabayashi2022directed}}] 
    \label{thm:dtw_to_web}
    Let $h,v,k \in \N$.
    There exists a function $\bound{thm:dtw_to_web}{t}{}: \N \times \N \times \N \to \N$ such that every digraph $D$ with $\dtw{D} \geq \bound{thm:dtw_to_web}{t}{h,v,k}$ contains a cylindrical grid of order $k$ as a butterfly minor or a $\Brace{h,v}$-web $\Brace{\mathcal{H},\mathcal{V}}$ where $\End{\mathcal{V}} \cup \Start{\mathcal{V}}$ is a well-linked set in $D$ and $\mathcal{H}$ is $\mathcal{V}$-minimal.
\end{theorem}

To achieve an elementary bound, we utilise a result known as \emph{Lovász Local Lemma},
obtaining even a polynomial bound for a previously non-elementary step.
In total, however, the gap between directed \treewidth and the size of the ordered web we can guarantee is upper bounded by a super-polynomial function.
Recall that \(e = 2.71828\ldots\) is \emph{Euler's number}.

\begin{lemma}[Lovász Local Lemma \cite{EL1974,SPENCER197769}]
	\label{state:lovasz-local-lemma}
	Consider a set \(\mathcal{E}\) of events such that for each \(A \in \mathcal{E}\)
	\begin{itemize}
		\item \(\Pr{A} \leq p < 1\), and
		\item \(A\) is mutually independent of a set of all but at most \(d\) other events.
	\end{itemize}
	If \(ep(d + 1) < 1\), then with positive probability, none of the events in \(\mathcal{E}\) occur.
\end{lemma}

While~\cref{state:lovasz-local-lemma} above is not constructive, \cite{MoserT10} provided a randomised algorithm and \cite{CGH2013deterministic} provided a deterministic algorithm for finding an assignment of the random variables which avoids all events $\mathcal{E}$.

The proof of~\cite{kawarabayashi2022directed} for obtaining a web works by starting with a \emph{bramble} of high order and then showing that the existence of such a bramble implies the existence of an object called a \emph{\pathsystem}.
We repeat the definition of \pathsystem below (see~\cref{fig:clean-path-system} for an illustration).

\begin{definition}[\pathsystem \cite{kawarabayashi2022directed}]
	\label{def:path-system}
	Let $D$ be a digraph and let $\ell, p \geq 1$.
	An \emph{$\ell$-linked \pathsystem of order $p$} is a sequence $\SSS \coloneqq (\PPP, \LLL, \AAA)$, where
	\begin{itemize}
		\item  $\AAA \coloneqq \big(A_i^{\variablestyle{in}}, A_i^{\variablestyle{out}}\big)_{1\leq i\leq p}$ such that $|A_i^{\variablestyle{in}}| = |A_i^{\variablestyle{out}}| = \ell$ and $A \coloneqq \bigcup_{1\leq i \leq p} A_i^{\variablestyle{in}} \cup A_i^{\variablestyle{out}} \subseteq V(D)$ is a well-linked set of order $2 \ell p$, for all $1 \leq i \leq p$,
		\item $\PPP \coloneqq (P_1, \dots, P_p)$ is a sequence of pairwise vertex-disjoint paths and for all $1\leq i \leq p$, $A_i^{\variablestyle{in}}, A_i^{\variablestyle{out}} \subseteq V(P_i)$ and all $v\in A^{\variablestyle{in}}_i$ occur on $P_i$ before any $v'\in A_i^{\variablestyle{out}}$ and the first vertex of $P_i$ is in $A_i^{\variablestyle{in}}$ and the last vertex of $P_i$ is in $A_i^{\variablestyle{out}}$ and
		\item $\LLL \coloneqq (L_{i,j})_{1\leq i\not= j \leq p}$ is a sequence of linkages such that for all $1\leq i\not=j\leq p$, $L_{i,j}$ is a linkage of order $\ell$ from $A_i^{\variablestyle{out}}$ to $A_j^{\variablestyle{in}}$.
	\end{itemize}
	The \pathsystem $\SSS$ is \emph{clean} if $Q \cap P_s = \emptyset$ holds for all pairwise distinct $i,j,s \in \Set{1, 2, \ldots, p}$ and all $Q \in L_{i,j}$.
\end{definition}

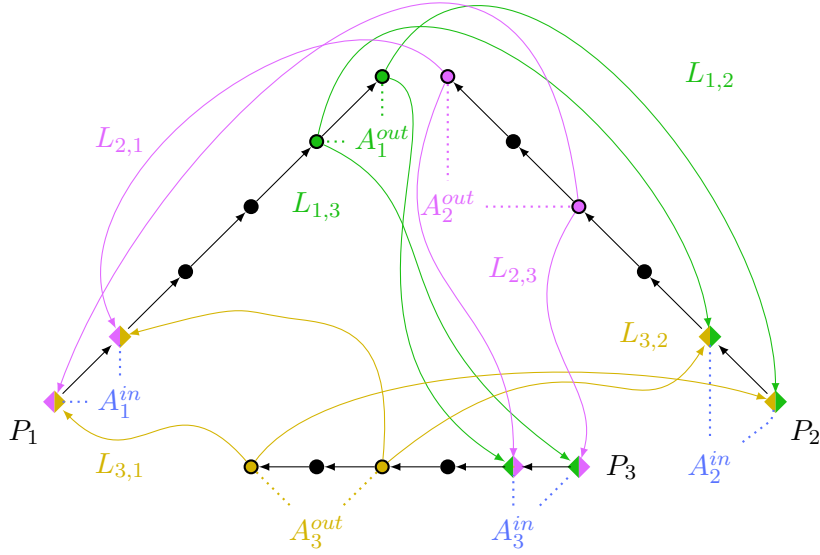
\begin{figure}[!ht]
	\centering
	\vspace{-3cm}	
	\begin{tikzpicture}[scale=0.75]
		\node[rectangle]
	(n0) at (1.15, 2.3) [align = left]{\color{blue}{$A_1^{\variablestyle{in}}$}};
\node[rectangle]
	(n11) at (5.75, 6.9) [align = left]{\color{green}{$A_1^{\variablestyle{out}}$}};
\node[rectangle]
	(n14) at (6.9, 5.75) [align = left]{\color{purple}{$A_2^{\variablestyle{out}}$}};
\twocolouredVertex{n15}{0,2.3}{purple}{yellow}
\node (n15-label) at ($(n15)+(225:0.75)$) {\(P_1\)};
\twocolouredVertex{n16}{1.15,3.45}{purple}{yellow}
\node[vertex, fill = black]
	(n23) at (2.3, 4.6){};
\node[vertex,  fill = black]
	(n27) at (3.45, 5.75){};
\node[vertex,  fill = green]
	(n28) at (4.6, 6.9){};
\node[vertex,  fill = green]
	(n29) at (5.75, 8.05){};
\twocolouredVertex{n1}{12.65,2.3}{yellow}{green}
\node (n1-label) at ($(n1)+(-45:0.75)$) {\(P_2\)};
\twocolouredVertex{n2}{11.5,3.45}{yellow}{green}
\node[vertex,  fill = black]
	(n3) at (10.35, 4.6){};
\node[vertex,  fill = purple]
	(n4) at (9.2, 5.75){};
\node[vertex,  fill = black]
	(n5) at (8.05, 6.9){};
\node[vertex,  fill = purple]
	(n6) at (6.9, 8.05){};
\node[vertex,  fill = yellow]
	(n7) at (3.45, 1.15){};
\node[vertex,  fill = black]
	(n8) at (4.6, 1.15){};
\node[vertex,  fill = yellow]
	(n9) at (5.75, 1.15){};
\node[vertex,  fill = black]
	(n10) at (6.9, 1.15){};
\twocolouredVertex{n12}{8.05,1.15}{green}{purple}
\twocolouredVertex{n13}{9.2,1.15}{green}{purple}
\node (n13-label) at ($(n13)+(0:0.75)$) {$P_3$};
\node[]
	(n17) at (4.6, 5.75) [align = left]{\color{green}{$L_{1,3}$}};
\node[]
	(n18) at (11.5, 8.05) [align = left]{\color{green}{$L_{1,2}$}};
\node[]
	(n19) at (1.15, 1.15) [align = left]{\color{yellow}{$L_{3,1}$}};
\node[]
	(n20) at (10.35, 3.45) [align = left]{\color{yellow}{$L_{3,2}$}};
\node[]
	(n21) at (1.15, 6.9) [align = left]{\color{purple}{$L_{2,1}$}};
\node[]
	(n22) at (8.05, 4.6) [align = left]{\color{purple}{$L_{2,3}$}};
\node[rectangle,  fill = white]
	(n24) at (11.5, 1.15) [align = left]{\color{blue}{$A_2^{\variablestyle{in}}$}};
\node[rectangle,  fill = white]
	(n25) at (8.05, 0) [align = left]{\color{blue}{$A_3^{\variablestyle{in}}$}};
\node[rectangle,  fill = white]
	(n26) at (4.6, 0) [align = left]{\color{yellow}{$A_3^{\variablestyle{out}}$}};
\path[-latex,  draw = black]
	(n15) to (n16);
\path[-latex,  draw = black]
	(n16) to (n23);
\path[-latex,  draw = black]
	(n23) to (n27);
\path[-latex,  draw = black]
	(n27) to (n28);
\path[-latex,  draw = black]
	(n28) to (n29);
\path[-latex,  draw = black]
	(n13) to (n12);
\path[-latex,  draw = black]
	(n12) to (n10);
\path[-latex,  draw = black]
	(n10) to (n9);
\path[-latex,  draw = black]
	(n9) to (n8);
\path[-latex,  draw = black]
	(n8) to (n7);
\path[-latex,  draw = black]
	(n1) to (n2);
\path[-latex,  draw = black]
	(n2) to (n3);
\path[-latex,  draw = black]
	(n3) to (n4);
\path[-latex,  draw = black]
	(n4) to (n5);
\path[-latex,  draw = black]
	(n5) to (n6);
\path[-latex,  draw = green]
	(n28) .. controls (7.078, 5.763) and (5.423, 3.885) .. (n13);
\path[-latex,  draw = green]
	(n29) .. controls (7.255, 7.669) and (4.463, 3.39) .. (n12);
\path[-latex,  draw = yellow]
	(n9) .. controls (5.951, 3.334) and (5.252, 3.64) .. (4.319, 3.752) .. controls (3.139, 3.893) and (3.543, 4.105) .. (n16);
\path[-latex,  draw = yellow]
	(n7) .. controls (2.019, 2.828) and (1.484, 0.823) .. (n15);
\path[-latex,  draw = green]
	(n28) .. controls (5.632, 11.672) and (11.083, 6.733) .. (n2);
\path[-latex,  draw = green]
	(n29) .. controls (7.729, 11.8) and (12.643, 5.854) .. (n1);
\path[-latex,  draw = yellow]
	(n9) .. controls (9.316, 4.13) and (10.213, 1.374) .. (n2);
\path[-latex,  draw = yellow]
	(n7) .. controls (4.824, 3.05) and (10.1, 3.116) .. (n1);
\path[-latex,  draw = purple]
	(n4) .. controls (7.795, 3.562) and (9.65, 2.699) .. (n13);
\path[-latex,  draw = purple]
	(n6) .. controls (5.262, 4.291) and (8.068, 4.393) .. (n12);
\path[-latex,  draw = purple]
	(n6) .. controls (4.979, 10.148) and (-0.259, 6.25) .. (n16);
\path[-latex,  draw = purple]
	(n4) .. controls (8.294, 13.366) and (1.695, 7.08) .. (n15);
\path[ draw = blue, dotted, thick]
	(n2) to (n24);
\path[ draw = blue, dotted, thick]
	(n24) to (12.576, 1.972) to (n1);
\path[ draw = yellow, dotted, thick]
	(n7) to (n26);
\path[ draw = yellow, dotted, thick]
	(n9) to (n26);
\path[ draw = blue, dotted, thick]
	(n25) to (n12);
\path[ draw = blue, dotted, thick]
	(n25) to (n13);
\path[ draw = blue, dotted, thick]
	(n15) to (n0);
\path[ draw = blue, dotted, thick]
	(n0) to (n16);
\path[ draw = green, dotted, thick]
	(n11) to (n28);
\path[ draw = green, dotted, thick]
	(n11) to (n29);
\path[ draw = purple, dotted, thick]
	(n6) to (n14);
\path[ draw = purple, dotted, thick]
	(n14) to (n4);

	\end{tikzpicture}
	\caption{A clean \(2\)-linked \pathsystem of order 3.}
	\label{fig:clean-path-system}
\end{figure}

A \pathsystem can be obtained from a bramble using the following lemma.
We define \(\bound{state:bramble-contains-linked-path-system}{k}{\ell, p} = (4 \ell p)(2 \ell p + 1)\)\boundDef{state:bramble-contains-linked-path-system}{k}{\ell, p} and observe that \(\bound{state:bramble-contains-linked-path-system}{k}{\ell,p} \in \Oh( \ell^{2} p^{2})\).
\begin{lemma}[{\cite[Lemma 4.6]{kawarabayashi2022directed}}]
	\label{state:bramble-contains-linked-path-system}
	Let \(D\) be a digraph and let \(\ell, p \geq 1\).
	If \(D\) contains a bramble of order \(\bound{state:bramble-contains-linked-path-system}{k}{\ell, p}\), then \(D\) contains an \(\ell\)-linked \pathsystem \(\mathcal{S}\) of order \(p\).
\end{lemma}

\subsection{Semi-webs}

Our proof uses an intermediate object called a \emph{semi-web}.
Unlike a web, we do not require from a semi-web \((\mathcal{P}, \mathcal{Q})\) that the paths in \(\mathcal{P}\) and \(\mathcal{Q}\) pairwise intersect each other.
Instead, a \emph{semi-web} has a degree controlling the intersections between \(\mathcal{P}\) and \(\mathcal{Q}\).

A web can be obtained from a semi-web using~\cref{state:semi-web-to-web} below.
Towards this end, we use~\cref{state:preserve-minimality-deletion,state:preserve-minimality-disjoint}, which summarise some basic properties of a linkage \(\mathcal{L}\) which is minimal with respect to another linkage \(\mathcal{P}\).

\begin{observation}[{\cite[Lemma 2.14]{kawarabayashi2022directed}}]
	\label{state:preserve-minimality-deletion}
	Let \(D\) be a digraph.
	Let \(\mathcal{P}, \mathcal{L}\) be linkages such that \(\mathcal{L}\) is minimal with respect to \(\mathcal{P}\).
	Then \(\mathcal{L}\) is minimal with respect to \(\mathcal{P}'\) for every \(\mathcal{P}' \subseteq \mathcal{P}\).
\end{observation}

\begin{observation}
	\label{state:preserve-minimality-disjoint}
	Let \(\mathcal{P}\), \(\mathcal{Q}\) be two linkages such that \(\mathcal{P}\) is minimal with respect to \(\mathcal{Q}\).
	Let \(P \in \mathcal{P}\).
	If \(P\) does not intersect any path in \(\mathcal{Q}\), then \(\mathcal{P} \setminus \{P\}\) is minimal with respect to \(\mathcal{Q}\).
\end{observation}

We adapt the following statement from~\cite{kawarabayashi2022directed} to our notation, while also fixing some minor mistakes in their proof.
We first define
\begin{align*}
			\boundDefAlign{state:semi-web-to-web}{q}{p', q, k}
	\bound{state:semi-web-to-web}{q}{p', q, k} & = q k (p')^k.
\end{align*}

\begin{lemma}[{\cite[Lemma 4.10]{kawarabayashi2022directed}}]
	\label{state:semi-web-to-web}
	Let \(\Brace{\mathcal{P}', \mathcal{Q}'}\) be a \(\Brace{p',q'}\)-semi-web of degree \(k\) in a digraph \(D\) such that \(\mathcal{P}'\) is minimal with respect to \(\mathcal{Q}'\).
	If \(q' \geq \bound{state:semi-web-to-web}{q}{p', q, k}\), then \(D\) contains a well-linked \(\Brace{p_1,q}\)-web \(\Brace{\mathcal{P}, \mathcal{Q}}\) where
	\(\mathcal{P}\) is minimal with respect to \(\mathcal{Q}\),
	\(\mathcal{P} \subseteq \mathcal{P}'\),
	\(\mathcal{Q} \subseteq \mathcal{Q}'\) and
	\(p_1 \geq k\).
\end{lemma}
\begin{proof}
	Define a function \(f\) as \(f(Q) = \Set{P \in \mathcal{P}' \mid \V{P} \cap \V{Q} = \emptyset}\) for each \(Q \in \mathcal{Q}'\).
	By assumption, \(\Abs{f(Q)} \leq p' - k\) for each \(Q \in \mathcal{Q}'\).

	As there are \(\binom{p'}{\Abs{f(Q)}}\) choices for each \(f(Q)\), and as \(\sum_{i=0}^k \binom{p'}{p'-k} = \sum_{i=0}^k \binom{p'}{k} \leq k (p')^k \), by the pigeon-hole principle there is some \(\mathcal{Q} \subseteq \mathcal{Q}'\) of order \(q\) such that \(X \coloneqq f(Q_a) = f(Q_b)\) holds for all \(Q_a, Q_b \in \mathcal{Q}\).
	
	Let \(\mathcal{P} = \mathcal{P}' \setminus X\) and let \(p_1 = \Abs{\mathcal{P}}\).
	Note that \(p_1 \geq k\).
	By~\cref{state:preserve-minimality-deletion,state:preserve-minimality-disjoint}, \(\mathcal{P}\) is minimal with respect to \(\mathcal{Q}\).
	Hence, \(\Brace{\mathcal{P}, \mathcal{Q}}\) is a \((p_1, q)\)-web where \(\mathcal{P}\) is minimal with respect to \(\mathcal{Q}\), as desired.
\end{proof}

\Cref{state:path-system-to-clean-path-system} essentially corresponds to Lemma 4.7 from~\cite{kawarabayashi2022directed}, and our proof is based on theirs.
The idea is to attempt to construct a semi-web from some linkage \(L_{a,b} \in \mathcal{L}\) and some \(P \in \mathcal{P}\).
If we do not find any semi-web, it means that the linkages in \(\mathcal{L}\) are mostly disjoint from \(\mathcal{P}\).
We then use this observation to argue that, for each pair \(P_i, P_j \in \mathcal{P}\), there are only few other paths \(P_r \in \mathcal{P}\) which are \say{bad} for the choice \(P_i, P_j\), that is, we cannot easily construct a clean \pathsystem if we take \(P_i, P_j\) and \(P_r\).
This allows us to construct our \say{bad} events in order to apply Lovász Local Lemma.

We define
\begin{align*}
	\Func{d'}{p_2} & = 3(p_2)^2/2 - 15p_2/2 + 10,
	\\[0em]
	\boundDefAlign{state:path-system-to-clean-path-system}{\ell}{p_2, \ell_2, d_1}
	\bound{state:path-system-to-clean-path-system}{\ell}{p_2, \ell_2, d_1} & = \ell_2 + (p_2 - 2)d_1,
	\\[0em]
	\boundDefAlign{state:path-system-to-clean-path-system}{p}{q_1, p_2}
	\bound{state:path-system-to-clean-path-system}{p}{q_1, p_2} & = (2e(\Func{d'}{p_2} + 1)q_1 + 1)p_2.
\end{align*}
Note that \(\bound{state:path-system-to-clean-path-system}{\ell}{p_2,\ell_2,d_1} \in \Oh(\ell_{2} + d_{1} p_{2})\) and \(\bound{state:path-system-to-clean-path-system}{p}{q_1,p_2} \in \Oh(q_1 (p_2)^3)\).
\begin{lemma}
	\label{state:path-system-to-clean-path-system}
	Let \(d_1, q_1, p_2, \ell_2\) be integers.
	Let \(\mathcal{S} = (\mathcal{P}, \mathcal{L}, \mathcal{A})\) be an \(\ell\)-linked \pathsystem of order \(p\).
	If \(\ell \geq \bound{state:path-system-to-clean-path-system}{\ell}{p_2, \ell_2, d_1}\) and \(p \geq \bound{state:path-system-to-clean-path-system}{p}{q_1, p_2}\), then \(\ToDigraph{\mathcal{S}}\) contains one of the following
	\begin{enamerate}{C}{item:path-system-to-clean-path-system:last}
	\item a well-linked \(\Brace{\ell, q_1}\)-semi-web \(\Brace{\mathcal{P}_1, \mathcal{Q}_1}\) of degree \(d_1\) where \(\mathcal{P}_1 \in \mathcal{L}\), \(\mathcal{Q}_1 \subseteq \mathcal{P}\) and \(\mathcal{P}_1\) is minimal with respect to \(\mathcal{Q}_1\), or
		\label{item:path-system-to-clean-path-system:web}
	\item a clean \(\ell_2\)-linked \pathsystem \((\mathcal{P}_2, \mathcal{L}_2, \mathcal{A}_2)\) of order \(p_2\).
		\label{item:path-system-to-clean-path-system:clean}
		\label{item:path-system-to-clean-path-system:last}
	\end{enamerate}
\end{lemma}
\begin{proof}
	Let \(d = 3(p_2)^2/2 - 15p_2/2 + 10\) and \(p_2' = \ceil{2eq_1(d + 1)}\) + 1.

	First, if there is some \(\mathcal{L}_{s,t} \in \mathcal{L}\) such that \(\mathcal{L}_{s,t}\) is not minimal with respect to \(\mathcal{P}\), then replace such a linkage with another \(\Start{\mathcal{L}_{s,t}}\)-\(\End{\mathcal{L}_{s,t}}\)-linkage of the same order which is minimal with respect to \(\mathcal{P}\).
	This does not alter the fact that \(\mathcal{S}\) is an \(\ell\)-linked \pathsystem of order \(p\).
	Further, if \(p_2 = 1\), we can trivially obtain a clean \(\ell_2\)-linked \pathsystem satisfying \ref{item:path-system-to-clean-path-system:clean} by taking any path in \(\mathcal{P}\) and setting \(\mathcal{L} = \emptyset\).
	Hence, we can assume that \(p_2 \geq 2\).

	Define a function \(\gamma\) as follows.
	For each distinct \(P_s, P_t \in \mathcal{P}\) we set (see~\cref{fig:gamma-set} for an illustration) 
	\begin{align*}
		\gamma(P_s, P_t) = \{P \in \mathcal{P} \setminus \Set{P_s, P_t} \ST &~ P \text{ intersects at least \(d_1\) paths of \(\mathcal{L}_{s,t}\) or}
		\\[0em] &~ P \text{ intersects at least \(d_1\) paths of \(\mathcal{L}_{t,s}\)}\}.
	\end{align*}

\begin{figure}[!ht]
	\begin{minipage}[b]{0.5\textwidth}
		\centering
		\begin{tikzpicture}[scale=0.75]
			\node[vertex, label = above:{\color{black}{$P_s$}}, line width = 1, fill = black]
	(n0) at (0, 8.05){};
\node[vertex, line width = 1, fill = green]
	(n11) at (1.15, 8.05){};
\node[vertex, line width = 1, fill = green]
	(n16) at (3.45, 6.9){};
\node[vertex, line width = 1, fill = green]
	(n25) at (4.6, 5.75){};
\node[vertex, line width = 1, fill = green]
	(n26) at (2.3, 5.75){};
\node[vertex, line width = 1, fill = green]
	(n27) at (5.75, 6.9){};
\node[vertex, line width = 1, fill = green]
	(n33) at (6.9, 8.05){};
\node[vertex, line width = 1, fill = red]
	(n34) at (1.15, 1.15){};
\node[vertex, line width = 1, fill = red]
	(n36) at (2.3, 0){};
\node[vertex, line width = 1, fill = red]
	(n1) at (4.6, 2.3){};
\node[vertex, line width = 1, fill = red]
	(n2) at (3.45, 2.3){};
\node[vertex, line width = 1, fill = black]
	(n3) at (0, 6.9){};
\node[vertex, line width = 1, fill = black]
	(n4) at (0, 5.75){};
\node[vertex, line width = 1, fill = black]
	(n5) at (0, 4.6){};
\node[vertex, line width = 1, fill = black]
	(n6) at (0, 3.45){};
\node[vertex, line width = 1, fill = blue]
	(n7) at (6.9, 1.15){};
\node[vertex, line width = 1, fill = blue]
	(n8) at (4.6, 4.6){};
\node[vertex, line width = 1, fill = blue]
	(n9) at (1.15, 6.9){};
\node[vertex, line width = 1, fill = blue]
	(n10) at (2.3, 4.6){};
\node[vertex, line width = 1, fill = blue]
	(n12) at (6.9, 4.6){};
\node[vertex, line width = 1, fill = black]
	(n13) at (0, 2.3){};
\node[vertex, line width = 1, fill = black]
	(n14) at (0, 1.15){};
\node[vertex, line width = 1, fill = black]
	(n15) at (0, 0){};
\node[vertex, label = above:{\color{black}{$P_t$}}, line width = 1, fill = black]
	(n17) at (8.05, 8.05){};
\node[vertex, line width = 1, fill = black]
	(n18) at (8.05, 6.9){};
\node[vertex, line width = 1, fill = black]
	(n19) at (8.05, 5.75){};
\node[vertex, line width = 1, fill = black]
	(n20) at (8.05, 4.6){};
\node[vertex, line width = 1, fill = black]
	(n21) at (8.05, 3.45){};
\node[vertex, line width = 1, fill = black]
	(n22) at (8.05, 2.3){};
\node[vertex, line width = 1, fill = black]
	(n23) at (8.05, 1.15){};
\node[vertex, line width = 1, fill = black]
	(n24) at (8.05, 0){};
\node[]
	(n28) at (4.6, 8.05) [align = left]{\color{green}{$L_{s,t}$}};
\node[]
	(n29) at (5.75, 0) [align = left]{\color{red}{$L_{t,s}$}};
\node[vertex, line width = 1, fill = blue]
	(n30) at (2.3, 2.3){};
\node[vertex, line width = 1, fill = blue]
	(n31) at (2.3, 1.15){};
\node[vertex, line width = 1, fill = blue]
	(n32) at (1.15, 2.3){};
\node[]
	(n35) at (3.45, 3.45) [align = left]{\color{blue}{$\gamma(P_s, P_t) = \gamma(P_t, P_s)$}};
\path[latex-, line width = 1, draw = black]
	(n0) to (n3);
\path[latex-, line width = 1, draw = black]
	(n3) to (n4);
\path[latex-, line width = 1, draw = black]
	(n4) to (n5);
\path[latex-, line width = 1, draw = black]
	(n5) to (n6);
\path[latex-, line width = 1, draw = black]
	(n6) to (n13);
\path[latex-, line width = 1, draw = black]
	(n13) to (n14);
\path[latex-, line width = 1, draw = black]
	(n14) to (n15);
\path[-latex, line width = 1, draw = red]
	(n22) .. controls (7.148, 2.693) and (6.332, 2.431) .. (n1);
\path[-latex, line width = 1, draw = red]
	(n1) .. controls (3.174, 2.092) and (4.088, -0.232) .. (n36);
\path[-latex, line width = 1, draw = red]
	(n36) to (n15);
\path[-latex, line width = 1, draw = red]
	(n23) .. controls (4.893, 0.381) and (1.002, 1.616) .. (n13);
\path[-latex, line width = 1, draw = red]
	(n24) .. controls (5.325, 1.27) and (5.284, 0.531) .. (n2);
\path[-latex, line width = 1, draw = red]
	(n2) .. controls (2.147, 3.053) and (1.106, 1.725) .. (n34);
\path[-latex, line width = 1, draw = red]
	(n34) .. controls (1.134, 0.281) and (0.345, 0.878) .. (n14);
\path[-latex, line width = 1, draw = green]
	(n0) .. controls (0.551, 8.433) and (0.768, 8.15) .. (n11);
\path[-latex, line width = 1, draw = green]
	(n11) .. controls (3.356, 7.195) and (6.867, 6.997) .. (n33);
\path[-latex, line width = 1, draw = green]
	(n33) .. controls (6.966, 8.617) and (7.501, 8.135) .. (n17);
\path[-latex, line width = 1, draw = green]
	(n4) .. controls (1.087, 5.398) and (1.215, 5.31) .. (n26);
\path[-latex, line width = 1, draw = green]
	(n26) .. controls (3.821, 6.854) and (5.267, 6.917) .. (n27);
\path[-latex, line width = 1, draw = green]
	(n27) .. controls (7.658, 6.705) and (7.27, 5.741) .. (n19);
\path[-latex, line width = 1, draw = green]
	(n3) .. controls (1.626, 6.014) and (2.191, 7.164) .. (n16);
\path[-latex, line width = 1, draw = green]
	(n16) .. controls (4.248, 6.69) and (3.674, 5.729) .. (n25);
\path[-latex, line width = 1, draw = green]
	(n25) .. controls (5.748, 5.817) and (6.355, 7.121) .. (n18);
\path[latex-, line width = 1, draw = black]
	(n24) to (n23);
\path[latex-, line width = 1, draw = black]
	(n23) to (n22);
\path[latex-, line width = 1, draw = black]
	(n22) to (n21);
\path[latex-, line width = 1, draw = black]
	(n21) to (n20);
\path[latex-, line width = 1, draw = black]
	(n20) to (n19);
\path[latex-, line width = 1, draw = black]
	(n19) to (n18);
\path[latex-, line width = 1, draw = black]
	(n18) to (n17);
\path[-latex, line width = 1, draw = blue]
	(n7) .. controls (5.97, 1.824) and (5.424, 1.196) .. (n1);
\path[-latex, line width = 1, draw = blue]
	(n1) .. controls (4.21, 2.754) and (3.937, 3.027) .. (n2);
\path[-latex, line width = 1, draw = blue]
	(n2) .. controls (3.092, 1.511) and (2.343, 1.639) .. (n30);
\path[-latex, line width = 1, draw = blue]
	(n31) .. controls (2.683, 0.937) and (2.637, 0.363) .. (n36);
\path[-latex, line width = 1, draw = blue]
	(n36) .. controls (1.599, -0.665) and (1.806, 0.654) .. (n34);
\path[-latex, line width = 1, draw = blue]
	(n34) .. controls (0.663, 1.414) and (0.462, 1.928) .. (n32);
\path[-latex, line width = 1, draw = blue]
	(n8) to (n25);
\path[-latex, line width = 1, draw = blue]
	(n25) .. controls (4.632, 6.459) and (5.43, 6.22) .. (n27);
\path[-latex, line width = 1, draw = blue]
	(n27) .. controls (6.031, 7.612) and (5.975, 8.075) .. (n33);
\path[-latex, line width = 1, draw = blue]
	(n33) .. controls (7.863, 8.046) and (7.215, 5.554) .. (n12);
\path[-latex, line width = 1, draw = blue]
	(n9) .. controls (0.637, 7.209) and (0.672, 7.69) .. (n11);
\path[-latex, line width = 1, draw = blue]
	(n11) .. controls (2.887, 8.682) and (3.455, 7.5) .. (n16);
\path[-latex, line width = 1, draw = blue]
	(n16) .. controls (3.344, 6.225) and (2.365, 6.709) .. (n26);
\path[-latex, line width = 1, draw = blue]
	(n26) to (n10);

		\end{tikzpicture}
	\end{minipage}
	\begin{minipage}[b]{0.49\textwidth}
        \begin{center}
		\caption{Illustration of the set \(\gamma(P_s, P_t)\) used in the proof of~\cref{state:path-system-to-clean-path-system}, given in blue.
		This set consists of the paths of \(\mathcal{P}\) which intersect many paths in at least one of the linkages between \(P_s\) and \(P_t\).\label{fig:gamma-set}}
        \end{center}
        \hfill
	\end{minipage}
\end{figure}
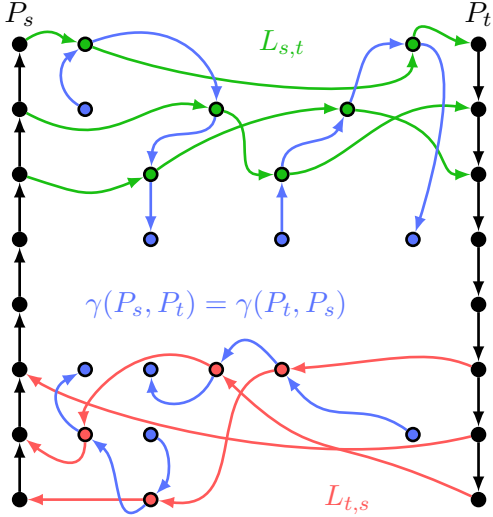

	If there is a pair of distinct \(P_s, P_t \in \mathcal{P}\) such that \(\Abs{\gamma(P_s, P_t)} \geq 2q_1\), then we construct our pair \((\mathcal{P}_1, \mathcal{Q}_1)\) as follows.
	By the pigeon-hole principle, there is a choice of \(\mathcal{P}_1 \in \Set{\mathcal{L}_{s,t}, \mathcal{L}_{t,s}}\) and a set \(\mathcal{Q}_1 \subseteq \gamma(P_s, P_t)\) of size \(q_1\) such that every \(Q \in \mathcal{Q}_1\) intersects at least \(d_1\) paths of \(\mathcal{P}_1\).
	By assumption, \(\mathcal{P}_1\) is minimal with respect to \(\mathcal{Q}_1\).
	Furthermore, \(\End{\mathcal{Q}_1}\) is well-linked to \(\Start{\mathcal{Q}_1}\).
	Hence, \((\mathcal{P}_1, \mathcal{Q}_1)\) satisfies~\ref{item:path-system-to-clean-path-system:web}.

	We now assume that \(\Abs{\gamma(P_s, P_t)} < 2q_1\) holds for all distinct \(P_s, P_t \in \mathcal{P}\).
	We construct a set \(\mathcal{P}_2\) as follows.

	First, distribute the elements of \(\mathcal{P}\) arbitrarily into \(p_2\) disjoint sets \(\mathcal{X}_1, \ldots, \mathcal{X}_{p_2}\), each of size \(p_2'\).
	For each \(\mathcal{X}_i\), define a random variable \(x_i\) which corresponds to sampling one element of \(\mathcal{X}_i\) from a uniform distribution.

	For each three distinct \(a,b,c \in \Set{1,2, \ldots, p_2}\), let \(A_{a,b,c}\) be the event that \(x_a \in \gamma(x_b, x_c)\).
	
	Since the event \(A_{a,b,c}\) depends only on the values of \(x_a, x_b\) and \(x_c\), we know that \(A_{a,b,c}\) is independent from \(A_{a',b',c'}\) if \(\Set{a,b,c} \cap \Set{a',b',c'} = \emptyset\).
	Hence, \(A_{a,b,c}\) is independent from all but at most \(\binom{p_2}{3} - \binom{p_2 - 3}{3} = d\) other events.

	We now bound the value of \(\Pr{A_{a,b,c}}\).
	As there are \((p_2')^3\) distinct choices for the tuple \((x_a, x_b, x_c)\) and for each choice of \(x_b, x_c\) there are at most \(2q_1\) choices of \(x_a\) such that \(x_a \in \gamma(x_b, x_c)\), we have that \(\Pr{A_{a,b,c}} \leq (2q_1 (p_2')^2)/(p_2')^3 = 2q_1/p_2'\).
	
	Because \(p_2' \geq e \cdot (d + 1) \cdot 2q_1 + 1\), from~\cref{state:lovasz-local-lemma} we know that the probability that none of the events \(A_{a,b,c}\) occur is positive.
	That is, there is some choice of \(x_{1}, x_{2}, \ldots, x_{p_2}\) such that \(x_a \not \in \gamma(x_b, x_c)\) for all three distinct \(a,b,c \in \Set{1,2,\ldots, p_2}\).
	We set \(\mathcal{P}_2 = \Set{x_{1}, x_{2}, \ldots, x_{p_2}}\).

	For each distinct \(P_s, P_t \in \mathcal{P}_2\) define \(\mathcal{L}_{s,t}' = \{L \in \mathcal{L}_{s,t} \ST{} \text{for all \(P \in \mathcal{P}_2 \setminus \Set{P_s, P_t}\) we have}\) \(\V{P} \cap L = \emptyset\}\).

	By choice of \(\mathcal{P}_2\), we have that \(P_r \not\in \gamma(P_s, P_t)\) holds for all pairwise distinct \(P_s, P_t, P_r \in \mathcal{P}_2\).
	Hence, \(\Abs{\mathcal{L}_{s,t}'} \geq \Abs{\mathcal{L}_{s,t}} - (p_2 - 2)d_1 \geq \ell_2\).

	Finally, choose \(\mathcal{A}_2\) as the elements \(A_i^{in}, A_i^{out}\) of \(\mathcal{A}\) satisfying \(P_i \in \mathcal{P}_2\).
	Clearly, \((\mathcal{P}_2, \mathcal{L}_2, \mathcal{A}_2)\) is a clean \(\ell_2\)-linked \pathsystem of order \(p_2\), satisfying \ref{item:path-system-to-clean-path-system:clean}.
\end{proof}

We attempt to obtain a bidirected clique from a clean \pathsystem by first iteratively trying to obtain disjoint paths inside \(\mathcal{L}\) pairwise connecting the paths of \(\mathcal{P}\).
In case this construction fails, we obtain a well-linked semi-web.

\subsection{Splits or segmentations with elementary bounds}

As seen before in~\cref{state:semi-web-to-web}, in order to obtain a web from a semi-web \(\Brace{\mathcal{P}, \mathcal{Q}}\) of degree \(d\), we require that \(\mathcal{Q}\) is much larger than \(\mathcal{P}\).
Unfortunately, it is not possible to directly use the results of \cite{kawarabayashi2022directed} to obtain the required web, as the sizes of the linkages \(\mathcal{P}\) and \(\mathcal{Q}\) provided by their statements do not match.
Instead, we need to modify the proof of~\cite[Lemma 4.8]{kawarabayashi2022directed}, ensuring that in each step of the iteration described at the end of the previous subsection we obtain a sufficiently large gap between \(\mathcal{Q}\) and \(\mathcal{P}\) so that we can apply~\cref{state:semi-web-to-web}.
Similarly, we again require a large gap between \(\mathcal{Q}\) and \(\mathcal{P}\) when obtaining a split or segmentation from a web using~\cref{lemma:y-split or x-segmentation}.
Hence, we need to \say{pay} the function from both~\cref{lemma:y-split or x-segmentation} and~\cref{state:semi-web-to-web} during each step of the iteration in our proof.

In order to avoid repetition, we only present here the part of the proof from~\cite{kawarabayashi2022directed} which we modify, which is~\cref{state:clean-path-system-to-web-or-clique} below.
The remainder of the proof of~\cite[Lemma 4.8]{kawarabayashi2022directed} is given by~\cref{state:clean-and-disjoint-path-system-to-clique}.

\begin{lemma}[{\cite[proof of Lemma 4.8]{kawarabayashi2022directed}}]
	\label{state:clean-and-disjoint-path-system-to-clique}
	Let \(k\) be an integer and let \(\mathcal{S} = \Brace{\mathcal{P}, \mathcal{L}, \mathcal{A}}\) be a clean 1-linked \pathsystem of order \(p\) in a digraph \(D\) such that all paths in \(\mathcal{L}\) are pairwise vertex-disjoint.
	If \(p \geq \bound{state:clean-and-disjoint-path-system-to-clique}{p}{k} \coloneqq 3k\), then \(D\) contains a bidirected clique of order \(k\) as a butterfly minor.
\end{lemma}

The proof of~\cref{state:clean-path-system-to-web-or-clique} works by iteratively constructing pairwise disjoint paths \(R_r\).
If in some step we cannot construct the desired path \(R_r\), then we argue that some linkage \(\mathcal{L}_{i,j}\) in the clean \pathsystem intersects many paths of some other linkage \(\mathcal{L}_{s,t}\).
Our modifications of the proof of~\cite[Lemma 4.8]{kawarabayashi2022directed} are essentially focused on ensuring that both linkages forming the semi-web described above are large enough for us to apply~\cref{state:semi-web-to-web,lemma:y-split or x-segmentation}.

We define
\begin{align*}
	\Func{k'}{k} & = 2 \binom{3k}{2},
	\\[0em]
	\boundDefAlign{state:clean-path-system-to-web-or-clique}{p}{k}
	\bound{state:clean-path-system-to-web-or-clique}{p}{k} & = 3k,
	\\[0em]
	\boundDefAlign{state:clean-path-system-to-web-or-clique}{\ell}{x, y, q, k}
	\bound{state:clean-path-system-to-web-or-clique}{\ell}{x, y, q, k} & = 
	\left(2 x q \Func{k'}{k} \right)^{2^{\Func{k'}{k} \left(y \left(x - 1\right) + 1\right)} \left(3 x\right)^{\Func{k'}{k}}}.
\end{align*}
Observe that \(\bound{state:clean-path-system-to-web-or-clique}{\ell}{x,y,q,k} \in \PowerTower{2}{\Polynomial{5}{x,y,q,k}}\). 
\begin{lemma}
	\label{state:clean-path-system-to-web-or-clique}
	Let \(x,y, q\) and \(k\) be integers.
	Let \(\mathcal{S} = \Brace{\mathcal{P}, \mathcal{L}, \mathcal{A}}\) be a clean \(\ell'\)-linked \pathsystem of order \(p'\) in a digraph \(D\).
	If \(\ell' \geq \bound{state:clean-path-system-to-web-or-clique}{\ell}{x, y, q, k}\) and \(p' \geq \bound{state:clean-path-system-to-web-or-clique}{p}{k}\), then there are \(\mathcal{P}_1 \subseteq \mathcal{L}_P \in \mathcal{L}\) and \(\mathcal{Q}_1 \subseteq \mathcal{L}_Q \in \mathcal{L}\), where \(\mathcal{L}_P \neq \mathcal{L}_Q\), such that \(D\) contains one of the following:
	\begin{enamerate}{W}{item:clean-path-system-to-web-or-clique:last}
		\item a bidirected clique of order \(k\) as a butterfly minor,
			\label{item:clean-path-system-to-web-or-clique:clique}
		\item \label{item:clean-path-system-to-web-or-clique:split}
			a $(y, q)$-split \(\Brace{\mathcal{P}', \mathcal{Q}'}\) of $(\mathcal{P}_1, \mathcal{Q}_1)$ where \(\End{\mathcal{Q}'}\) is well-linked to \(\Start{\mathcal{Q}'}\), or
		\item \label{item:clean-path-system-to-web-or-clique:segmentation}
			an $(x, q)$-segmentation $(\mathcal{P}', \mathcal{Q}')$ of $(\mathcal{P}_1, \mathcal{Q}_1)$ where \(\End{\mathcal{P}'}\) is well-linked to \(\Start{\mathcal{P}'}\).
			\label{item:clean-path-system-to-web-or-clique:last}
	\end{enamerate}
\end{lemma}
\begin{proof}
		Let \(k_1 = 2 \binom{3k}{2}\).
		We define functions \(q', f, p'\) and \(q''\) recursively as follows.
		We start by setting
		\begin{align*}
			q'(k_1) & = \bound{state:semi-web-to-web}{q}{x, q, x},
			& & & f(k_1 + 1) & = 0, f(k_1) = q'(k_1) + 1,
			\\[0em]
			p'(k_1) & = x & & \text{ and } & q''(k_1) & = \bound{lemma:y-split or x-segmentation}{q}{p'(k_1), q, x, y, p'(k_1)}.
		\end{align*}

		For \(1 \leq r < k_1\), we set
		\begin{align*}
			p'(r) & = f(r + 1) + x - 1, & & &
			q''(r) & = \bound{lemma:y-split or x-segmentation}{q}{p'(r), q, x, y, p'(r)},
			\\[0em]
			q'(r) & = \bound{state:semi-web-to-web}{q}{p'(r), q''(r), x} & & \text{ and } & 
			f(r) & = (k_1 - r + 1) q'(r) + 1.
		\end{align*}

		By repeatedly applying the functions above, we obtain the following recursive equality, which we use later
			\begin{align*}
				f(r) & = 1 + x \left(k_{1} - r + 1\right) \left(x + f{\left(r + 1 \right)} - 1\right)^{x}
				\\[0em]
				& \quad \cdot
				\left(q \left(x + f{\left(r + 1 \right)} - 1\right) \left(x + q + f{\left(r + 1 \right)} - 1\right)\right)^{2^{y \left(x - 1\right) + 1}}.
			\end{align*} 
   			
		Before proceeding with the proof, we give upper bounds for the functions defined above.
		\begin{claim}
			\label{state:cpstwc:bounds}
			For all \(0 \leq r \leq k_1 - 1\), we have
			\begin{align*}
				f(k_1 - r) \leq \left(2 k_{1} x q\right)^{2^{\left(r + 1\right) \left(y \left(x - 1\right) + 1\right)} \left(3 x\right)^{r + 1}}
			\end{align*} 
   					\end{claim}
		\begin{claimproof}
			We prove the statement iteratively by starting at \(r = 0\).
			\begin{align*}
				f(k_1) = x^{x + 1} q + 1
				\leq \left(2 k_{1} x q\right)^{3 x}
				\leq \left(2 k_{1} x q\right)^{3 x 2^{y \left(x - 1\right) + 1}}.
			\end{align*}
			Hence, the bounds mentioned above hold for \(r = 0\).
			Now assume the bounds hold for some \(r - 1 \in \Set{0, 1, \ldots, k_1 - 2}\).
			We show that they also hold for \(r\).
			To aid readability, we replace \(y ( x - 1) + 1\) with \(w\).
			\begingroup
			\allowdisplaybreaks
			\begin{align*}
				& f(k_1 - r) = x \left(r + 1\right)  \left(x + f{\left(k_{1} - r + 1 \right)} - 1\right)^{x}
				\\[0em]
				& \phantom{f(k_1 - r) = } \cdot
					(q \left(x + f{\left(k_{1} - r + 1 \right)} - 1\right)
					\\[0em]
					& \phantom{f(k_1 - r) = \cdot ( q } \cdot \left(x + q + f{\left(k_{1} - r + 1 \right)} - 1\right) )^{2^{y \left(x - 1\right) + 1}} +1
				\\[0em]
																&{} \leq x \left(r + 1\right) \left(x + \left(2 k_{1} x q\right)^{2^{r w} \left(3 x\right)^{r}}\right)^{x} 
				\\[0em]
				& \quad {} \cdot 
				\left(q \left(x + \left(2 k_{1} x q\right)^{2^{r w} \left(3 x\right)^{r}}\right) \left(x + q + \left(2 k_{1} x q\right)^{2^{r w} \left(3 x\right)^{r}}\right)\right)^{2^{w}} && (\text{induction})
\\[0em]
  & \leq x \left(r + 1\right) \left(2 \left(2 k_{1} x q\right)^{2^{r w} \left(3 x\right)^{r}}\right)^{x} \left(4 q \left(2 k_{1} x q\right)^{2^{r w + 1} \left(3 x\right)^{r}}\right)^{2^{w}} 
				&& (x + q \leq 2 k_1 x q)
\\[0em]
  & \leq \left(2 \left(2 k_{1} x q\right)^{2^{r w} \left(3 x\right)^{r}}\right)^{x} \left(4 k_{1} x q \left(2 k_{1} x q\right)^{2^{r w + 1} \left(3 x\right)^{r}}\right)^{2^{w}}
				&& (xk_1 \leq (x k_1)^{2^w} )
\\[0em]
  & = \left(2 \left(2 k_{1} x q\right)^{2^{r w} \left(3 x\right)^{r}}\right)^{x} \left(2 \left(2 k_{1} x q\right)^{2^{r w + 1} \left(3 x\right)^{r} + 1}\right)^{2^{w}}
\\[0em]
  & = 2^{2^{w} + x} \left(2 k_{1} x q\right)^{2^{w} \left(2^{r w + 1} \left(3 x\right)^{r} + 1\right) + 2^{r w} x \left(3 x\right)^{r}}
\\[0em]
  & \leq \left(2 k_{1} x q\right)^{2^{w} \left(2^{r w + 1} \left(3 x\right)^{r} + 1\right) + 2^{w} + x + 2^{r w} x \left(3 x\right)^{r}}
	&& (2 \leq 2 k_1 x q)
\\[0em]
  & = \left(2 k_{1} x q\right)^{2^{w \left(r + 1\right) + 1} \left(3 x\right)^{r} + 2^{w + 1} + x + 2^{r w} x \left(3 x\right)^{r}}
\\[0em]
  & \leq \left(2 k_{1} x q\right)^{3 \cdot 2^{w \left(r + 1\right)} 3^{r} x^{r + 1}}
	&& (2^{w \left(r + 1\right) + 1} \left(3 x\right)^{r},
\\[0em]
				&&& 2^{r w} x \left(3 x\right)^{r},
\\[0em]
				&&& 2^{w + 1} + x \leq 2^{(r + 1) w} 3^r x^{r + 1})
\\[0em]
 & = \left(2 k_{1} x q\right)^{2^{\left(r + 1\right) \left(y \left(x - 1\right) + 1\right)} \left(3 x\right)^{r + 1}}
				&& (\text{def.~of } w)
			\end{align*}
			\endgroup

			Hence, the claim follows by induction.
		\end{claimproof}

		From~\cref{state:cpstwc:bounds} we have that \(f(1) \leq \bound{state:clean-path-system-to-web-or-clique}{\ell}{x, y, q, k}\).
				Let \((\mathcal{P} = \Brace{P_{1}, P_{2}, \ldots, P_{3k}}, \mathcal{L} = \Brace{\mathcal{L}_{i,j}}, \mathcal{A} = \Brace{A_i^{\text{in}}, A_i^{\text{out}}}) \coloneqq \mathcal{S}\).
		Choose an arbitrary bijection \(\sigma : [k_1] \to \Set{(i,j) \mid 1 \leq i,j \leq 3k, i \neq j}\), where \([k_1] = \Set{1,2,\ldots,k_1}\).

		We iteratively construct linkages \(\mathcal{L}_{i,j}^r\) and paths \(R^r\), where \(1 \leq r \leq k_1\), satisfying the following:
		\begin{enamerate}{L}{prop:cpstwc:last}
		\item \label{prop:cpstwc:disjoint}
			\(R^r\) is a path from \(A_i^{\text{out}}\) to \(A_j^{\text{in}}\), where \(\Brace{i,j} \coloneqq \sigma(r)\), and \(R^r\) does not share any internal vertex with any path in \(\mathcal{P}\) or in any \(\mathcal{L}_{\sigma(q)}^r\) where \(q > r\).
		\item \label{prop:cpstwc:size}
			\(\Abs{\mathcal{L}_{\sigma(r)}^r} = f(r)\),
		\item \label{prop:cpstwc:minimal}
			for all \(r < q \leq k_1\) we have \(\Abs{\mathcal{L}}_{\sigma(q)}^r = p'(r)\), and \(\mathcal{L}_{\sigma(q)}^r\) is \(\mathcal{L}_{\sigma(r)}^r\)-minimal, and
		\item \label{prop:cpstwc:clean}
			for all \(1 \leq i, j \leq 3k\) where \(i \neq j\) and for all \(P \in \mathcal{L}_{\sigma^{-1}(i,j)}^r\) the path \(P\) has no vertex in common with any \(P_t\) for \(i \neq t \neq j\).
			\label{prop:cpstwc:last}
		\end{enamerate}

		We show that, if~\cref{prop:cpstwc:size,prop:cpstwc:minimal,prop:cpstwc:clean} hold on step \(1 \leq r \leq k_1\), then~\cref{prop:cpstwc:disjoint} holds on step \(r\) and if~\cref{prop:cpstwc:size,prop:cpstwc:minimal,prop:cpstwc:clean} hold on step \(1 \leq r < k_1\), then~\cref{prop:cpstwc:size,prop:cpstwc:minimal,prop:cpstwc:clean} also hold on step \(r + 1\).

		For \(r = 1\), we pick \(\mathcal{L}^1_{\sigma(1)} \subseteq \mathcal{L}_{s,t}\) arbitrarily, where $(s, t) = \sigma(1)$, so that \(\Abs{\mathcal{L}^1_{\sigma(1)}} = f(1)\), satisfying~\cref{prop:cpstwc:size} for \(r = 1\).
		Further, for each \(1 < q \leq k_1\), we choose \(\mathcal{L}^1_{\sigma(q)}\) as a $\mathcal{L}^1_{\sigma(1)}$-minimal \(\Start{\mathcal{L}_{\sigma(q)}}\)-\(\End{\mathcal{L}_{\sigma(q)}}\)-linkage in \(\ToDigraph{\mathcal{L}^1_{\sigma(1)} \cup \mathcal{L}_{\sigma(q)}}\) of order \(p'(1)\).
								This satisfies~\cref{prop:cpstwc:minimal} for \(r = 1\).
		Observe that~\cref{prop:cpstwc:clean} is satisfied for \(r = 1\) because \(\mathcal{S}\) is a clean \pathsystem.

		Now assume that~\cref{prop:cpstwc:minimal,prop:cpstwc:size,prop:cpstwc:clean} hold for step \(r \geq 1\).
		We construct the path \(R^r\) as follows.
		
		First, let \(\Brace{i,j} = \sigma(r)\).
		We consider two cases.
		\begin{CaseDistinction}
			\Case{There is a path \(P \in \mathcal{L}^{r}_{i,j}\) which, for each \(r < q \leq k_1\), is internally disjoint from at least \(f(r + 1)\) paths in \(\mathcal{L}^{r}_{\sigma(q)}\).}
			\label{case:cpstwc:web}

			We set \(R^{r} \coloneqq P\), satisfying~\cref{prop:cpstwc:disjoint} for \(r\).
			If \(r = k_1\), we are done with the iteration.
			Otherwise, let \(\Brace{s,t} \coloneqq \sigma(r + 1)\) and let \(\mathcal{L}^{r+1}_{s,t} \subseteq \mathcal{L}^r_{s,t}\) be an \(A_s^{\text{out}}\)-\(A_t^{\text{in}}\)-linkage of order \(f(r + 1) \leq p'(r)\) (satisfying~\cref{prop:cpstwc:size}) such that no path in \(\mathcal{L}^{r+1}_{s,t}\) has an internal vertex in \(\V{P} \cup \bigcup_{r' = 1}^r \V{R^{r'}}\) (towards satisfying~\cref{prop:cpstwc:disjoint} for \(r + 1\)).
			Because~\cref{prop:cpstwc:disjoint} holds for \(r\), we know that \(\mathcal{L}^r_{s,t}\) is internally disjoint from all \(R^{r'}\) with \(1 \leq r' < r\).
			Hence, such a linkage \(\mathcal{L}^{r+1}_{s,t}\) exists.
			Furthermore, as~\cref{prop:cpstwc:disjoint} holds for \(r\), we have that every path in \(\mathcal{L}^{r + 1}_{s,t}\) is disjoint from all \(P \in \mathcal{P} \setminus \Set{P_s, P_t}\) (towards satisfying~\cref{prop:cpstwc:clean}).

			For each \(r + 1 < q \leq k_1\), let \(\Brace{s', t'} = \sigma(q)\) and choose an \(A_{s'}^{\text{out}}\)-\(A_{t'}^{\text{in}}\)-linkage \(\mathcal{L}^{r+1}_{\sigma(q)}\) of order \(p'(r + 1)\) inside \(\ToDigraph{\mathcal{L}^{r}_{\sigma(q)} \cup \mathcal{L}^{r + 1}_{s,t}}\) which satisfies~\cref{prop:cpstwc:clean} such that every path in \(\mathcal{L}^{r+1}_{\sigma(q)}\) has no inner vertex in \(\V{P} \cup \bigcup_{r' = 1}^r \V{\mathcal{L}^r_{\sigma(r')}}\) and which is \(\mathcal{L}^{r+1}_{s,t}\)-minimal (satisfying~\cref{prop:cpstwc:minimal}).
			Thus,~\cref{prop:cpstwc:disjoint,prop:cpstwc:size,prop:cpstwc:minimal,prop:cpstwc:clean} hold for the step~\(r + 1\).

			\Case{For every \(P_z \in \mathcal{L}^{r}_{i,j}\) there is some \(r < q_z \leq k_1\) for which \(P_z\) intersects least \(p'(r) - f(r + 1) + 1 = x\) paths in \(\mathcal{L}^{r}_{\sigma(q_z)}\).}

			Let \(\Brace{i',j'} = \sigma(q_z)\).
			As \(\Abs{\mathcal{L}^r_{i,j}} = f(r) = (k_1 - r - 1) q'(r) + 1\), by the pigeon-hole principle there is a \(r < w \leq k_1\) and a \(\mathcal{Q} \subseteq \mathcal{L}^r_{i,j}\) of order \(q'(r)\) such that all paths in \(\mathcal{Q}\) intersect at least \(x\) paths in \(\mathcal{L}^r_{\sigma(w)}\).
			Hence, \(\Brace{\mathcal{L}^r_{\sigma(w)}, \mathcal{Q}}\) is a \(\Brace{p'(r), q'(r)}\)-semi-web of degree \(x\).
			Finally, as the starting points and endpoints of both \(\mathcal{Q}\) and \(\mathcal{L}^r_{\sigma(w)}\) lie in the well-linked set \(A_i^{\text{out}} \cup A_j^{\text{in}} \subseteq A\), we have that \(\Brace{\mathcal{L}^r_{\sigma(w)}, \mathcal{Q}}\) is also a well-linked semi-web.

			Applying~\cref{state:semi-web-to-web} to \(\Brace{\mathcal{L}^r_{\sigma(w)}, \mathcal{Q}}\) yields a well-linked \(\Brace{p_2, q''(r)}\)-web \(\Brace{\mathcal{P}_2, \mathcal{Q}_2}\) where \(\mathcal{P}_2\) is minimal with respect to \(\mathcal{Q}_2\), where \(p'(r) \geq p_2 \geq x\).
			As \(q''(r) \geq \bound{lemma:y-split or x-segmentation}{q}{p'(r), q, x, y, p'(r)}\) and \(\mathcal{P}_2\) is also \(p_2\)-minimal with respect to \(\mathcal{Q}_2\), we can apply~\cref{lemma:y-split or x-segmentation} to \(\Brace{\mathcal{P}_2, \mathcal{Q}_2}\), obtaining two cases.
			
			If~\cref{lemma:y-split or x-segmentation}\ref{item:y-split or x-segmentation:segmentation} holds, then we satisfy~\cref{item:clean-path-system-to-web-or-clique:segmentation}.
			Otherwise,~\cref{lemma:y-split or x-segmentation}\cref{item:y-split or x-segmentation:split} holds, satisfying~\cref{item:clean-path-system-to-web-or-clique:split}.
		\end{CaseDistinction}

		If Case~2 above never occurs, then we obtain a sequence of pairwise disjoint paths \(\mathcal{R} \coloneqq (R_{1}, R_{2}, \ldots, R_{k_1})\) such that for all \(1 \leq r \leq k_1\), the path \(R_r\) is an \(A_{i}^{\text{out}}\)-\(A_j^{\text{in}}\)-path which is disjoint from \(\mathcal{P}\), where \(\Brace{i,j} = \sigma(r)\).
		By~\cref{state:clean-and-disjoint-path-system-to-clique}, we obtain a bidirected clique of size \(k\) as a butterfly minor, satisfying~\cref{item:clean-path-system-to-web-or-clique:clique}.
\end{proof}

We conclude this section by combining the main statements proven above, yielding the following theorem, which is used later on, both in this paper and in~\cite{COSSII}.
The statement of the theorem is slightly stronger than what is required in this paper, as we only need the \(k\) disjoint cycles from the outcome~\ref{item:high-dtw-to-web:grid} of~\cref{state:high-dtw-to-segmentation-or-split}.
However, we need the statement in this form when improving the bounds of the Directed Grid Theorem.

We define
\begin{align*}
	\Func{\ell'}{x,y,q,k} & = \bound{state:path-system-to-clean-path-system}{\ell}{
				\bound{state:clean-path-system-to-web-or-clique}{p}{2k^2},
				\bound{state:clean-path-system-to-web-or-clique}{\ell}{x,y,q,2k^2}, x},
	\\[0em]
	\boundDefAlign{state:high-dtw-to-segmentation-or-split}{t}{x, y, q, k}
	\bound{state:high-dtw-to-segmentation-or-split}{t}{x, y, q, k} & =
		2(\bound{state:bramble-contains-linked-path-system}{k}{}(
			\Func{\ell'}{x,y,q,k},
				\\[0em] & \phantom{ = 2(\bound{state:bramble-contains-linked-path-system}{k}{}(~}
				\bound{state:path-system-to-clean-path-system}{p}{}(
					\bound{state:semi-web-to-web}{q}{}(
						\Func{\ell'}{x,y,q,k},
							\\[0em] & 
							\phantom{ = 2(\bound{state:bramble-contains-linked-path-system}{k}{}(\bound{state:path-system-to-clean-path-system}{p}{}(~}
							\bound{lemma:y-split or x-segmentation}{q}{}(
								\Func{\ell'}{x,y,q,k},
									q, x, y,
									\\[0em] &
									\phantom{ = 2(\bound{state:bramble-contains-linked-path-system}{k}{}(\bound{state:path-system-to-clean-path-system}{p}{}(\bound{lemma:y-split or x-segmentation}{q}{}(~}
									\Func{\ell'}{x,y,q,k}), x),
							\bound{state:clean-path-system-to-web-or-clique}{p}{2k^2})))).
\end{align*}
Note that \(\bound{state:high-dtw-to-segmentation-or-split}{t}{x,y,q,k} \in \PowerTower{5}{\Polynomial{9}{x,y,q,k}}\).

\begin{theorem}
	\label{state:high-dtw-to-segmentation-or-split}
	Let \(D\) be a digraph.
	Let \(k, x, y, q\) be integers.
	If \(\dtw{D} \geq \bound{state:high-dtw-to-segmentation-or-split}{t}{x,y,q,k}\), then \(D\) contains one of the following
	\begin{enamerate}{D}{item:high-dtw-to-web:last}
	\item
		\label{item:high-dtw-to-web:grid}
		a cylindrical grid of order \(k\) as a butterfly minor,
	\item 
		\label{item:high-dtw-to-web:split}
		a $(y, q)$-split \(\Brace{\mathcal{P}', \mathcal{Q}'}\) of some pair $(\mathcal{P}_1, \mathcal{Q}_1)$ in \(D\), where \(\End{\mathcal{Q}'}\) is well-linked to \(\Start{\mathcal{Q}'}\), or
	\item
		\label{item:high-dtw-to-web:segmentation}
		an $(x, q)$-segmentation $(\mathcal{P}', \mathcal{Q}')$ of some pair $(\mathcal{P}_1, \mathcal{Q}_1)$ in \(D\), where \(\End{\mathcal{P}'}\) is well-linked to \(\Start{\mathcal{P}'}\).
		\label{item:high-dtw-to-web:last}
	\end{enamerate}
\end{theorem}
\begin{proof}
	Let 
	\(k_6 = 2k^2\),
	\(\ell_5 = \bound{state:clean-path-system-to-web-or-clique}{\ell}{x,y,q,k_6}\),
	\(p_5 = \bound{state:clean-path-system-to-web-or-clique}{p}{k_6}\),
	\(d_3 = x\),
	\(\ell_2 = \bound{state:path-system-to-clean-path-system}{\ell}{p_5, \ell_5, d_3}\),
	\(q_4 = \bound{lemma:y-split or x-segmentation}{q}{\ell_2, q, x, y, \ell_2}\)
	\(q_3 = \bound{state:semi-web-to-web}{q}{\ell_2, q_4, d_3}\),
	\(p_2 = \bound{state:path-system-to-clean-path-system}{p}{q_3, p_5}\),
	\(k_1 = \bound{state:bramble-contains-linked-path-system}{k}{\ell_2, p_2}\).

	Observe that \(\ell_2 = \Func{\ell'}{x,y,q,k} \geq x\) and that \(\bound{state:high-dtw-to-segmentation-or-split}{t}{x,y,q,k} \geq 2k_1\).

	By~\cref{state:dtw-to-bramble}, \(D\) contains a bramble \(\mathcal{B}_1\) of order at least \(k_1\).
	By~\cref{state:bramble-contains-linked-path-system}, \(D\) contains an \(\ell_2\)-linked path-system \(\mathcal{S}_2\) of order \(p_2\).
	By applying~\cref{state:path-system-to-clean-path-system} to \(\mathcal{S}_2\), we obtain two cases.

	\begin{CaseDistinction}
		\Case{\cref{state:path-system-to-clean-path-system}\ref{item:path-system-to-clean-path-system:web} holds.}

		Then \(D\) contains a well-linked \(\Brace{\ell_2, q_3}\)-semi-web \(\Brace{\mathcal{P}_3, \mathcal{Q}_3}\) of degree \(d_3\) where \(\mathcal{P}_3\) is minimal with respect to \(\mathcal{Q}_3\) and \(\mathcal{P}_3 \in \mathcal{L}\).
		In particular, \(\End{\mathcal{P}_3}\) is well-linked to \(\Start{\mathcal{P}_3}\) as \(\Start{\mathcal{L}} \cup \End{\mathcal{L}}\) are vertices of a well-linked set, and \(\End{\mathcal{Q}_3}\) is well-linked to \(\Start{\mathcal{Q}_3}\) as the starting and endpoints of the paths in \(\mathcal{P}\) are vertices of a well-linked set.

		By~\cref{state:semi-web-to-web}, there is some \(p_4\) such that \(\ell_2 \geq p_4 \geq d_3\) and \(D\) contains a well-linked \(\Brace{p_4, q_4}\)-web \(\Brace{\mathcal{P}_4, \mathcal{Q}_4}\) where \(\mathcal{P}_4\) is minimal with respect to \(\mathcal{Q}_4\) and \(\mathcal{P}_4 \subseteq \mathcal{P}_3\), \(\mathcal{Q}_4 \subseteq \mathcal{Q}_3\).
		In particular, \(\mathcal{P}_4\) is also weakly \(p_4\)-minimal with respect to \(\mathcal{Q}_4\), \(\End{\mathcal{P}_4}\) is also well-linked to \(\Start{\mathcal{P}_4}\) and \(\End{\mathcal{Q}_4}\) is also well-linked to \(\Start{\mathcal{Q}_4}\)
		Applying~\cref{lemma:y-split or x-segmentation} to \(\Brace{\mathcal{P}_4, \mathcal{Q}_4}\) yields two cases.
		If~\cref{lemma:y-split or x-segmentation}\ref{item:y-split or x-segmentation:split} holds, then~\cref{item:high-dtw-to-web:split} is satisfied.
		Otherwise,~\cref{lemma:y-split or x-segmentation}\ref{item:y-split or x-segmentation:segmentation} holds, satisfying~\cref{item:high-dtw-to-web:segmentation}.

		\Case{\cref{state:path-system-to-clean-path-system}\ref{item:path-system-to-clean-path-system:clean} holds.}

		That is, \(D\) contains a clean \(\ell_5\)-linked \pathsystem \(\mathcal{S}_5\) of order \(p_5\).
		Applying~\cref{state:clean-path-system-to-web-or-clique} to \(\mathcal{S}_5\)	yields three cases.

		If~\cref{state:clean-path-system-to-web-or-clique}\ref{item:clean-path-system-to-web-or-clique:clique} holds, then \(D\) contains a bidirected clique of order \(k_6\) as a butterfly minor.
		As a cylindrical grid of order \(k\) contains \(k_6\) vertices, \(D\) also contains a cylindrical grid of order \(k\) as a butterfly minor, satisfying \ref{item:high-dtw-to-web:grid}.

		If~\cref{state:clean-path-system-to-web-or-clique}\ref{item:clean-path-system-to-web-or-clique:split} holds, then we obtain a \(\Brace{y,q}\)-split, satisfying \ref{item:high-dtw-to-web:split}.

		If~\cref{state:clean-path-system-to-web-or-clique}\ref{item:clean-path-system-to-web-or-clique:segmentation} holds, then we obtain an \(\Brace{x,q}\)-segmentation, satisfying \ref{item:high-dtw-to-web:segmentation}.
															\end{CaseDistinction}
\end{proof}

\section{Temporal digraphs and routings}
\label{sec:temporal}

In our proof, we are frequently faced with problems of the following form.
We have already constructed two linkages $\PPP$ and $\QQQ$, say.
However, while the paths within the same linkage are disjoint by definition, a pair of paths from different linkages may intersect arbitrarily or not at all.
So the intersection pattern between the two linkages can be arbitrarily complex.
The problem then is to find a subgraph of a specific form in $\bigcup \PPP \cup \bigcup \QQQ$.
Problems of this form occur frequently in this research area, including in the proof of the Directed Grid Theorem~\cite{kawarabayashi2015directed}.

In this section, we develop a framework based on temporal digraphs which allows us to formulate these problems more abstractly.
This abstraction allows us to simplify many arguments and unify proofs by isolating the core ideas common to several of these proofs.
Moreover, our framework allows us to obtain much better bounds and to prove elementary bounds for results that require non-elementary bounds in~\cite{kawarabayashi2015directed}.

There are several different definitions of temporal graphs and temporal walks, each applicable in a different context.
Here, we use the notation from~\cite{CasteigtsHMZ20, Molter20} and adapt it to the directed setting. 
We first define our notion of temporal digraphs and walks, and then discuss how they arise in our context. 

\begin{definition}
	\label{def:temporal-digraph}
	A \emph{temporal digraph} is a pair $T = (V, \AAA)$ consisting of a vertex set $V$ and sequence of arc sets $\mathcal{A} = \Brace{A_1, A_2, \dots A_{\ell}}$ such that $\Layer{T}{t} \coloneqq \Brace{V, A_t}$ is a digraph for all $1 \leq t \leq \ell$. 
	We also refer to $\Layer{T}{t}$ as \emph{layer $t$} of $T$ and call $t$ a \emph{\timestep}.
	The \emph{lifetime} of $D$ is given by $\Lifetime{D} \coloneqq \ell$.
\end{definition}

Next, we define paths and walks in the temporal setting.
A \emph{temporal walk} in a temporal digraph $T$ is required to obey the \say{timeline} of $T$, i.e.~the order in which arcs occur on the walk must respect the \timesteps of $T$. 
In our setting, we need a more restrictive definition that only allows a temporal walk to use at most one arc of each layer.

\begin{definition}
	\label{def:temporal-walks-and-paths}
	A \emph{temporal walk} of length $n$ from $v_0$ to $v_n$ in a temporal digraph $T$ is a sequence $W \coloneqq (v_0, t_0), (v_1, t_1), \ldots, (v_n, t_n)$ such that $(v_i, v_{i+1}) \in A_{t_{i}}$ and $t_i < t_{i+1} \leq \Lifetime{T}$ for all $0 \leq i \leq n - 1$.
	If such a walk exists, we say that $v_0$ \emph{temporally reaches} $v_n$.
	A temporal walk is said to be a \emph{temporal path} if no vertex appears twice in the sequence.
	Finally, we say that $W$ \emph{departs} at $t_0$ and \emph{arrives} at $t_n$, and that $t_n - t_0$ is the \emph{duration} of $W$.
\end{definition}

We often have situations where some linkage \(\mathcal{P}\) intersects pairwise disjoint digraphs \(Q_{1}, \ldots, Q_{q}\).
If the paths in \(\mathcal{P}\) \say{agree} on the order of the digraphs \(Q_i\), then we can construct an auxiliary temporal digraph in the following way.

\begin{definition}
	\label{def:routing-temporal-digraph}
	Let $\mathcal{P}$ be a linkage and let $\mathcal{Q} = \Set{Q_1, Q_2, \dots, Q_{q}}$ be a set of pairwise disjoint digraphs such that each path $P_i \in \mathcal{P}$ can be partitioned as $P_i^1 \cdot P_i^2 \cdot \ldots \cdot P_i^{q} = P_i$ such that  $\V{P_i^j} \cap \V{\mathcal{Q}} \subseteq \V{Q_j}$ for all $1 \leq j \leq q$.
	
	The \emph{routing temporal digraph $\Brace{V, \mathcal{A}}$ of $\mathcal{P}$ through $\mathcal{Q}$} is constructed as follows.
    We set $V = \mathcal{P}$ and for each $1 \leq j \leq q$ we define $A_j = \{(P_a, P_b) \mid P_a, P_b \in \mathcal{P}$ and there is a path from \(\V{P_a}\) to \(\V{P_b}\) inside \(Q_j\) which is internally disjoint from \(\mathcal{P}\)\(\}\).
\end{definition}

We refer the reader to \cref{fig:routing_temporal_digraph} for an illustration of the definition above
and to \cref{sec:obtaining-a-path-of-well-linked-sets} for a discussion of why routing temporal digraphs are defined precisely in this way.

\subsection{\texorpdfstring{$H$}{H}-routings}
\label{sec:H-routings}

Since the connectivity provided by each layer of a temporal digraph can differ considerably, we need some way to describe connectivity involving several layers.
We also want to describe connectivity between specific sets of vertices.
We are particularly interested in deciding when we can construct a path $R$ between two paths $P_1, P_2$ such that $R$ is disjoint from some specific set $\mathcal{P}' \subseteq \mathcal{P}$.

The concept of \emph{$H$-routings}, given in~\cref{def:h-routings} below, allows us to use a digraph $H$ in order to describe which connections we need.
Since both digraphs and temporal digraphs have a concept of paths, we can define $H$-routings for both digraphs and temporal digraphs analogously.

\begin{definition}
	\label{def:h-routings}
	Let $H$ be a digraph, $D$ be a digraph or temporal digraph, and $S\subseteq \V{D}$.
	An \emph{$H$-routing (over $S$)} is a bijection $\varphi : \V{H} \to S$ such that for each $v$-$u$-path $P$ in $H$ we can find a $\varphi(v)$-$\varphi(u)$-path (or temporal path, resp.) in $D$ which is disjoint from $S \setminus \Fkt{\varphi}{\V{P}}$. 
\end{definition}

Note that reachability in temporal digraphs is not transitive, as the example in~\cref{fig:routing_temporal_digraph} demonstrates: in the temporal digraph $T_{1,3} \coloneqq \Brace{\{P_1, P_2, P_3\}, \{A_1, A_3\}}$ containing only the layers $A_1$ and $A_3$, $P_a$ is reachable from $P_c$ and $P_b$ is reachable from $P_a$ but $P_b$ is not reachable from $P_c$.
To get a meaningful concept of $H$-routings, we therefore have to require that not only for every arc in $H$ but also for every path in $H$ there is a temporal path in $T$ with the same start and endpoint.

Compared to other containment relations, \(H\)-routings are rather \emph{weak} in the sense given by the three observations below.
We first recall the definitions of strong immersions and strong minors.

An immersion of a digraph \(H\) in a digraph \(D\) is a function \(\phi\) from vertices and arcs of \(H\) such that
\begin{itemize}
	\item 
		\(\phi(v)\) is a vertex of \(D\) for each \(v \in \V{H}\), and the restriction of φ to \(\V{H}\) is injective.
	\item
		\(\phi(a)\) is a directed path (or cycle) in \(D\) for each \(a ∈ \A{H}\), and if \(b \in \A{H}\) is distinct from \(a\), then \(\phi(a)\) and \(\phi(b)\) are arc-disjoint.
    \item 		
	   If \(a \in \A{H}\) has tail \(u \in \V{H}\) and head \(v \in \V{H}\), then \(\phi(a)\) is a \(\phi(u)\)-\(\phi(v)\)-path, and if \(a\) is a loop, then \(\phi(a)\) is a cycle passing through \(\phi(v)\).
\end{itemize}
An immersion \(\phi\) is \emph{strong} if it additionally satisfies the following condition:
If \(a ∈ \A{H}\) is not incident with \(v ∈ \V{H}\), then \(φ(a)\) does not contain \(φ(v)\).
We say an immersion is \emph{weak} if we do not require it to be strong.

\begin{observation}
If a digraph \(D\) contains \(H\) as a strong immersion, then it also contains an \(H\)-routing.
\end{observation}
\begin{proof}
	Let φ be an immersion of \(H\) in \(D\) and let \(φ'\) be the restriction of φ to \(\V{D}\).
		
	By definition, \(φ'\) is injective.
	Further, for every arc \((u,v) ∈ \A{H}\), the \(φ(u)\)-\(φ(v)\)-path given by \(φ((u,v))\) witnesses the existence of a \(φ'(u)\)-\(φ'(v)\)-path in \(D\) which is disjoint from all vertices in the image of \(φ'\), except for \(u\) and \(v\).
	Hence, \(φ'\) is an \(H\)-routing in \(D\).
\end{proof}

The statement above is not true if we take weak immersions instead, as illustrated in \cref{fig:weak-immersion-vs-H-routing}.

\begin{figure}[!h]
	\centering
	\begin{minipage}[b]{0.38\textwidth}
		\centering
		\begin{tikzpicture}
			\node[vertex,fill = green] (v0) at (0,0) {};
			\node[below = 0.1cm of v0] (l0) {\(v_0\)};
			\node[vertex, fill = blue] (v1) at (1,0) {};
			\node[below = 0.1cm of v1] (l1) {\(v_1\)};
			\node[vertex, fill = purple] (v2) at (2,0) {};
			\node[below = 0.1cm of v2] (l2) {\(v_2\)};
			\node[vertex, fill = red] (v3) at (1,1) {};
			\node[above = 0.1cm of v3] (l3) {\(v_3\)};
			\path[directededge,draw=black] (v0) to (v1);
			\path[directededge,draw=black] (v1) to (v2);
			\path[directededge,draw=black] (v1) to (v3);
			\node (D) at (-0.5,0.5) {\(D\):};

			\begin{scope}[shift={(3.5, 0)}]
			\node[vertex, fill = blue] (u0) at (0,0) {};
			\node[below = 0.1cm of u0] (l0) {\(u_1\)};
			\node[vertex,fill = red] (u1) at (1,0) {};
			\node[below = 0.1cm of u1] (l1) {\(u_3\)};
			\node[vertex, fill = green] (u2) at (0,1) {};
			\node[above = 0.1cm of u2] (l2) {\(u_0\)};
			\node[vertex, fill = purple] (u3) at (1,1) {};
			\node[above = 0.1cm of u3] (l3) {\(u_2\)};
			\path[directededge,draw=black] (u0) to (u1);
			\path[directededge,draw=black] (u2) to (u3);
			\node (H) at (-0.5,0.5) {\(H\):};
			\end{scope}
		\end{tikzpicture}
	\end{minipage}
	\hfill
	\begin{minipage}[b]{0.6\textwidth}
		\caption{The digraph \(D\) contains \(H\) as a weak immersion, as witnessed by the mapping \(u_i \mapsto v_i\).
		However, it does not contain an \(H\)-routing.}
		\label{fig:weak-immersion-vs-H-routing}
	\end{minipage}
\end{figure}

In order to define strong minors, we need to first define a digraph operation.
Let \(D, H\) be two digraphs.
A \emph{strong contraction} of a strongly-connected component \(X\) of \(D\)
is the operation of replacing \(X\) by a single vertex \(x\) and,
for each \(v \in \V{D - X}\),
adding the arc \((v, x)\) if \(\OutN{}{v} \cap X \neq \emptyset\) and
adding the arc \((x, v)\) if \(\InN{}{v}  \cap X \neq \emptyset\).

We say that \(H\) is a \emph{strong minor} of \(D\)
if we can obtain \(H\) from a subgraph of \(D\) 
through strong contractions.

\begin{observation}
	If a digraph \(D\) contains \(H\) as a strong minor,
then it also contains an \(H\)-routing.
\end{observation}
\begin{proof}
	For each vertex in \(v ∈ \V{H}\),
	let \(V_v\) be the set of vertices which were strongly contracted in \(D\)
	when constructing \(H\) as a strong minor.
	Construct a function \(φ : \V{H} \to \V{D}\)
	by fixing some arbitrary \(u_v \in V_v\) for each \(v ∈ \V{H}\) and
	setting \(φ(v) = u_v\).

	Because strongly connected sets of vertices are contracted into a single vertex
	during strong contractions, the sets \(V_v\) defined above are pairwise vertex-disjoint.
	Hence, φ is injective.

	Let \((u, v) ∈ \A{H}\).
	By definition of strong minors, there is some vertex \(u' \in V_u\) and some \(v' \in V_v\)
	such that
	\((u', v') ∈ \A{D}\).
	Since \(V_u\) and \(V_v\) are strongly connected sets,
	there is a path from \(φ(u)\) to \(φ(v)\)
	using only vertices of \(V_u \cup V_v\).
	Thus, \(φ\) is an \(H\)-routing in \(D\).
\end{proof}

\begin{observation}
    If a digraph \(D\) contains \(H\) as a butterfly minor, then it also contains an \(H\)-routing.
\end{observation}
\begin{proof}
    For each vertex \(v \in \V{H}\), let \(V_v\) be the set of vertices of \(D\) which were butterfly contracted into \(v\).
    By definition of butterfly contraction, there is some vertex \(u_v \in V_v\) such that, for all \(x \in V_v\), \(u_v\) can reach \(x\) or \(x\) can reach \(u_v\).
    We then define \(φ(v) = u_v\) for every \(v ∈ \V{H}\).
    The fact that \(φ\) is an \(H\)-routing in \(D\) follows from the fact that the sets \(V_v\) defined above are pairwise vertex-disjoint and the paths connecting \(V_{v}\) to \(V_{u}\) for each arc \((v,u)\) in \(H\) do not intersect other sets \(V_w\) if \(v \neq w \neq u\).
\end{proof}

To see that \(H\)-routings are distinct from the previous containment relations, consider the bidirected star with \(k\) leaves.
It contains a \(\biK{k}\)-routing, but it does not contain even a \(\biK{3}\) as a strong immersion, butterfly minor or strong minor.

We now consider how to apply \(H\)-routings in our setting.
Let us briefly consider the following statement proved by Leaf and Seymour for undirected graphs to motivate our next results. 

\begin{lemma}[{\cite[statement 2.3]{leaf2015tree}}]
	\label{state:many-leaves-or-long-paths}
	Let $r \geq 1$ and $h \geq 3$ be integers, let $G$ be a connected graph with $\Abs{\V{G}} \geq (r + 2)(2h - 5) + 2$.
	Then, $G$ contains one of the following
	\begin{itemize}
		\item a path with $r$ vertices whose internal vertices have degree two in $G$, or
		\item a spanning tree $T$ with at least $h$ leaves.
	\end{itemize}
\end{lemma}

In the undirected setting, both cases of the previous lemma can be useful regarding the connectivity they provide:
For every pair of leaves, a tree contains a path connecting them without intersecting any other leaves, whereas a long path with internal vertices of degree two yields can have a useful special meaning if the graph is an auxiliary graph describing the intersections between some objects,
as such a path provide a unique way to traverse the objects corresponding to the vertices of the path.

In the directed setting, however, neither out-trees nor in-trees provide any connectivity between their leaves.
To get a statement that we can use in the sequel, we therefore have to find an alternative for the trees used in~\cref{state:many-leaves-or-long-paths}.
We first need the following well-known result about acyclic digraphs, often stated in terms of chains and anti-chains in partial orderings.

\begin{observation}
	\label{lem:DagIS}
	Every acyclic digraph $D$ with more than $\ell \cdot p$ vertices but no $\Pk{p}$ as a subgraph contains a set $X \subseteq \V{D}$ of size $\ell$ such that no vertex in $X$ can reach any other vertex in $X$.
\end{observation}
\begin{proof}
	For each $i \geq 0$ let $L_i \coloneqq \{ v \in \V{D} \sth \text{the longest path from a source to $v$ in $T$} \text{ has}\allowbreak \text{length } i \}$. 
	As, by assumption, $D$ does not contain a path of length $p$, $L_i = \emptyset$ for all $i \geq p$.
	Furthermore, by construction, no vertex $v \in L_i$ can reach any other vertex $u \in L_i$, as otherwise the longest path from a source to $u$ would be longer than $i$. 
	
	Every vertex of $D$ lies in some $L_i$.
    By the pigeon-hole principle, at least one $L_i$ must contain at least $\ell$ vertices, proving the claim. 
\end{proof}

In the next~\namecref{lemma:many-strong-components-implies-path-or-Kk}, we establish a simple base case where we are guaranteed to either find a long path or a $\biK{k}$-routing. 
\begin{lemma}
	\label{lemma:many-strong-components-implies-path-or-Kk}
	Let $D$ be a strongly connected digraph.
	Let $s \in \V{D}$ be a vertex such that $D - \Set{s}$ contains at least $kp$ strongly connected components.
	Then, $D$ contains one of the following:
	\begin{enamerate}{B}{item:many-strong-components-implies-path-or-Kk:Kk}
		\item \label{item:many-strong-components-implies-path-or-Kk:path}
		a $\Pk{p}$ as a subgraph, or
		\item \label{item:many-strong-components-implies-path-or-Kk:Kk}
		a $\biK{k}$-routing over some $S \subseteq \V{D}$.
	\end{enamerate}
\end{lemma}
\begin{proof}
	We show that~\cref{item:many-strong-components-implies-path-or-Kk:Kk} holds if~\cref{item:many-strong-components-implies-path-or-Kk:path} does not hold.
	
	Let $T$ be the acyclic digraph of strongly connected components of $D - \Set{s}$.
	As $D$ has no path of length $p$, $T$ also has no such path.
	Thus, by~\cref{lem:DagIS}, $T$ contains a set $X' \subseteq \V{T}$ of size $k$ such that no vertex in $X'$ can reach any other vertex in $X'$ in $T$.
	
	Let $X \subseteq V(D)$ be a set of size $|X| = |X'|$ which contains a vertex $v_i \in \V{C_i}$ for each strong component $C_i \in X'$ of $D - \Set{s}$.
	Let $\varphi : \V{\biK{k}} \to \Set{v_1, v_2, \dots, v_{k}}$ be a bijection.
	We show that $\varphi$ is a $\biK{k}$-routing in $D$.
	
	Let $v_i, v_j \in X$ be two distinct vertices.
	Let $C_t \in \V{T}$ be a sink in $T$ which is reachable from $C_i$ and let $C_r \in \V{T}$ be a source in $T$ which can reach $C_j$.
	
	Since $D$ is strongly connected and $D - \Set{s}$ is not, there is some $u_t \in \V{C_t}$ and some $u_r \in \V{C_r}$ such that $(u_t, s)$ and $(s, u_r)$ are arcs in $D$.
	
	Because $C_i$ can reach $C_t$ and $C_r$ can reach $C_j$ in $T$, there is a $v_i$-$u_t$-path $P_{i,t}$ and a $u_r$-$v_j$-path $P_{r,j}$ in $D$ such that $P_{i,t}$ and $P_{r,j}$ do not intersect any vertex in $X \setminus \Set{v_i, v_j}$.
	Hence, $P_{i,t} \cdot (u_t,s) \cdot (s,u_s) \cdot P_{r,j}$ is a $v_i$-$v_j$-path in $D$ which is disjoint from $X \setminus \Set{v_i, v_j}$.
	
	We conclude that $\varphi$ is a $\biK{k}$-routing in $D$, and so~\cref{item:many-strong-components-implies-path-or-Kk:Kk} holds, as desired.
\end{proof}

We now prove an analogous statement to~\cref{state:many-leaves-or-long-paths} for directed graphs.

\begin{theorem}
	\label{thm:routing star or path}
	Let \boundDef{thm:routing star or path}{n}{k, p} $\bound{thm:routing star or path}{n}{k, p} \coloneqq 2k^2p^3$.
	Every strongly connected digraph $D$ with $\Abs{\V{D}} \geq \bound{thm:routing star or path}{n}{k, p}$  contains one of the following: 
	\begin{enamerate}{Q}{item:routing star or path:Kk}
		\item \label{item:routing star or path:path}
		a $\Pk{p}$ as a subgraph, or
		\item \label{item:routing star or path:Kk}
		a $\biK{k}$-routing over some $S \subseteq \V{D}$.
	\end{enamerate}
\end{theorem}
\begin{proof}
	We show that~\cref{item:routing star or path:Kk} holds if~\cref{item:routing star or path:path} does not hold.
	
	We iterate from $1$ to $2kp$, potentially stopping earlier.
	On step $i$, we construct a vertex sequence $X_i$ and a digraph $D_i$ satisfying all of the following.
	\begin{enamerate}{A}{item:routing star or path:D_i-many-vertices}
		\item \label{item:routing star or path:|X|}
		$\Brace{v_1, v_2, \dots, v_{i}} = X_i$ and so $\Abs{X_i} \geq i$,
		\item \label{item:routing star or path:reach-D_i}
		$D_i$ is strongly connected component of $D_{i-1} - \Set{v_{i}}$ (and so $v_i$ is not on $D_i$),
				\item \label{item:routing star or path:v_i-on-D_i}
		for every $1 \leq j \leq i$, $v_j$ lies on $D_{j - 1}$,
				\item	\label{item:routing star or path:D_i-many-vertices}
		$\V{D_i} \geq (2kp - i)kp^2$.
	\end{enamerate}
	
	Start by setting $X_0$ as the empty sequence and $D_0 = D$.
	Clearly,~\cref{item:routing star or path:|X|,item:routing star or path:reach-D_i,item:routing star or path:v_i-on-D_i,item:routing star or path:D_i-many-vertices} hold for 0 (to simplify notation, we set $D_{-1} \coloneqq D$ and replace $\Set{v_0}$ with the empty set so that~\cref{item:routing star or path:reach-D_i} is well-defined for $i=0$).
	On step $i \leq 2kp$, we consider the following cases.
	
	\textit{Case 1.} There is a $v \in V(D_{i-1})$ such that $D_{i-1} - v$ contains at least $kp$ strongly connected components.
	
	As we assume that~\cref{item:routing star or path:path} does not hold, we know from~\cref{lemma:many-strong-components-implies-path-or-Kk} that~\cref{item:routing star or path:Kk} holds, and we are done with the construction and the proof.
	
	\textit{Case 2.} There is a $v \in V(D_{i-1})$ such that $D_{i-1} - v$ is strongly connected.
	
	Then set $X_i \coloneqq X_{i-1} \cdot (v)$ and set $D_{i} \coloneqq D_{i - 1} - v$.
	It is immediate from the choice of $v$ and our assumption on $D_{i-1} - v$ that~\cref{item:routing star or path:v_i-on-D_i,item:routing star or path:|X|,item:routing star or path:reach-D_i,item:routing star or path:D_i-many-vertices} hold for $i$ because they hold for $i-1$.
	
	\textit{Case 3.} There is a $v \in \V{D_i}$ such that the largest strongly connected component $C$ of $D_{i-1} - v$ has at least $\Abs{\V{D_{i-1}}} - kp^2$ many vertices.
	
	Then set $X_i \coloneqq X_{i-1} \cdot (v)$ and $D_i \coloneqq C$.
	Note that $\Abs{\V{D_{i-1}}} - kp^2 \geq (2kp - i)kp^2$ as~\cref{item:routing star or path:D_i-many-vertices} holds for $i-1$.
	Hence, it is again immediate from the choice of $v$ and from
	our assumption over $C$ that~\cref{item:routing star or
		path:v_i-on-D_i,item:routing star or path:|X|,item:routing
		star or path:reach-D_i,item:routing star or
		path:D_i-many-vertices} hold for $i$ as they also hold for
	$i-1$. 
	
	This completes the case distinction above.
	Now, assume towards a contradiction that none of the three cases above apply and $i \leq 2kp$.
	
	For every $v \in V(D_{i - 1})$,  we know that $D_{i-1} - v$ has fewer than $kp$ strong components because \textit{Case 1} does not apply.
	Further, $D_{i-1} - v$ is not strongly connected, as \textit{Case 2} does not apply.
	Finally, we know that each strong component of $D_{i-1} - v$ has fewer than $|\V{D_{i-1}} - kp^2|$ vertices because \textit{Case 3} does not apply.
	
	For each $v \in V(D_{i-1})$, let $\mathcal{C}_v$ be the set of strong components of $D'_v  \coloneqq  D_{i-1} - v$.
	Let $A_v = \{ (u,w) \sth u, w \in V(D'_v)$ and there is no path from $u$ to $w$ in $D'_v \}$.
	Let $n_v = \Abs{V(D'_v)}$.
	
	Now let $v \in \V{D_{i-1}}$ be arbitrary.
	For any two distinct components $C_1, C_2 \in \CCC_v$, there is no path in $D'_v$ from any  $u \in V(C_1)$ to any $w \in V(C_2)$ or vice versa.
	Thus, $A_v$ contains all possible arcs from vertices in $C_1$ to vertices in $C_2$ or all possible arcs from vertices in $C_2$ to vertices in $C_1$.
	Fixing some arbitrary ordering $\Brace{C_1, C_2, \dots, C_{c}}$ of the elements of $\mathcal{C}_v$, we deduce that
	\begin{align*}
		|A_v| \geq \sum_{\substack{C_a, C_b \in \mathcal{C}_v, \\[0em] a < b}} |V(C_a)| \cdot |V(C_b)|.
	\end{align*}
	
	Let $C \in \CCC_v$ be a strong component of $D'_v$ with the maximal number of vertices among all components in $\CCC_v$.
	The previous inequality implies that $|A_v| \geq |V(C)| \cdot \sum_{C_a \in \mathcal{C}_v \setminus \{C\}} |V(C_a)|$.
	
	By assumption $|V(C)| \leq n_v - kp^2 - 1$, and since the strong components of $D'_v$ form a partition of $D'_v$, we obtain that $\sum_{C_a \in \mathcal{C} \setminus \{ C \}}|V(C_a)| \geq kp^2$.
	Note that $n_v - kp^2 - 1 \geq kp^2$ since~\cref{item:routing star or path:D_i-many-vertices} holds for $i < 2kp$.
	Since $D'_v$ contains fewer than $kp$ strong components, we also obtain $|V(C)| \geq (n_v - 1)/kp$.
	From the inequality above, we obtain 
	\begin{align*}
		|A_v| & \geq  |V(C)| \cdot \sum_{C_a \in \mathcal{C} \setminus \{C\}} |V(C_a)|\\[0em]
		& \geq \frac{n_v - 1}{kp} \cdot \sum_{C_a \in \mathcal{C} \setminus \{C\}} |V(C_a)|\\[0em]
		& \geq \frac{n_v - 1}{kp} \cdot kp^2 = (n_v - 1) \cdot p
	\end{align*}
	Hence, $\sum_{v \in V(D_{i-1})}|A_v| \geq n_v \cdot (n_v - 1) \cdot p$.
	Since $A_v$ does not contain any reflexive tuples and there are $n_v(n_v - 1)$ non-reflexive tuples in the set $V(D_{i-1}) \times V(D_{i-1})$, by the pigeon-hole principle we deduce that there are $u, v \in V(D_{i-1})$ and there are $v_1, \dots, v_{p} \in V(D_{i-1})$ such that $(u, v) \in A_{v_j}$ for all $1 \leq j \leq p$.
	Thus, every path from $u$ to $w$ in $D_{i-1}$ must contain each of the vertices $v_1, \dots v_p$ and thus be of length at least $p+1$ (by definition, $(u,w) \not\in A_u \cup A_w$), a contradiction to the assumption that~\cref{item:routing star or path:path} does not hold, that is, that $D$ does not contain a path of length $p$.
	Hence, one of the three cases must apply, and we can complete the construction above.
	
	If at any point during the construction we end up at Case 1,
	then, as argued above,~\cref{item:routing star or path:Kk} is true and we are done.
	Otherwise, we know that~\cref{item:routing star or path:|X|,item:routing star or path:reach-D_i,item:routing star or path:v_i-on-D_i,item:routing star or path:D_i-many-vertices} hold for $2kp$.
	We now show that~\cref{item:routing star or path:Kk} holds.
	
	Let $\Brace{v_1, v_2, \dots, v_{2kp}} \coloneqq X_{2kp}$.
	We inductively construct disjoint sets $A_i, B_i \subseteq \{ v_1, \dots,\allowbreak v_{2pk}\}$ and we construct for each $v_j \in B_i$ a path $P_j^+ \subseteq D_j$ from $v_j$ to $\V{D_{2kp}}$ and a path $P_j^- \subseteq D_j$ from $\V{D_{2kp}}$ to $v_j$ such that $\big(V(P_j^+) \cup V(P_j^-) \big) \cap A_i = \emptyset$ for all $v_j \in B_i$.
	Finally, $|B_i| = i$, $|A_i| \geq 2kp - 2pi$ and the elements of $B_i$ are contained in $X_i$.
	
	We start by setting $A_0 \coloneqq \{ v_1, \dots, v_{2kp}\}$ and $B_0 \coloneqq \emptyset$, which obviously meet the requirements.
	On step $i < k$, let $j$ be minimal such that $v_j \in A_{i-1}$.
	
	Let $P_j^+$ be a shortest path in $D_{j-1}$ from $v_j$ to a vertex in $D_{2kp}$ and let $P_j^-$ be a shortest path from a vertex in $D_{2kp}$ to $v_j$, again in $D_{j - 1}$.
	As $D_{j - 1}$ is strongly connected and $D_{2kp} \subseteq D_{j - 1}$ due to~\cref{item:routing star or path:reach-D_i}, such paths exist.
	Furthermore, $P_j^+$ and $P_j^-$ are both of length at most $p$ and all internal vertices of $P_j^+$ and $P_j^-$ are disjoint from $X_{j-1}$ (and so from $B_{i - 1}$) due to~\cref{item:routing star or path:reach-D_i}.
	
	We define $B_{i} = B_{i - 1} \cup \{ v_j \}$ and $A_{i} = A_{i - 1} \setminus \big( V(P_j^+ \cup P_j^-) \big)$.
	As $|V(P_j^+)|, |V(P_j^-)| < p$ and $|A_{i - 1}| \geq 2kp - 2p(i + 1)$,
	we immediately get that $|A_{i}| \geq 2kp - 2pi$, as required.
	Clearly, the other conditions are satisfied as well.
	
	The construction stops after $k$ steps with a set $B_k$ such that
	for each $v_j \in B_k$ there are paths $P^+_j, P^-_j$ such that $P^+_j$ is a $v_j$-$\V{D_{2kp}}$-path and $P^-_j$ is a $\V{D_{2kp}}$-$v$-path and both paths are internally disjoint from $B_k$.
	
	Let $\varphi : \V{\biK{k}} \to B_k$ be a bijection.
	We show that $\varphi$ is a $\biK{k}$-routing in $D$.
	To see this, let $v_i, v_{j} \in B_k$ and let $P_{i,j}$ be a path in $D_{2kp}$ from the end vertex of $P_i^+$ to the start of $P_j^-$.
	Then $P_i^+ \cup P_{i,j} \cup P_j^-$ contains a $v_i{-}v_j$-path disjoint from $B_k$.
	Hence,~\cref{item:routing star or path:Kk} holds, concluding the proof of the theorem.
\end{proof}

\subsection{Finding \texorpdfstring{$\Pk{k}$}{P\_k}-routings in temporal digraphs}

Of particular interest to us are $H$-routings in temporal digraphs
where $H$ is just a simple path $\Pk{k}$.
This is natural in our context, as, for example, an acyclic grid is nothing but a
\say{horizontal} linkage $\QQQ \coloneqq \Brace{Q_1, \dots, Q_q}$  that intersects a sequence $\PPP \coloneqq \{P_1, \dots, P_p\}$ of pairwise disjoint digraphs $P_i$ in the order $P_1, \dots, P_p$ where each $P_i$ happens to be a simple path intersecting the paths in $\QQQ$ in the order $Q_1, \dots, Q_q$. 

Our next goal, therefore, is to identify properties of the
individual layers of a temporal digraph $T$ that guarantee the existence of
a $\Pk{k}$-routing in $T$.

Recall from~\cref{sec:preliminaries} that a  digraph $D$ is 
\emph{\onewayconnected} if, for every pair $u, v \in V(D)$ of vertices,
$v$ can reach $u$ or $u$ can reach $v$. 
As proved in \cite[Theorem~3.10]{zbMATH03226832}, a digraph $D$ is
unilateral if and only if there is a walk visiting all the vertices of $D$.
	
We need a stronger characterisation of unilateral digraphs for our results, and we need the following classic result of R\'{e}dei to prove it.

\begin{theorem}[\cite{redei1934kombinatorischer}, Theorem 1]
	\label{thm:redei}
	Given a tournament $T$, the number of Hamiltonian paths in $T$ is odd.
\end{theorem}

In particular,~\cref{thm:redei} implies that a tournament always contains at least one Hamiltonian path.

\begin{lemma}
	\label{lemma:one-way connected iff hamiltonian walk}
	A digraph $D$ is \onewayconnected{} if and only if for every $S \subseteq \V{D}$ there is a walk $W$ of length at most $kn$ in $D$ with $S \subseteq \V{W}$, where $k = \Abs{S}$ and $n = \Abs{\V{D}}$.
\end{lemma}

\begin{proof}
	If there is a walk $W$ in $D$ with $\V{W} = \V{D}$, then $D$ is clearly \onewayconnected.
	
	Now assume $D$ is \onewayconnected{} and let $S = \{v_1, v_2, \dots v_{k}\} \subseteq \V{D}$.
	We construct an auxiliary tournament $T$ with $V(T) =S$ in the following way. 
	Given vertices $v_i, v_j  \in S$, with $i < j$, if there exists in $D$ a path $P_{i,j}$ in $D$ starting in $v_i$ and ending in $v_j$ we add the arc $v_iv_j$ to $E(T)$.
	Otherwise, there must be a path $P_{j,i}$ starting in $v_j$ and ending in $v_i$, and we add the arc $v_jv_i$ to $E(T)$. 
	From~\cref{thm:redei}, $T$ must contain a Hamiltonian path $P$.
	From $P$ we construct the wanted walk $W$ by replacing the arc $v_iv_j$ of $P$ with the path $P_{i,j}$, for all $1 \leq i,j \leq k$.
	Clearly, $W$ is a walk in $D$, and its length is at most $kn$.  
	
	Thus, the walk $W_k$ satisfies the statement of the lemma.
\end{proof}

Finding long walks in unilateral digraphs is easy. The task becomes more complicated in temporal digraphs as the connectivity provided by individual layers may differ significantly.
As we show next, one direction of the previous lemma can be retained in the temporal setting.
Observe that \(\bound{lemma:temporal one-way connected contains walk with many vertices}{\Lifetime{}}{n,k} \in \Oh(k^2 n^{ k n + 2})\).

\begin{lemma}
	\label{lemma:temporal one-way connected contains walk with many vertices}
	Let $\bound{lemma:temporal one-way connected contains walk with many vertices}{\Lifetime{}}{n,k} \coloneqq kn \sum_{i=1}^{kn}n^i$.\boundDef{lemma:temporal one-way connected contains walk with many vertices}{\Lifetime{}}{n,k}
	Let $T$ be a temporal digraph with $n$ vertices where each layer is unilateral, and let $S \subseteq \V{T}$ be a set of size $k$.
	If $\Lifetime{T} \geq \bound{lemma:temporal one-way connected contains walk with many vertices}{\Lifetime{}}{n,k}$, then $T$ contains a temporal walk $W$ with $S \subseteq \V{W}$.
\end{lemma}
\begin{proof}
	By~\cref{lemma:one-way connected iff hamiltonian walk}, for each $1 \leq i \leq \Lifetime{T}$ there is a walk $W_i$ of length at most $kn$ in $\Layer{T}{i}$ such that $S \subseteq \V{W_i}$.
	Note that there are $\sum_{i=1}^{kn}n^i$ distinct walks of length at most $kn$ over the vertex set of $T$.
	
	As $\Lifetime{T} \geq kn \sum_{i=1}^{kn}n^i$, by the pigeon-hole principle, there is some walk $W'$ which appears on at least $kn$ different layers.
	Let $t_1, t_2, \dots, t_{kn}$ be \timesteps such that each $\Layer{T}{t_i}$ contains the walk $W'$.
	Now set $W \coloneqq \Brace{(v_i, t_i) \mid 1 \leq i \leq kn \text{ and $v_i$ is the $i$th vertex on $W'$}}$.
	
	Since $\V{W} = \V{W'}$, we also have $S \subseteq \V{W}$, as desired.
\end{proof}

The next lemma establishes a special case where a temporal digraph is guaranteed to contain a $\Pk{k}$-routing.
Together with~\cref{lemma:temporal one-way connected contains walk with many vertices}, this implies~\cref{theorem:one-way connected temporal digraph contains P_k routing}.

\begin{lemma}
	\label{lemma:walk with many vertices implies P_k routing}
	Let $D$ be a temporal digraph and $W$ be a temporal walk in $D$.
	If $\Abs{\V{W}} \geq \frac{k^2}{2} $, then $W$ contains a $\Pk{k}$-routing.
\end{lemma}
\begin{proof}
	If $W$ contains a $\Pk{k}$ as a temporal subpath, then this subgraph also contains a $\Pk{k}$-routing.
	So, assume $W$ does not contain any $\Pk{k}$ as a temporal subpath.
	
	Let $W'$ be a minimal temporal subwalk of $W$ such that $V(W') = V(W)$ and consider a maximal set $S = \Brace{s_1, s_2, \dots, s_{r}}$ of vertices appearing only once on $W'$. 
	
	If $\Abs{S} \geq k$ then $W'$ contains a $\Pk{k}$ routing on $S$.
	Therefore, we can assume $\Abs{S} \leq k -1$. 
	Let us denote by  $W'_i$ with $1 \leq i < r$ the temporal subwalk of $W'$ starting in $s_{i}$ and ending in $s_{i+1}$.
    Further, we denote by $W'_0$ the subwalk of $W'$ starting in  $\Start{W'}$ and ending in $s_1$, and by $W'_r$ the temporal subwalk of $W'$ starting in $s_r$ and ending in $\End{W'}$.
	We claim that, for all $0 \leq i \leq r$, all vertices of $W'_i$ appear once on $W'_i$. 
	Suppose not and consider a vertex $v$ appearing twice on $W'_i$.
    We denote by $W_v$ the temporal subwalk joining the first and second occurrence of $v$ on $W'_i$. 
	Since $W'$ is minimal, $W_v$ must contain a vertex appearing only once on $W'$, otherwise the concatenation of the temporal subwalk of $W'$ going from $\Start{W'}$ to the first occurrence of $v$ in $W'_i$ and from the second occurrence of $v$ in $W'_i$ to $\End{W'}$ would be a shorter subwalk of $W$ containing the same amount of vertices, a contradiction to the minimality of $W'$. 
	Hence, the subwalks $W'_{i}$ do not contain repeated vertices. 
	Since $W$ does not contain any $\Pk{k}$ as a temporal subpath, $\Abs{V(W_i)} \leq k - 1$.
	Further, for all $0 \leq i \leq r$, the internal vertices of $W'_i$ must occur in a subwalk $W'_j$ for some  $j \neq i$ with $0 \leq j \leq r$ and  $\Abs{\bigcup_{0\leq i \leq r}V(W'_i)} \leq \frac{k(k - 2)}{2} + k - 1$. 
	Since $\bigcup_{0\leq i \leq r}V(W'_i) = V(W')$ and $\Abs{V(W')} \geq \frac{k^2}{2}$ this is a contradiction.
\end{proof}

The previous two lemmas immediately imply the following result.
Observe that \(\bound{theorem:one-way connected temporal digraph contains P_k routing}{\Lifetime{}}{n,k} \in \PowerTower{1}{\Polynomial{6}{k, n}}\).

\begin{theorem}
	\label{theorem:one-way connected temporal digraph contains P_k routing}
	Let $\bound{theorem:one-way connected temporal digraph contains P_k routing}{\Lifetime{}}{n, k} \coloneqq \bound{lemma:temporal one-way connected contains walk with many vertices}{\Lifetime{}}{n, \frac{k^2}{2}}$\boundDef{theorem:one-way connected temporal digraph contains P_k routing}{\Lifetime{}}{n, k}. 
	Let $T$ be a temporal digraph where each layer is unilateral.
	If $\Lifetime{T} \geq \bound{theorem:one-way connected temporal digraph contains P_k routing}{\Lifetime{}}{n, k}$ and $n \coloneqq \Abs{\V{T}} \geq \frac{k^2}{2}$, then there is some set $S \subseteq \V{T}$ such that $T$ contains a $\Pk{k}$-routing over $S$. 
\end{theorem}
\begin{proof}
	Let $S' \subseteq \V{D}$ with $\Abs{S'} = \frac{k^2}{2}$.
	By~\cref{lemma:temporal one-way connected contains walk with many vertices}, $T$ contains a temporal walk $W$ such that $S' \subseteq \V{W}$.
	In particular, $\Abs{\V{W}} \geq \frac{k^2}{2}$.
	By~\cref{lemma:walk with many vertices implies P_k routing}, there is some $S \subseteq \V{W}$ such that $W$ and, hence, $T$ contain a $\Pk{k}$ routing over $S$.
\end{proof}

\subsection{Finding \texorpdfstring{$\Ck{k}$}{C\_k} and \texorpdfstring{$\biPk{k}$}{P\_k}-routings in temporal digraphs}

As discussed at the beginning of the previous subsection, the existence of $\Pk{k}$-routings in a routing temporal digraph relates to the connectivity provided by the columns of an acyclic grid, which only allows routing from top to bottom and from left to right.
If we consider a fence instead of an acyclic grid, the fence allows us to route upwards as well as downwards, as the columns alternate in direction.
Two consecutive columns taken together allow one to go from any row to any other row and in this way resemble a strongly connected digraph like a cycle $\Ck{k}$ or a bidirected $\biPk{k}$ much more than a $\Pk{k}$.
In this section, we aim to find $H$-routings that provide this higher level of connectivity.

We first define 
\begin{align*}
	\boundDefAlign{lemma:walk with many vertices implies choosable P_k or biP_k routing}{s}{k_1, k_2}
	\bound{lemma:walk with many vertices implies choosable P_k or biP_k routing}{s}{k_1, k_2} 
	& \coloneqq 6k_1(k_2)^2 - 8k_1k_2 + 2k_1 - 2(k_2)^2 + 3k_2
\end{align*}
and prove the following technical~\namecref{lemma:walk with many vertices implies choosable P_k or biP_k routing}.
Note that \(\bound{lemma:walk with many vertices implies choosable P_k or biP_k routing}{s}{k_1,k_2} \in \Oh(k_{1} (k_{2})^{2})\).

\begin{lemma}
	\label{lemma:walk with many vertices implies choosable P_k or biP_k routing}
	Let $T$ be a temporal digraph, let $W$ be a temporal walk in $T$, let $k_1, k_2$ be integers, and let $S \subseteq \V{W}$ be a set of size at least $\bound{lemma:walk with many vertices implies choosable P_k or biP_k routing}{s}{k_1, k_2}$. 
	Then there is some $S' \subseteq S$ such that one of the following is true:
	\begin{enamerate}{R}{item:choosable-P_k:P_k}
		\item \label{item:choosable-P_k:biP_k}
		      $W$ contains a $\biPk{k_1}$-routing over $S'$, or
		\item \label{item:choosable-P_k:P_k}
		      there are (possibly arcless) walks $W_1, W_a, W_b, W_c$ in $D$ such that $W_1$ is a subwalk of $W$ leaving and arriving at the same \timesteps as $W$, $W_a \cdot W_b \cdot W_c = W_1$, $W_a$ and $W_c$ are internally disjoint from $S'$, and $W_b$ contains a $\Pk{k_2}$-routing over $S'$ where the first vertex of the $\Pk{k_2}$ is mapped to $\Start{W_b}$ and the last vertex of the $\Pk{k_2}$ is mapped to $\End{W_b}$. 
	\end{enamerate}
\end{lemma}
\begin{proof}
	To simplify arithmetic steps, we define $k_3 = (k_1 - 1)(k_2 - 1)$ and $s_1 = 2(k_3 + k_2) + k_2 - 1$.
	Note that $(k_2 - 1)(s_1 - 1) + k_2(4k_3 + k_2) = 6k_1(k_2)^2 - 8k_1k_2 + 2k_1 - 2(k_2)^2 + 3k_2 = \bound{lemma:walk with many vertices implies choosable P_k or biP_k routing}{s}{k_1, k_2} \leq \Abs{S}$. 
	
	In the following claim, we identify a base case for the proof, which is used several times later.
	
	\begin{claim}
		\label{claim:split-walk-implies-choosable-Pk-or-biPk-routing}
		Let $\hat{W} = \hat{W}_a \cdot \hat{W}_b \cdot \hat{W}_c$ be a temporal walk inside $W$ such that $\Start{\hat{W}} = \Start{W}$ and $\End{\hat{W}} = \End{W}$ and let $\hat{S} \subseteq \V{\hat{W}_a} \cap \V{\hat{W}_c} \cap S$ be a set such that each vertex of $\hat{S}$ appears exactly once on $\hat{W}_a$ and exactly once on $\hat{W}_c$. 
		If $|\hat{S}| \geq k_3 + 1$, then there is some $S' \subseteq \hat{S}$ such that~\cref{item:choosable-P_k:biP_k} or~\cref{item:choosable-P_k:P_k} holds. 
	\end{claim}
	\begin{claimproof}
		Since each vertex of $\hat{S}$ appears exactly once on $\hat{W}_a$ and exactly once on $\hat{W}_c$, each of these walks induces an ordering over the vertices of $\hat{S}$.
		By~\cref{thm:erdos_szekeres}, we obtain two cases.
		
		\textbf{Case 1:} There is some $S' \subseteq \hat{S}$ of size $k_1$ such that the vertices of $S'$ appear on $\hat{W}_c$ in the reverse order compared to their order on $\hat{W}_a$.
		
		Let $W_a$ be a shortest temporal subpath of $\hat{W}_a$ containing every vertex of $S'$, and let $W_c$ be a shortest temporal subpath of $\hat{W}_c$ containing every vertex of $S'$.
		Note that $\End{W_a} = \Start{W_c}$.
		We show that $W_a \cdot W_c$ contains a $\biPk{k_1}$-routing over $S'$.
		Let $\Set{u_1, u_2, \dots, u_{k_1}}$ be the vertices of the $\biPk{k_1}$ sorted according to their occurrence on the path.
		
		Let $u_i, u_j \in \Set{u_1, u_2, \dots, u_{k_1}}$.
		If $i < j$, then $W_a$ contains a $u_i$-$u_j$-path avoiding $S' \setminus \Set{u_i, \dots, u_j}$.
		If $j > i$, then $W_c$ contains a $u_i$-$u_j$-path avoiding $S' \setminus \Set{u_j, \dots, u_{i}}$.
		Since both $W_a$ and $W_c$ are temporal paths, we have that $W_a \cdot W_c$ contains a $\biPk{k_1}$-routing over $S'$, satisfying~\cref{item:choosable-P_k:biP_k}.
		
		\textbf{Case 2:} There is some $S' \subseteq \hat{S}$ of size $k_2$ such that the vertices of $S'$ appear in $\hat{W}_c$ in the same order as in $\hat{W}_a$.
		
		Let $W_b$ be the shortest temporal subpath of $\hat{W}_a$ containing every vertex of $S'$.
		Let $W_a$ be a temporal $\Start{\hat{W}}$-$\Start{W_b}$-path in $\hat{W}_a$ and let $\hat{W}_c$ be a temporal $\End{W_b}$-$\End{\hat{W}}$-path in $\hat{W}_c$.
		
		Since every vertex of $S'$ appears exactly once in \(W_b\), $W_b$ contains a $\Pk{k_2}$-routing over $S'$ where the first vertex of the $\Pk{k_2}$ is mapped to $\Start{W_b}$ and the last vertex of the $\Pk{k_2}$ is mapped to $\End{W_b}$.
		Further, $W_a$ and $W_c$ are internally disjoint from $S'$.
		Thus, the temporal walk $W_1 \coloneqq W_a \cdot W_b \cdot W_c$ satisfies~\cref{item:choosable-P_k:P_k}.
	\end{claimproof}
	
	Let $W'$ be a minimal temporal subwalk of $W$ such that $S \subseteq \V{W'}$.
	We say that a temporal subwalk $R$ of $W$ is a \emph{return (around a vertex $u \in S$)} if $u$ appears exactly twice on $R$, $R$ starts and ends on $u$ and all vertices of $(\V{R} \cap S) \setminus \Set{u}$ appear exactly once on $R$.
	Note that, by minimality of $W'$, each return $R$ around a vertex $u$ must contain a vertex $u' \in S$ whose only occurrence on $W'$ is on $R$.
	
	Let $R$ be a return in $W'$ maximising the cardinality of $S_1 \coloneqq \V{R} \cap S$.
	We consider two cases.
	
	\textbf{Case 1:} $\Abs{S_1} \leq s_1 - 1$.
	
	We decompose $W'$ into $Q_1 \cdot R_1 \cdot Q_2 \cdot R_2 \cdot \ldots \cdot Q_x \cdot R_x \cdot Q_{x + 1} = W'$, where each $Q_i$ is a temporal walk on which no vertex of $S$ appears twice and each $R_i$ is a return in $W'$.
	By definition of return, such a decomposition is unique.
	We distinguish between two subcases.
	
	\textbf{Case 1.1:} $x \geq k_2$.
	
	For each $1 \leq i \leq k_2$ let $u_i \in \V{R_i} \cap S$ be a vertex which occurs exactly once on $W'$.
	Let $W_a$ be a temporal $\Start{W'}$-$u_1$-path in $W'$, let $W_c$ be a temporal $u_{k_2}$-$\End{W'}$-path in $W'$ and let $W_b$ be a temporal $u_1$-$u_{k_2}$-walk in $W'$ which contains every vertex of $S' \coloneqq \Set{u_1, u_2, \dots, u_{k_2}}$ exactly once.
	
	Because each vertex of $S'$ appears exactly once in \(W_b\), $W_b$ contains a $\Pk{k_2}$-routing over $S'$ where $u_1 = \Start{W_b}$ is the first vertex of the $\Pk{k_2}$ and $u_{k_2} = \End{W_b}$ is the last vertex of the $\Pk{k_2}$.
	Further, $W_a$ and $W_c$ are temporal walks that are internally disjoint from $S'$.
	Hence, $W_1 \coloneqq W_a \cdot W_b \cdot W_c$ satisfies~\cref{item:choosable-P_k:P_k}.
	
	\textbf{Case 1.2:} $x < k_2$.	
	
	Since $\Abs{S} \geq (k_2 - 1)(s_1 - 1) + k_2(4k_3 + k_2)$ and $\Abs{\V{R_i} \cap S} \leq \Abs{S_1} \leq s_1 - 1$ for every $1 \leq i \leq x$, there is some $1 \leq b \leq x$ such that $\Abs{\V{Q_b} \cap S} \geq (\Abs{S} - (k_2 - 1)(s_1 - 1)) / k_2 \geq 4k_3 + k_2$.
	
	Let $t_1$ be the \timestep in which $Q_b$ departs, and let $t_2$ be the \timestep in which $Q_b$ arrives.
	Let $Q_a'$ be a temporal $\Start{W'}$-$\Start{Q_b}$-path in $W'$ arriving at $t_1$ and let $Q_c'$ be a temporal $\End{Q_b}$-$\End{W'}$-path in $W'$ departing at $t_2$.
	Let $S_a = \V{Q_a'} \cap S$, $S_b = \V{Q_b} \cap S$ and $S_c = \V{Q_c'} \cap S$.
	We consider three further subcases.
	
	\textbf{Case 1.2.1:} $\Abs{S_b \setminus (S_a \cup S_c)} \geq k_2$.
	
	Let $S' \subseteq S_b \setminus (S_a \cup S_c)$ be a set of size $k_2$.
	Let $W_b$ be a shortest walk inside $Q_b$ containing every vertex of $S'$ exactly once.
	This is possible by the construction of $Q_b$.
	Let $W_a$ be a temporal $\Start{W'}$-$\Start{W_b}$-walk in $Q_a' \cdot Q_b$ and let $W_c$ be a temporal $\End{W_b}$-$\End{W'}$-walk in $Q_b \cdot Q_c'$.
	
	By construction, we have that $W_a$ and $W_c$ are internally disjoint from $S'$.
	Further, since each vertex of $S'$ appears exactly once in \(W_b\), $W_b$ contains a $\Pk{k_2}$-routing over $S'$ where the first vertex of the $\Pk{k_2}$ is mapped to $\Start{W_b}$ and the last vertex of the $\Pk{k_2}$ is mapped to $\End{W_b}$.
	Hence, $W_1 \coloneqq W_a \cdot W_b \cdot W_c$ satisfies~\cref{item:choosable-P_k:P_k}.
	
	\textbf{Case 1.2.2:} $\Abs{(S_b \cap S_a) \setminus S_c} \geq k_3 + 1$ or $\Abs{(S_b \cap S_c) \setminus S_a} \geq k_3 + 1$.
	
	Without loss of generality, we assume that $\Abs{(S_b \cap S_a) \setminus S_c} \geq k_3 + 1$.
	The other case follows analogously.
	
	Let $S_2 = (S_b \cap S_a) \setminus S_c$.
	Let $\hat{W}_c = Q_b \cdot Q_c'$.
	Each vertex of $\hat{S}$ appears exactly once on $Q_a$ and exactly once on $\hat{W}_c$.
	Hence, by~\cref{claim:split-walk-implies-choosable-Pk-or-biPk-routing}, there is some $S' \subseteq S_2$ such that~\cref{item:choosable-P_k:biP_k} or~\cref{item:choosable-P_k:P_k} holds.
	
	\textbf{Case 1.2.3:} The conditions of \textbf{Case 1.2.1} and \textbf{Case 1.2.2} do not apply.
	
	We first show that $\Abs{S_a \cap S_c} \geq k_3 + 1$.
	Since $\Abs{S_b \setminus (S_a \cup S_c)} \leq k_2 - 1$ and $\Abs{S_b} \geq 4k_3 + k_2$, we have that $\Abs{S_b \cap (S_a \cup S_c)} \geq 4k_3 + 1$.
	Hence, $\Abs{S_b \cap S_a} \geq 2k_3 + 1$ or $\Abs{S_b \cap S_c} \geq 2k_3 + 1$.
	
	Assume, without loss of generality, that $\Abs{S_b \cap S_a} \geq 2k_3 + 1$ holds.
	Because $\Abs{(S_b \cap S_a) \setminus S_c} \leq k_3$, we have that $\Abs{S_a \cap S_c} \geq k_3 + 1$, as desired.
	
	Let $S_2 \subseteq S_a \cap S_c$ be a set of size $k_3 + 1$.
	Since $Q_a'$ and $Q_c$ are temporal paths, each vertex of $S_2$ appears exactly once on $Q_a$ and exactly once on $Q_c$.
	Hence, by~\cref{claim:split-walk-implies-choosable-Pk-or-biPk-routing}, there is some $S' \subseteq S_2$ such that~\cref{item:choosable-P_k:biP_k} or~\cref{item:choosable-P_k:P_k} holds.
	
	\textbf{Case 2:} $\Abs{S_1} \geq s_1$.
	
	Let $Q_a, Q_b$ be two temporal paths inside $W'$ such that $Q_a \cdot R \cdot Q_b$ is a walk starting at $\Start{W'}$ and ending at $\End{W'}$.
	Note that, as $\Start{R} = \End{R}$, $Q_a \cdot Q_b$ is also a temporal walk.
	Let $S_2 \subseteq S_1$ be the vertices of $S_1$ which occur exactly once on $Q_a \cdot R \cdot Q_b$.
	
	\textbf{Case 2.1:} $\Abs{S_2} \geq k_2$.
	
	Let $S' \subseteq S_2$ be a set of size $k_2$.
	Let $W_b$ be the shortest temporal subpath of $R$ which contains every vertex in $S'$.
	As every internal vertex of $R$ appears exactly once on $R$, such a path $W_b$ exists.
	Now let $W_a$ be a temporal $\Start{W'}$-$\Start{W_b}$-path inside $W'$ and let $W_c$ be an temporal $\End{W_b}$-$\End{W'}$-path inside $W'$.
	The temporal paths $W_a$ and $W_c$ are internally disjoint from $S'$ since the only occurrence of the vertices of $S'$ along $W'$ is on $R$.
	Since $W_b$ is a temporal path containing every vertex of $S'$, it also contains a $\Pk{k_2}$-routing $\varphi$ over $S'$.
	By choice of $W_b$, we also have that $\varphi$ maps the first vertex of the $\Pk{k_2}$ to $\Start{W_b}$ and the last vertex of the $\Pk{k_2}$ to $\End{W_b}$.
	By setting $W_1 \coloneqq W_a \cdot W_b \cdot W_c$, we satisfy~\cref{item:choosable-P_k:P_k}.
	
	\textbf{Case 2.2:} $\Abs{S_2} < k_2$.
	
	Because $\Abs{S_1} \geq s_1 = 2(k_3 + k_2) + k_2 - 1$, we have that $\Abs{S_1 \cap (\V{Q_a} \cup \V{Q_b})} = \Abs{S_1 \setminus S_2} \geq 2(k_3 + k_2)$.
	Hence, $\Abs{\V{Q_a} \cap S_1} \geq k_3 + k_2$ or $\Abs{\V{Q_b} \cap S_1} \geq k_3 + k_2$.
	
	We assume without loss of generality that $\Abs{\V{Q_a} \cap S_1} \geq k_3 + k_2$, as the other case follows analogously.
	Since every vertex of $\V{Q_a} \cap S_1$ is either in $\V{Q_b}$ or not, we obtain two subcases.
	
	\textbf{Case 2.2.1:} $\Abs{(\V{Q_a} \cap S_1) \setminus \V{Q_b}} \geq k_2$.
	
	Let $S' \subseteq (\V{Q_a} \cap S_1) \setminus \V{Q_b}$ be a set of size $k_2$ and let $W_b$ be a minimal temporal subpath of $Q_a$ containing every vertex of $S'$.
	Let $W_a$ be a temporal $\Start{W'}$-$\Start{W_b}$-walk in $Q_a \cdot Q_b$ and let $W_c$ be a temporal $\End{W_b}$-$\End{W'}$-walk in $Q_a \cdot Q_b$.
	
	Every vertex of $S'$ appears exactly once in the temporal walk $Q_a \cdot Q_b$.
	Since every occurrence of $S'$ is on $W_b$ and $W_b$ is a temporal path, we have that $W_b$ contains a $\Pk{k_2}$-routing over $S'$.
	Since $W_b$ was chosen minimal, the first vertex of the $\Pk{k_2}$ is mapped to $\Start{W_b}$ and the last vertex of the $\Pk{k_2}$ is mapped to $\End{W_b}$.
	Further, $W_1 = W_a \cdot W_b \cdot W_c$ and $W_a$ and $W_b$ are internally disjoint from $S'$, satisfying~\cref{item:choosable-P_k:P_k}.
	
	\textbf{Case 2.2.2:} $\Abs{\V{Q_a} \cap \V{Q_b} \cap S_1} \geq k_3 + 1$.
	
	Let $S_3 \subseteq \V{Q_a} \cap \V{Q_b} \cap S_1$ be a set of size $k_3 + 1$.
	As $Q_a$ and $Q_b$ are temporal paths, each vertex of $S_3$ appears exactly once in each of those temporal paths.
	Hence, by~\cref{claim:split-walk-implies-choosable-Pk-or-biPk-routing}, there is some $S' \subseteq S_3$ such that~\cref{item:choosable-P_k:biP_k} or~\cref{item:choosable-P_k:P_k} holds.
\end{proof}

The next lemma allows us to transfer the strong connectivity of each individual layer of a temporal digraph to the temporal digraph as a whole. 

\begin{lemma}
	\label{lem:strongly connected temporal digraph contains u-v path}
	Let $T$ be a temporal digraph in which each layer is strongly connected.
	If $\Lifetime{T} \geq \Abs{\V{T}} - 1$, then every $u \in \V{T}$ temporally reaches every $v \in \V{T}$.
\end{lemma}
\begin{proof}
	Let $u \in \V{T}$ and let $n = \Abs{\V{T}} - 1$.
	
	For each $0 \leq i \leq n$ let $R_i$ be the set of vertices of $T$ which $u$ temporally reaches in at most $i$ \timesteps.
	Clearly $u \in R_0$ and so $\Abs{R_0} = 1$.
	Further, $\Abs{R_i} \leq \Abs{R_j}$ if $i \leq j$.
	
	We show that, for every $0 \leq i < n$, if $\Abs{R_i} < \V{T}$, then $\Abs{R_{i+1}} > \Abs{R_{i}}$.
	Let $R_i$ be such a set and let $X = \V{T} \setminus \V{R_i}$.
	Since $\Layer{T}{i+1}$ is strongly connected, there is some $w \in R_i$ and some $v \in X$ such that $(w, v) \in \A{\Layer{T}{i+1}}$.
	By assumption, there is a temporal walk $W$ from $u$ to $w$ within the first $i$ layers in $D$.
	Extending $W$ with the arc $\Brace{w,v}$ allows us to obtain a walk from $u$ to $v$ within the first $i+1$ layers.
	Hence, $\Abs{R_{i+1}} > \Abs{R_i}$.
	
	Since $n \geq \Abs{\V{T} - 1}$ and $\Abs{R_0} = 1$, we have that $\Abs{R_n} = n+1$.
	Thus, $u$ temporally reaches every $v \in \V{T}$.
\end{proof}

\subsection{Making the layers strongly connected}

We are almost ready to prove the main result of this part, which allows
us to construct $\Ck{k}$- or $\biPk{k}$-routings in temporal digraphs
with strongly connected layers. But first, we need the following \namecref{lemma:temporal strongly connected contains walk with many vertices}.

\begin{lemma}
	\label{lemma:temporal strongly connected contains walk with many vertices}
	Let $D$ be a temporal digraph in which each layer is strongly
	connected, let $S \subseteq \V{D}$, let $v \in \V{D}$ and let $s \in S$. 
	If $\Lifetime{D} \geq \Abs{S} \cdot (\Abs{\V{D}} - 1)$, then $D$ contains a
	temporal $v$-$s$-walk $W$ with $S \subseteq \V{W}$.	 
\end{lemma}
\begin{proof}
	Let $\Set{ s_1, \dots, s_k } \coloneqq S$ be an arbitrary ordering of $S$ such that $s = s_k$, and let $n \coloneqq \Abs{\V{D}}$.
	We iteratively construct temporal walks $W_1, W_2, \dots, W_{k}$ such
	that $W_i$ is a walk from $v$ to $s_i$ within the first $i \cdot (n - 1)$
	layers and $W_i$ contains $s_1, \dots, s_i$. 
	
	Start by taking some temporal $v$-$s_1$-walk $W_1$ within the first $n - 1$ layers.
	By~\cref{lem:strongly connected temporal digraph contains u-v path}, such a walk exists.
	
	On step $i \geq 2$, let $W'$ be the temporal $s_{i-1}$-$s_i$-walk from
	layer $i \cdot (n - 1) + 1$ to layer $(i+1) \cdot (n - 1)$ in $D$. 
	By~\cref{lem:strongly connected temporal digraph contains u-v path}, such a walk exists.
	Now set $W_{i+1} = W_i \cdot W'$.
	Since $W_i$ arrives on $\End{W_i} = \Start{W'}$ on \timestep $i(n-1)$
	and $W'$ leaves $\Start{W'}$ on \timestep $i(n-1) + 1$, we have that
	$W_{i+1}$ is a temporal $v$-$s_{i+1}$-walk as desired. 
	
	Thus, the walk $W_k$ is a temporal $v$-$s$-walk within the first
	$\Abs{S} \cdot (n - 1)$ layers which contains all vertices of $S$. 
\end{proof}

We are now ready to prove the following result, which guarantees an $H$-routing for some $H \in \Set{ \biPk{k}, \Ck{k} }$ in any temporal digraph of sufficiently large lifetime as long as each layer is strongly connected.
Moreover, we even have some control over the vertex set of the $H$-routing.
Note, however, that we have no control over which of the two possible routings we obtain.

We define the following functions:
\begin{align*}
	\boundDefAlign{theorem:strongly connected temporal digraph contains H routing}{s}{k}
		\bound{theorem:strongly connected temporal digraph contains H routing}{s}{k} & \coloneqq \bound{lemma:walk with many vertices implies choosable P_k or biP_k routing}{s}{k, \bound{lemma:walk with many vertices implies choosable P_k or biP_k routing}{s}{k,(k-1)^2 + 1}},\\[0em]
	\boundDefAlign{theorem:strongly connected temporal digraph contains H routing}{\Lifetime{}}{k}
		\bound{theorem:strongly connected temporal digraph contains H routing}{\Lifetime{}}{n,k} & \coloneqq \bound{theorem:strongly connected temporal digraph contains H routing}{s}{k} + \bound{lemma:walk with many vertices implies choosable P_k or biP_k routing}{s}{k,(k-1)^2 + 1} \cdot \Brace{n - 1}.
\end{align*}
Observe that \(\bound{theorem:strongly connected temporal digraph contains H routing}{s}{k} \in \Oh(k^{11})\) and
\(\bound{theorem:strongly connected temporal digraph contains H routing}{\Lifetime{}}{n,k} \in \Oh(k^{11} + k^{5} n)\).

\begin{theorem}
	\label{theorem:strongly connected temporal digraph contains H routing}
	Let $T$ be a temporal digraph such that $\Layer{T}{i}$ is strongly connected for all $1 \leq i \leq \Lifetime{T}$.
  If $\Lifetime{T} \geq \bound{theorem:strongly connected temporal digraph contains H routing}{\Lifetime{}}{\Abs{\V{T}}, k}$, then for every set $S \subseteq \V{T}$ with $\Abs{S} \geq \bound{theorem:strongly connected temporal digraph contains H routing}{s}{k}$ there is a subset $S' \subseteq S$ with $\Abs{S'} = k$ such that $D$ contains an $H$-routing over $S'$ for some $H \in \Set{ \Ck{k},  \biPk{k}}$.
\end{theorem}
\begin{proof}
	Let $k_2 = (k-1)^2 + 1$ and let $k_1 = \bound{lemma:walk with many vertices implies choosable P_k or biP_k routing}{s}{k,k_2}$.
	Let $S_0 \subseteq S$ be a set of size $\bound{lemma:walk with many vertices implies choosable P_k or biP_k routing}{s}{k,k_1}$.
	Note that $\Lifetime{T} \geq (\Abs{S_0} + k_1) \cdot (\Abs{\V{T}} - 1)$.
	
	Let $W_1$ be a temporal walk of minimal length which contains all vertices of $S_0$ within the first $\Abs{S_0} \cdot \Brace{\Abs{\V{T}} - 1}$ layers of $D$.
	By~\cref{lemma:temporal strongly connected contains walk with many vertices}, such a walk $W_1$ exists.
	
	If~\cref{lemma:walk with many vertices implies choosable P_k or biP_k routing}\cref{item:choosable-P_k:biP_k} holds, then $W_1$ contains a $\biPk{k}$-routing over some $S' \subseteq S_0$ and we are done.
	Otherwise,~\cref{lemma:walk with many vertices implies choosable P_k or biP_k routing}\cref{item:choosable-P_k:P_k} holds.
	That is, there is some $S_1 \subseteq S_0$ and there are (possibly arcless) walks $W_2, W_a, W_b, W_c$ in $W_1$ such that $W_2$ is a subwalk of $W_1$ departing and arriving at the same \timesteps as $W_1$, $W_a \cdot W_b \cdot W_c = W_2$, $W_a$ and $W_c$ are internally disjoint from $S_1$, and $W_b$ contains a $\Pk{k_1}$-routing over $S_1$ where the first vertex of the $\Pk{k_1}$ is mapped to $\Start{W_b}$ and the last vertex of the $\Pk{k_1}$ is mapped to $\End{W_b}$.
	Let $\varphi_1$ be the bijection of this $\Pk{k_1}$-routing.
	
	Let $t_1 \leq (\Abs{S_0}\cdot(\V{T} - 1))$ be the \timestep in which $W_1$ arrives and let $W_3$ be a temporal walk departing on $t_1$ and of duration at most $\Abs{S_1} \cdot \Brace{\Abs{\V{T}} - 1}$ which visits all vertices of $S_1$.
	By~\cref{lemma:temporal strongly connected contains walk with many vertices}, such a walk $W_3$ exists.
	
	If~\cref{lemma:walk with many vertices implies choosable P_k or biP_k routing}\cref{item:choosable-P_k:biP_k} holds, then $W_2$ contains a $\biPk{k}$-routing over some $S' \subseteq S_1$ and we are done.
	Otherwise,~\cref{lemma:walk with many vertices implies choosable P_k or biP_k routing}\cref{item:choosable-P_k:P_k} holds.
	That is, there is some $S_2 \subseteq S_1$ and there are (possibly arcless) walks $W_4, W_d, W_e, W_f$ in $W_3$ such that $W_4$ is a subwalk of $W_3$ departing and arriving at the same \timesteps as $W_3$, $W_d \cdot W_e \cdot W_f = W_4$, $W_d$ and $W_f$ are internally disjoint from $S_2$, and $W_f$ contains a $\Pk{k_2}$-routing over $S_2$ where  the first vertex of the $\Pk{k_2}$ is mapped to $\Start{W_f}$ and the last vertex of the $\Pk{k_2}$ is mapped to $\End{W_f}$.
	Let $\varphi_2$ be the bijection of this $\Pk{k_2}$-routing.
	
	By~\cref{thm:erdos_szekeres}, some $S_3 \subseteq S_2$ of size $k$ satisfies one of the following two cases.
	Let $\varphi_1' = \FktRest{\varphi_1}{S_3}$ and $\varphi_2' = \FktRest{\varphi_2}{S_3}$.
	
	\textbf{Case 1:}
	$\varphi_1'$ and $\varphi_2'$ induce two $\Pk{k}$-routings over $S_3$ where the order of the vertices along the $\Pk{k}$ are the same.
	We show that $\varphi_1'$ also induces a $\Ck{k}$-routing in $D$.
	Let $u_1, u_2, \dots, u_{k}$ be the vertices of $\Pk{k}$ sorted according to their order along $\Pk{k}$.
	Let $u_i, u_j \in \V{\Pk{k}}$.
	
	If $i < j$, then \(W_b\) contains a temporal $\varphi'_1(u_i)$-$\varphi'_1(u_j)$-path which is disjoint from $S_3 \setminus \Set{\varphi'_1(u_i), \dots,\allowbreak \varphi'_1(u_j)}$.
	
	If $i > j$, we construct the desired temporal path $P'$ as follows.
	Let $Q_1$ be a temporal $\varphi'_1(u_i)$-$\varphi'_1(u_{k})$-walk in $W_b$ which is disjoint from $S_3 \setminus \Set{\varphi_1'(u_i), \ldots, \varphi_1'(u_{k})}$ such that $\End{Q_1} = \End{W_b} = \Start{W_c}$.
	Since $W_b$ contains a $\Pk{k}$-routing and $\varphi'_1(u_{k}) = \End{W_b}$, such a walk $Q_1$ exists.
	
	Let $Q_2$ be a temporal $\varphi'_1(u_1)$-$\varphi'_1(u_{j})$-walk in $W_e$ which is disjoint from $S_3 \setminus \{\varphi_1'(u_1), \ldots, \varphi'_1(u_{j})\}$ such that $\Start{Q_2} = \Start{W_e} = \End{W_d}$.
	Since $W_e$ contains a $\Pk{k}$-routing and $\varphi'_1(u_{1}) = \End{W_e}$, such a walk $Q_2$ exists.
	
	We now have that $Q_1 \cdot W_c \cdot W_d \cdot Q_2$ is a temporal $\varphi'_1(u_i)$-$\varphi_1'(u_j)$-walk which is disjoint from $S_3 \setminus (\Set{\varphi'_1(u_i), \ldots, \varphi'_1(u_{k})} \cup \Set{\varphi'_1(u_1), \ldots, \varphi'_1(u_{j})})$ in $D$.
	Thus, $Q_1 \cdot W_c \cdot W_d \cdot Q_2$ contains the desired temporal $\varphi'_1(u_i)$-$\varphi'_1(u_j)$-path $P'$.
	Hence, $\varphi'_1$ induces a $\Ck{k}$-routing over $S_3 \subseteq S$ in $D$.
	
	\textbf{Case 2:}
	$\varphi_1'$ and $\varphi_2'$ induce two $\Pk{k}$-routings over $S_3$ where the vertices along the $\Pk{k}$ of $\varphi'_2$ are ordered in reverse compared to those of the $\Pk{k}$ of $\varphi_1'$.
	We show that $\varphi'_1$ induces a $\biPk{k}$-routing over $S_3 \subseteq S$ in $D$.
	Let $u_1, u_2, \dots, u_{k}$ be the vertices of $\Pk{k}$ sorted according to their order along the $\Pk{k}$ for $\varphi'_1$.
	Let $u_i, u_j \in \V{\Pk{k}}$.
	
	If $i < j$, then $W_b$ contains a temporal $\varphi'_1(u_i)$-$\varphi'_1(u_j)$-path which is disjoint from $S_3 \setminus \Set{\varphi'_1(u_i), \ldots,\allowbreak \varphi'_1(u_j)}$.
	
	If $i > j$, we take a temporal $\varphi'_1(u_i)$-$\varphi'_1(u_j)$-path $P'$ which is disjoint from $S_3 \setminus \Set{\varphi'_1(u_j), \ldots,\allowbreak \varphi'_1(u_i)}$ in $W_e$.
	Since $\varphi'_2$ induces a $\Pk{k}$-routing in $W_e$ where the vertices of the $\Pk{k}$ are ordered in reverse when compared to the $\Pk{k}$ of $\varphi'_1$, such a path $P'$ exists.
	Hence, $\varphi'_1$ induces a $\biPk{k}$-routing over $S_3 \subseteq S$ in $D$.
\end{proof}

Our next goal is to relate $H$-routings, for $H \in \{ \Ck{k}, \biPk{k}\}$, to well-linkedness.
Towards this end, we first observe the following.

\begin{observation}
	\label{obs:sequence-A-B-partition}
	Let $A$ and $B$ be disjoint sets of equal cardinality and let $S$ be a sequence containing each element of $A \uplus B$ exactly once.
	Then there are sequences $S_1, S_2$ such that $S_1 \cdot S_2 = S$ and the following holds
	\begin{enamerate}{P}{item:sequence-A-B-partition:balanced}
		\item \label{item:sequence-A-B-partition:end}
		      $S_1$ starts in $A$ and ends in $B$ or starts in $B$ and ends in $A$, and
		\item \label{item:sequence-A-B-partition:balanced}
		      each of $S_1$ and $S_2$ contains as many elements of $A$ as elements of $B$.
	\end{enamerate}
\end{observation}
\begin{proof}
	Without loss of generality, we assume that $S$ starts at an element of $A$.
	The other case follows analogously by swapping $A$ and $B$.
	
	Let $S_1$ be the shortest prefix of $S$ containing the same number of elements in $A$ and $B$.
    Such a prefix exists, as $S$ itself is a valid choice.
	Since the first element of $S_1$ lies on $A$, its last element must lie on $B$.
	If this were not the case, then $S_1$ would contain a prefix with more elements of $B$ than elements of $A$, which implies that $S_1$ also contains a shorter prefix with as many elements of $A$ as elements of $B$, a contradiction to the choice of $S_1$. 
	Hence,~\cref{item:sequence-A-B-partition:end} holds.
	
	Let $S_2$ be such that $S_1 \cdot S_2 = S$.
	Since both $S_1$ and $S$ contain as many elements of $A$ as elements of $B$, we have that $S_2$ also contains as many elements of $A$ as elements of $B$. 
	Thus,~\cref{item:sequence-A-B-partition:balanced} holds.
\end{proof}

Note that in~\cref{obs:sequence-A-B-partition} if $S$ does not start and end in an element of the same set, then we can always take $S_1$ to be $S$ and $S_2$ to be empty.

The next observation is used when obtaining well-linkedness in the case of a $\Ck{k}$-routing. 

\begin{observation}
	\label{obs:number-cycle-non-negative}
	Let $C$ be a directed cycle and let $f: \V{C} \to \Z$ be a function such that $\sum_{v \in \V{C}}f(v) = 0$.
	Then there is some $v \in \V{C}$ such that for every subpath $P$ of $C$ starting at $v$ we have $\sum_{v \in \V{P}}f(v) \geq 0$. 
\end{observation}
\begin{proof}
	The statement is true if $f(v) =0$ for all $v \in \V{C}$.
	So assume otherwise and take a subpath $P$ of $C$ minimising $\sum_{v \in \V{P}}f(v)$.
	Note that the weight of this path is negative and that $P$ is a proper subpath of $C$.
	Let $u$ be the vertex on $C$ after $\End{P}$.
	We claim that $u$ has the desired property.
	Suppose not, and let $P'$ be a negative subpath of $C$ starting in $u$.
	If $\End{P'} \notin \V{P}$, then $P\cdot P'$ is a proper subpath of $C$ with lower weight than $P$, a contradiction.
	Thus $\End{P'} \in \V{P}$.
    As $P$ was chosen to minimise $\sum_{v \in \V{P}}f(v)$, the subpath $P \cap P'$ cannot be of positive weight, and thus we obtain a contradiction to $C$ being of total weight 0.
\end{proof}

We obtain well-linkedness by progressively ``moving'' the paths around the \(\Ck{k}\) or \(\biPk{k}\)-routing.
In both cases, we need to carefully choose the order in which we reroute the paths in the linkage we are constructing, and we have to use the structure of the \(\Ck{k}\) and of the \(\biPk{k}\) differently in order to complete the linkage in both situations.

\begin{lemma}
	\label{state:Ck-or-bidirected-Pk-implies-well-linked}
	Let \(h\) be some integer, \(D\) be a digraph, \(\mathcal{L}\) be a linkage of order \(k\) in \(D\) and \(T\) be the routing temporal digraph of \(\mathcal{L}\) through \(\mathcal{H} \coloneqq \Brace{H_{1}, H_{2}, \ldots, H_{h}}\), where each $H_i$ is a subgraph of $D$.
	If there is some \(R \in \Set{\biPk{k}, \Ck{k}}\) and there are some temporally disjoint subgraphs \(T_{1}, T_{2}, \ldots, T_{k}\) of \(T\) such that for each \(1 \leq i \leq k\) there is an \(R\)-routing \(\varphi_i\) over \(\mathcal{L}\) in \(T_i\) where \(\varphi_i = \varphi_j\) for all \(1 \leq i,j \leq k\), then $\Start{\mathcal{L}}$ is well-linked to $\End{\mathcal{L}}$ in $\ToDigraph{\mathcal{L} \cup \mathcal{H}}$.
\end{lemma}
\begin{proof}
	Let $A \subseteq \Start{\mathcal{L}}$ and $B \subseteq \End{\mathcal{L}}$ be sets with $n \coloneqq \Abs{A} = \Abs{B}$.
	Let \(\varphi\) be the \(R\)-routing over \(\mathcal{L}\) in each \(T_i\).
	We construct an $A$-$B$-linkage $\mathcal{Q}$ as follows.

	Let $\mathcal{L}^A = \Set{L \in \mathcal{L} \mid \Start{L} \subseteq A}$ and let $\mathcal{L}^B = \Set{L \in \mathcal{L} \mid \End{L} \subseteq B}$.
	Let $\varphi_{A,B} = \FktRest{\varphi}{\mathcal{L}^A \cup \mathcal{L}^B}$.
			Let $\Set{a_1, a_2, \dots, a_{n}} \coloneqq A$ and let $\Set{b_1, b_2, \dots, b_{n}} \coloneqq B$.

	We start by constructing temporal walks $\mathcal{W} \coloneqq \Set{W_1, W_2, \dots, W_{n}}$ in $T$ and by constructing sets $X_0, \dots, X_{n} \subseteq \mathcal{L}^A$ and $Y_0, \dots, Y_{n} \subseteq \mathcal{L}^B$ such that, for each $0 \leq i \leq n$, we have
	\begin{enamerate}{W}{item:routing-implies-well-linked:A-B}
		\item \label{item:routing-implies-well-linked:X-Y}
			$\Abs{X_i} = i = \Abs{Y_i}$,
		\item \label{item:routing-implies-well-linked:W-X-Y}
			$W_i$ is a temporal walk in $T_{i}$ which is disjoint from $(\mathcal{L}^A \setminus X_i) \cup Y_{i-1}$, and
		\item \label{item:routing-implies-well-linked:A-B}
			$\Start{\mathcal{W}} = \mathcal{L}^A$ and $\End{\mathcal{W}} = \mathcal{L}^B$.
	\end{enamerate}	
	We consider two cases.

	\textbf{Case 1:} \(R = \Ck{k}\).

	Let $R' = \Ck{n}$ and note that $\varphi_{A,B}$ is a $\Ck{n}$-routing in each $T_i$.
	Partition $\V{R'}$ into a sequence of subpaths $\mathcal{Q} = \Brace{Q_1, Q_2, \dots, Q_{x}}$ of $R'$ where each $Q_i$ can be decomposed into $Q_i^a \cdot Q_i^b$ such that $\Abs{\V{Q_i^a}} \geq 1$, $\Abs{\V{Q_i^b}} \geq 1$, $\varphi_{A,B}(\V{Q_i^a}) \subseteq \mathcal{L}^A$ and $\varphi_{A,B}(\V{Q_i^b}) \subseteq \mathcal{L}^B$.
	Since $\varphi(\V{R'}) = \mathcal{L}^A \cup \mathcal{L}^B$, such a decomposition exists.
	Now define the function $f : \mathcal{Q} \to \Z$ with $f(Q_i = Q_i^a \cdot Q_i^b) = \Abs{\varphi(Q_i^a) \cap \mathcal{L}^A} - \Abs{\varphi(Q_i^b) \cap \mathcal{L}^B}$.

	From~\cref{obs:number-cycle-non-negative} we know there is some $v \in \V{R'}$ for which every subpath $P$ of $R'$ starting at $v$ satisfies $\Abs{\varphi_{A,B}(\V{P}) \cap \mathcal{L}^A} \geq \Abs{\varphi_{A,B}(\V{P}) \cap \mathcal{L}^B}$.
	Let $Q$ be the subpath of $R'$ starting at $v$ and containing every vertex of $R'$, and let $\Set{u_1, u_2, \dots, u_{n}} \coloneqq \V{Q}$ be an ordering of the vertices of $Q$ according to their occurrence along $Q$.

	We iteratively construct the desired walks $W_i$ and sets $X_i, Y_i$ such that, for each $0 \leq i \leq n-1$, $Q$ can be decomposed as $Q_1^i \cdot Q_2^i \cdot Q_3^i = Q$ satisfying the following properties
	\begin{enamerate}{C}{item:routing-implies-well-linked:C:3}
		\item \label{item:routing-implies-well-linked:C:1}
			$\varphi_{A,B}(\End{Q_1^i}) \in \mathcal{L}^A \setminus X_i$,
		\item \label{item:routing-implies-well-linked:C:2}
			$\varphi_{A,B}(\V{Q_2^i} \setminus \Set{\Start{Q_2^i}}) \subseteq (\mathcal{L}^B \cup X_i) \setminus Y_i$, $\Abs{\V{Q_2^i} \cap (\mathcal{L}^B \setminus Y_i)} \geq 1$, and
		\item \label{item:routing-implies-well-linked:C:3}
			$Y_i \subseteq \varphi_{A,B}(\V{Q_3^i})$ and $\varphi_{A,B}(\V{Q_3^i}) \subseteq Y_i \cup X_i$.
	\end{enamerate}

	Start by setting $X_0 \coloneqq \emptyset$ and $Y_0 \coloneqq \emptyset$.
	By choice of $Q$,~\cref{item:routing-implies-well-linked:C:1},\cref{item:routing-implies-well-linked:C:2},~\cref{item:routing-implies-well-linked:C:3} and~\cref{item:routing-implies-well-linked:X-Y} hold for $i = 0$.

	On step $1 \leq i \leq n$, let $Q_1^{i-1} \cdot Q_2^{i-1} \cdot Q_3^{i-1} = Q$ be a decomposition of $Q$ satisfying~\cref{item:routing-implies-well-linked:C:1},~\cref{item:routing-implies-well-linked:C:2} and~\cref{item:routing-implies-well-linked:C:3} for $i-1$.
	Let $u_a = \End{Q_1^{i-1}}$ and let $u_b \in \V{Q_2^{i-1}} \cap (\mathcal{L}^B \setminus Y_{i-1})$ be the last such vertex on $Q_2^{i-1}$.
	As~\cref{item:routing-implies-well-linked:C:1} and~\cref{item:routing-implies-well-linked:C:2} hold for $i-1$, $u_a \in \mathcal{L}^A \setminus X_{i-1}$ holds and such a vertex $u_b$ exists.

	Let $W_i$ be a temporal $\varphi_{A,B}(u_a)$-$\varphi_{A,B}(u_b)$-walk in $T_{i}$ avoiding $(\mathcal{L}^A \cup \mathcal{L}^B)  \setminus \varphi_{A,B}(\{u_a, \dots,\allowbreak u_{b}\})$.
	Since $\varphi_{A,B}$ is an $R'$-routing in $T_{i}$, such a walk exists.
	Set $X_i = X_{i-1} \cup \Set{\varphi_{A,B}(u_a)}$ and $Y_i = Y_{i-1} \cup \Set{ \varphi_{A,B}(u_b)}$.
	The walk $W_i$ satisfies~\cref{item:routing-implies-well-linked:W-X-Y} because $Y_{i-1} \subseteq Q_3^{i-1}$ and~\cref{item:routing-implies-well-linked:C:3} holds for $i-1$.

	We now show that $X_i, Y_i$ and $W_i$ satisfy the required properties.
	Clearly~\cref{item:routing-implies-well-linked:X-Y} holds for $i$.
	If $i < n$, then $\Abs{\mathcal{L}^A \setminus X_i} \geq 1$.
	As~\cref{item:routing-implies-well-linked:C:2} and~\cref{item:routing-implies-well-linked:C:3} hold for $i-1$, we have $\mathcal{L}^A \setminus X_{i-1} \subseteq \varphi(\V{Q_1^{i-1}})$.

	Let $Q_1^i$ be the shortest subpath of $Q_1^{i-1}$ containing every vertex of $\varphi^{-1}(\mathcal{L}^A \setminus X_{i})$ and let $Q_3^i$ be the $u_b$-$\End{Q}$ subpath of $Q$.
	By construction,~\cref{item:routing-implies-well-linked:C:1} holds for $i$.
	Further, by choice of $u_b$,~\cref{item:routing-implies-well-linked:C:3} holds for $i$.

	Let $Q_2^i$ be the $\End{Q_1^i}$-$\Start{Q_3^i}$ subpath of $Q$.
	Since $Q_1^i$ is a subpath of $H'$ starting at $v$ and ending in a vertex $u_a'$ with $\varphi(u_a') \in \mathcal{L}^A \setminus X_i$, we have that $Q_2^i \cdot Q_3^i$ must contain some vertex of $\varphi^{-1}(\mathcal{L}^B)$.
	Further, as $\varphi(\V{Q_3^i}) \subseteq Y_i \cup X_i$ due to~\cref{item:routing-implies-well-linked:C:3}, we have that $Q_2^i$ contains some vertex of $\mathcal{L}^B \setminus Y_i$.
	Finally, $\mathcal{L}^A \setminus X_i \subseteq \V{Q_1^i}$ and so $\varphi_{A,B}(\V{Q_2^i} \setminus \Set{\Start{Q_2^i}}) \subseteq \mathcal{L}^B \cup (X_i \setminus Y_i)$.
	Hence,~\cref{item:routing-implies-well-linked:C:2} holds for $i$.

	After $n$ steps, it is immediate that~\cref{item:routing-implies-well-linked:A-B} holds by choice of $W_1, W_2, \dots, W_{n}$.

	\textbf{Case 2:} \(\varphi\) is a \(\biPk{k}\)-routing.

	Let $R' = \biPk{n}$ and note that $\varphi_{A,B}$ is a $\biPk{n}$-routing in $T_i$ for each $i$.

	We iteratively construct the desired walks $W_i$ and sets $X_i, Y_i \subseteq \mathcal{L}'$ such that, for each $0 \leq i \leq n-1$, the following statement holds
	\begin{namerise}{P}{item:routing-implies-well-linked:biPk}
	\item \label{item:routing-implies-well-linked:biPk}
		for each strongly connected component $Z_i$ of $R' - \varphi_{A,B}^{-1}(Y_i)$, the set $\varphi_{A,B}(\V{Z_i})$ contains as many elements of $\mathcal{L}^A \setminus X_i$ as elements of $\mathcal{L}^B \setminus Y_i$.
	\end{namerise}

	Start by setting $X_0 \coloneqq \emptyset$ and $Y_0 \coloneqq \emptyset$.
	Clearly,~\cref{item:routing-implies-well-linked:X-Y} and~\cref{item:routing-implies-well-linked:biPk} hold for $i = 0$.

	On step $1 \leq i \leq n$, let $Z$ be a strong component (and hence a subpath) of $R' - \varphi_{A,B}^{-1}(X_{i-1})$ such that $\varphi_{A,B}(\V{Z})$ contains at least one element of $\mathcal{L}^A \setminus X_{i-1}$ and one element of $\mathcal{L}^B \setminus Y_{i-1}$.
	Since~\cref{item:routing-implies-well-linked:X-Y} and~\cref{item:routing-implies-well-linked:biPk} hold for $i-1$, such a strong component exists.

	Because $Z$ is a bidirected path, it induces a sequence over the elements of $\mathcal{L}^A \setminus X_{i-1}$ and of $\mathcal{L}^B \setminus Y_{i-1}$.
	Let $Z'$ be the shortest subpath of $Z$ satisfying $\varphi_{A,B}(\V{Z'}) \cap ((\mathcal{L}^A \setminus X_{i-1}) \cup (\mathcal{L}^B \setminus Y_{i-1})) = \varphi(\V{Z}) \cap ((\mathcal{L}^A \setminus X_{i-1}) \cup (\mathcal{L}^B \setminus Y_{i-1}))$.
	By~\cref{obs:sequence-A-B-partition}, there is a subpath $Z''$ of $Z'$ starting at one of the endpoints of $Z'$ such that one endpoint of $Z''$ is in $\mathcal{L}^A \setminus X_{i-1}$ and the other is in $\mathcal{L}^B \setminus Y_{i-1}$, and both $Z''$ and the rest of $Z'$ contain as many elements of $\mathcal{L}^A \setminus X_{i-1}$ as they contain elements of $\mathcal{L}^B \setminus Y_{i-1}$.
	Let $Z''$ be the shortest such subpath of $Z'$.

	Let $\Set{z_1, z_2, \dots, z_{j}}$ be the vertices of $Z''$ sorted according to their occurrence along $Z''$.
	Without loss of generality, we have $\varphi_{A,B}(z_1) \in \mathcal{L}^B \setminus X_{i-1}$ and $\varphi_{A,B}(z_{j}) \in \mathcal{L}^A \setminus Y_{i-1}$.

	Let $j_a$ be the smallest index such that $\varphi_{A,B}(z_{j_a}) \in \mathcal{L}^A \setminus X_{i-1}$ .
	We set $W_i$ as a temporal $\varphi_{A,B}(z_{j_a})$-$\varphi_{A,B}(z_{1})$-walk in $T_{i}$ which is disjoint from $(\mathcal{L}^A \setminus X_{i-1}) \cup Y_{i-1}$.
	By choice of $j_a$ and because $\varphi_{A,B}$ is an $R'$-routing over $\mathcal{L}^A \cup \mathcal{L}^B$ in $T_{i}$ and $Z$ is a component of $R' - \varphi_{A,B}(X_{i-1})$, such a walk $W_i$ exists.

	We set $X_i = X_{i-1} \cup \Set{\varphi_{A,B}(z_{j_a})}$ and $Y_i = Y_{i-1} \cup \Set{\varphi_{A,B}(z_{1})}$.
	The vertex $z_1$ is an endpoint of $Z''$, \cref{item:routing-implies-well-linked:biPk} holds for $i-1$ and $\varphi_{A,B}(\V{Z''} \setminus \Set{z_1})$ contains one less vertex of $\mathcal{L}^B \setminus Y_{i-1}$ and one less vertex of $\mathcal{L}^A \setminus X_{i-1}$ when compared to $Z$.
	Hence,~\cref{item:routing-implies-well-linked:biPk} holds for $i$.
	Further, $\Abs{X_i} = \Abs{X_{i-1}} + 1 = \Abs{Y_{i-1}} + 1 = \Abs{Y_i}$, and so~\cref{item:routing-implies-well-linked:X-Y} holds.

	This completes the case distinction above and the construction of $W_1, W_2, \dots, W_{n}$.
	We construct an $A$-$B$-linkage $\mathcal{L}'$ as follows.
	For each $1 \leq i \leq n$, let $L_i^a = \Start{W_i}$, $L^b_i = \End{W_i}$, $a_i = \Start{L_i^a}$ and $b_i = \End{L_i^b}$.
	We construct a path $Q_i = Q_i^a \cdot Q_i^t \cdot Q_i^b$ such that $\ToDigraph{Q_i^a} \subseteq L_i^a - \ToDigraph{\Set{W_{j} \mid s_{i+1} \leq j \leq \Lifetime{T}}}$, $\ToDigraph{Q_i^t} \subseteq \ToDigraph{\Set{H_{j} \mid s_i \leq j \leq s_{i+1}}}$ and $\ToDigraph{Q_i^b} \subseteq L_i^b - \ToDigraph{\Set{W_{j} \mid s_{1} \leq j \leq s_i}}$.

	Since~\cref{item:routing-implies-well-linked:W-X-Y} holds, each arc of $W_i$ corresponds to some path in $D$ which is disjoint from $\Set{L_j^a \mid i < j \leq n} \cup \Set{L_j^b \mid 1 \leq j < i}$.
	Furthermore, $W_i$ corresponds to some path $Q_i^t$ in $\ToDigraph{\Set{W_{j} \mid t_i \leq j \leq t_{i+1}}}$.

	We set $Q_i^a$ as the subpath of $L_i^a$ ending at $\Start{Q_i^t}$ and we set $Q_i^b$ as the subpath of $L_i^b$ starting at $\End{Q_i^t}$.
	By construction, $Q_i \coloneqq Q_i^a \cdot Q_i^t \cdot Q_i^b$ is an $a_i$-$b_i$-path in $\ToDigraph{\mathcal{L}} \cup \ToDigraph{\mathcal{W}}$ which is disjoint from all $Q_j$ for $1 \leq j < i$.
	Hence, $\mathcal{Q} \coloneqq \Set{Q_i \mid 1 \leq i \leq n}$ is an $A$-$B$-linkage as desired.

	Because we can construct such an $A$-$B$-linkage for any choice of $A,B$, we have that $\Start{\mathcal{L}'}$ is well-linked to $\End{\mathcal{L}'}$ in $\ToDigraph{\mathcal{L} \cup \mathcal{W}}$, as desired.
\end{proof}

We can now show how to obtain well-linkedness from the routing temporal digraph of some linkage $\mathcal{L}$.
The idea is to use~\cref{theorem:strongly connected temporal digraph contains H routing} to obtain many $\Ck{k}$ and $\biPk{k}$-routings.
With the pigeon-hole principle, we get many equal routings.
We then use~\cref{obs:sequence-A-B-partition,obs:number-cycle-non-negative} in each case to argue that certain sets are well-linked.
We start by defining
\begin{align*}
	\boundDefAlign{lemma:routing-implies-well-linked}{\ell}{k}
	\bound{lemma:routing-implies-well-linked}{\ell}{k} & \coloneqq
	  \bound{theorem:strongly connected temporal digraph contains H routing}{s}{k}, \\[0em]
	\boundDefAlign{lemma:routing-implies-well-linked}{h}{k}
  \bound{lemma:routing-implies-well-linked}{h}{k} & \coloneqq \bound{theorem:strongly connected temporal digraph contains H routing}{\Lifetime{}}{\bound{lemma:routing-implies-well-linked}{\ell}{k}, k} \cdot 2k
		\binom{\bound{theorem:strongly connected temporal digraph contains H routing}{s}{k}}{k}
	 k!.
\end{align*}
Note that \(\bound{lemma:routing-implies-well-linked}{\ell}{k} \in \Oh(k^{11})\) and \(\bound{lemma:routing-implies-well-linked}{h}{k} \in \PowerTower{1}{\Polynomial{2}{k}}\).
Using the pigeon-hole principle, we can combine~\cref{theorem:strongly connected temporal digraph contains H routing} and~\cref{state:Ck-or-bidirected-Pk-implies-well-linked} to obtain the desired statement.

\begin{proposition}
	\label{lemma:routing-implies-well-linked}
    Let $k$ be an integer, $h \geq \bound{lemma:routing-implies-well-linked}{h}{k}$, $D$ be a digraph, $\mathcal{L}$ be a linkage of order $\bound{lemma:routing-implies-well-linked}{\ell}{k}$ in $D$ and $T$ be the routing temporal digraph of $\mathcal{L}$ through $\mathcal{H} \coloneqq \Set{H_1, \dots, H_h}$, where each $H_i$ is a subgraph of $D$.
	If each layer $\Layer{T}{i}$ is strongly connected, then there exists some $\mathcal{L}' \subseteq \mathcal{L}$ of order $k$ such that $\Start{\mathcal{L}'}$ is well-linked to $\End{\mathcal{L}'}$ in $\ToDigraph{\mathcal{L} \cup \mathcal{H}}$.
\end{proposition}
\begin{proof}
	Let $k_1 = 2k \binom{\bound{theorem:strongly connected temporal digraph contains H routing}{s}{k}}{k} k!$.
	Define $s_1 = 1$ and for each $1 \leq i \leq k_1$ define $s_i = (i - 1) \cdot \bound{theorem:strongly connected temporal digraph contains H routing}{\Lifetime{}}{\bound{lemma:routing-implies-well-linked}{\ell}{k}, k} + 1$.
	For each $1 \leq i \leq k_1$ let $T_i$ be the temporal subgraph of $T$ from \timestep $s_i$ to $s_{i+1} - 1$.
	Note that $\Lifetime{T_i} = s_{i+1} - s_i = \bound{theorem:strongly connected temporal digraph contains H routing}{\Lifetime{}}{\bound{lemma:routing-implies-well-linked}{\ell}{k},k}$ and that $\Abs{\mathcal{L}} = \Abs{\V{T_i}} = \bound{theorem:strongly connected temporal digraph contains H routing}{s}{k} = \bound{lemma:routing-implies-well-linked}{\ell}{k}$.

	By~\cref{theorem:strongly connected temporal digraph contains H routing} each $T_i$ contains a $\Ck{k}$-routing $\varphi_i$ or a $\biPk{k}$-routing $\varphi_i$ over some set $\mathcal{L}_i \subseteq \mathcal{L}$ of size $k$.
	As there are $k_1 = 2k \binom{\bound{theorem:strongly connected temporal digraph contains H routing}{s}{k}}{k} k!$ temporal digraphs $T_i$, there is some $I \subseteq \Set{1, \dots, k_1}$ of size $k$ and some $H \in \Set{\Ck{k} ,\biPk{k}}$ such that, for every $i,j \in I$, both $T_i$ and $T_j$ have an $H$-routing $\varphi \coloneqq \varphi_i = \varphi_j$ over $\mathcal{L}' \coloneqq \mathcal{L}_i = \mathcal{L}_j$.
	By~\cref{state:Ck-or-bidirected-Pk-implies-well-linked}, \(\Start{\mathcal{L}'}\) is well-linked to \(\End{\mathcal{L}'}\) in \(\ToDigraph{\mathcal{L}'} \cup \mathcal{H}\).
\end{proof}

\section{Paths of order-linked sets and acyclic grids}
\label{sec:order-linked}

The previous section contains some discussion of the similarities between $\Pk{k}$-routings in routing temporal digraphs and acyclic grids without yet constructing such a grid or any grid-like structure.
We now develop a more abstract framework to model these intuitive observations.
This enables us to lift specific properties of acyclic grids to a more abstract setting.
The techniques and results we develop in this section play an important role both in obtaining the path of well-linked sets, as in~\cref{thm:high_dtw_to_POSS_plus_back-linkage}, and also in our proof of the Directed Grid Theorem in~\cite{COSSII}.

\begin{figure}[!ht]
	\begin{center}
			\includegraphics{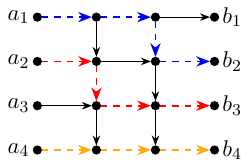}
			\caption{An illustration of $r$-shifts in acyclic grids.
            The $2$-shift $(b_2, b_3, b_4)$ of $(a_1, a_2, a_4)$ is routable in the grid as illustrated here.
            However the $3$-shift $(b_2, b_3, b_4)$ of $(a_1, a_2, a_3)$ is not	routable in the grid.}
			\label{fig:r-shift-grid}
		\end{center}
	\end{figure}
	\DeferTask{figure out how to centre the caption}

	To motivate the following definitions, consider the acyclic grid illustrated in~\cref{fig:r-shift-grid}.
	Suppose we want to connect some vertex $a_i$ on the left of the grid to a vertex $b_j$ on the right.
	As in an acyclic grid we can never route upwards,
	connecting \(a_i\) to \(b_j\) is possible if and only if $i \leq j$.
	
	Let $A$ be the ordered set containing the left-most vertices of the grid, i.e.~$\Set{a_1, \dots, a_4}$ in the example in~\cref{fig:r-shift-grid}, ordered from top to bottom and let $B$ be the ordered set containing the vertices at the right, i.e.~$\{b_1, \dots, b_4\}$ in the example, again ordered from top to bottom.
	
	Suppose we are given a subset $A' \subseteq A$ and an equal-sized subset $B' \subseteq B$.
	Under what conditions can we connect $A'$ to $B'$ by a linkage $\LLL$ in the grid?
    As discussed, we can only connect $a_i \in A'$ to $b_j \in B'$ if $i \leq j$.
    Furthermore, as the grid is planar, if $\{ a_{i_1}, \dots, a_{i_l}\}$ are the vertices of $A'$ ordered by their order in $A$ and likewise $\{b_{j_1}, \dots, b_{j_l}\}$ are the ordered vertices of $B'$, then we have to connect $a_{i_s}$ to $b_{j_s}$, for all $1 \leq s \leq l$.
    This implies that  $i_s \leq j_s$ for all $1 \leq s \leq l$.
	But even if $A'$ and $B'$ satisfy this condition, it may still be impossible to connect $A'$ to $B'$.
    As the example in~\cref{fig:r-shift-grid} demonstrates, there may not be enough columns in the grid to route all paths downwards that connect pairs $a_{i_s},b_{j_s}$ with $i_s < j_s$.
    So we may have to restrict the number of pairs $(a_{i_s}, b_{j_s})$ for which we allow $i_s < j_s$.
	
	This idea is formalised in the next definition by the concept of \emph{$r$-shifts}.
	
	\begin{definition}
		\label{def:A-shift}
		Let $A = (a_1,\dots, a_n)$ and $B = (b_1,\dots, b_m)$ be ordered sets.
		Let $r \in \N$, let $A'$ be an ordered subset of $A$ and let $B'$ be an ordered subset of $B$ such that $\Abs{A'} = \Abs{B'}$. 
		We say that $B'$ is an \emph{$r$-shift of $A'$} if there is a bijection $\pi : A' \rightarrow B'$ such that 
		\begin{enumerate}
			\item for all $a_i \in A'$ we have that $\pi(a_i) = b_j$ implies $i \leq j$;
			\item there are at most $r$ vertices $a_i \in A'$ with $\pi(a_i) \neq b_i$; and
			\item $\pi$ is order preserving, that is, for all $a_i, a_j \in A'$, if $a_i \leq_A a_j$, then $\pi(a_i) \leq_B \pi(a_j)$.
		\end{enumerate}
	\end{definition}
	
	In the example of~\cref{fig:r-shift-grid}, $(b_2, b_3, b_4)$ is a $2$-shift of $(a_1, a_2, a_4)$ and
	a $3$-shift of $(a_1, a_2, a_3)$.
	We are interested in finding pairs of equal-sized ordered sets $A$ and $B$ inside a digraph D for which,
	given a subset $A' \subseteq A$,
	we can find, inside \(D\), a linkage from $A'$ to any possible $r$-shifts $B' \subseteq B$ of $A'$.
  This property is formalised in the following definition. 
	Recall from~\cref{sec:preliminaries} that we may consider a linkage $\LLL$ as a function $\LLL \sth \Start{\LLL} \rightarrow \End{\LLL}$ where $\LLL(a)$ is the endpoint of the path in $\LLL$ starting at $a$. 
	\begin{definition}
		\label{def:order-linkedness}
		Let $H$ be a digraph, let $A = \Brace{a_1, \dots, a_n}, B = \Brace{b_1, \dots, b_m} \subseteq \V{H}$ be ordered sets, and let $r \in \N$. 
		We say that $A$ is $r$-\emph{order-linked} to $B$ in $H$ if for every $A' \subseteq A$ and every $B' \subseteq B$ with $\Abs{A'} = \Abs{B'}$ where $B'$ is an $r$-shift of $A'$ witnessed by the bijection $\pi$ there is an $A'$-$B'$-linkage $\mathcal{L}$ in $H$ satisfying $\Fkt{\pi}{a} = \Fkt{\mathcal{L}}{a}$ for all $a \in A'$. 
		
		For (unordered) sets $A,B \subseteq \V{H}$, we say that $A$ is $r$-\emph{order-linked} to $B$ in $H$ if there exist orderings $A_1$ and $B_1$ of $A$ and $B$, respectively, such that $A_1$ is $r$-\emph{order-linked} to $B_1$ in $H$. 
	\end{definition}
	
	To give an example, it is easily seen that, in any acyclic $(r,r)$-grid $(\PPP, \QQQ)$, the set $\Start{\QQQ}$ is $r$-order-linked to $\End{\QQQ}$.
	
	We now define a new type of structure which can be seen as an abstraction of acyclic grids. 
	\begin{definition}[path of $r$-order-linked sets]
		\label{def:path-of-order-linked-sets}
		A \emph{path of $r$-order-linked sets} of width $w$ and length $\ell$ is a tuple $\Brace{\mathcal{S},\mathscr{P}}$ such that
		\begin{enumerate}
			\item $\mathcal{S}$ is a sequence of $\ell + 1$ pairwise disjoint subgraphs $\Brace{S_0,\dots,S_{\ell}}$, which are called \emph{clusters}, 
			\item for every $0 \leq i \leq \ell$ there are disjoint sets $A(S_i),B(S_i) \subseteq \V{S_i}$ of size $w$ such that $A(S_i)$ is $r$-order-linked to $B(S_i)$ in $S_i$, 
			\item $\mathscr{P}$ is a sequence of $\ell$ pairwise disjoint linkages $\Brace{\mathcal{P}_0, \mathcal{P}_1, \dots, \mathcal{P}_{\ell - 1}}$ such that, for every $0 \leq i < \ell$, $\mathcal{P}_i$ is a $B(S_i)$-$A(S_{i+1})$-linkage of order $w$ which is internally disjoint from $S_i$ and $S_{i+1}$ and disjoint from every $S \in \mathcal{S} \setminus \Set{S_i, S_{i+1}}$. 
		\end{enumerate}
		
		By definition, for each $1 \leq i \leq \ell$ there are orderings $\leq_{A(S_i)}$ of $A(S_i)$ and $\leq_{B(S_i)}$ of $B(S_i)$ witnessing the $r$-order-linkedness of $A(S_i)$ and $B(S_i)$ in $S_i$.
		
		We say that $\Brace{\mathcal{S}, \mathscr{P}}$ is \emph{uniform} if for all $1 \leq i \leq \ell$ we can choose orderings $\leq_{A(S_i)}$ and $\leq_{B(S_i)}$ witnessing that $A(S_i)$ is $r$-order-linked to $B(S_i)$ so that for all $0\leq  i < \ell$ and all $b_1,b_2 \in B(S_i)$: if $b_1 \leq_{B(S_i)} b_2$, then $\mathcal{P}_i(b_1) \leq_{A(S_{i+1})} \mathcal{P}_i(b_2)$.
	\end{definition}
	
	The following notation is used frequently later on.
    Given a path of $r$-order-linked sets $(\SSS \coloneqq (S_0, \dots, S_\ell),  \PPPP\coloneqq(\PPP_0, \dots, \PPP_{\ell-1}))$ and indices $0 \leq i \leq j \leq \ell$, we define $(\SSS, \PPPP)[i, j]$ as the path of $r$-order-linked sets from cluster $i$ to cluster $j$.
    That is, $(\SSS, \PPPP)[i, j] \coloneqq \Brace{(S_i, \dots, S_j), (\PPP_i, \dots, \PPP_{j-1})}$.\Symbol{SPij@$(\SSS, \PPPP)[i, j]$} 
	
	To give an example, let $(\PPP, \QQQ)$ be an acyclic $(r, r^2)$-grid, where $\QQQ \coloneqq (Q_1^1, \dots, Q_1^r, Q_2^1, \dots$, $Q_2^r, \dots, Q_r^r)$.
    Then $(\PPP, \QQQ)$ contains a path of $r$-order-linked sets as follows.
		The cluster $S_i$ is obtained as the union of the columns $Q_i^j$, for $1 \leq j \leq r$, and, for each $P \in \PPP$, the subpath $P^i$ of $P$ starting at the first vertex of $P$ on $Q_i^1$ and ending at the last vertex of $P$ on $Q_i^r$.
    We define $A(S_i) \coloneqq \{ \Start{P^i} \sth P \in \PPP\}$ and $B(S_i) \coloneqq \{ \End{P^i} \sth P \in \PPP\}$.
    Finally, the linkages $\PPP_i$, for $0 \leq i < r$, are obtained by taking the subpaths of the rows connecting $S_i$ to $S_{i+1}$ in the obvious way.
	Then $\Brace{(S_0, \dots, S_r), (\PPP_0, \dots, \PPP_{r-1})}$ is a path of $r$-order-linked sets.
	
	As the example shows, we can easily obtain a path of $r$-order linked sets from a sufficiently large acyclic grid.
    Our next goal is to show that the converse is also true, albeit with bigger bounds.
	
	We first observe the following. 
	
	\begin{observation}
		\label{lemma:linkage-inside-pools}
		Let $D = (\mathcal{S} = \Brace{S_0, S_1, \dots, S_{\ell}}, \mathscr{P})$ be a
		path of $0$-order-linked sets of width $w$. 
		For every $0 \leq i < j \leq \ell$, every $A' \in \Set{A(S_i), B(S_i)}$, and
		every $B' \in \Set{A(S_j), B(S_j)}$ there is an $A'$-$B'$-linkage $\mathcal{L}$
		of order $w$ in $D$. 
		Furthermore, for all $i < k < j$ every path in $\mathcal{L}$ must intersect $A(S_k)$ and $B(S_k)$.
	\end{observation}
	\begin{proof}
		We show the case where $A' = A(S_i)$ and $B' = B(S_j)$.
		The other cases follow analogously.
		
		For each $i \leq t \leq j - 1$ construct sets $A_t, B_t$ and a linkage $\mathcal{L}_t$ as follows.
		Start by setting $A_{i-1} = A'$ and $\mathcal{L}_{i-1}$ as the linkage containing only the vertices of $A'$.
		
		On step $t$, let $B_t$ be a 0-shift of $A_{t - 1}$ and let $\mathcal{R}_t$ be an $A_{t - 1}$-$B_t$-linkage of order $w$ in $S_t$.
		Since $A(S_t)$ is $0$-order-linked to $B(S_t)$, such a linkage $\mathcal{R}_t$ exists.
		Let $\mathcal{R}_t' \subseteq \mathcal{P}_t$ be the set of paths with $\Start{\mathcal{R}_t'} = \End{\mathcal{R}_t}$.
		Now set $A_t = \End{\mathcal{R}_t'}$ and set $\mathcal{L}_t = \mathcal{L}_{t-1} \cdot \mathcal{R}_t \cdot \mathcal{R}_t'$.
		
		It is immediate from the construction that $\mathcal{L}_j$ is an $A'$-$B'$-linkage of order $w$, as desired.
	\end{proof}
	
	We are now ready to show that every path of $1$-order-linked sets
	contains an acyclic grid. Here we make use of our framework of
	$H$-routings in temporal digraphs.
    The idea is to use the
	$\Pk{k}$-routings constructed in~\cref{theorem:one-way connected temporal digraph contains P_k routing} to obtain the columns of the grid.
    We start by defining
	\begin{align*}
		\boundDefAlign{thm:order_linked_to_acyclic_grid}{w}{k}
		\bound{thm:order_linked_to_acyclic_grid}{w}{k} & \coloneqq k^2 - 1,\\[0em]
		\boundDefAlign{thm:order_linked_to_acyclic_grid}{\ell}{k}
		\bound{thm:order_linked_to_acyclic_grid}{\ell}{k} & \coloneqq \Brace{(k^2 -k - 1) \cdot \binom{\bound{thm:order_linked_to_acyclic_grid}{w}{k}}{k} \cdot
				k! + 1} \cdot \bound{lemma:temporal one-way connected contains walk with many vertices}{\Lifetime{}}{\bound{thm:order_linked_to_acyclic_grid}{w}{k},\bound{thm:order_linked_to_acyclic_grid}{w}{k}}.
	\end{align*}
	Observe that \(\bound{thm:order_linked_to_acyclic_grid}{w}{k} \in \Oh(k^{2})\) and 
	\(\bound{thm:order_linked_to_acyclic_grid}{\ell}{k} \in \PowerTower{1}{\Polynomial{7}{k}}\).
	
	\begin{theorem}
		\label{thm:order_linked_to_acyclic_grid}
		Every path of $1$-order-linked sets of width at least $w = \bound{thm:order_linked_to_acyclic_grid}{w}{k}$ and length at least $\bound{thm:order_linked_to_acyclic_grid}{\ell}{k}$ contains an acyclic $(k,k)$-grid.
	\end{theorem}
	\begin{proof}
		Let $\Brace{\mathcal{S} \coloneqq \Brace{S_0, S_1, \dots, S_{\ell}}, \mathscr{P} \coloneqq \Brace{\mathcal{P}_0, \mathcal{P}_1, \dots, \mathcal{P}_{\ell - 1}}}$ be a path of 1-order-linked sets of width at least $w \coloneqq \bound{thm:order_linked_to_acyclic_grid}{w}{k}$ and length $\ell \geq \bound{thm:order_linked_to_acyclic_grid}{\ell}{k}$.
		Let $D = \ToDigraph{\Brace{\mathcal{S}, \mathscr{P}}}$.
		By~\cref{lemma:linkage-inside-pools}, there is an $A(S_0)$-$B(S_\ell)$-linkage $\mathcal{L}$ of order $w$ in $D$.
		Note that every path in $\mathcal{L}$ must intersect every $A(S_i)$ and every $B(S_i)$.
		
		Let $T$ be the routing temporal digraph of $\mathcal{L}$ through $\mathcal{S}$.
		Since $A(S_i)$ is $1$-order-linked to $B(S_i)$ for every $S_i \in \mathcal{S}$, we have that each layer of $T$ is unilateral.
		
		Let $k_1 = k(k-1)$, let $k_2 = (k_1 - 1) \cdot \binom{w}{k} \cdot k! + 1$.
		For each $1 \leq i \leq k_2$, let $t_i = (i - 1) \cdot \bound{lemma:temporal one-way connected contains walk with many vertices}{\Lifetime{}}{w,w}$ and let $T_i$ be the temporal subgraph of $T$ from \timestep $t_i$ to $t_{i+1} - 1$.
		Note that $\Lifetime{T_i} = \bound{lemma:temporal one-way connected contains walk with many vertices}{\Lifetime{}}{w, w}$ and $\Abs{\V{T_i}} = w$.
		
		By~\cref{theorem:one-way connected temporal digraph contains P_k routing}, each $T_i$ contains a $\Pk{k}$-routing $\varphi_i$.
		Since $\Lifetime{T} \geq \bound{thm:order_linked_to_acyclic_grid}{\ell}{k} = k_2 \cdot \bound{lemma:temporal one-way connected contains walk with many vertices}{\Lifetime{}}{w, w}$, there are at least $k_2$ subgraphs $T_i$.
		By the pigeon-hole principle, there is some set $\mathcal{T} \subseteq \Set{T_{1}, T_{2}, \ldots, T_{k_2}}$ of size $k_1$
		such that
		$\varphi \coloneqq \varphi_i = \varphi_j$ for all $T_i, T_j \in \mathcal{T}$.
		
		Let $\Brace{T'_1, T'_2, \dots, T'_{k_1}} \coloneqq \mathcal{T}$ be sorted according to the corresponding \timesteps, let $\mathcal{Q}$ be the image of $\varphi$.
		
		Let $u_1, u_2, \dots, u_{k}$ be the vertices of the $\Pk{k}$ ordered according to their occurrence on the $\Pk{k}$.
		Let \(\mathcal{P} = \Set{\varphi(u_1), \varphi(u_2), \ldots, \varphi(u_k)}\).
		We construct a sequence $\mathcal{Q}$ of $k$ paths \(Q_{1}, Q_{2}, \ldots, Q_{k}\) where, for each $1 \leq i \leq k$, the path $Q_i$ is constructed as follows.
		
		For each $1 \leq j < k$, let $t_{i,j} = (i - 1) \cdot (k - 1) + j$ and let $R_{i,j}$ be a temporal $\varphi(u_j)$-$\varphi(u_{j+1})$-path in $T'_{t_{i,j}}$ which does not contain any path in $\mathcal{P} \setminus \Set{\varphi(u_j), \varphi(u_{j+1})}$.
		Note that $t_{i,k - 1} = t_{i+1,1} - 1$.
		Since $\varphi$ is a $\Pk{k}$-routing in $T'_{t_{i,j}}$, such a path $R_{i,j}$ exists.
		Finally, $R_{i,j}$ corresponds to a $\V{\varphi(u_j)}$-$\V{\varphi(u_{j+1})}$-path $Q_{i,j,2}$ in $D$.
		Let $Q_{i,j,1}$ be the $\End{Q_{i,j-1,2}}$-$\Start{Q_{i,j,2}}$-path in $\ToDigraph{\varphi(u_j)}$ (to simplify notation, we choose $\End{Q_{i,0,2}}$ as $\Start{Q_{i,1,2}}$).
		
		We now set $Q_i = Q_{i,1,1} \cdot Q_{i,1,2} \cdot Q_{i,2,1} \cdot Q_{i,2,2} \cdot \ldots \cdot Q_{i,k - 1,2}$.
		After constructing all $Q_i$, set $\mathcal{Q} = \Brace{Q_1, Q_2, \dots, Q_{k}}$.
		Note that the paths in $\mathcal{Q}$ are pairwise disjoint.
		It is now immediate from the construction that $(\mathcal{P}, \mathcal{Q})$ is an acyclic $(k, k)$-grid.
	\end{proof}
	
	The previous results show that we can convert an acyclic grid into a path of $r$-order linked sets and vice versa.

    \subsection{Constructing a path of order-linked sets}
    
    We now turn to the problem of actually constructing a path of $r$-order-linked sets in a digraph.
	We show first how to construct a path of $1$-order-linked sets from routing temporal digraphs containing $\Pk{k}$-routings.
    Similar to a column in an acyclic grid, such a $\Pk{k}$-routing allows us to shift one path to its destination without intersecting the other paths in the linkage we construct. 
	
	\begin{lemma}
		\label{lemma:P_k-routing-implies-1-order-linked}
		Let $h,k$ be integers.
		Let $T$ be the routing temporal digraph of some linkage $\mathcal{L}$ through
		a sequence $\Brace{H_1, H_2, \dots, H_{h}}$ of disjoint digraphs. 
		Let $\mathcal{L}' \subseteq \mathcal{L}$ be a linkage of order at most $k$.
		If $T$ contains a $\Pk{k}$-routing on the paths $L_1, L_2, \dots,
		L_{k} \in \mathcal{L}'$, ordered according to their occurrence on the
		$\Pk{k}$-routing, then $A$ is $1$-order-linked to $B$ in
		$\ToDigraph{\mathcal{L}' \cup \bigcup_{i=1}^{h} H_i}$, where $A = \Set{a_i \mid a_i \text{
				is the first vertex of \(L_i\) on \(H_1\)}}$ and $B = \Set{b_i \mid
			b_i \text{ is the last vertex of \(L_i\) on \(H_h\)}}$. 
	\end{lemma}
	\begin{proof}
		Let $A' \subseteq A$ and $B' \subseteq B$ such that $B'$ is a 1-shift of $A'$.
		Let $\pi: A' \rightarrow B'$ be the bijection witnessing that $B$ is a 1-shift of $A'$.
		If $\pi(a_x) = b_x$ for all $a_x \in A'$, then $\mathcal{L}'$ contains an $A'$-$B'$-linkage $\mathcal{R}$
		such that
		for all $a_x \in A'$ there is an $a_x$-$b_x$-path in $\mathcal{R}$, and so we are done.
		
		Otherwise, let $x \in \Set{1,\dots,k}$ be such that $a_x \in A'$ is the unique vertex such that $\pi(a_x) \neq b_x$ and let $b_y = \pi(a_x)$.
		As $B'$ is a 1-shift of $A'$, we know that $x < y$ and $a_i \not\in A'$ for all $i \in \{x+1, \ldots, y - 1\}$.
		
		We construct an $A'$-$B'$-linkage $\mathcal{R}$ satisfying $\pi(a_x) = \mathcal{R}(a_x)$ for all $a_x \in A'$ as follows.
				For each $a_j \in A'$, let $W_j$ be the temporal walk in the $\Pk{k}$-routing starting in $L_j$ and ending in $L_{j+1}$ and let $W$ be the concatenation of $W_x \cdot W_{x+1} \cdot \ldots \cdot W_{y-1}$.
		The temporal walk $W$ connects $L_x$ to $L_y$ in $T$, and we can assume it starts at \timestep $1$ and ends at \timestep $h$.
		If $\pi(a_j) = b_j$, we add the path $L_j$ to $\mathcal{R}$.
		Construct $L'_x$ as follows.
		
		Let $(v_i, t_i)$ and $(v_j, t_j)$ be two consecutive elements in the sequence of $W$.
		We follow $L_i$ from $H_{t_i}$ to $H_{t_j}$, then take a path $P$ in $H_{t_j}$ connecting $L_i$ to $L_j$.
		By construction of $T$ and because $a_n \not\in A'$ for all $n \in \{x+1, \ldots, y - 1\}$, the path $P$ in $H_{t_j}$ does not intersect any other path of $\mathcal{L}$.
		Further, by definition of $\Pk{k}$-routing, $L_i$ and $L_j$ only intersect $A'$ at $a_x$ or $a_y$.
		Hence, $L'_x$ is disjoint from all $L_i$ we chose earlier.
		Thus, we obtain an $A'$-$B'$-linkage $\mathcal{R}$ such that $\mathcal{R}(a_x) = \pi(a_x)$ for all $a_x \in A'$ as desired.
	\end{proof}
	
	The previous lemma allows us to construct a path of $1$-order-linked sets.
    Next, we show that we can increase the order-linkedness of the clusters at the expense of the length of the path of order-linked sets we obtain.
    The idea is to \say{merge} a set of consecutive clusters of a path of $r$-order-linked sets into a single cluster, increasing the order-linkedness.
	
	This idea is much easier to implement in uniform paths of $r$-order-linked sets than in the general case and yields much better bounds.
  As the uniform case is sufficient for our application, it is the only case considered here. 
	
	The next lemma explains how to construct a single cluster of higher order-linkedness in a path of $r$-order-linked sets by merging the existing clusters. 
	\begin{lemma}
		\label{lemma:increase-order-linkedness}
		Let $r,c,w$ be integers.
		Let $D = \Brace{\mathcal{S} = \Brace{S_0, S_1, \dots, S_{\ell}}, \mathscr{P}}$ be a uniform path of $r$-order-linked sets of width $w$ and length at least $c-1$. 
		Then $A(S_0)$ is $cr$-order-linked to $B(S_\ell)$ in $D$.
	\end{lemma}
	\begin{proof}
		Let $\Brace{\mathcal{P}_0, \mathcal{P}_1, \dots, \mathcal{P}_{\ell - 1}} \coloneqq \mathscr{P}$.
		For each $0 \leq i \leq \ell$, let $\varphi_i : A(S_i) \to B(S_i)$ be the bijection witnessing that $A(S_i)$ is $r$-order-linked to $B(S_i)$.
		
		We define for each $0 \leq i \leq \ell$ two bijections $\alpha_i : A(S_i) \to \Set{1,2,\ldots, w}$ and $\beta_i : B(S_i) \to \Set{1, 2, \ldots, w}$ according to $\leq_{A(S_i)}$ and $\leq_{B(S_i)}$, that is, $\alpha_i(a_j) \leq \alpha_i(a_k)$ holds if and only if $a_j \leq_{A(S_i)} a_k$ holds (and analogously for $\beta_i$).
		In particular, we have $\varphi_i = \beta_i^{-1} \circ \alpha_i$.
		Since $(\mathcal{S}, \mathscr{P})$ is uniform, we also have that $\beta_i(b) = \alpha_{i+1}(\mathcal{P}_i(b))$ for all $0 \leq i \leq \ell - 1$ and all $b \in B(S_i)$.
		
		Let $A' \subseteq A(S_{0})$ and let $B' \subseteq B(S_{\ell})$ be sets of size $k$ such that $B'$ is a $cr$-shift of $A'$ as witnessed by the bijection $\varphi : A' \to B'$.
		We also define $\pi \coloneqq \beta_\ell \circ \varphi \circ \alpha^{-1}_0$.
		
		For each $0 \leq i \leq c - 1$ we construct an $A'$-$B(S_{i})$-linkage $\mathcal{R}_i$ of order $\Abs{A'}$ satisfying the following,
		\begin{enamerate}{L}{item:increase-order-linkedness:pi-hits}
			\item \label{item:increase-order-linkedness:pi-hits}
			$\Abs{\Set{ a \in A' \mid \beta_i(\mathcal{R}_i(a)) = \beta_{\ell}(\varphi(a))}} \geq (i+1)r$.
		\end{enamerate}
		To simplify notation we set $\End{\mathcal{R}_{-1}}$ as $A'$.
		
		On step $i$, let $\mathcal{L}^1_i \subseteq \mathcal{P}_{i - 1}$ be such that $\Start{\mathcal{L}^1_i} = \End{\mathcal{R}_{i-1}}$ and let $\hat{\mathcal{R}}_{i-1} = \mathcal{R}_{i-1} \cdot \mathcal{L}^1_i$.
		
		Choose the largest possible $A'' \subseteq A'$ of size at most $r$ by starting at the largest elements of $A'$ with respect to $\leq_{A(S_0)}$ and proceeding in descending order such that $\alpha_i(\hat{\mathcal{R}}_{i-1}(a)) \neq \pi(a)$ for all $a \in A''$.
		Let $\hat{A} = \hat{\mathcal{R}}_{i-1}(A'')$.
		
		Let $B'' = \beta_i^{-1}(\pi(\alpha_i(\hat{A})) \cup \beta_i^{-1}(\alpha_i(\End{\hat{\mathcal{R}}_{i-1}} \setminus \hat{A}))$.
		Let $\varphi_i' : \End{\hat{\mathcal{R}}_{i-1}} \to B''$ be the bijection defined as follows
		\begin{align*}
			\varphi_i'(a) & \coloneqq
			\begin{cases}
				\beta_i^{-1}(\pi(\alpha_i(a))), & a \in \hat{A}\\[0em]
				\beta_i^{-1}(\alpha_i(a)), & a \in \End{\hat{\mathcal{R}}_{i-1}} \setminus \hat{A}.
			\end{cases}
		\end{align*}
		
		Because $A''$ was constructed by taking the largest elements of $A'$ with respect to $\leq_{A(S_0)}$, we have that $\alpha_i(a) \not\in \pi(\alpha_i(\hat{A}))$ for all $a \in \End{\hat{\mathcal{R}}_{i-1}} \setminus \hat{A}$.
		Hence, the set $B''$ is an $r$-shift of $\End{\hat{\mathcal{R}}_{i-1}}$ witnessed by $\varphi_i'$.
		
		Since $A(S_i)$ is $r$-order-linked to $B(S_i)$ in the cluster $S_i$, there is a linkage $\mathcal{L}_i^2$ in $S_i$ such that $\varphi_i'(\End{\hat{\mathcal{R}}_{i-1}}) = \mathcal{L}_i^2(\End{\hat{\mathcal{R}}_{i-1}})$.
		We now set $\mathcal{R}_i = \hat{\mathcal{R}}_{i-1} \cdot \mathcal{L}_i^2$.
		
		For every $a \in A''$ we now have $\alpha_{i+1}(\mathcal{R}_i(a)) = \varphi(a)$.
		Since~\cref{item:increase-order-linkedness:pi-hits} holds for $i-1$, we also have that~\cref{item:increase-order-linkedness:pi-hits} holds for $i$.
		
		After iterating $c$ steps, we have that $\mathcal{R}_{c-1}(A') = \varphi(A')$ since~\cref{item:increase-order-linkedness:pi-hits} holds for $i = c$ and $B'$ is an $cr$-shift of $A'$.
		Thus, $\mathcal{R}_{c-1}$ is an $A'$-$B'$-linkage.
		Hence, $A(S_0)$ is $cr$-order-linked to $B(S_\ell)$.
	\end{proof}
	
	By applying~\cref{lemma:increase-order-linkedness} repeatedly, we obtain the following
	theorem. 
	
	\begin{theorem}
		\label{lemma:merging path of order-linked sets}
		Every uniform path of $r$-order-linked sets $D = (\mathcal{S} = \Brace{S_0, S_1, \dots, S_{\ell}}, \mathscr{P} = \Brace{\mathcal{P}_0, \mathcal{P}_1, \dots, \mathcal{P}_{\ell - 1}})$ of length at least $c\ell$ and width $w$ contains a uniform path of $cr$-order-linked sets $\Brace{\mathcal{S}' = \Brace{ S'_0, S'_1, \dots, S'_{\ell}}, \mathscr{P}' = \Brace{ \mathcal{P}'_0, \mathcal{P}'_1, \dots, \mathcal{P}'_{\ell - 1}}}$ of length $\ell$ and width $w$.
        Additionally, for every $0 \leq i \leq \ell$ we have $S_i' \subseteq \SubPOSS{D}{ci}{c(i+1)-1}$, $A(S_i') \subseteq A(S_{ci})$ and $B(S_i') \subseteq B(S_{c(i + 1) - 1})$, and for $0 \leq i < \ell$ we have $\mathcal{P}'_i \subseteq \mathcal{P}_{(c-1)(i+1)}$.
	\end{theorem}
	\begin{proof}
		For each $0 \leq i \leq \ell$ let $S_i' = \SubPOSS{D}{ci}{c(i+1) - 1}$ and set $A(S_i') \coloneqq A(S_{ci})$ and $B(S_i') \coloneqq B(S_{c(i+1) - 1})$.
		Note that each $S_i'$ is a path of $r$-order-linked sets of width $w$ and length $c-1$.
		From~\cref{lemma:increase-order-linkedness}, we have that $A(S_i')$ is $cr$-order-linked to $B(S_i')$ in $S_i'$.
		
		Let $\mathscr{P}' \coloneqq \Brace{\mathcal{P}_{c - 1}, \mathcal{P}_{2(c - 1)}, \ldots, \mathcal{P}_{(c - 1)\ell}}$.
		It is immediate that $(\mathcal{S}' \coloneqq \Brace{S'_0, S'_1, \dots, S'_{\ell}}, \mathscr{P}')$ is a uniform path of $cr$-order-linked sets of width $w$ and length $\ell$ satisfying the requirements in the statement.
	\end{proof}
	
	\section{Paths of well-linked sets and fences}
	\label{sec:well-linked}
	
	The previous section introduces \emph{path of $r$-order-linked sets} as a suitable abstraction of acyclic grids.
	We now want to extend this idea to find a similar abstraction of fences as well.
	
	The main difference between a fence and an acyclic grid $(\PPP, \QQQ)$ is that if we are interested in routing from left to right, that is, from $\Start{P}$ to $\End{P}$, then, if $(\PPP, \QQQ)$ is an acyclic grid, the \(\Start{\mathcal{P}}\) is only order-linked to \(\End{\mathcal{P}}\), whereas in a fence these two sets are well-linked.
  Consequently, we modify this requirement of the path of $r$-order-linked sets to obtain a suitable abstraction of fences.

	\begin{definition}[path of well-linked sets]
		\label{def:path-of-well-linked-sets}
		A \emph{path of well-linked sets}\Index{path of well-linked sets} of width $w$ and length $\ell$ is a tuple $\Brace{\mathcal{S},\mathscr{P}}$ such that
		\begin{enumerate}             
			\item $\mathcal{S}$ is a sequence of $\ell + 1$ pairwise disjoint subgraphs $\Brace{S_0,\dots,S_{\ell}}$, called \emph{clusters},
			\item for every $0 \leq i \leq \ell$ there are disjoint sets $A(S_i),B(S_i) \subseteq \V{S_i}$ of size $w$ such that $A(S_i)$ is well-linked to $B(S_i)$ in $S_i$, and
			\item $\mathscr{P}$ is a sequence of $\ell$ pairwise disjoint linkages $\Brace{\mathcal{P}_0, \mathcal{P}_1, \dots, \mathcal{P}_{\ell - 1}}$ such that, for every $0 \leq i < \ell$, $\mathcal{P}_i$ is a $B(S_i)$-$A(S_{i+1})$-linkage of order $w$ which is internally disjoint from $S_i$ and $S_{i+1}$ and is disjoint from every $S \in \mathcal{S} \setminus \Set{S_i, S_{i+1}}$. 
		\end{enumerate}
		We call $(\SSS, \PPPP)$ \emph{strict} if within each cluster $S_i$ every vertex $v \in V(S_i)$ occurs on a path from $A(S_i)$ to $B(S_i)$ in $S_i$.
		
		As before we define $(\SSS, \PPPP)[i, j] \coloneqq \Brace{(S_i, \dots, S_j), (\PPP_i, \dots, \PPP_{j-1})}$, where $0 \leq i \leq j \leq \ell$. 
	\end{definition}
	\begin{figure}[!ht]
		\includegraphics{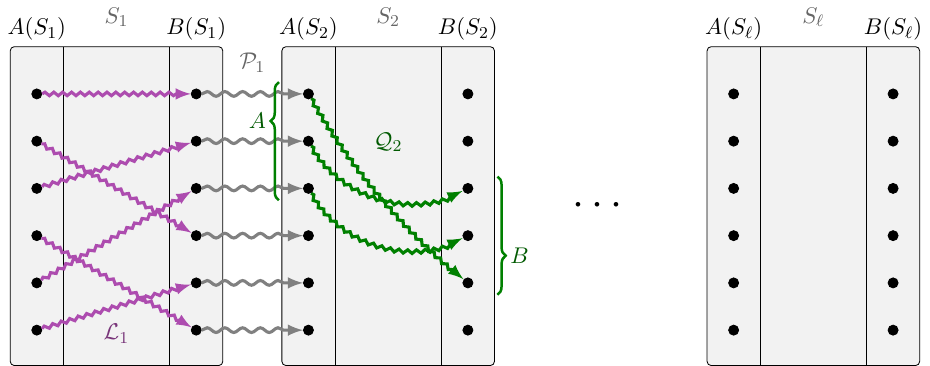}
		\caption{A path of well-linked sets of width $w$ and length $\ell$.
			The linkage $\mathcal{L}_1$ connects all of $\Fkt{A}{S_1}$ to all of $\Fkt{B}{S_1}$ while there are also linkages from every subset of $\Fkt{A}{S_i}$ to every subset of $\Fkt{B}{S_i}$ as $\mathcal{Q}_2$ in $S_2$ illustrates for example.}
		\label{fig:poss}
	\end{figure}
	
	While acyclic grids and fences may look quite different, Reed et al.~proved~\cite{reed1996packing} that it is possible to construct a fence from any given acyclic grid. 
	
	\begin{lemma}[{\cite[statement (4.7)]{reed1996packing}}]
		\label{obs:acyclic-grid-to-fence}
		Every acyclic $(pq + 1, pq +1)$-grid contains a $(p,q)$-fence.
	\end{lemma}
	
	Next, we show that a similar relation to the one proved in the previous lemma is also true for paths of order-linked sets and paths of well-linked sets.
	
	\begin{lemma}
		\label{proposition:order-linked to path of well-linked sets}
				Let $\bound{proposition:order-linked to path of well-linked sets}{w}{w, \ell} \coloneqq w(\ell + 1)$.
        \boundDef{proposition:order-linked to path of well-linked sets}{w}{w, \ell} 
		Every path of $w$-order-linked sets $(\mathcal{S} = \Brace{S_0, S_1, \dots, S_{\ell}}$, $\mathscr{P} = \Brace{\mathcal{P}_0, \mathcal{P}_1, \dots, \mathcal{P}_{\ell - 1}})$ of width at least $\bound{proposition:order-linked to path of well-linked sets}{w}{w, \ell}$ and length at least $\ell$ contains a path of well-linked sets $\Brace{\mathcal{S}' = \Brace{S'_0, S'_1, \dots, S'_{\ell}}, \mathscr{P}' = \Brace{\mathcal{P}'_0, \mathcal{P}'_1, \dots, \mathcal{P}'_{\ell - 1}}}$ of width $w$ and length $\ell$.
		Further, for every $0 \leq i \leq \ell$ we have $A(S_i') \subseteq A(S_i)$, $B(S_i') \subseteq B(S_i)$, $S_i' \subseteq S_i$ and for every $0 \leq i < \ell$ we have $\mathcal{P}_i' \subseteq \mathcal{P}_i$.
	\end{lemma}
	\begin{proof}
				Recall that $\leq_{A(S_i)}$ is the order of the vertices of $A(S_i)$ witnessing that $A(S_i)$ is $w$-order-linked to $B(S_i)$ in $S_i$.
		Further, let $\pi_i$ be the bijection witnessing this property.

		In the following we construct the sets $A_i'$ and $B_{i-1}'$ for each $1 \leq i \leq \ell$ such that $A_i'$ is a subset of the smallest $(i+1)w$ elements of $\leq_{A(S_i)}$ and	$\mathcal{P}_i$ contains a $B_{i-1}'$-$A_i'$-linkage of order $\Abs{A_i'}$ (see~\cref{fig:order-linked-to-well-linked} for an illustration).

        \begin{figure}[!ht]
        \resizebox{\textwidth}{!}{		\includegraphics{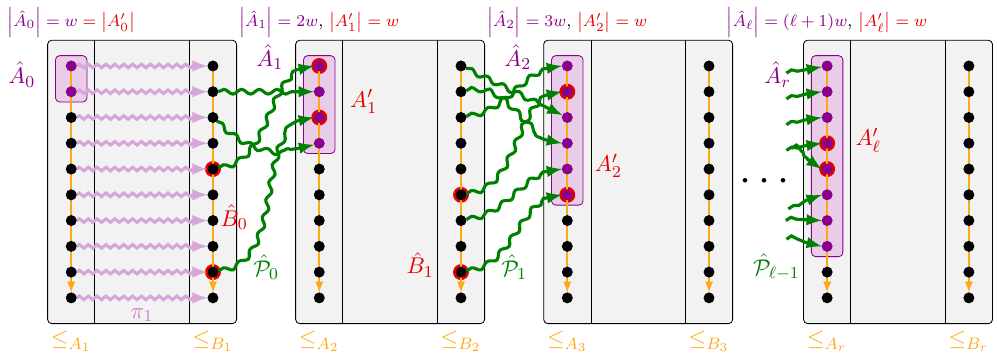}}
		\caption{A path of well-linked sets of width $w$ and length $\ell$.
			The linkage $\mathcal{L}_1$ connects all of $\Fkt{A}{S_1}$ to all of $\Fkt{B}{S_1}$ while there are also linkages from every subset of $\Fkt{A}{S_i}$ to every subset of $\Fkt{B}{S_i}$ as $\mathcal{Q}_2$ in $S_2$ illustrates for example.}
		\label{fig:order-linked-to-well-linked}
	\end{figure}

        First let $\hat{A}_0$ be the $w$ smallest elements of $\leq_{A(S_0)}$.
		For $0 < i \leq \ell$, let $\hat{A}_i$ be the $(i + 1)w$ smallest elements of $\leq_{A(S_i)}$.
		Since $\Abs{A(S_i)} \geq w(\ell + 1)$ and $i \leq \ell$, such a set exists.
  
		Now, let $\hat{\mathcal{P}}_{i-1}$ be the paths of $\mathcal{P}_{i-1}$ such that $\End{\hat{\mathcal{P}}_{i-1}} = \hat{A}_i$.
        Since $\hat{A}_{i-1}$ contains the smallest $iw$ elements of $\leq_{A(S_{i-1})}$, there is some $B_{i-1}' \subseteq \Start{\hat{\mathcal{P}}_{i-1}}$ of size $w$ such that $\pi_{i-1}(a) \leq_{B(S_{i-1})} b$ for all $a \in \hat{A}_{i-1}$ and all $b \in B_{i-1}'$.

        Finally, choose $A'_{i} \coloneqq \End{\hat{\mathcal{P}}_{i-1}}$ for all $0 < i \leq \ell$ and $A'_0 \coloneqq \hat{A}_0$.
  
		        As $A'_i \subseteq \hat{A}_i$ for all $0 \leq i \leq \ell$, we have $\pi_{i-1}(a) \leq_{B(S_{i-1})} b$ for all $a \in A_{i-1}'$ and all $b \in B_{i-1}'$.
		Hence, for every $A' \subseteq A_{i-1}'$ and every $B' \subseteq B_{i-1}'$ with $\Abs{A'} = \Abs{B'}$ we have that $B'$ is an $r$-shift of $A'$.
		Thus, $A_{i-1}'$ is well-linked to $B_{i-1}'$ in $S_{i - 1}$.
		Let $S_{i - 1}'$ be a minimal subgraph of $S_{i - 1}$ in which $A_{i-1}'$ is well-linked to $B_{i-1}'$.
		We set $A(S_{i-1}') \coloneqq A_{i-1}'$ and $B(S_{i-1}') \coloneqq B_{i-1}'$.
		Choose $\mathcal{P}'_{i-1} \subseteq \mathcal{P}_{i-1}$ such that $\Start{\mathcal{P}'_{i-1}} = B_{i-1}'$ and set $A_{i}' \coloneqq \End{\mathcal{P}'_{i-1}}$.
		
		After constructing all sets above, we choose $B_{\ell}'$ as the $w$ largest elements of $\leq_{B(S_{\ell})}$.
		As argued above, the set $A_{\ell}'$ is well-linked to $B_{\ell}'$ in some $S_i' \subseteq S_i$, where we choose $S_i'$ minimal.
		
		We set $\mathcal{S}' = \Brace{S'_0, S'_1, \dots, S'_{\ell}}$ and $\mathscr{P}' = \Brace{\mathcal{P}'_0, \mathcal{P}'_1, \dots, \mathcal{P}'_{\ell - 1}}$.
		It is immediate from the construction above that $\Brace{\mathcal{S}', \mathscr{P}'}$ is a path of well-linked sets of width $w$ and length $\ell$ satisfying the conditions in the statement.
	\end{proof}

    \subsection{Paths of well-linked sets contain fences}
	Next, we show that every path of well-linked sets contains a fence.
	Towards this end, we first show that the well-linkedness of $A(S_i)$ to $B(S_i)$ within an individual cluster $S_i$  can be preserved when going from one cluster to the next, i.e.~the set $A(S_i)$ is also well-linked to every  $A(S_j)$ and $B(S_j)$ for clusters $S_j$ with $j > i$ appearing later on the path of well-linked sets.
	
	\begin{lemma}
		\label{lem:poss-simple-routing}
		Let $(\SSS \coloneqq (S_0, \dots, S_{\ell}), \PPPP \coloneqq (\mathcal{P}_0, \dots, \mathcal{P}_{\ell-1}))$ be a path of well-linked sets of width $w$ and length $\ell$.
        Then for every $0 \leq i < j \leq \ell$, for each $A' \in \Set{A(S_i), B(S_i)}$ and for each $B' \in \Set{B(S_j), A(S_j)}$ we have that $A'$ is well-linked to $B'$ in $\SubPOSS{\Brace{\mathcal{S}, \mathscr{P}}}{i}{j}$. 
	\end{lemma}
	\begin{proof}
		We show the case where $A' = A(S_i)$ and $B' = B(S_j)$.
		The other cases follow analogously.
		
		Let $X \subseteq A'$ and $Y \subseteq B'$ be sets of size $k$.
		We prove by induction on $j - i$ that there is an $X$-$Y$-linkage of order $k$ in $\Brace{\mathcal{S}, \mathscr{P}}$.
		
		If $j-i = 1$, then let $B_i \subseteq B(S_i)$ be a set of size $k$ and let $A_{j} \subseteq A(S_j)$ the set of size $k$ such that $\mathcal{P}_i(B_i) = A_j$.
		Since $A(S_i)$ is well-linked to $B(S_i)$ in $S_i$, there is an $A'$-$B_i$-linkage $\mathcal{R}_i$ of order $k$ in $S_i$.
		Similarly, there is an $A_j$-$B'$-linkage $\mathcal{R}_j$ in $S_j$.
		Let $\mathcal{R}_i' \subseteq \mathcal{P}_i$ be the paths of $\mathcal{P}_i$ such that $\Start{\mathcal{R}_i'} = \End{\mathcal{R}_i}$.
		Clearly, $\mathcal{R}_i \cdot \mathcal{R}_i' \cdot \mathcal{R}_j$ is an $A'$-$B'$-linkage of order $k$.
		
		Now consider the case where $j - i > 1$.
		Choose any subset $B_i \subseteq B(S_{i})$ of order $\Abs{A'} = k$.
		As before there is an $A'$-$B_i$-linkage $\mathcal{R}_1$ of order $k$ in $S_i$.
		Let $\mathcal{R}_2 \subseteq \mathcal{P}_i$ be the linkage with $\Start{\mathcal{R}_2} = \End{\mathcal{R}_1}$.
		Note that $\End{\mathcal{R}_2} \subseteq A(S_{i+1})$.
		By induction, there is an $\End{\mathcal{R}_2}$-$B'$-linkage $\mathcal{R}_3$ of order $k$, and so $\mathcal{R}_1 \cdot \mathcal{R}_2 \cdot \mathcal{R}_3$ is an $A'$-$B'$-linkage of order $k$, as desired.
	\end{proof}
	
	Before we can use~\cref{thm:routing star or path} to construct our fence, we need to lift the connectivity provided by \(H\)-routings on digraphs to temporal digraphs.

	\begin{lemma}
		\label{statement:lifting-H-routings-to-temporal-digraphs}
		Let \(T\) be a temporal digraph with \(\Lifetime{T} \geq \Abs{\V{T}}\) and let \(H\) be a digraph.
		If there is a function \(\varphi : \V{H} \to \V{T}\) which is an \(H\)-routing for every layer of \(T\), then \(\varphi\) is an \(H\)-routing for \(T\).
	\end{lemma}
	\begin{proof}
		Let \(P\) be a path in \(H\).
		Let \(u = \varphi(\Start{P})\) and \(v = \varphi(\End{P})\).
				Let \(X = \Image{\varphi} \setminus \varphi(\V{P})\).
		We show that there is a temporal \(u\)-\(v\)-path in \(T\) which is disjoint from \(X\).
		For each \(1 \leq i \leq \Lifetime{T}\) let \(V_i\) be the set of vertices temporally reachable from \(u\) in at most \(i\) \timesteps when starting at \timestep one without visiting any vertex in \(X\).
		To simplify notation, let \(V_0 = \emptyset\).
		We prove that, if \(v \notin V_i\), then \(V_{i} \neq V_{i - 1}\).

		For \(i = 1\) the statement is trivial, as \(u \in V_1 \setminus V_0\).
		Now assume \(i > 1\).
		Let \(L\) be the \(u\)-\(v\)-path witnessing that \(\varphi\) is an \(H\)-routing in \(\Layer{T}{i}\).
		If \(\V{L} \subseteq V_{i - 1}\), then we are done and \(u\) can temporally reach \(v\) in \(T\) without intersecting \(X\).

		Otherwise, as \(u \in V_{i - 1}\), there is some \(w_1 \in V_{i - 1} \setminus \V{L}\) and some \(w_2 \in \V{L} \setminus V_{i - 1}\) such that \((w_1, w_2) \in \A{\Layer{T}{i}}\).
		Let \(R\) be a temporal \(u\)-\(w_1\)-path in \(T\) ending on \timestep \(i - 1\) without intersecting \(X\).
		Then by adding the arc \((w_1, w_2)\) on \timestep \(i\) to \(R\), we obtain a temporal \(u\)-\(w_2\)-path.
		Further, as \(R\) comes from the \(H\)-routing, it does not contain any vertex from \(X\).
		Hence, \(w_2 \in V_{i} \setminus V_{i - 1}\).

		Since \(\Lifetime{T} \geq \Abs{V}\), after \(t = \Abs{V}\) steps we have that \(v \in V_{t}\), and so there is a temporal \(u\)-\(v\)-path in \(T\) which does not intersect \(X\).
	\end{proof}

	We now apply our framework of $\Pk{k}$-routings in temporal digraphs to construct a fence from a path of well-linked sets.
    The idea is first to construct an acyclic grid using $\Pk{k}$-routings and then apply~\cref{obs:acyclic-grid-to-fence} to obtain a fence. 
	We define
	\begin{align*}
		\boundDefAlign{thm:poss-to-fence}{w}{p,q} 
		\bound{thm:poss-to-fence}{w}{p,q} 
		& \coloneqq \bound{thm:routing star or path}{n}{pq + 1, pq + 1},
		\\[0em]
		\bound{thm:poss-to-fence}{\ell}{p,q}
		& \coloneqq (pq(pq + 1)\bound{thm:poss-to-fence}{w}{p,q} + (pq + 1)!) \cdot \binom{\bound{thm:poss-to-fence}{w}{p,q}}{pq + 1} + pq(pq + 1) 
	\end{align*}
	Observe that \(\bound{thm:poss-to-fence}{w}{p,q} \in \Oh(p^{5} q^{5})\) and \(\bound{thm:poss-to-fence}{\ell}{p,q} \in \PowerTower{1}{\Polynomial{15}{p, q}}\).

	\begin{figure}[!ht]
        \resizebox{\textwidth}{!}{		\includegraphics{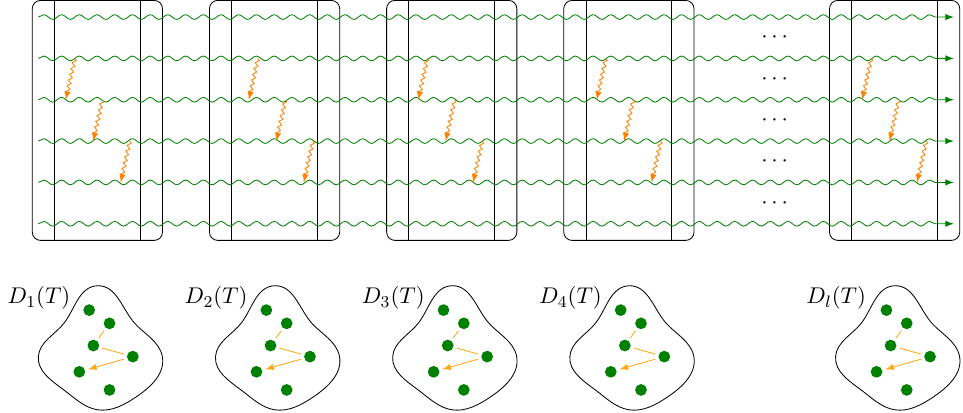}}
		\caption{
			For each cluster $S_i$, we obtain a digraph $D_{i+1}(T)$ from the temporal digraph $T$.
			Every $D_i$ contains a path of length $k_1$ or a $\biK{k_1}$-routing, which means it contains a $\Pk{k_1}$-routing in any case.
 			As there are enough clusters, we can find $k_4-1$ agreeing on the vertices and their order, shown in orange.}
		\label{fig:poss-to-fence}
	\end{figure}

	\begin{theorem}
		\label{thm:poss-to-fence}
		Every path of well-linked sets $\Brace{\mathcal{S} = (S_{0}, S_{1}, \ldots, S_{\ell}), \mathscr{P}}$ of width at least $\bound{thm:poss-to-fence}{w}{p,q}$ and length $\ell \geq \bound{thm:poss-to-fence}{\ell}{p,q}$ contains a $(p,q)$-fence \((\mathcal{P}, \mathcal{Q})\).
		Moreover, \(\Start{\mathcal{P}} \subseteq A(S_0)\) and \(\End{\mathcal{P}} \subseteq B(S_\ell)\).
	\end{theorem}
	\begin{proof}
		Let $k_2 = pq + 1$,
		$k_1 = \bound{thm:routing star or path}{n}{k_2, k_2}$,
		\(\ell_7 = k_1 k_2 (k_2 - 1)\),
		\(\ell_6 = \binom{k_1}{k_2}\),
		\(\ell_5 = \ell_6 \ell_7\),
		\(\ell_4 = k_2(k_2 - 1)\),
		\(\ell_3 = \binom{k_1}{k_2} \cdot (k_2)!\),
		\(\ell_2 = \ell_3 + \ell_4\), and
		\(\ell_1 = \ell_2 + \ell_5\).
		Note that \(\bound{thm:poss-to-fence}{\ell}{p,q} = \ell_1\) and \(\bound{thm:poss-to-fence}{w}{p,q} = k_1\).

		Let $D = \ToDigraph{\Brace{\mathcal{S}, \mathscr{P}}}$.
		Let $\Brace{S_0, S_1, \dots, S_{\ell}} = \mathcal{S}$ and let $\mathcal{L}$ be an $A(S_0)$-$B(S_\ell)$-linkage of order $k_1$ within $\Brace{\mathcal{S}, \mathscr{P}}$.
		By~\cref{lem:poss-simple-routing}, such a linkage exists.

		Let \(T\) be the routing temporal digraph of \(\mathcal{L}\) through \(\mathcal{S}\) (see~\cref{fig:poss-to-fence}).
		Since $A(S_i)$ is well-linked to $B(S_i)$ and every path in $\mathcal{L}$ must intersect both $A(S_i)$ and $B(S_i)$ for every $S_i \in \mathcal{S}$, we have that every layer $\Layer{T}{t}$ is strongly connected.

		By~\cref{thm:routing star or path}, every $\Layer{T}{t}$ contains a \(\Pk{k_2}\) subgraph or a $\biK{k_2}$-routing \(\varphi_t\).

		We define a concatenation operation for temporal digraphs as follows.
		Given two temporal digraphs \(T_1 = (V_1, \mathcal{A}_1)\) and \(T_2 = (V_2, \mathcal{A}_2)\),
		we define \(T_1 \cdot T_2 \coloneqq (V_1 \cup V_2, \mathcal{A}_1 \cdot \mathcal{A}_2)\).

		We show that we can decompose \(T\) into
		\begin{equation*}
			T = T_{1,1} \cdot T_{1,2} \cdot \ldots \cdot T_{1,k_2 - 1} \cdot T_{2,1} \cdot \ldots \cdot T_{2,k_2 - 1} \cdot \ldots \cdot T_{k_2,1} \cdot \ldots \cdot T_{k_2, k_2 - 1}
		\end{equation*}
		in such a way that there exists vertices \(u_{1}, u_{2}, \ldots, u_{k_2} \in \V{T}\) such that, for each \(1 \leq i \leq k_2\) and each \(1 \leq j < k_2\), there is a temporal \(u_j\)-\(u_{j + 1}\)-path \(R_{i,j}\) in \(T_{i,j}\) which does not intersect any vertex \(u_a \notin \{u_j, u_{j + 1}\}\).
		As \(\Lifetime{T} \geq \ell_1 = \ell_2 + \ell_5\), we obtain two cases.
        
		\begin{CaseDistinction}
			\Case{There are \(\ell_2\) layers of \(T\) which contain some path on \(k_2\) vertices.}

			There are \(\ell_3 = \binom{k_1}{k_2} \cdot (k_2)!\) many different choices for the path \(\Pk{k_2}\).
			As \(\ell_2 = \ell_3 \cdot \ell_4\), there are \(\ell_4\) layers of \(T\) with exactly the same path \(\Pk{k_2}\).

			Let \(u_{1}, u_{2}, \ldots, u_{k_2}\) be the vertices of the \(\Pk{k_2}\) above sorted according to their order along the path.
			Let \(t_{1,1}, t_{1,2}, \ldots, t_{k_2, k_2 - 1}\) be the \timesteps in which \(\Layer{T}{t_{i,j}}\) contains the path \(\Pk{k_2}\) above.
			Hence, we can decompose \(T\) into temporal digraphs \(T_{i,j}\) where each \(T_{i,j}\) starts at \timestep \(t_{i,j}\) and ends before the next \timestep in the sequence above (or at the end of \(T\)).

			For each \(1 \leq i \leq k_2\) and each \(1 \leq j < k_2\), we choose the arc \((u_{j}, u_{j + 1})\) in \(\Layer{T}{t_{i,j}}\) as the desired path \(R_{i,j}\).
		
			\Case{There are \(\ell_5\) layers of \(T\) which contain some \(\biK{k_2}\)-routing.}
			
			There are \(\ell_6 = \binom{k_1}{k_2}\) different choices for the \(\biK{k_2}\)-routing \(\varphi_t\).
			As \(\ell_5 = \ell_6 \cdot \ell_7\), there are, by the pigeon-hole principle, \(\ell_7\) layers \(I = \{t_1, t_2, \ldots, t_{\ell_7}\}\) with the same \(\biK{k_2}\)-routing \(\varphi\).
			We sort \(I\) such that \(t_a < t_b\) if \(a < b\).

			Let \(f(i,j) = (i - 1)(k_1(k_2 - 1)) + jk_1\).
			Note that \(f(i,j + 1) - f(i,j) = k_1\) and that \(f(i+1,1) - f(i,k_2 - 1) = k_1\).
			For each \(1 \leq i \leq k_2\) and each \(1 \leq j \leq k_2 - 1\), let \(I_{i,j} = \{t_{f(i,j)}, \ldots, t_{f(i,j) + k_1 - 2}\}\) and let \(T_{i,j}\) be the temporal subdigraph of \(T\) obtained by taking only the \timesteps in \(I_{i,j}\).
			As \(\varphi\) is a \(\biK{k_2}\)-routing for every layer of \(T_{i,j}\), by~\cref{statement:lifting-H-routings-to-temporal-digraphs} we have that \(\varphi\) is a \(\biK{k_2}\)-routing for \(T_{i,j}\).
			We choose \(R_{i,j}\) as the temporal \(u_{j}\)-\(u_{j+1}\)-path given by \(\varphi\) in \(T_{i,j}\).
			As \(\varphi\) is a \(\biK{k_2}\)-routing in \(T_{i,j}\), the temporal path \(R_{i,j}\) does not intersect any \(u_a\) other than \(u_j\) and \(u_{j + 1}\).
		\end{CaseDistinction}

		Now, we construct an acyclic \((k_2, k_2)\)-grid as follows.
		Observe that the vertices \(u_{1}, u_{2}, \ldots, u_{k_2}\) defined above correspond to paths \(U_{1}, U_{2}, \ldots, U_{k_2}\) in \(D\).
		Further, $R_{i,j}$ corresponds to a $\V{U_j}$-$\V{U_{j+1}}$-path in \(D\), which we call $Q_{i,j,2}$.
		Define \(\mathcal{P} = \Set{U_{1}, U_{2}, \ldots, U_{k_2}}\).
		Let $Q_{i,j,1}$ be the $\End{Q_{i,j-1,2}}$-$\Start{Q_{i,j,2}}$-path in $U_j$ (to simplify notation, we choose $\End{Q_{i,0,2}}$ as $\Start{Q_{i,1,2}}$).
		
		We now set $Q_i = Q_{i,1,1} \cdot Q_{i,1,2} \cdot Q_{i,2,1} \cdot Q_{i,2,2} \cdot \ldots \cdot Q_{i,k_2 - 1,2}$.
		After constructing all $Q_i$, set $\mathcal{Q} = \Brace{Q_1, Q_2, \dots, Q_{k_2}}$.
		Note that the paths in $\mathcal{Q}$ are pairwise disjoint.
		
		It is immediate from the construction that $(\mathcal{P}, \mathcal{Q})$ is an acyclic $(k_2, k_2)$-grid.
		By~\cref{obs:acyclic-grid-to-fence}, $\ToDigraph{\mathcal{P} \cup \mathcal{Q}}$ contains a $(p,q)$-fence, as desired.
	\end{proof}

    \subsection{Routing through a path of well-linked sets}
    
	We close this section by exhibiting various routing properties of paths of well-linked sets used later.
	We first observe the following simple property.
	
	\begin{observation}
		\label{obs:large-linkage-containing-vertex-well-linked}
		Let $D$ be a digraph and $A, B \subseteq \V{D}$ be sets in $D$ such that $A$ is well-linked to $B$.
		Let $v \in \V{D}$ be a vertex on some $A$-$B$-path.
		Then there is an $(A \cup \Set{v})$-$B$-linkage $\mathcal{L}$ of order $\Abs{A}$ such that $v \in \Start{\mathcal{L}}$.
	\end{observation}
	\begin{proof}
		Let $\mathcal{R}$ be some $A$-$B$-linkage of order $\Abs{A}$.
		Let $P$ be some $A$-$B$-path containing $v$ and let $P'$ be the $v$-$B$-subpath of $P$.
		Let $P''$ be the largest subpath of $P'$ with $\Start{P''} = \Start{P'}$ which is internally disjoint from $\mathcal{R}$ and let $R \in \mathcal{R}$ be the path of $\mathcal{R}$ intersecting $P''$.
		Finally, let $R'$ be the $\End{P''}$-$\End{R}$ subpath of $R$.
		It is now immediate that $\mathcal{R}' \coloneqq (\mathcal{R} \setminus \Set{R}) \cup \Set{P'' \cdot R'}$ is a linkage of order $\Abs{R}$ with $v \in \Start{\mathcal{R}'}$.
	\end{proof}
	
	When working with paths of well-linked sets later on, we are often in a situation where we are given two equal-sized sets $X$ and $Y$ of vertices in a path of well-linked sets $(\SSS, \PPPP)$ and we want to find a linkage connecting $X$ to $Y$ within $(\SSS, \PPPP)$.
    In the next~\namecref{lem:linkage_inside_poss}, we identify several cases in which these linkages are guaranteed to exist.
    This~\namecref{lem:linkage_inside_poss} is frequently applied in the next steps of the proof. 
	
	\begin{lemma}
        \label{lem:linkage_inside_poss}
        Let $\Brace{\mathcal{S} = \Brace{S_0, S_1, \dots, S_{\ell}}, \mathscr{P} = \Brace{\mathcal{P}_0, \mathcal{P}_1, \dots, \mathcal{P}_{\ell - 1}}}$ be a path of well-linked sets of width  $w$ and length $\ell$.
		Let $X,Y \subseteq \V{(\mathcal{S}, \mathscr{P})}$ such that $\Abs{X} = \Abs{Y} = k$.
		Let $f : X \cup Y \rightarrow \N$ be a function such that $v \in S_{\Fkt{f}{v}} \cup \mathcal{P}_{\Fkt{f}{v}}$ for all $v \in X \cup Y$.
		There is an $X$-$Y$-linkage $\mathcal{L}$ in $(\mathcal{S}, \mathscr{P})$ if $\Fkt{f}{x} \leq \Fkt{f}{y} - 2$ for all $x \in X$ and all $y \in Y$ and at least one of the following is true:
		\begin{enamerate}{N}{last-item-linkage-inside-POSS}
			\item there are $0 \leq i < j \leq \ell$ such that $X \subseteq B(S_{i})$ and $Y \subseteq A(S_{j})$,
                \label{case:linkage_inside_poss_B_A}
			\item $\Abs{\Fkt{f}{x_1} - \Fkt{f}{x_2}} \geq 2$ for all $x_1, x_2 \in X$ with $x_1 \neq x_2$ and there is some $0 \leq i \leq \ell$ such that $Y \subseteq A(S_{i})$,
		          \label{case:linkage_inside_poss_scattered_A}
			\item $\Abs{\Fkt{f}{y_1} - \Fkt{f}{y_2}} \geq 2$ for all $y_1, y_2 \in Y$ with $y_1 \neq y_2$ and there is some $0 \leq i \leq \ell$ such that $X \subseteq B(S_{i})$, or
		          \label{case:linkage_inside_poss_B_scattered}
			\item $\Abs{\Fkt{f}{x_1} - \Fkt{f}{x_2}} \geq 2$ for all $x_1, x_2 \in X$ with $x_1 \neq x_2$ and $\Abs{\Fkt{f}{y_1} - \Fkt{f}{y_2}} \geq 2$ for all $y_1, y_2 \in Y$ with $y_1 \neq y_2$.
                \label{case:linkage_inside_poss_scattered_scattered}
                \label{last-item-linkage-inside-POSS}
		\end{enamerate}
		Furthermore, $\mathcal{L}$ is contained inside $\SubPOSS{\Brace{\mathcal{S}, \mathscr{P}}}{i}{j}$ where $i$ is minimal with $S_i$ containing a vertex from $X$ and $j$ is maximal with $S_j$ containing a vertex from $Y$.
	\end{lemma}
	\begin{proof}
		The case where~\cref{case:linkage_inside_poss_B_A} holds follows directly from~\cref{lem:poss-simple-routing}.
		
		If~\cref{case:linkage_inside_poss_scattered_A} holds, we construct an $X$-$A(S_{i-1})$-linkage of order $k$ as follows.
		We first rename the vertices of $X$ such that $x_r \in \V{S_r}$ for all $x_r \in X$.
		For each $x_r \in X$, let $k_r$ be the number of vertices in $X$ which appear before $x_r$ along $\Brace{\mathcal{S}, \mathscr{P}}$.
		
		By~\cref{obs:large-linkage-containing-vertex-well-linked}, there is an $(A(S_{r}) \cup \Set{x_{r}})$-$B(S_{r})$-linkage $\mathcal{R}_r$ of order $k_r + 1$ in $S_{r}$ such that $x_{r} \in \Start{\mathcal{R}_r}$.
		If $k_r > 0$, then by~\cref{lem:poss-simple-routing} there is an $\End{\mathcal{R}_{r-1}}$-$\Start{\mathcal{R}_{r}}$-linkage $\mathcal{L}_{r-1}$ of order $k_r = \Abs{\mathcal{R}_{r}} - 1$.
		
		Clearly, the concatenation of all $\mathcal{R}_r$ and all $\mathcal{L}_r$ above (in the only order possible) yields an $X$-$B(S_{r'})$-linkage of order $\Abs{X}$, where $r'$ is the smallest index such that all vertices of $X$ appear before $S_{r'}$ along $\Brace{\mathcal{S}, \mathscr{P}}$.
		Now, by~\cref{lem:poss-simple-routing}, we have an $\End{\mathcal{R}_{r' - 1}}$-$Y$-linkage of order $k$, as desired. 		
		The proof for the case where~\cref{case:linkage_inside_poss_B_scattered} holds is analogous to the one of where~\cref{case:linkage_inside_poss_scattered_A} holds and so we omit it.
		
		If~\cref{case:linkage_inside_poss_scattered_scattered} holds, let $r_y$ be the smallest index such that $S_{r_y}$ contains a vertex of $Y$ and let $r_x$ be the largest index such that $S_{r_x}$ contains a vertex of $X$.
		
		Construct linkages $\mathcal{R}_r$ and $\mathcal{L}_r$ as in the proof of the case when~\cref{case:linkage_inside_poss_scattered_A} holds.
		Let $\mathcal{X}$ be the linkage obtained by concatenating all $\mathcal{R}_r$ and all $\mathcal{L}_r$ (in the only possible order) belonging to the vertices of $X$.
		Similarly, let $\mathcal{Y}$ be the linkage obtained by concatenating all $\mathcal{R}_r$ and all $\mathcal{L}_r$ (in the only possible order) belonging to the vertices of $Y$.
		
		Note that $\End{\mathcal{X}} \subseteq B(S_{r_x})$ and that $\Start{\mathcal{Y}} \subseteq A(S_{r_y})$.
		Hence, by~\cref{lem:poss-simple-routing} there is an $\End{\mathcal{X}}$-$\Start{\mathcal{Y}}$-linkage of order $k$, as desired. 	\end{proof}
	
	The last statement we prove in this section helps us to deal with a situation in which we already have a path of well-linked sets $(\SSS \coloneqq (S_0, \dots, S_\ell), \PPPP)$ but we would like to restrict the system so that it \say{starts} at a specific set $A \subseteq A(S_0)$ and ends at some fixed set $B \subseteq B(S_\ell)$. 
	
	\begin{observation}
		\label{obs:restricting_width_poss}
		Let $\Brace{\mathcal{S} = \Brace{S_0, S_1, \dots, S_{\ell}}, \mathscr{P} = \Brace{\mathcal{P}_0, \mathcal{P}_1, \dots, \mathcal{P}_{\ell - 1}}}$ be a path of well-linked sets of width at least $w$ and length $\ell$.
		Let $A_{0} \subseteq A(S_0)$ and $B_{\ell} \subseteq B(S_{\ell})$ with $\Abs{A(S_{0})} = \Abs{B_\ell} = w$.
		Then, $\Brace{\mathcal{S}, \mathscr{P}}$ contains a path of well-linked sets $(\mathcal{S}' = \Brace{S'_0, S'_1, \dots, S'_{\ell}}, \mathscr{P}' = \Brace{\mathcal{P}'_0, \mathcal{P}'_1, \dots, \mathcal{P}'_{\ell - 1}})$ of width $w$ and length $\ell$ such that $B(S_\ell') = B_\ell$, $A(S_0') = A_0$, $S'_i \subseteq S_i$ for all $0 \leq i \leq \ell$ and $\mathcal{P}'_i \subseteq \mathcal{P}_i$ for all $0 \leq i < \ell$.
	\end{observation}
	\begin{proof}
		For each $0 \leq i < \ell$ choose some $B_{i} \subseteq B(S_i)$ of size $w$ and let $\mathcal{P}'_i \subseteq \mathcal{P}_i$ be such that $\Start{\mathcal{P}_i'} = B_{i}$.
		For each $1 \leq i \leq \ell$ let $A_{i} = \mathcal{P}_{i-1}'(B_{i-1})$.
		
		For each $0 \leq i \leq \ell$ let $S_i' \subseteq S_i$ be a maximal subgraph of $S_i$ such that $A_{i}$ is well-linked to $B_{i}$ in $S_i'$ and for each $v \in \V{S_i'}$ there is some $A_{i}$-$B_{i}$-path $P$ in $S_i'$ containing $v$.
		Clearly, if no such path $P$ exists for some vertex $v$, then we can remove $v$ from $S_i'$ while preserving the property that $A_{i}$ is well-linked to $B_{i}$.
		Hence, such a subgraph $S_i'$ exists.
		We then set $A(S_i') \coloneqq A_{i}$ and $B(S_i') \coloneqq B_{i}$.
		
		By construction, $\Brace{\Brace{S_0', \dots, S_{\ell}'}, \Brace{\PPP_0', \dots, \PPP_{\ell-1}'}}$ is a path of well-linked sets of width $w$ and length $\ell$, as desired.
	\end{proof}

	\section{Constructing a path of well-linked sets}
	\label{sec:constructing-pows}
	
	We show how to obtain a path of well-linked sets from splits and segmentations by using the results from~\cref{sec:temporal}, where we defined the \emph{routing temporal digraph} of a linkage $\mathcal{L}$ through a sequence of disjoint digraphs $H_1, H_2, \dots, H_{t}$.

	In order to construct the routing temporal digraph, the linkage $\mathcal{L}$ must intersect all $H_i$ in an \emph{ordered} fashion.
	This means that, if one of the linkages in a web $(\mathcal{H}, \mathcal{V})$ is \emph{ordered} with respect to the other, then we can construct such a routing temporal digraph.
	This leads us to the following definition of \emph{ordered web} (see~\cref{fig:ordered-web} for an example of an ordered web).

	\begin{definition}
		\label{def:ordered_web}
		Let $\Brace{\mathcal{H}, \mathcal{V}}$ be an $\Brace{h,v}$-web.
		We say that $\Brace{\mathcal{H}, \mathcal{V}}$ is an \emph{ordered web} if there is an ordering of $\mathcal{V} = \Brace{V_1, V_2, \dots, V_{v}}$ for which each path $H \in \mathcal{H}$ can be decomposed into $H = H_1 \cdot H_2 \cdot \ldots \cdot H_{v}$ such that $H_i$ intersects $V_j$ if and only if $i = j$.
	\end{definition}

    \begin{figure}[!ht]
        \centering
        \includegraphics{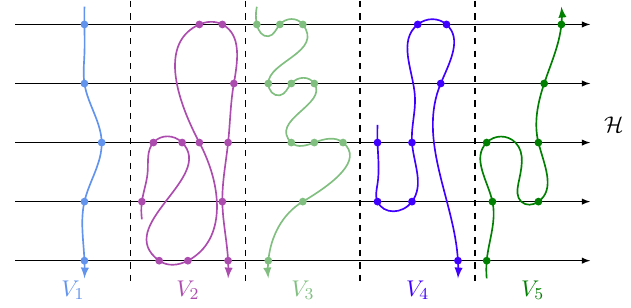}
        \caption{A $(5,3)$-ordered web $(\mathcal{H},\Set{V_1,V_2,V_3,V_4,V_5})$.}
        \label{fig:ordered-web}
    \end{figure}

  \subsection{From ordered webs}

	Next, we show how to construct a path of $1$-order-linked sets from an ordered web using the framework of $H$-routings developed in~\cref{sec:temporal}.
	The statement we present is more general than what we need here,
	but we require this stronger version in \cite{COSSII}.
	We start by defining
	\begin{align*}
		\boundDefAlign{lemma:ordered-web-to-path-of-order-linked-sets-with-forward-linkage}{h}{w}
		\bound{lemma:ordered-web-to-path-of-order-linked-sets-with-forward-linkage}{h}{w} & \coloneqq w^2 - 1,\\[0em]
		\boundDefAlign{lemma:ordered-web-to-path-of-order-linked-sets-with-forward-linkage}{v}{w,\ell}
		\bound{lemma:ordered-web-to-path-of-order-linked-sets-with-forward-linkage}{v}{w,\ell} & \coloneqq 
		(w \ell \cdot \binom{\bound{lemma:ordered-web-to-path-of-order-linked-sets-with-forward-linkage}{h}{w}}{w} \cdot w! + 1) \cdot \bound{theorem:one-way connected temporal digraph contains P_k routing}{\Lifetime{}}{w,\bound{lemma:ordered-web-to-path-of-order-linked-sets-with-forward-linkage}{h}{w}} - 1.
	\end{align*}
	Observe that \(\bound{lemma:ordered-web-to-path-of-order-linked-sets-with-forward-linkage}{v}{w,\ell} \in \PowerTower{1}{\Polynomial{13}{\ell, w}}\).
	
	\begin{lemma}
		\label{lemma:ordered-web-to-path-of-order-linked-sets-with-forward-linkage}
		Let $\Brace{\mathcal{H}, \mathcal{V}}$ be an ordered $(h,v)$-web where $h = \bound{lemma:ordered-web-to-path-of-order-linked-sets-with-forward-linkage}{h}{w}$ and $v \geq \bound{lemma:ordered-web-to-path-of-order-linked-sets-with-forward-linkage}{v}{w,\ell}$.
		Then $(\mathcal{H}, \mathcal{V})$ contains a path of $w$-order-linked sets $\Brace{\mathcal{S} = \Brace{ S_0, S_1, \dots, S_{\ell} }, \mathscr{P}}$ of width $w$ and length $\ell$ with the following additional properties.
		\begin{itemize}
			\item There is a $\Start{\mathcal{H}}$-$\End{\mathcal{H}}$-linkage $\mathcal{L} = \mathcal{L}_1 \cdot \mathcal{L}_2 \cdot \mathcal{L}_3$ of order $w$ contained in $\mathcal{H}$ such that $\mathcal{L}_2$ is an $A(S_0)$-$B(S_\ell)$-linkage and both $\mathcal{L}_1$ and $\mathcal{L}_3$ are internally disjoint from $\Brace{\mathcal{S}, \mathscr{P}}$.
			\item There is a linkage $\mathcal{X} \subseteq \mathcal{V}$ of order $\ell + 1$ and a bijection $\pi : \mathcal{S} \rightarrow \mathcal{X}$ such that $A(S_i) \subseteq \V{\pi(S_i)}$ and $\V{\pi(S_i)} \cap \V{(\mathcal{S}, \mathscr{P})} \subseteq \V{S_i}$ for each $0 \leq i \leq \ell$.
		\end{itemize}
	\end{lemma}
	\begin{proof}
		Let $\ell_1 \coloneqq w \ell + 1$ and let $\ell_2 \coloneqq (\ell_1 - 1) \cdot \binom{h}{w} \cdot w! + 1$.
		
        We define $f(i) \coloneqq (i-1)\cdot \bound{theorem:one-way connected temporal digraph contains P_k routing}{\Lifetime{}}{w, h} + 1$ and observe that $f(i) - f(i-1) = \bound{theorem:one-way connected temporal digraph contains P_k routing}{\Lifetime{}}{w,h}$.
		Let $\Brace{V_1 \cdot V_2 \cdot \ldots \cdot V_{v}} \coloneqq \mathcal{V}$ be an ordering of $\mathcal{V}$ witnessing that $\Brace{\mathcal{H}, \mathcal{V}}$ is an ordered web.
		Observe that $f(\ell_2 + 1) - 1 = \bound{lemma:ordered-web-to-path-of-order-linked-sets-with-forward-linkage}{v}{w,\ell} \leq v$.
		
		Decompose $\mathcal{H}$ into $\mathcal{H} = \mathcal{H}^0 \cdot \mathcal{H}^1 \cdot \ldots \cdot \mathcal{H}^{\ell_2}$, where $\Start{\mathcal{H}_0} = \Start{\mathcal{H}}$, $\End{\mathcal{H}_{\ell_2}} = \End{\mathcal{H}}$ and, for each $1 \leq i \leq \ell_2 - 1$, the sublinkage $\mathcal{H}_i$ starts at the first intersections of $\mathcal{H}$ with $V_{f(i)}$ and ends at the first intersections of $\mathcal{H}$ with $V_{f(i+1)}$.
		For each $1 \leq t \leq \ell_2$, let $\mathcal{V}^t \coloneqq \Brace{V_{f(t)}, V_{f(t) + 1}, \ldots, V_{f(t) + \bound{theorem:one-way connected temporal digraph contains P_k routing}{\Lifetime{}}{w, h} - 1}}$ and let $T_t$ be the routing temporal digraph of $\mathcal{H}$ through $\mathcal{V}^{t}$.
		
		Each layer of each $T_i$ is unilateral since every path in $\mathcal{H}$ intersects every path in $\mathcal{V}$.
		As $\Lifetime{T_i} = \bound{theorem:one-way connected temporal digraph contains P_k routing}{\Lifetime{}}{h, w}$, by~\cref{theorem:one-way connected temporal digraph contains P_k routing} each $T_i$ contains some $\Pk{w}$-routing $\varphi_i$ over some paths of $\mathcal{H}$.
		
		There are at most $\binom{h}{w} \cdot w!$ distinct $\Pk{w}$-routings $\varphi_i$.
		Hence, by the pigeon-hole principle, there is a subset $\mathcal{T} = \Set{T_{t_1}, T_{t_2}, \dots, T_{t_{\ell_1}}}$ of the temporal digraphs above of size $\ell_1$ such that $\varphi \coloneqq \varphi_i = \varphi_j$ for all $T_i, T_j \in \mathcal{T}$.
		
		Let $\Brace{u_1, u_2, \dots, u_{w}}$ be the vertices of $\Pk{w}$ sorted according to their order along $\Pk{w}$.
		For each $i \in \{1, \ldots, \ell_1\}$ let
		\begin{align*}
			S_i' & = \ToDigraph{\mathcal{H}^{t_{i}} \cup \mathcal{V}^{t_{i}}},\\[0em]
			A(S_i') &= \Set{ a_{i,j} \mid 1 \leq j \leq w \text{ and } a_{i,j} \text{ is the first vertex of $\varphi(u_j)$ on $V_{f'(t_{i})}$}} \text{ and }\\[0em]
			B(S_i') &= \Set{ b_{i,j} \mid 1 \leq j \leq w \text{ and } b_{i,j} \text{ is the last vertex of $\varphi(u_j)$ on $V_{f'(t_i + 1) - 1}$}}.
		\end{align*}
		Let $T'_i$ be the routing temporal digraph of $\mathcal{H}^i$ through $\mathcal{V}^i$.
		Since $T'_i$ is isomorphic to $T_i$, the bijection $\varphi$ induces a $\Pk{w}$-routing on $T_i'$ as well.
		By~\cref{lemma:P_k-routing-implies-1-order-linked}, each $A(S_i')$ is $1$-order-linked to $B(S_i')$ in $S_i'$.
		By choice of $b_{i,j}$ and $a_{i+1,j}$, the path $\varphi(u_j)$ contains a $b_{i,j}$-$a_{i+1,j}$-path.
		Hence, for each $1 \leq i \leq \ell_1$ there is a $B(S_{i}')$-$A(S_{i+1}')$-linkage $\mathcal{P}_i'$ such that $(\mathcal{S}' \coloneqq \Brace{ S_1', S_2', \dots, S_{\ell_1}' }, \mathscr{P}' \coloneqq \Brace{ \mathcal{P}_1', \mathcal{P}_2', \dots, \mathcal{P}_{\ell_1 - 1}' })$ is a uniform path of $1$-order-linked sets of width $w$ and length $\ell_1 - 1 = \ell w$.
		
		By~\cref{lemma:merging path of order-linked sets}, $\Brace{\mathcal{S}', \mathscr{P}'}$ contains as a subgraph a uniform path of $w$-order-linked sets $(\mathcal{S} = \Brace{ S_0, S_1, \dots, S_{\ell}}, \mathscr{P} = \Brace{ \mathcal{P}_0, \mathcal{P}_1, \dots, \mathcal{P}_{\ell - 1}})$ of length $\ell$ and width $w$.
		Additionally, for every $0 \leq i \leq \ell$ we have $S_i \subseteq \SubPOSS{(\mathcal{S}', \mathscr{P}')}{wi + 1}{w(i+1)}$, $A(S_i) \subseteq A(S_{wi + 1}')$ and $B(S_i) \subseteq B(S'_{w(i + 1)})$, and for $0 \leq i < \ell$ we have $\mathcal{P}_i \subseteq \mathcal{P}'_{(w-1)(i+1) + 1}$.
		
		By construction of each $S_i'$, we have that $A(S_i) \subseteq \V{V_{f(t_{wi + 1})}}$.
		Let $\mathcal{X} = \{V_{f(t_{wi + 1})} \mid 0 \leq i \leq \ell\}$.
		Define the bijection $\pi : \mathcal{S} \rightarrow \mathcal{X}$ as $\pi(S_i) = V_{f(t_{wi + 1})}$.
		Hence, $\mathcal{X}$ is a linkage of order $\ell + 1$ inside $\mathcal{V}$ such that $A(S_i) \subseteq \V{\pi(S_i)}$ for all $0 \leq i \leq \ell$.
		Furthermore, by construction of each $S_i$, it is immediate that $\V{\pi(S_i)} \cap \V{(\mathcal{S}, \mathscr{P})} \subseteq \V{S_i}$ for each $0 \leq i \leq \ell$.
		
		We construct the linkage $\mathcal{L}$ as follows.
		Let $\mathcal{Q}$ be the image of $\varphi$ and, for each $0 \leq i \leq \ell_2$, let $\mathcal{Q}^i \subseteq \mathcal{H}^i$ be the paths of $\mathcal{H}^i$ which are subpaths of $\mathcal{Q}$.
		
		Let $\mathcal{L}_1 \coloneqq \mathcal{Q}^0$, $\mathcal{L}_2 \coloneqq \mathcal{Q}^1 \cdot \mathcal{Q}^2 \cdot \ldots \cdot \mathcal{Q}^{\ell_2}$ and let $\mathcal{L}_3$ be the $B(S_\ell)$-$\End{\mathcal{Q}}$-linkage inside $\mathcal{Q}$.
		By construction, $\mathcal{L}_1 \cdot \mathcal{L}_2 \cdot \mathcal{L}_3$ is a $\Start{\mathcal{H}}$-$\End{\mathcal{H}}$-linkage of order $w$, $\mathcal{L}_2$ is an $A(S_0)$-$B(S_\ell)$-linkage and both $\mathcal{L}_1$ and $\mathcal{L}_3$ are internally disjoint from $(\mathcal{S}, \mathscr{P})$, as desired.
	\end{proof}
	
	Combining the previous~\namecref{lemma:ordered-web-to-path-of-order-linked-sets-with-forward-linkage} and~\cref{proposition:order-linked to path of well-linked sets} allows us to construct a path of well-linked sets from an ordered web.
  At a later part of our proof, we need additional information about how the linkage $\VVV$ intersects the individual clusters of the path of well-linked sets.
	This is captured by the bijection \(\pi\) in the statement of the following result.
	
	We define
	\begin{align*}
		\boundDefAlign{lemma:ordered-web-to-path-of-well-linked-sets-with-side-linkage}{h}{w, \ell}
		\bound{lemma:ordered-web-to-path-of-well-linked-sets-with-side-linkage}{h}{w, \ell} & {}\coloneqq 
		\bound{lemma:ordered-web-to-path-of-order-linked-sets-with-forward-linkage}{h}{w(\ell + 1)},\\[0em]
				\boundDefAlign{lemma:ordered-web-to-path-of-well-linked-sets-with-side-linkage}{v}{w, \ell}
		\bound{lemma:ordered-web-to-path-of-well-linked-sets-with-side-linkage}{v}{w, \ell} & {}\coloneqq 
		\bound{lemma:ordered-web-to-path-of-order-linked-sets-with-forward-linkage}{v}{w(\ell + 1), \ell}.
	\end{align*}
	Note that \(\bound{lemma:ordered-web-to-path-of-well-linked-sets-with-side-linkage}{h}{w,\ell} \in \Oh(w^{2} \ell^{2})\) and \(\bound{lemma:ordered-web-to-path-of-well-linked-sets-with-side-linkage}{v}{w,\ell} \in \PowerTower{1}{\Polynomial{25}{w, \ell}}\).
	
	        	\begin{corollary}
		\label{lemma:ordered-web-to-path-of-well-linked-sets-with-side-linkage}
		Let $\Brace{\mathcal{H},\mathcal{V}}$ be an ordered $(h,v)$-web such that $h \geq \bound{lemma:ordered-web-to-path-of-well-linked-sets-with-side-linkage}{h}{w, \ell}$ and $v \geq \bound{lemma:ordered-web-to-path-of-well-linked-sets-with-side-linkage}{v}{w, \ell}$.
		Then, there is a path of well-linked sets $\Brace{\mathcal{S} = \Brace{ S_0, S_1, \dots, S_{\ell} }, \mathscr{P}}$ of width $w$ and length $\ell$ in $\ToDigraph{\mathcal{H} \cup \mathcal{V}}$ such that $B(S_{\ell}) \subseteq \End{\mathcal{H}}$.
		Additionally, there is a linkage $\mathcal{X} \subseteq \mathcal{V}$ of order $\ell + 1$ and a bijection $\pi : \mathcal{S} \rightarrow \mathcal{X}$ such that $A(S_i) \subseteq \V{\pi(S_i)}$ and $\V{\pi(S_i)} \cap \V{(\mathcal{S}, \mathscr{P})} \subseteq \V{S_i}$ for each $0 \leq i \leq \ell$.
	\end{corollary}
	\begin{proof}
		By~\cref{lemma:ordered-web-to-path-of-order-linked-sets-with-forward-linkage}, $\Brace{\mathcal{H}, \mathcal{V}}$ contains a path of $w$-order-linked sets $(\mathcal{S}' = \Brace{ S'_0, S'_1, \dots, S'_{\ell}}, \mathscr{P}' = \Brace{ \mathcal{P}_0', \mathcal{P}_1', \dots, \mathcal{P}_{\ell - 1}'})$ of width $\ell(w + 1)$ and length $\ell$.
		Further, there is a linkage $\mathcal{X}' \subseteq \mathcal{V}$ of order $\ell + 1$ and a bijection $\pi' : \mathcal{S}' \rightarrow \mathcal{X}'$ such that $A(S_i') \subseteq \V{\pi'(S_i')}$ and $\V{\pi'(S_i')} \cap \V{(\mathcal{S}', \mathscr{P}')} \subseteq \V{S_i'}$ for each $0 \leq i \leq \ell$.
		
		By~\cref{proposition:order-linked to path of well-linked sets}, the path of $w$-order-linked sets $(\mathcal{S}', \mathscr{P}')$ contains a path of well-linked sets $(\mathcal{S} = (S_0, S_1, \dots, S_{\ell}), \allowbreak \mathscr{P} = ( \mathcal{P}_0, \mathcal{P}_1, \dots, \mathcal{P}_{\ell - 1}))$ of length $\ell$ and width $w$.
		Additionally, for each $0 \leq i \leq \ell$ we have that $S_i \subseteq S_i'$ and that $A(S_i) \subseteq A(S_i')$.
		
		Finally, let $\mathcal{X} = \Set{\pi'(S_{i}') \mid 0 \leq i \leq \ell}$ and let $\pi : \mathcal{S} \to \mathcal{X}$ be the bijection given by $\pi(S_i) = \pi'(S_i')$.
		It is immediate that $\mathcal{X}$ and $\pi$ satisfy the desired conditions.
	\end{proof}
	
We can manipulate the path of well-linked sets given by~\cref{lemma:ordered-web-to-path-of-well-linked-sets-with-side-linkage} above to ensure that the extremities of the path of well-linked sets are contained in the extremities of \(\mathcal{H}\).
This is useful when we need the end of the path of well-linked sets to be well-linked to its beginning.

We define
\begin{align*}
	\boundDefAlign{state:ordered-web-to-well-linked-pows}{h}{w, \ell}
	\bound{state:ordered-web-to-well-linked-pows}{h}{w, \ell} & =
		\bound{lemma:ordered-web-to-path-of-well-linked-sets-with-side-linkage}{h}{w, \ell},
	\\[0em]
	\boundDefAlign{state:ordered-web-to-well-linked-pows}{v}{w, \ell}
	\bound{state:ordered-web-to-well-linked-pows}{v}{w, \ell} & =
		\bound{lemma:ordered-web-to-path-of-well-linked-sets-with-side-linkage}{v}{w, \ell + 4w}.
\end{align*}
Note that \(\bound{state:ordered-web-to-well-linked-pows}{h}{w,\ell} \in \Oh(w^{2} \ell^{2} )\) and
\(\bound{state:ordered-web-to-well-linked-pows}{v}{w,\ell} \in \PowerTower{1}{\Polynomial{25}{w, \ell}}\).
\begin{lemma}
	\label{state:ordered-web-to-well-linked-pows}
	Let $\Brace{\mathcal{H}, \mathcal{V}}$ be an ordered $(h,v)$-web.
	If $h \geq \bound{state:ordered-web-to-well-linked-pows}{h}{w, \ell}$ and $v \geq \bound{state:ordered-web-to-well-linked-pows}{v}{w, \ell}$, then $\Brace{\mathcal{H}, \mathcal{V}}$ contains a path of well-linked sets $\Brace{\mathcal{S} = \Brace{ S_0, S_1, \dots, S_{\ell}}, \mathscr{P}}$ of width $w$ and length $\ell$.
	Additionally, \(A(S_0) \subseteq \Start{\mathcal{H}}\) and \(B(S_\ell) \subseteq \End{\mathcal{H}}\).
\end{lemma}
\begin{proof}
	Let \(\ell_1 = \ell + 4w\).
	By~\cref{lemma:ordered-web-to-path-of-well-linked-sets-with-side-linkage}, $\ToDigraph{\Brace{\mathcal{H}, \mathcal{V}}}$ contains a path of well-linked sets $(\mathcal{S}' = \Brace{ S'_0, S'_1, \dots, S'_{\ell_1}},\allowbreak \mathscr{P}' = \left( \mathcal{P}'_0, \mathcal{P}'_1, \dots, \mathcal{P}'_{\ell_1 - 1}\right))$ of width $w$ and length $\ell_1$.
	Additionally, there is a linkage $\mathcal{X} \subseteq \mathcal{V}$ of order $\ell_1 + 1$ and a bijection $\pi : \mathcal{S}' \rightarrow \mathcal{X}$ such that $A(S_i') \subseteq \V{\pi(S_i')}$ and $\V{\pi(S_i')} \cap \V{(\mathcal{S}', \mathscr{P}')} \subseteq \V{S_i'}$ for each $0 \leq i \leq \ell_1$.

	We construct a path of well-linked sets $(\mathcal{S}, \mathscr{P})$ of width $w$ and length $\ell$ as follows.
	For each $0 \leq i \leq \ell_1$, let $X_i = \pi(S_i')$, let $a_i$ be the first intersection of $X_i$ with $\V{S_i'}$ and let $b_i$ be the last intersection of $X_i$ with $\V{S_i'}$.
	Let $A' = \Set{a_0, a_2, \ldots, a_{2(w - 1)}}$ and $B' = \Set{b_{\ell_1}, b_{\ell_1 - 2}, \ldots, b_{\ell_1 - 2(w - 1)}}$, let $\mathcal{X}_A$ be the $\Start{\mathcal{X}}$-$A'$-linkage of order $w$ inside $\mathcal{X}$ and let $\mathcal{X}_B$ be the $B'$-$\End{\mathcal{X}}$-linkage of order $w$ inside $\mathcal{X}$.

	By~\cref{lem:linkage_inside_poss}\ref{case:linkage_inside_poss_scattered_A}, there is an $A'$-$A(S_{2w}')$-linkage $\mathcal{L}_A$ of order $w$ in $\SubPOSS{(\mathcal{S}', \mathscr{P}')}{0}{2w}$.
    And analogously, there is a $B(S_{\ell_1 - 2w}')$-$B'$-linkage $\mathcal{L}_B$ of order $w$ in $\SubPOSS{(\mathcal{S}', \mathscr{P}')}{\ell_1 - 2w}{\ell_1}$ by \cref{lem:linkage_inside_poss}\ref{case:linkage_inside_poss_B_scattered}.

	Since $\V{X_i} \cap \V{(\mathcal{S}', \mathscr{P}')} \subseteq \V{S_i'}$ for all $0 \leq i \leq \ell_1$, we have that $\mathcal{Y}_A \coloneqq \mathcal{X}_A \cdot \mathcal{L}_A$ is a $\Start{\mathcal{V}}$-$A(S_{2w}')$-linkage of order $w$ and $\mathcal{Y}_B \coloneqq \mathcal{L}_B \cdot \mathcal{X}_B$ is a $B(S_{\ell_1 - 2w}')$-$\End{\mathcal{V}}$-linkage of order $w$.

	Let $S_0 = \ToDigraph{S_{2w}' \cup \mathcal{Y}_A}$, $A(S_0) = \Start{\mathcal{Y}_A}$, $B(S_0) = B(S_{2w}')$, $S_{\ell} = \ToDigraph{S_{2w + \ell}' \cup \mathcal{Y}_B}$, $A(S_{\ell}) = A(S_{2w + \ell}')$ and $B(S_\ell) = \End{\mathcal{Y}_B}$.
	Let $\mathcal{S} = \Brace{S_0, S_{2w + 1}', S_{2w + 2}' \dots, S_{2w + \ell - 1}', S_{\ell}}$ and $\mathscr{P} = (\mathcal{P}_{2w},\allowbreak \mathcal{P}_{2w + 1}, \dots, \mathcal{P}_{2w + \ell - 1})$.
	Clearly, $(\mathcal{S}, \mathscr{P})$ is a path of well-linked sets of width $w$ and length $\ell$.
	Finally, we have $A(S_0) \subseteq \Start{\mathcal{L}_1} \subseteq \Start{\mathcal{V}}$ and $B(S_{\ell}) \subseteq \End{\mathcal{L}_3} \subseteq\End{\mathcal{V}}$.
\end{proof}

	The construction from~\cref{lemma:ordered-web-to-path-of-well-linked-sets-with-side-linkage} above does not guarantee that the paths in $\HHH$ intersect many clusters of the resulting path of well-linked sets.
  The reason is that the layers of the routing temporal digraph constructed from an ordered web are only unilateral and not strongly connected.
  Therefore, there is no guarantee that $\Start{\HHH}$ and $\End{\HHH}$ are well-linked.
	While this does not affect our results,
	this property is certainly useful in other contexts and we show how to obtain it in the next subsection.
	Moreover, in \cite{COSSII} we require this stronger property in some of our steps.

  \subsection{Preserving the horizontal paths in a folded ordered web}
		
	The issue mentioned at the end of the last section is unavoidable
	when working with ordered webs, as the web might be an acyclic grid,
	where it is no longer possible to ``come back'' to a row once we leave it.
	To be able to preserve the paths in the linkage \(\mathcal{H}\),
	we need each path in \(\mathcal{V}\) to provide us a possibility of
	revisiting a path in \(\mathcal{H}\) which we might have intersected earlier in the web.
	This is achieved by the following definition of \emph{folded web}
	(see~\cref{fig:folded-web} for an illustration of the concept).
	
	\begin{definition}
		\label{def:folded-web}
		An $(h,v)$-web $\Brace{\mathcal{H}, \mathcal{V}}$ is a \emph{folded web} if every $V_i \in \mathcal{V}$ can be split as $V_i^a \cdot V_i^b \coloneqq V_i$ such that both $V_i^a$ and $V_i^b$ intersect all paths of $\mathcal{H}$.
  \end{definition}
	
	\begin{figure}[!ht]
	    \centering
	    \includegraphics{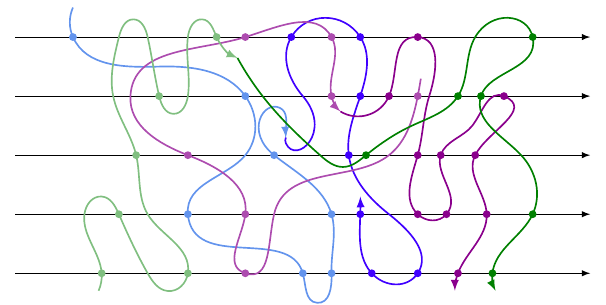}
	    \caption{A folded $\Brace{5,3}$-web.}
	    \label{fig:folded-web}
	\end{figure}

	Folded ordered webs correspond to splits from~\cref{def:split-edge}\ref{def:split}.
    The example shown in~\cref{fig:split_to_folded_order_web} illustrates the connection between splits and folded ordered webs, which we make precise in the following observation.
	
 	 	   \begin{figure}[!ht]
      \centering
      \includegraphics[scale=1.2]{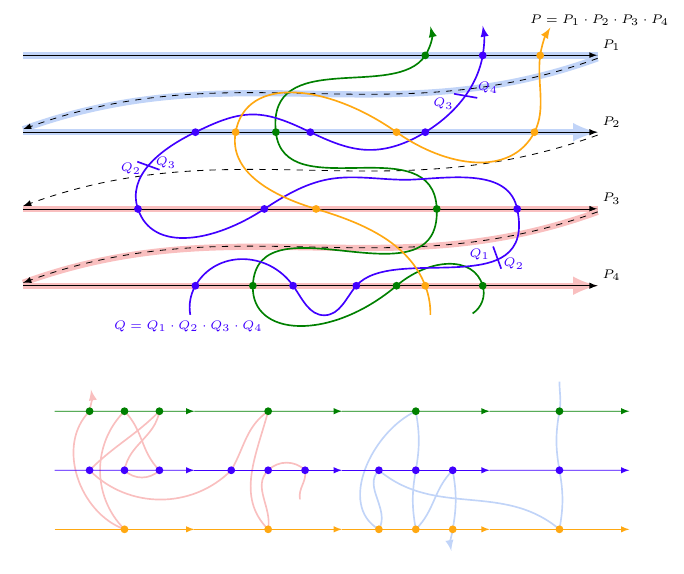}
      \caption{An example of how a $(3,2)$ folded ordered web is obtained from a $(4,3)$-split.}
      \label{fig:split_to_folded_order_web}
  \end{figure}
	
	\begin{observation}
		\label{obs:split-to-folded-web}
		Let $\Brace{\mathcal{P}', \mathcal{Q}'}$ be a $\Brace{2p, q}$-split of $\Brace{\mathcal{P}, \mathcal{Q}}$.
		Then there is some $\mathcal{P}''$ containing only subpaths of $\mathcal{P}$ such that $\Brace{\mathcal{Q}', \mathcal{P}''}$ is a folded ordered $(q, p)$-web.
	\end{observation}
	\begin{proof}
		Let $\Brace{ P_1, P_2, \dots, P_{2p}} \coloneqq \mathcal{P'}$ be and ordering of $\mathcal{P}'$ witnessing that $\Brace{\mathcal{P}', \mathcal{Q}'}$ is a $\Brace{2p, q}$-split.
		For each $1 \leq i \leq p$ let $P_i' = P_{2i - 1} \cdot e_{2i - 1} \cdot P_{2i}$, where $e_{2i - 1}$ is the arc inside $\mathcal{P}'$ such that $P_{2i - 1} \cdot e_{2i - 1} \cdot P_{2i}$ is a subpath of $\mathcal{P}'$.
		Let $\mathcal{P}'' = \Brace{P'_1, P'_2, \dots, P'_{p}}$.
		Now $\Brace{\mathcal{Q}', \mathcal{P}''}$ is a folded ordered $(q, p)$-web, which can be seen by partitioning $P_i'$ into $P_{2i-1}$ and $P_{2i}$.
	\end{proof}
	
		Next, we show that if we construct a path of well-linked sets starting from an ordered web that is also folded, then we can construct the path of well-linked sets in a way that the paths in $\HHH$ are guaranteed to intersect the individual clusters of the resulting path of well-linked sets.
    The idea of the construction is similar to the proof of~\cref{lemma:ordered-web-to-path-of-order-linked-sets-with-forward-linkage}, but now we can use~\cref{theorem:strongly connected temporal digraph contains H routing} in the construction, which yields the extra properties we need.
	
	We define 
	\begin{align*}
		\boundDefAlign{lemma:folded-web-to-pows}{h}{w}
		\bound{lemma:folded-web-to-pows}{h}{w} & \coloneqq
		\bound{lemma:routing-implies-well-linked}{\ell}{w},
		\\[0em]
		\boundDefAlign{lemma:folded-web-to-pows}{v}{w, \ell}
		\bound{lemma:folded-web-to-pows}{v}{w, \ell} & \coloneqq
		\bound{lemma:routing-implies-well-linked}{h}{w} \Big(\ell \binom{\bound{lemma:folded-web-to-pows}{h}{w}}{w} + 1 \Big).
	\end{align*}
	We observe that \(\bound{lemma:folded-web-to-pows}{h}{w} \in \Oh(w^{11})\) and 
	\(\bound{lemma:folded-web-to-pows}{v}{w,\ell} \in \PowerTower{1}{\Polynomial{2}{w, \ell}}\).
	
	\begin{lemma}
		\label{lemma:folded-web-to-pows}
		Let $\Brace{\mathcal{H}, \mathcal{V}}$ be a folded ordered $(h,v)$-web.
		If $h \geq \bound{lemma:folded-web-to-pows}{h}{w}$ and $v \geq \bound{lemma:folded-web-to-pows}{v}{w, \ell}$, then $\Brace{\mathcal{H}, \mathcal{V}}$ contains a path of well-linked sets $\Brace{\mathcal{S} = \Brace{ S_0, S_1, \dots, S_{\ell}}, \mathscr{P}}$ of width $w$ and length $\ell$.
		Additionally, there is a $\Start{\mathcal{H}}$-$\End{\mathcal{H}}$-linkage $\mathcal{L} = \mathcal{L}_1 \cdot \mathcal{L}_2 \cdot \mathcal{L}_3$ using only arcs of $\mathcal{H}$ such that $\mathcal{L}_2$ is an $A(S_0)$-$B(S_\ell)$-linkage of order $w$ inside $\Brace{\mathcal{S}, \mathscr{P}}$ and $\mathcal{L}_1$ and $\mathcal{L}_3$ are internally disjoint from $\Brace{\mathcal{S}, \mathscr{P}}$.
	\end{lemma}
	\begin{proof}
		Assume, without loss of generality, that $h = \bound{lemma:folded-web-to-pows}{h}{w}$, as any $\mathcal{H}' \subseteq \mathcal{H}$ of size $\bound{lemma:folded-web-to-pows}{h}{w}$ also satisfies the assumptions of the statement above.
		
		Let $\ell_1 = \bound{lemma:routing-implies-well-linked}{h}{w}$, and $\ell_2 = \ell \binom{h}{w} + 1$.
		Let $\Brace{V_0, V_1, \dots, V_{v - 1}}$ be an ordering of $\mathcal{V}$ witnessing that $\Brace{\mathcal{H}, \mathcal{V}}$ is a folded ordered web.
		
		For each $1 \leq i \leq \ell_2$ let $\mathcal{H}^i$ be the maximal linkage inside $\mathcal{H}$ such that $\Start{\mathcal{H}^i} \subseteq \V{V_{(i-1) \ell_1}}$ and $\End{\mathcal{H}^i} \subseteq \V{V_{i \ell_1 - 1}}$.
		Additionally, let $T_i$ be the routing temporal digraph of $\mathcal{H}^i$ through $\mathcal{V}^i \coloneqq \Brace{V_{(i - 1) \ell_1}, \ldots, V_{i \ell_1 - 1}}$.
		Because $\Brace{\mathcal{H}, \mathcal{V}}$ is a folded web, for every $0 \leq j \leq \ell_1 - 1$ and every pair of paths $H_a^i, H_b^i \in \mathcal{H}^i$ there is a  subpath of $V_{(i - 1) \ell_1 + j}$ from $\V{H_a^i}$ to $\V{H_b^i}$.
		Hence, each layer of $T_i$ is strongly connected.
		
		By construction, $\Lifetime{T_i} = \ell_1$ and by assumption $\Abs{\V{T_i}} = \Abs{\mathcal{H}} = \bound{lemma:folded-web-to-pows}{h}{w}$.
		By~\cref{lemma:routing-implies-well-linked}, for every $1 \leq i \leq \ell_2$ there is some $\mathcal{L}_i \subseteq \mathcal{H}_i$ of order $w$ such that $\Start{\mathcal{L}_i}$ is well-linked to $\End{\mathcal{L}_i}$ inside $\ToDigraph{\mathcal{H}^i \cup \mathcal{V}^i}$.
		
		By the pigeon-hole principle, there is some $\mathcal{H}' \subseteq \mathcal{H}$ of order $w$ and some $\mathcal{I} \subseteq \Set{1, \ldots, \ell_2}$ of order $\ell + 1$ such that $\mathcal{L}_i$ is a sublinkage of $\mathcal{H}'$ of order $w$ for all $i \in \mathcal{I}$.
		Let $\Brace{t_0, t_1, \ldots, t_{\ell}} \coloneqq \mathcal{I}$ be the ascending order of the elements of $\mathcal{I}$.
		
		For each $0 \leq i \leq \ell$ let $S_i \coloneqq \ToDigraph{\mathcal{L}_{t_i} \cup \mathcal{V}_{t_i}}$ and set $A(S_i) = \Start{\mathcal{L}_{t_i}}$ and $B(S_i) = \End{\mathcal{L}_{t_{i}}}$.
		For each $1 \leq i \leq \ell - 1$ let $\mathcal{P}_i$ be the $\End{\mathcal{L}_{t_i}}$-$\Start{\mathcal{L}_{t_{i + 1}}}$-linkage inside $\mathcal{H}$.
		
		By construction, $A(S_i)$ is well-linked to $B(S_i)$ inside $S_i$ for all $i$.
		This implies that $(\mathcal{S} = \{ S_0, S_1$, $\dots, S_{\ell}\}, \mathscr{P} = \Brace{ \mathcal{P}_0, \mathcal{P}_1, \dots, \mathcal{P}_{\ell - 1}})$ is a path of well-linked sets of width $w$ and length $\ell$.
		Further, $\mathcal{L}$ is an $A(S_0)$-$B(S_\ell)$-linkage of order $w$ inside $\Brace{\mathcal{S}, \mathscr{P}}$ using only arcs of~$\mathcal{H}$.
	\end{proof}

In a way similar to~\cref{state:ordered-web-to-well-linked-pows} above, we can manipulate the path of well-linked sets obtained from~\cref{lemma:folded-web-to-pows} in order to ensure that the extremities of the path of well-linked sets are subsets of the extremities of \(\mathcal{H}\).

\begin{corollary}
	\label{state:folded-ordered-web-to-well-linked-pows}
	Let $\Brace{\mathcal{H}, \mathcal{V}}$ be a folded ordered $(h,v)$-web.
	If $h \geq \bound{lemma:folded-web-to-pows}{h}{w}$ and $v \geq \bound{lemma:folded-web-to-pows}{v}{w, \ell}$, then $\Brace{\mathcal{H}, \mathcal{V}}$ contains a path of well-linked sets $\Brace{\mathcal{S} = \Brace{ S_0, S_1, \dots, S_{\ell}}, \mathscr{P}}$ of width $w$ and length $\ell$.
	Additionally, \(A(S_0) \subseteq \Start{\mathcal{H}}\) and \(B(S_\ell) \subseteq \End{\mathcal{H}}\).
\end{corollary}
\begin{proof}
	By~\cref{lemma:folded-web-to-pows}, $\Brace{\mathcal{V}', \mathcal{H}''}$ contains a path of well-linked sets $\Brace{\mathcal{S} = \Brace{ S_0, S_1, \dots, S_{\ell}}, \mathscr{P}}$ of width $w$ and length $\ell$.
	Additionally, there is a $\Start{\mathcal{V}'}$-$\End{\mathcal{V}'}$-linkage $\mathcal{L} = \mathcal{L}_1 \cdot \mathcal{L}_2 \cdot \mathcal{L}_3$ of order $w$ such that $\mathcal{L}_2$ is an $A(S_0)$-$B(S_\ell)$-linkage of order $w$ inside $\Brace{\mathcal{S}, \mathscr{P}}$ and $\mathcal{L}_1$ and $\mathcal{L}_3$ are internally disjoint from $\Brace{\mathcal{S}, \mathscr{P}}$.

	Set $S_0' = \ToDigraph{S_0 \cup \mathcal{L}_1}$, $A(S_0') = \Start{\mathcal{L}_1}$, $B(S_0') = B(S_0)$, $S_{\ell}' = \ToDigraph{S_\ell \cup \mathcal{L}_3}$, $A(S_{\ell}') = A(S_\ell)$ and $B(S_{\ell}') = \End{\mathcal{L}_3}$.
	Because $\End{\mathcal{L}_1} = A(S_0)$ and $\Start{\mathcal{L}_3} = B(S_{\ell})$, we have that $\Brace{\mathcal{S}' \coloneqq \Brace{S_0', S_1, S_2, \dots, S_{\ell}'}, \mathscr{P}}$ is a path of well-linked sets of width $w$ and length $\ell$ such that $A(S_0') \subseteq \Start{\mathcal{V}}$ and $B(S_0') \subseteq \End{\mathcal{V}}$.
\end{proof}

\subsection{Back to finding cycles}

We conclude this section by showing that digraphs of large directed \treewidth contain a large path of well-linked sets where the last cluster is well-linked to the first.
Here, we use the well-linkedness property to construct our set of pairwise vertex-disjoint cycles.
Later, in~\cite{COSSII}, we use this well-linkedness property to construct the linkage required to \emph{close} the path of well-linked sets into a \emph{cycle of well-linked sets}.

Define 
\begin{align*}
	\Func{v'}{w, \ell} & = \bound{lemma:folded-web-to-pows}{h}{w} + \bound{state:ordered-web-to-well-linked-pows}{v}{w, \ell},
	\\[0em]
	\boundDefAlign{thm:high_dtw_to_POSS_plus_back-linkage}{t}{w, \ell}
	\bound{thm:high_dtw_to_POSS_plus_back-linkage}{t}{w, \ell}
		& = \bound{state:high-dtw-to-segmentation-or-split}{t}{2 \bound{lemma:folded-web-to-pows}{v}{w, \ell}, \bound{lemma:ordered x-segmentation}{p}{\bound{state:ordered-web-to-well-linked-pows}{h}{w, \ell}, \Func{v'}{w, \ell}}, \Func{v'}{w, \ell},  (\ell + 1)w}.	
\end{align*}
Note that \(\bound{thm:high_dtw_to_POSS_plus_back-linkage}{t}{w,\ell} \in \PowerTower{7}{\Polynomial{25}{w, \ell}}\).
\begin{theorem}
    \label{thm:high_dtw_to_POSS_plus_back-linkage}
		Every digraph $D$ with $\dtw{D} \geq \bound{thm:high_dtw_to_POSS_plus_back-linkage}{t}{w, \ell}$ contains a path of well-linked sets $\Brace{\mathcal{S} = \Brace{ S_0, S_1, \dots, S_{\ell}},\mathscr{P}}$ of width $w$ and length $\ell$ such that $B(S_\ell)$ is well-linked to $A(S_0)$ in $D$.
\end{theorem}
\begin{proof}	
	We define
	\(\ell_1 = \ell + 1\),
	\(h_3 = \bound{lemma:folded-web-to-pows}{v}{w, \ell}\),
	\(h_2 = 2 h_3\),
	\(v_2 = \bound{lemma:folded-web-to-pows}{h}{w} + \bound{state:ordered-web-to-well-linked-pows}{v}{w, \ell}\),
	\(h_5 = \bound{state:ordered-web-to-well-linked-pows}{h}{w, \ell}\),
	\(h_4 = \bound{lemma:ordered x-segmentation}{p}{h_5, v_2}\).
	Observe that \(\bound{thm:high_dtw_to_POSS_plus_back-linkage}{t}{w, \ell} = \bound{state:high-dtw-to-segmentation-or-split}{t}{h_2, h_4, v_2, w \ell_1}\).

	By~\cref{state:high-dtw-to-segmentation-or-split}, we obtain three cases.

	If~\cref{state:high-dtw-to-segmentation-or-split}\ref{item:high-dtw-to-web:grid} holds, then \(D\) contains a cylindrical grid of order \(w \ell_1 \), which in turn contains a path of well-linked sets \(\Brace{\mathcal{S}^1 = \Brace{S^1_{0}, S^1_{1}, \ldots, S^1_{\ell_1}}, \mathscr{P}^1 = \Brace{\mathcal{P}^1_{0}, \mathcal{P}^1_{1}, \ldots, \mathcal{P}^1_{\ell_1 - 1}}}\) of width \(w\) and length \(\ell_1\) together with a \(B(S_{\ell_1})\)-\(A(S_0))\)-linkage \(\mathcal{P}^1_{\ell_1}\) which is internally disjoint from \((\mathcal{S}^1, \mathscr{P}^1)\).
	
	Let \(S_0 = \ToDigraph{\mathcal{P}^1_{\ell_1} \cup S^1_0}\) and \(S_\ell = \ToDigraph{S^1_\ell \cup \mathcal{P}^1_\ell}\).
	Set \(A(S_0) = \Start{\mathcal{P}^1_{\ell_1}}\), \(B(S_0) = B(S^1_0)\), \(A(S_\ell) = A(S^1_\ell\) and \(B(S_\ell) = \End{\mathcal{P}^1_\ell}\).
	For each \(1 \leq i \leq \ell - 1\), set \(A(S_i) = A(S^1_i)\) and \(B(S_i) = B(S^1_i)\).
	It is immediate that \(\Brace{\mathcal{S} \coloneqq \Brace{S_{0}, S^1_{1}, \ldots, S^1_{\ell-1}, S_{\ell}}, \mathscr{P} \coloneqq \Brace{\mathcal{P}^1_{0}, \mathcal{P}^1_{1}, \ldots, \mathcal{P}^1_{\ell}}}\) is a path of well-linked sets of width \(w\) and length \(\ell\).
	Further, as \(B(S_\ell) \subseteq A(S^1_{\ell_1})\) and \(A(S_0) \subseteq B(S^1_{\ell_1})\), we have that \(B(S_\ell)\) is well-linked to \(A(S_0)\), as desired.

	If~\cref{state:high-dtw-to-segmentation-or-split}\ref{item:high-dtw-to-web:split} holds, then \(D\) contains a \(\Brace{h_2, v_2}\)-split \(\Brace{\mathcal{H}_2, \mathcal{V}_2}\) where \(\End{\mathcal{V}_2}\) is well-linked to \(\Start{\mathcal{V}_2}\).
	By~\cref{obs:split-to-folded-web}, there is some \(\mathcal{H}_3 \subseteq \mathcal{H}_2\) of order \(h_3\) such that \(\Brace{ \mathcal{V}_2, \mathcal{H}_3}\) is a folded ordered \(\Brace{v_2, h_3}\)-web.
	Applying~\cref{state:folded-ordered-web-to-well-linked-pows} to \(\Brace{ \mathcal{V}_2, \mathcal{H}_3}\) yields a path of well-linked sets \(\Brace{\mathcal{S} = \Brace{S_{0}, S_{1}, \ldots, S_{\ell}}, \mathscr{P}}\) of width \(w\) and length \(\ell\) such that \(A(S_0) \subseteq \Start{\mathcal{V}_2}\) and \(B(S_\ell) \subseteq \End{\mathcal{V}_2}\).
	As \(\End{\mathcal{V}_2}\) is well-linked to \(\Start{\mathcal{V}_2}\), we have that \(A(S_0)\) is well-linked to \(B(S_\ell)\), as desired.

	Finally, if~\cref{state:high-dtw-to-segmentation-or-split}\ref{item:high-dtw-to-web:segmentation} holds, then \(D\) contains an \(\Brace{h_4, v_2}\)-segmentation \(\Brace{\mathcal{H}_4, \mathcal{V}_4}\) where \(\End{\mathcal{H}_4}\) is well-linked to \(\Start{\mathcal{H}_4}\).
	By~\cref{lemma:ordered x-segmentation}, there is some \(\mathcal{H}_5 \subseteq \mathcal{H}_4\) of order \(h_5\) such that \(\Brace{\mathcal{H}_5, \mathcal{V}_4}\) is an ordered segmentation.
	By definition, \(\Brace{\mathcal{H}_5, \mathcal{V}_4}\) is an ordered web.
	By~\cref{state:ordered-web-to-well-linked-pows}, \(\ToDigraph{\Brace{\mathcal{H}_5, \mathcal{V}_4}}\) contains a path of well-linked sets \(\Brace{\mathcal{S} = \Brace{S_{0}, S_{1}, \ldots, S_{\ell}}, \mathscr{P}}\) of width \(w\) and length \(\ell\) such that \(A(S_0) \subseteq \Start{\mathcal{H}_5}\) and \(B(S_\ell) \subseteq \End{\mathcal{H}_5}\).
	As \(\End{\mathcal{H}_5}\) is well-linked to \(\Start{\mathcal{H}_5}\), we have that \(A(S_0)\) is well-linked to \(B(S_\ell)\), as desired.
\end{proof}

In order to conclude the proof of~\cref{statement:elementary-younger}, we need Statement (3.1) from \cite{reed1996packing}, which requires the following functions:
\begin{align*}
	\Func{n}{k}
	& \coloneqq
	2k^2 - 3k + 2,
	\\[0em]
	\bound{statement:fence-plus-back-linkage-implies-cycles}{r}{k}
	& \coloneqq
	8(\bound{statement:fence-plus-back-linkage-implies-cycles}{r}{\ceil{k/2}} + \Func{n}{k}),
	\\[0em]
	\bound{statement:fence-plus-back-linkage-implies-cycles}{q}{k}
	& \coloneqq
	16k(\bound{statement:fence-plus-back-linkage-implies-cycles}{r}{\ceil{k/2}} + \Func{n}{k}) + 2k + 2\Func{q'}{k},
	\\[0em]
	\Func{q'}{k}
	& \coloneqq
	\bound{statement:fence-plus-back-linkage-implies-cycles}{q}{\ceil{k/2}} + 
	2 \Func{n}{k} + 1.
\end{align*}

\fenceToCycles*

Because \cite{reed1996packing} does not give explicit upper bounds on the functions \(\bound{statement:fence-plus-back-linkage-implies-cycles}{r}{}\) and \(\bound{statement:fence-plus-back-linkage-implies-cycles}{q}{}\) above, we first need to determine these bounds.
To this end, we apply a classical result known as \emph{Master Theorem}, stated below.

\begin{theorem}[{Master Theorem~\cite{CLRS2009ItAc4p94,BHS1980gmsdcr}}]
	\label{statement:master-theorem}
	Let \(a \geq 1\) and \(b \geq 1\), let \(f\) be a function and let \(T(n)\) be defined on non-negative integers by the recurrence \(T(n)  = a T(n / b) + f(n)\), where we interpret \(n / b\) to mean either \(\floor{n / b}\) or \(\ceil{n/b}\).
	Then \(T(n)\) has the following asymptotic bounds.
	\begin{enamerate}{M}{item:master-theorem:last}
	\item
		\label{item:master-theorem:small-f}
		If \(f(n) \in O(n^{\log_b a - \epsilon})\) for some constant \(\epsilon > 0\), then \(T(n) = \Theta(n^{\log_b a})\).
	\item
		\label{item:master-theorem:average-f}
		If \(f(n) \in \Theta(n^{\log_b a})\), then \(T(n) \in \Theta(n^{\log_b a} \log n)\).
	\item
		\label{item:master-theorem:large-f}
		If \(f(n) \in \Omega(n^{\log_b a + \epsilon})\) for some constant \(\epsilon > 0\), and if \(a f(n/b) \leq c f(n)\) for some constant \(c < 1\) and all sufficiently large \(n\), then \(T(n) \in \Theta(f(n))\).
		\label{item:master-theorem:last}
	\end{enamerate}
\end{theorem}

It is fairly straightforward to apply~\cref{statement:master-theorem} in order to obtain polynomial bounds for \(\bound{statement:fence-plus-back-linkage-implies-cycles}{r}{}\) and \(\bound{statement:fence-plus-back-linkage-implies-cycles}{q}{}\).

\begin{observation}
	\label{statement:bounds-for-fence-plus-back-linkage-implies-cycles}
	\(\bound{statement:fence-plus-back-linkage-implies-cycles}{r}{k} \in \Theta(k^3)\) and
	\(\bound{statement:fence-plus-back-linkage-implies-cycles}{q}{k} \in \Theta(k^3)\).
\end{observation}
\begin{proof}
	We apply~\cref{statement:master-theorem}.

	We rewrite \(\bound{statement:fence-plus-back-linkage-implies-cycles}{r}{k}\) in the form \(T_1(k) = a_1 T_1(k/b_1) + f_1(k)\) by defining \(a_1 = 8 = 2^3\), \(b_1 = 2\) and \(f_1(k) = 8 \Func{n}{k} \in O(k^2)\).
	From~\cref{statement:master-theorem}\ref{item:master-theorem:small-f} we thus obtain that \(\bound{statement:fence-plus-back-linkage-implies-cycles}{r}{k} \in \Theta(k^3)\).

	Also, we rewrite \(\bound{statement:fence-plus-back-linkage-implies-cycles}{q}{k}\) in the form \(T_2(k) = a_2 T_2(k / b_2) + f_2(k)\) by defining \(a_2 = 2\), \(b_2 = 2\) and \(f_2(k) = 2k^3 + 10k^2 - 13k + 12 \in O(k^3)\).
	In order to apply~\cref{statement:master-theorem}\ref{item:master-theorem:large-f}, we first show that \(a_2 f_2(k / b_2) < c f_2(k)\) holds for some \(c < 1\) and for all sufficiently large \(k\).

	\begin{align*}
		&& a_2 f_2(k / b_2) & {}< c f_2(k)
		\\[0em]
		\Leftrightarrow && 
		2 f_2(k / 2)  - c f_2(k) & {}< 0
		\\[0em]
		\Leftrightarrow &&
		k^3/2 + 5k^2 - 13 k + 24 - c(2k^3 + 10k^2 - 13k + 12) & {}< 0 
		\\[0em]
		\Leftrightarrow &&
		k^3/2 + 5k^2 - 13 k + 24 -c2k^3 - c10k^2 + c13k - c12 & {}< 0
		\\[0em]
		\Leftrightarrow &&
		(1/2 - 2c)k^3 + (5 - 10c)k^2 + (13c - 13)k - 12c + 24 & {}< 0
	\end{align*}

	The inequality above holds, for example, for \(c = 1/2\) and all \(k \geq 3\).
	By applying~\cref{statement:master-theorem}\ref{item:master-theorem:large-f}, we conclude that \(\bound{statement:fence-plus-back-linkage-implies-cycles}{q}{k} \in \Theta(k^3)\).
\end{proof}

We split the proof of~\cref{statement:elementary-younger} into two cases.
The case where the given digraph has a large directed \treewidth can be handled using the statements proven so far.

We define
\begin{align*}
		\bound{statement:high-dtw-implies-cycles}{t}{k}
		& \coloneqq
		\bound{thm:high_dtw_to_POSS_plus_back-linkage}{t}{\bound{thm:poss-to-fence}{w}{			\bound{statement:fence-plus-back-linkage-implies-cycles}{r}{k},			\bound{statement:fence-plus-back-linkage-implies-cycles}{q}{k}}			, \bound{thm:poss-to-fence}{\ell}{			\bound{statement:fence-plus-back-linkage-implies-cycles}{r}{k},			\bound{statement:fence-plus-back-linkage-implies-cycles}{q}{k}}}.
\end{align*}
Observe that \(\bound{statement:high-dtw-implies-cycles}{t}{k} \in \PowerTower{8}{\Polynomial{43}{k}}\).

\begin{lemma}
	\label{statement:high-dtw-implies-cycles}
	Let \(D\) be a digraph with \(\DTreewidth{D} \geq \bound{statement:high-dtw-implies-cycles}{t}{k}\).
	Then, \(D\) contains \(k\) pairwise vertex-disjoint cycles.
\end{lemma}
\begin{proof}
	Let 
	\(q      = \bound{statement:fence-plus-back-linkage-implies-cycles}{q}{k}\),
	\(r      = \bound{statement:fence-plus-back-linkage-implies-cycles}{r}{k}\),
	\(w_1    = \bound{thm:poss-to-fence}{w}{r, q}\),
	\(\ell_1 = \bound{thm:poss-to-fence}{\ell}{r, q}\).
	Note that \(\bound{statement:high-dtw-implies-cycles}{t}{k} \geq \bound{thm:high_dtw_to_POSS_plus_back-linkage}{t}{w_1, \ell_1}\).

	By~\cref{thm:high_dtw_to_POSS_plus_back-linkage}, \(D\) contains a path of well-linked sets \((\mathcal{S} = (S_{0}, S_{1}, \ldots, S_{\ell_1}), \mathscr{P})\) of width \(w_1\) and length \(\ell_1\).
	Moreover, \(B(S_{\ell_1})\) is well-linked to \(A(S_0)\).

	By~\cref{thm:poss-to-fence},
	the digraph \(\ToDigraph{(\mathcal{S} = (S_{0}, S_{1}, \ldots, S_{\ell_1}), \mathscr{P})}\)
	contains an \((r, q)\)-fence \((\mathcal{P}, \mathcal{Q})\)
	where \(\Start{\mathcal{P}} \allowbreak \subseteq A(S_0)\) and
	\(\End{\mathcal{P}} \subseteq B(S_{\ell_1})\).
	Hence, \(\End{\mathcal{P}}\) is well-linked to \(\Start{\mathcal{P}}\).
	Because \(\Abs{\Start{\mathcal{P}}} = r\),
	there is an \(\End{\mathcal{P}}\)-\(\Start{\mathcal{P}}\)-linkage \(\mathcal{R}\) of order \(r\) in \(D\).

	By~\cref{statement:fence-plus-back-linkage-implies-cycles}, \(D\) contains \(k\) pairwise vertex-disjoint cycles inside \((\mathcal{P}, \mathcal{Q})\) and \(\mathcal{R}\).
\end{proof}

For the case where the digraph has small directed \treewidth, we use the following result due to \cite{amiri2016erdos}.

\begin{customlem}{}{lem:akkw2}[{\cite[Lemma 4.2]{amiri2016erdos}}]
    Let $D$ be a digraph with $\dtw{D} \leq w$. For each strongly connected digraph $H$, the digraph $D$ either has $k$ disjoint copies of $H$ as a topological minor, or contains a set $T$ of at most $k \cdot (w+1)$ vertices such that $H$ is not a topological minor of $D - T$. 
\end{customlem}

Combining~\cref{lem:akkw2,statement:high-dtw-implies-cycles} yields our final main result.

We define \(\bound{statement:elementary-younger}{f}{k}  \coloneqq \bound{statement:high-dtw-implies-cycles}{t}{k}k\) and observe that \(\bound{statement:elementary-younger}{f}{k} \in \PowerTower{8}{\Polynomial{43}{k}}\).

\younger*
\begin{proof}
	Let \(t_1 = \bound{statement:high-dtw-implies-cycles}{t}{k}\).

	Let \(H = \Ck{2}\).
	If \(\dtw{D} \leq t_1 - 1\), then by~\cref{lem:akkw2} \(D\) contains \(k\) vertex-disjoint copies of \(H\) as a topological minor or \(D\) contains a set \(X \subseteq \V{D}\) of size at most \(kt_1\) such that \(D - X\) does not contain \(H\) as a topological minor.
	By choice of \(H\), both cases imply the initial statement, and we are done.

	Otherwise, \(\dtw{D} \geq t_1\).
	By~\cref{statement:high-dtw-implies-cycles}, \(D\) contains \(k\) pairwise vertex-disjoint cycles.
\end{proof}

\section{Conclusion}
\label{sec:conclusion}

While we considerably improve the function for Younger's conjecture, our bounds are likely not tight.
Indeed, the only known lower bound for our function \(\bound{statement:elementary-younger}{f}{k}\) is \(\Omega(k \log k)\), noted by Alon (unpublished, see~\cite{reed1996packing}, but possibly uses the same techniques as the lower bound for the undirected case due to~\cite{posa1965independent}).
It is thus natural to ask how far this function can be improved further.
\begin{question}
	Is there a polynomial \(p\) such that, for every digraph \(D\), there are \(k\) pairwise vertex-disjoint cycles in \(D\) or there is some \(X \subseteq \V{D}\) of size at most \(\Func{p}{k}\) such that \(D - X\) is acyclic?
\end{question}

Similarly, the framework we introduce based on \(H\)-routings, temporal digraphs and paths of well-linked sets is still not fully studied.
In particular, can we improve~\cref{theorem:one-way connected temporal digraph contains P_k routing}?
How large do temporal digraphs need to be to certainly contain some \(\Pk{k}\)-routing?
\begin{question}
		Are there polynomials \(p_1, p_2\) such that every temporal digraph \(T\) with \(\Lifetime{T} \geq p_1(k)\), \(\Abs{\V{T}} \geq p_2(k)\) and where each layer is unilateral and contains a \(\Pk{k}\)-routing?
\end{question}

While acyclic grids and fences contain paths of 1-order-linked sets and paths of well-linked sets, respectively, of roughly the same size, the bounds given by~\cref{thm:order_linked_to_acyclic_grid,thm:poss-to-fence} are very large.
The proofs provided here, however, use the pigeon-hole principle, which tends to be very wasteful.
It seems, thus, plausible that those bounds could be improved.

\begin{question}
	Are there polynomials $p_1, p_2$ such that every uniform path of $p_1(w)$-order-linked sets of  and width $p_1(w)$ and length $p_2(w)$ contains a $(w,w)$-acyclic-grid?
\end{question}

\begin{question}
	Are there polynomials $p_1, p_2$ such that every path of well-linked sets of width $p_1(w)$ and length $p_2(w)$ contains a $(w,w)$-fence?
\end{question}

Finally, we did not consider any questions relating to the computability of $H$-routings in this paper.
While it is simple to verify that deciding whether a (temporal) digraph contains an $H$-routing is in \NP, we do not have an \NP-hardness reduction.

\begin{question}
	Can we decide in polynomial time if a (temporal) digraph $D$ contains an $H$-routing?
\end{question}

Another natural question is how much the framework above helps to obtain better bounds for the Directed Grid Theorem.
At first glance, the path of well-linked sets we obtained may seem to be very close to a cylindrical grid.
Indeed, we can find both a fence and many pairwise vertex-disjoint cycles in a path of well-linked sets.
Alas, the cycles we find may be short and interact poorly with the fence.

The main issue is that, while the end of the path of well-linked sets is itself well-linked to the beginning, the linkages we obtain this way might intersect the path of well-linked sets in arbitrary ways.
While this may quickly lead to many pairwise disjoint cycles if the linkage is going backwards along the path of well-linked sets, these cycles are not \say{concentric} in the way required for the cylindrical grid.

In order to obtain the cylindrical grid, we need to get a \say{back-linkage} from the end of the path of well-linked sets back to its beginning which is internally disjoint from said path of well-linked sets.
Formally, we obtain an object which we call a \emph{cycle of well-linked sets} (of width $w$ and length $\ell$), which is a pair \((\mathcal{S}, \mathscr{P} \cup \Set{\mathcal{P}_{\ell}})\) where \((\mathcal{S}, \mathscr{P})\) is a path of well-linked sets of width $w$ and length $\ell - 1$, and \(\mathcal{P}_\ell\) is a linkage from the $B$-set of the last cluster to the $A$-set of the first cluster that is internally disjoint from \((\mathcal{S}, \mathscr{P})\).

As the proof that we can indeed obtain such a back-linkage (and hence, a cycle of well-linked sets) is not trivial and requires several additional steps, we defer it to~\cite{COSSII}.

\bibliographystyle{alphaurl}
\bibliography{literature.bib}

@article{COSSII,
  author       = {Meike Hatzel and
                  Stephan Kreutzer and
                  Marcelo Garlet Milani and
                  Irene Muzi},
  title        = {Cycles of Well-Linked Sets II: an Elementary Bound for the {Directed Grid Theorem}},
  journal      = {CoRR},
  volume       = {abs/2602.11716},
  year         = {2026},
  doi          = {10.48550/arXiv.2602.11716},
  eprinttype    = {arXiv},
  eprint       = {2602.11716},
}

@inproceedings{hkmm24cows,
  author       = {Meike Hatzel and
                  Stephan Kreutzer and
                  Marcelo Garlet Milani and
                  Irene Muzi},
  title        = {Cycles of Well-Linked Sets and an Elementary Bound for the Directed Grid Theorem},
  booktitle    = {65th {IEEE} Annual Symposium on Foundations of Computer Science, {FOCS} 2024, Chicago, IL, USA, October 27-30, 2024},
  pages        = {1--20},
  publisher    = {{IEEE}},
  year         = {2024},
  doi          = {10.1109/FOCS61266.2024.00011},
  bibsource    = {dblp computer science bibliography, https://dblp.org}
}

@Misc{zbMATH03226832,
 Author = {Harary, Frank and Norman, R. Z. and Cartwright, D.},
 Title = {Structural models: {An} introduction to the theory of directed graphs},
 Year = {1965},
 Language = {English},
 HowPublished = {New {York}-{London}-{Sydney}: {John} {Wiley} and {Sons}, {Inc}. {IX}, 415 p. (1965).},
 zbMATH = {3226832},
 Zbl = {0139.41503}
}

@inproceedings{CasteigtsHMZ20,
  author    = {Arnaud Casteigts and
               Anne{-}Sophie Himmel and
               Hendrik Molter and
               Philipp Zschoche},
  editor    = {Yixin Cao and
               Siu{-}Wing Cheng and
               Minming Li},
  title     = {Finding Temporal Paths Under Waiting Time Constraints},
  booktitle = {31st International Symposium on Algorithms and Computation, {ISAAC}
               2020, December 14-18, 2020, Hong Kong, China (Virtual Conference)},
  series    = {LIPIcs},
  volume    = {181},
  pages     = {30:1--30:18},
  publisher = {Schloss Dagstuhl - Leibniz-Zentrum f{\"{u}}r Informatik},
  year      = {2020},
  doi       = {10.4230/LIPIcs.ISAAC.2020.30},
  bibsource = {dblp computer science bibliography, https://dblp.org}
}

@phdthesis{Molter20,
  author    = {Hendrik Molter},
  title     = {Classic graph problems made temporal - a parameterized complexity
               analysis},
  school    = {Technical University of Berlin, Germany},
  year      = {2020},
  url       = {https://nbn-resolving.org/urn:nbn:de:101:1-2020120901012282017374},
  urn       = {urn:nbn:de:101:1-2020120901012282017374},
  isbn      = {978-3-7983-3172-3},
  bibsource = {dblp computer science bibliography, https://dblp.org}
}

@phdthesis{Milani2024,
  author    = {Marcelo Garlet Milani},
  title     = {Digraph Tango: from structure to algorithms and back again},
  school    = {Technical University of Berlin, Germany},
  year      = {2024},
  doi       = {10.14279/depositonce-20065}
}

@phdthesis{amiri2017,
  author    = {Saeed Amiri},
  title     = {Structural Graph Theory Meets Algorithms: Covering and Connectivity Problems in Graphs},
  school    = {Technical University of Berlin, Germany},
  year      = {2017},
  doi       = {10.14279/depositonce-6538}
}

@Article{chekuri2016polynomial,
 Author = {Chandra {Chekuri} and Julia {Chuzhoy}},
 Title = {{Polynomial bounds for the grid-minor theorem}},
 FJournal = {{Journal of the ACM}},
 Journal = {{J. ACM}},
 ISSN = {0004-5411},
 Volume = {63},
 Number = {5},
 Pages = {65},
 Note = {Id/No 40},
 Year = {2016},
 Publisher = {Association for Computing Machinery (ACM), New York, NY},
 Language = {English},
 DOI = {10.1145/2820609},
 MSC2010 = {05C83 05C75 05C78 68Q17 68Q25 68W10 68R10},
 Zbl = {1410.05186}
}

@inproceedings{kawarabayashi2015directed,
  title={The {D}irected {G}rid {T}heorem},
  author={Kawarabayashi, Ken{-}ichi and Kreutzer, Stephan},
  booktitle={Proceedings of the forty-seventh annual ACM symposium on Theory of Computing},
  pages={655--664},
  year={2015}
}

@article{reed1996packing,
  title={Packing directed circuits},
  author={Reed, Bruce and Robertson, Neil and Seymour, Paul and Thomas, Robin},
  journal={Combinatorica},
  volume={16},
  number={4},
  pages={535--554},
  year={1996},
  publisher={Springer}
}

@Article{erdosszekeres1935,
 author = {Paul {Erd\H{o}s} and George {Szekeres}},
 Title = {{A combinatorial problem in geometry}},
 FJournal = {{Compositio Mathematica}},
 Journal = {{Compos. Math.}},
 ISSN = {0010-437X},
 Volume = {2},
 Pages = {463--470},
 Year = {1935},
 Publisher = {Cambridge University Press, Cambridge; London Mathematical Society, London},
 Language = {English},
 MSC2010 = {51-XX},
 Zbl = {0012.27010}
}

@article{johnson2001directed,
  title={Directed tree-width},
  author={Johnson, Thor and Robertson, Neil and Seymour, Paul and Thomas, Robin},
  journal={Journal of Combinatorial Theory, Series B},
  volume={82},
  number={1},
  pages={138--154},
  year={2001},
  publisher={Elsevier}
}

@unpublished{JohnsonRST2001,
  title = {Addendum to {{Directed Tree-Width}}},
  author = {Johnson, Thor and Robertson, Neil and Seymour, Paul and Thomas, Robin},
  year = {2001},
  note={Available e.g.~from \url{https://thomas.math.gatech.edu/PAP/diradd.pdf}}
}

@article{adler2007directed,
  title={Directed tree-width examples},
  author={Adler, Isolde},
  journal={Journal of Combinatorial Theory, Series B},
  volume={97},
  number={5},
  pages={718--725},
  year={2007},
  publisher={Elsevier}
}

@article{kawarabayashi2022directed,
  title={The {D}irected {G}rid {T}heorem},
  author={Kawarabayashi, Ken{-}ichi and Kreutzer, Stephan},
  journal={arXiv preprint arXiv:1411.5681},
	url={https://arxiv.org/abs/1411.5681v3},
  year={2022}
}

@article{kreutzer2014width,
  title={Width-measures for directed graphs and algorithmic applications},
  author={Kreutzer, Stephan and Ordyniak, Sebastian},
  journal={Quantitative Graph Theory: Mathematical Foundations and Applications. Springer},
  year={2014}
}

@article{reed1999introducing,
  title={Introducing directed tree width},
  author={Reed, Bruce},
  journal={Electronic Notes in Discrete Mathematics},
  volume={3},
  pages={222--229},
  year={1999},
  publisher={Elsevier}
}

@article{leaf2015tree,
  title={Tree-width and planar minors},
  author={Leaf, Alexander and Seymour, Paul},
  journal={Journal of Combinatorial Theory, Series B},
  volume={111},
  pages={38--53},
  year={2015},
  publisher={Elsevier}
}

@article{amiri2016erdos,
  title={The {E}rd{\H{o}}s-{P}{\'o}sa property for directed graphs},
  author={Amiri, Saeed Akhoondian and Kawarabayashi, Ken{-}ichi and Kreutzer, Stephan and Wollan, Paul},
  journal={arXiv preprint arXiv:1603.02504},
  year={2016}
}

@inproceedings{edwards2017half,
  title={Half-Integral Linkages in Highly Connected Directed Graphs},
  author={Edwards, Katherine and Muzi, Irene and Wollan, Paul},
  booktitle={25th Annual European Symposium on Algorithms (ESA 2017)},
  volume={87},
  pages={36},
  year={2017},
  organization={Schloss Dagstuhl--Leibniz-Zentrum fuer Informatik}
}

@article{EL1974,
author = {Erdős, Paul and Lovász, László},
year = {1974},
month = {01},
pages = {},
title = {Problems and results on 3-chromatic Hypergraphs and some related questions},
volume = {10},
journal = {Coll Math Soc J Bolyai}
}

@article{SPENCER197769,
title = {Asymptotic lower bounds for {R}amsey functions},
journal = {Discrete Mathematics},
volume = {20},
pages = {69-76},
year = {1977},
issn = {0012-365X},
doi = {https://doi.org/10.1016/0012-365X(77)90044-9},
author = {Joel Spencer}
}

@article{MoserT10,
  author       = {Robin A. Moser and
                  G{\'{a}}bor Tardos},
  title        = {A constructive proof of the general {L}ov{\'{a}}sz {L}ocal {L}emma},
  journal      = {J. {ACM}},
  volume       = {57},
  number       = {2},
  pages        = {11:1--11:15},
  year         = {2010},
  doi          = {10.1145/1667053.1667060},
  bibsource    = {dblp computer science bibliography, https://dblp.org}
}

@article{CGH2013deterministic,
author = {Chandrasekaran, Karthekeyan and Goyal, Navin and Haeupler, Bernhard},
title = {Deterministic Algorithms for the {L}ov\'{a}sz {L}ocal {L}emma},
year = {2013},
issue_date = {2013},
publisher = {Society for Industrial and Applied Mathematics},
address = {USA},
volume = {42},
number = {6},
issn = {0097-5397},
doi = {10.1137/100799642},
journal = {SIAM J. Comput.},
month = {jan},
pages = {2132–2155},
numpages = {24}
}

@article{redei1934kombinatorischer,
  title={Ein kombinatorischer {S}atz},
  author={R{\'e}dei, L{\'a}szl{\'o}},
  journal={Acta Litt. Szeged},
  volume={7},
  number={39-43},
  pages={97},
  year={1934}
}

@article{BHS1980gmsdcr,
  author       = {Jon Louis Bentley and
                  Dorothea Haken and
                  James B. Saxe},
  title        = {A general method for solving divide-and-conquer recurrences},
  journal      = {{SIGACT} News},
  volume       = {12},
  number       = {3},
  pages        = {36--44},
  year         = {1980},
  doi          = {10.1145/1008861.1008865},
  bibsource    = {dblp computer science bibliography, https://dblp.org}
}

@book{CLRS2009ItAc4p94,
  author       = {Thomas H. Cormen and
                  Charles E. Leiserson and
                  Ronald L. Rivest and
                  Clifford Stein},
  title        = {Introduction to Algorithms, 3rd Edition},
  publisher    = {{MIT} Press},
  year         = {2009},
	chapter      = {4},
	page         = {94},
  url          = {http://mitpress.mit.edu/books/introduction-algorithms},
  isbn         = {978-0-262-03384-8},
  bibsource    = {dblp computer science bibliography, https://dblp.org}
}

@article{menger,
 author = {Menger, K.},
 title = {Zur allgemeinen {Kurventheorie}.},
 fjournal = {Fundamenta Mathematicae},
 journal = {Fundam. Math.},
 issn = {0016-2736},
 volume = {10},
 pages = {96--115},
 year = {1927},
 language = {German},
 doi = {10.4064/fm-10-1-96-115},
 zbMATH = {2582908},
 JFM = {53.0561.01}
}

@Article{posa1965independent,
 Author = {Erd{\H{o}}s, Paul and P{\'o}sa, L.},
 Title = {On independent circuits contained in a graph},
 FJournal = {Canadian Journal of Mathematics},
 Journal = {Can. J. Math.},
 ISSN = {0008-414X},
 Volume = {17},
 Pages = {347--352},
 Year = {1965},
 Language = {English},
 DOI = {10.4153/CJM-1965-035-8},
 zbMATH = {3211580},
 Zbl = {0129.39904}
}

@inproceedings{younger1973,
  author = {D. H. Younger},
	title = {Graphs with interlinked directed circuits},
	year = 1973,
	booktitle = {Proceedings of the Midwest Symposium on Circuit Theory},
	volume = {2},
	pages = {XVI 2.1 - XVI 2.7}
}

@article{ep1962,
 author = {Erd{\H{o}}s, P{\'a}l and P{\'o}sa, L.},
 title = {On the maximal number of disjoint circuits of a graph},
 fjournal = {Publicationes Mathematicae Debrecen},
 journal = {Publ. Math. Debr.},
 issn = {0033-3883},
 volume = {9},
 pages = {3--12},
 year = {1962},
 language = {English},
 zbMATH = {3215864},
 Zbl = {0133.16701}
}

@article{RaymondT17,
  author       = {Jean{-}Florent Raymond and
                  Dimitrios M. Thilikos},
  title        = {Recent techniques and results on the {Erd{\H{o}}s-P{\'{o}}sa} property},
  journal      = {Discret. Appl. Math.},
  volume       = {231},
  pages        = {25--43},
  year         = {2017},
  doi          = {10.1016/J.DAM.2016.12.025},
  bibsource    = {dblp computer science bibliography, https://dblp.org}
}

@misc{raymondEPlist,
	author = {Jean{-}Florent Raymond},
	title  = {Dynamic {Erdős-Pósa} listing},
	url = {https://perso.ens-lyon.fr/jean-florent.raymond/Erd\%C5\%91s-P\%C3\%B3sa/},
	note = {Accessed on 2025-08-19.}
}

@inproceedings{BatenburgHJR19,
  author       = {Wouter Cames van Batenburg and
                  Tony Huynh and
                  Gwena{\"{e}}l Joret and
                  Jean{-}Florent Raymond},
  editor       = {Timothy M. Chan},
  title        = {A tight {Erd{\H{o}}s-P{\'{o}}sa} function for planar minors},
  booktitle    = {Proceedings of the Thirtieth Annual {ACM-SIAM} Symposium on Discrete
                  Algorithms, {SODA} 2019, San Diego, California, USA, January 6-9,
                  2019},
  pages        = {1485--1500},
  publisher    = {{SIAM}},
  year         = {2019},
  doi          = {10.1137/1.9781611975482.90},
  bibsource    = {dblp computer science bibliography, https://dblp.org}
}

@inproceedings{AhnGHK25,
  author       = {Jungho Ahn and
                  J. Pascal Gollin and
                  Tony Huynh and
                  O{-}joung Kwon},
  editor       = {Yossi Azar and
                  Debmalya Panigrahi},
  title        = {A coarse {Erd{\H{o}}s-P{\'{o}}sa} theorem},
  booktitle    = {Proceedings of the 2025 Annual {ACM-SIAM} Symposium on Discrete Algorithms,
                  {SODA} 2025, New Orleans, LA, USA, January 12-15, 2025},
  pages        = {3363--3381},
  publisher    = {{SIAM}},
  year         = {2025},
  url          = {https://doi.org/10.1137/1.9781611978322.109},
  doi          = {10.1137/1.9781611978322.109},
  bibsource    = {dblp computer science bibliography, https://dblp.org}
}

@article{MoussetNSW17,
  author       = {Frank Mousset and
                  Andreas Noever and
                  Nemanja Skoric and
                  Felix Weissenberger},
  title        = {A tight {Erd{\H{o}}s-P{\'{o}}sa} function for long cycles},
  journal      = {J. Comb. Theory {B}},
  volume       = {125},
  pages        = {21--32},
  year         = {2017},
  url          = {https://doi.org/10.1016/j.jctb.2017.01.004},
  doi          = {10.1016/J.JCTB.2017.01.004},
  bibsource    = {dblp computer science bibliography, https://dblp.org}
}

\end{document}